\documentclass[11pt]{article}
\usepackage[utf8]{inputenc}
\usepackage{amsmath,amsfonts,amssymb}
\usepackage{amsthm}
\usepackage{latexsym}
\usepackage[noadjust]{cite}
\usepackage{graphicx}
\usepackage{xcolor}
\usepackage{dirtytalk}
\usepackage{mathtools}
\usepackage[margin=1in]{geometry}
\usepackage{thm-restate}
\usepackage{enumitem}
\usepackage[linesnumbered,ruled]{algorithm2e}
\usepackage[colorlinks=true, allcolors=blue]{hyperref}
\usepackage[nameinlink,capitalise]{cleveref}
\hypersetup{
    citecolor={violet}
}
\usepackage{tikz}
\usetikzlibrary{decorations.pathreplacing}

\newtheorem{theorem}{Theorem}[section]
\newtheorem{corollary}[theorem]{Corollary}
\newtheorem{lemma}[theorem]{Lemma}

\newtheorem{proposition}[theorem]{Proposition}
\newtheorem{claim}[theorem]{Claim}
\theoremstyle{definition}
\newtheorem{definition}{Definition}[section]
\newtheorem{remark}[theorem]{Remark}

\newcommand{\E}{\mathbb{E}}
\newcommand{\R}{\mathbb{R}}
\newcommand{\Z}{\mathbb{Z}}

\newcommand{\calW}{\mathcal{W}}

\newcommand{\calL}{\mathcal{L}}
\newcommand{\calS}{\mathcal{S}}

\newcommand{\Def}{\text{Def}}

\newcommand{\zo}{\{0, 1\}}

\newcommand{\Contract}{\mathrm{Contract}}

\newcommand{\eps}{\epsilon}

\newcommand{\wt}{\mathrm{wt}}

\newcommand{\Supp}{\mathrm{Supp}}

\newcommand{\polylog}{\mathrm{polylog}}

\newcommand{\supp}{\mathrm{supp}}

\newcommand{\CVNRD}{\textrm{CVNRD}}

\newcommand{\BACNRD}{\textrm{BACNRD}}
\newcommand{\BADNRD}{\textrm{BADNRD}}
\newcommand{\NRD}{\textrm{NRD}}
\newcommand{\CL}{\mathrm{CL}}

\newcommand{\WS}{\mathrm{WS}}
\newcommand{\US}{\mathrm{US}}
\newcommand{\Class}{\mathrm{Class}}
\newcommand{\Ceven}{C^{\mathrm{even}}}
\newcommand{\Codd}{C^{\mathrm{odd}}}

\newif\ifdraft

\drafttrue

\title{A Unified Theory of Sparsification}
\date{\today}
\author{
    Sanjeev Khanna\thanks{Courant Institute School of Mathematics, Computing, and Data Science, New York University, New York, USA. Supported in part by NSF award CCF-2625203 and AFOSR award FA9550-25-1-0107. Email:\texttt{sanjeev.khanna@nyu.edu}.}
    \and
    Aaron Putterman\thanks{School of Engineering and Applied Sciences, Harvard University, Cambridge, Massachusetts, USA. Supported in part by Simons Investigator Awards to Madhu Sudan and Salil Vadhan, and AFOSR award FA9550-25-1-0112. Email:\texttt{aputterman@g.harvard.edu}} 
    \and   
    Madhu Sudan\thanks{School of Engineering and Applied Sciences, Harvard University, Cambridge, Massachusetts, USA. Supported in part by a Simons Investigator Award, NSF Award CCF 2152413 and AFOSR award FA9550-25-1-0112. Email:\texttt{madhu@cs.harvard.edu}} 
    }

\begin{document}

\maketitle

\begin{abstract}
We study the sparsifiability of \emph{real-valued codes}, a unifying abstraction that generalizes both combinatorial and continuous notions of sparsification, including spectral sparsification. In our setting, a code $C \subseteq \mathbb{R}_{\geq 0}^m$ is simply a collection of nonnegative real-valued vectors, and for a parameter $\epsilon > 0$, a \emph{$(1 \pm \epsilon)$-sparsifier} of $C$ is a subset $T \subseteq [m]$, together with weights $w \in \R_{\geq 0}^T$, such that, for every $c \in C$, $\sum_{i \in T} w_i c_i \in (1 \pm \epsilon)\sum_{i=1}^m c_i$. When $C \subseteq \{0,1\}^m$, this specializes to code sparsification, and hence captures CSP sparsification, as studied by Khanna--Putterman--Sudan (SODA 2024, STOC 2025) and Brakensiek--Guruswami (STOC 2025). Similarly, for a graph $G=(V,E)$, if one defines $C=\{c^{(x)}:x\in\mathbb R^V\}\subseteq\mathbb R_{\geq 0}^E$ by $c^{(x)}_{(u,v)}=(x_u-x_v)^2$, then sparsifying $C$ is exactly spectral graph sparsification, as studied by Spielman--Teng (SICOMP 2011).

Although the techniques driving combinatorial and continuous sparsification have traditionally been largely disjoint, our main result is a single structural theorem governing the sparsifiability of arbitrary real-valued codes $C\subseteq\mathbb{R}_{\geq 0}^m$. The central parameter is \emph{continuous-valued non-redundancy} ($\mathrm{CVNRD}$), a real-valued analogue of non-redundancy that captures the largest approximately block-diagonal obstruction contained in $C$. Our theorem gives sparsifiers of size nearly-linear in $\mathrm{CVNRD}$, and shows that $\mathrm{CVNRD}$ is also a lower-bound obstruction for the broad class of coordinate-wise unbiased randomized sparsification schemes. In the course of proving this result, we rely on new techniques; in particular, we completely forego the use of Gilmer's entropy method, matrix Chernoff concentration, or Talagrand's theory of chaining, and instead reveal the Sauer-Shelah lemma, its continuous relatives, and several extensions that we introduce, as the common combinatorial core for our sparsification results. 
As consequences of our main theorem, we obtain:
(1) the first near-linear sized sparisfiers for higher powers of graph and hypergraph spectra, (2) nearly tight bounds for sparsifying sums of bounded submodular functions, (3) tight structural characterizations of sparsifiability of valued CSPs and (4) A new proof bypassing Gilmer's entropy method, of Brakensiek and Guruswami's theorem (STOC 2025)  that \emph{non-redundancy} characterizes sparsification of $\{0,1\}$-valued codes.

\end{abstract}

\pagenumbering{gobble}

\pagebreak

\tableofcontents

\pagebreak
\pagenumbering{arabic}

\section{Introduction}

Given a code $C \subseteq \R_{\geq 0}^m$ and an accuracy parameter $\eps>0$, the problem of \emph{real-valued code sparsification} asks for a \emph{small reweighted subset} of coordinates that preserves the total weight of every codeword in $C$ up to a $(1\pm\eps)$ factor. In this work, we give a structural characterization of the sparsifiability of arbitrary real-valued codes.

\subsection{Background}

Sparsification is the task of selecting and possibly reweighting a small substructure of a large object so that it preserves a prescribed family of measurements on the original object.
A foundational example is the cut sparsification theorem of Bencz\'ur and Karger~\cite{BK96}, which shows that every weighted graph admits a sparse reweighted subgraph that approximately preserves the value of every cut.
This theorem has become a cornerstone of graph algorithms and combinatorial sparsification.
Since this seminal work, a central question has been to understand which objects are sparsifiable, and which measurements can be preserved by a small reweighted substructure. A sequence of works by Spielman and Teng~\cite{ST11}, Spielman and Srivastava~\cite{spielman2008graph}, and Batson, Spielman, and Srivastava~\cite{BSS09} extended the cut-sparsification paradigm to \emph{spectral sparsification}, where one preserves not only all cuts but the entire Laplacian quadratic form of the graph.

Subsequent work has pushed sparsification well beyond graphs. For hypergraphs, results of Kogan and Krauthgamer~\cite{KK15}, Chen, Khanna, and Nagda~\cite{CKN20}, and Jambulapati, Lee, Liu, and Sidford~\cite{Lee23, JLS22, JLLS23} show that both cut and spectral notions admit sparsifiers with only $\widetilde{O}(|V|/\eps^2)$ hyperedges.\footnote{Throughout this paper, we use $\widetilde{O}(\cdot)$ to hide factors of $\mathrm{polylog}(\cdot)$.}
Related notions of sparsification have also been studied for constraint satisfaction problems (CSPs)~\cite{KK15, FK17, BZ20, khanna2024code, khanna2025efficient, brakensiek2025redundancy}, with \cite{brakensiek2025redundancy} revealing surprising connections between CSP sparsification and classical questions in CSP approximation~\cite{bessiere2020chain, chen2020best, lagerkvist2020sparsification, carbonnel2022redundancy, brakensiek2025richness}. Sparsification has also been developed for families of positive semidefinite (PSD) matrices~\cite{CHS16, basu2025sparsifying, hsieh2025sparsifying, basu2026many, basu2026quantum}. Together, these results suggest that sparsification is not merely a graph-theoretic phenomenon, but a structural principle appearing across combinatorics, optimization, and algebra. 

\medskip
\paragraph{From Graphs to Codes.}
A more recent viewpoint recasts graph and hypergraph sparsification as special cases of sparsifying \emph{codes}. Here a code $C\subseteq\zo^m$ is simply a collection of $\{0,1\}$-valued vectors; we do not assume that it is linear unless explicitly stated. Equivalently, one may view $C$ as an $m\times |C|$ matrix whose columns are the \emph{codewords} and whose rows are the \emph{coordinates}. For an accuracy parameter $\eps\in(0,1)$, a $(1\pm\eps)$-sparsifier of $C$ is a small set of coordinates $T\subseteq[m]$, together with weights $w_i\ge 0$ for $i\in T$, such that for every codeword $c\in C$,
\[
    \sum_{i\in T} w_i c_i
    \in (1\pm\eps)\cdot \sum_{i\in[m]} c_i .
\]
Thus, sparsifying a code means preserving the Hamming weight of every codeword using only a small reweighted subset of coordinates.

This abstraction already contains classical graph cut sparsification. Given a graph $G=(V,E)$, its \emph{cut code} $C_{\mathrm{cut}}\subseteq\{0,1\}^E$ has one codeword for each cut $S\subseteq V$, namely the indicator vector of the edges crossing $S$. Preserving the weight of every codeword in this code is exactly the same as preserving the value of every cut. Khanna, Putterman, and Sudan~\cite{khanna2024code, khanna2025efficient} initiated a systematic study of code sparsification, focusing on \emph{linear} codes over $\mathbb{F}_2$.\footnote{These works also treat codes over general fields $\mathbb{F}_q$ and Abelian groups, but the weight of a codeword is its Hamming weight: the symbol $0$ contributes weight $0$, while every nonzero symbol contributes weight $1$. Thus the resulting sparsification problem remains effectively $\zo$-valued.}
They proved that every such code admits a $(1\pm\eps)$-sparsifier using a number of coordinates nearly linear in its dimension. Since the cut code of a graph is linear and has dimension at most $|V|-1$, their theorem recovers the Bencz\'ur--Karger sparsification theorem as a corollary.

\smallskip

Brakensiek and Guruswami~\cite{brakensiek2025redundancy} extended this theory from linear codes to arbitrary, potentially non-linear, \emph{binary} codes. Their characterization is governed by a structural parameter called \emph{non-redundancy} ($\NRD$).

\begin{definition}
    For a code $C \subseteq \zo^m$, $\NRD(C)$ is the largest size of a set $I \subseteq [m]$ such that for each $i \in I$, there is a codeword $c \in C$ with $c_i = 1$ and $c_j = 0$ for all $j \in I \setminus \{ i \}$.
\end{definition}

Equivalently, $\NRD(C)$ is the size of the largest diagonal submatrix that appears in the matrix representation of $C$, up to permuting rows and columns. It measures how many coordinates can be forced to behave independently across codewords. Brakensiek and Guruswami proved that every code $C\subseteq\{0,1\}^m$ admits a $(1\pm\eps)$-sparsifier with at most $\widetilde{O}(\NRD(C)/\eps^2)$ coordinates. This is essentially tight: for a linear code, $\NRD(C)$ equals its dimension, recovering the bounds of~\cite{khanna2024code, khanna2025efficient}; for graph cut codes, $\NRD$ corresponds to the size of a spanning forest, and is therefore at most $|V|-1$.

The tightness is also visible from the simplest obstruction: a diagonal code cannot be sparsified. Indeed, deleting any coordinate of an identity matrix turns the corresponding codeword from weight $1$ to weight $0$. Thus, for $\zo$-valued codes, non-redundancy captures the fundamental obstruction to sparsification up to the usual logarithmic and $\eps$-dependent factors. In particular, \cite{brakensiek2025redundancy} shows that
\[
\widetilde{O}(\NRD(C) / \eps^2)
    \geq
\max_{T \subseteq [m]}
\{\text{minimum sparsifier size of } C|_T \}
    \geq
\NRD(C),
\]
where $C|_T = \{c|_T : c\in C\}\subseteq\zo^T$.

\paragraph{Beyond Boolean Codes.}

Motivated by the sharp characterization of $\zo$-valued code sparsification, we now turn to the central problem of this paper: \emph{real-valued code sparsification}. A \emph{real-valued code} is an arbitrary set $C\subseteq\R_{\geq 0}^m$, and for an accuracy parameter $\eps>0$, a \emph{$(1\pm\eps)$-sparsifier} of $C$ is a small reweighted subset of coordinates $T\subseteq[m]$, with weights $w_i\ge 0$ for $i\in T$, such that for every $c\in C$,
\[
    \sum_{i\in T} w_i c_i
    \in
    (1\pm\eps)\cdot \sum_{i\in[m]} c_i .
\]
Thus, real-valued code sparsification strictly generalizes the Boolean setting. The challenge is to understand which structural features of a real-valued code govern the minimum possible size of such a sparsifier.

Real-valued code sparsification captures continuous sparsification problems whose objective values depend on magnitudes, not merely on supports. For example, fix a graph $G=(V,E)$ and a parameter $p\ge 1$. Define a code $C_G\subseteq\R_{\geq 0}^E$ with one codeword for each vector $x\in\R^V$, where the codeword $c^{(x)}$ is given by
\[
    c^{(x)}_{(u,v)} = |x_u-x_v|^p .
\]
The weight of this codeword is exactly the \emph{$p$-energy} of $x$ on the graph:
\[
    \sum_{e\in E} c^{(x)}_e
    =
    \sum_{e=(u,v)\in E} |x_u-x_v|^p .
\]
A sparsifier of $C_G$ is therefore a subset of edges $E'\subseteq E$, together with weights $w:E'\to\R_{\geq 0}$, such that for every $x\in\R^V$,
\[
    \sum_{e\in E'} w(e)c^{(x)}_e
    \in
    (1\pm\eps)\cdot \sum_{e\in E} c^{(x)}_e .
\]
Thus, preserving the \emph{$p$-energy} of every vector $x$ is exactly the same as preserving the weight of every codeword in $C_G$.

When $p=2$, this is precisely \emph{spectral graph sparsification} as a \emph{spectral sparsifier} is a small weighted edge set $E'\subseteq E$ such that, for every vector $x\in\R^V$,
\[
    \sum_{e=(u,v)\in E'} w(e)\cdot (x_u-x_v)^2
    \in
    (1\pm\eps)\cdot
    \sum_{e=(u,v)\in E} (x_u-x_v)^2 .
\]
In this sense, spectral sparsification is in fact  a \emph{special case} of real-valued code sparsification.

More generally, suppose we are given nonnegative functions
$f_1,\dots,f_m:\R^n\to\R_{\geq 0}$. The problem of \emph{sparsifying a sum of functions} asks for a small set $T\subseteq[m]$ and weights $w:T\to\R_{\geq 0}$ such that, for every $x\in\R^n$,
\[
    \sum_{i\in T} w(i) f_i(x)
    \in
    (1\pm\eps)\cdot \sum_{i=1}^m f_i(x).
\]
This is precisely \emph{real-valued code sparsification} applied to the code
\[
    C = \{c^{(x)} : x\in\R^n\}\subseteq\R_{\geq 0}^m,
    \qquad
    c^{(x)}_i = f_i(x).
\]
This formulation encompasses \emph{spectral hypergraph sparsification}~\cite{SY19, KKTY21a, KKTY21b,JLS22, Lee23}, \emph{sparsification of sums of norms}~\cite{JLLS23}, \emph{sparsification of submodular functions}~\cite{KK23, kudla2023sparsification, KZ23, KPS24b}, \emph{sparsification of discrete codes and CSPs}~\cite{khanna2024code, khanna2025efficient, brakensiek2025redundancy}, and \emph{sparsification of hypergraphs with general splitting functions}~\cite{veldt2020minimizing, VBK21, VBK22}.

\paragraph{Barriers to Real-Valued Sparsification.}

As observed in~\cite{brakensiek2025redundancy}, real-valued codes have obstructions to sparsification that do not arise for $\zo$-valued codes. Consider the \emph{full code} $C=\{1,2\}^m$. For $\eps=1/8$, every sparsifier of $C$ must retain $\Omega(m)$ coordinates. Indeed, suppose a sparsifier only looks at a set $S\subseteq[m]$ with $|S|<m/2$. Then the codewords $c=1^m$ and $c'$ defined by $c'_i=1$ for $i\in S$ and $c'_i=2$ for $i\notin S$ are indistinguishable to the sparsifier, since they agree on every retained coordinate. However,
\[
    \sum_i c_i = m
    \qquad\text{whereas}\qquad
    \sum_i c'_i \geq 3m/2,
\]
so no single reported value can approximate both within a $(1\pm 1/8)$ factor. This obstruction cannot be detected from the \emph{Boolean support pattern} of the code: every codeword in $\{1,2\}^m$ has full support, so replacing each nonzero coordinate by $1$ collapses the entire code to the single word $1^m$.

This example reveals a fundamental difference between Boolean and real-valued sparsification: in real-valued codes, a \emph{witness to unsparsifiability} may need to contain exponentially many codewords. Indeed, for any finite subcode $C'\subseteq\{1,2\}^m$, a standard sampling-and-union-bound argument gives a sparsifier with roughly $\log |C'|/\eps^2$ coordinates. Consequently, any finite subcode that witnesses the $\Omega(m)$ lower bound for the full code $C=\{1,2\}^m$ must have size $2^{\Omega(m)}$.

This behavior sharply contrasts with the $\zo$-valued setting. There, \cite{brakensiek2025redundancy} shows that every obstruction to sparsification has a concise \emph{diagonal witness}: one can collect a small set of coordinates and codewords that form a diagonal matrix. For real-valued codes, such a small witness cannot exist in general. Thus, the right goal is not to find a short list of codewords witnessing every obstruction, but rather to find a concise structural language for reasoning about these exponentially large witnesses.

A second barrier is that previous works on code sparsification~\cite{khanna2024code, khanna2025efficient, brakensiek2025redundancy} operate in \emph{discrete} settings with finitely many codewords. Their sparsification theorems rely on decomposition arguments: one peels off a small set of coordinates until the remaining code has sufficiently few relevant patterns, after which Chernoff bounds and union bounds show that random sampling succeeds. In the setting of \emph{real-valued code sparsification}, however, the code may be infinite. This already occurs in \emph{spectral graph sparsification}, where the codewords are indexed by all vectors $x\in\R^V$. Consequently, the sparsifiability of real-valued codes cannot be understood by counting codewords alone.

Together, these barriers motivate the central question of this work:
\begin{center}
    \emph{What governs the sparsifiability of real-valued codes?}
\end{center}

Our main contributions are (1) a structural characterization of these obstructions, and hence of the sparsifiability of arbitrary real-valued codes and (2) applications of this characterization to get new nearly-tight results in sparsification while also giving new proofs of some old sparsification results.

\subsection{Our Contributions}

In this subsection we first define a general structural parameter that we associate with  continuous-valued codes (\cref{def:CVNRDparameterIntro}) and then state a theorem showing that this parameter upper bounds sparsification  with nearly tight lower bounds for the common approach to sparsification (\cref{thm:CVNRDClassifyIntro}). We then state some applications in \cref{ssec:intro-app-spectral}. Then in \cref{ssec:intro-stronger-def} we specialize this definition to codes over alphabets of ``bounded aspect ratio'' (where the smallest and largest non-zero values in the code are bounded by a constant factor, see \cref{def:continuousRVNRDintro}). We state our theorem which gives nearly-tight upper and lower bounds on sparsification in terms of this parameter (\cref{thm:continuousRVNRDintro}). We then state  applications of this tighter characterization in \cref{sssec:intro-app-bounded}.

Before stating our results, we first make precise the sparsification quantities that we characterize. Recall that our goal is to understand the sparsifiability of arbitrary codes $C\subseteq\R_{\geq 0}^m$. Throughout the paper, we assume without loss of generality that $C$ is closed under positive scaling: if $c\in C$ and $\lambda>0$, then $\lambda c\in C$. This assumption does not change the sparsification problem, since any fixed set of reweighted coordinates that preserves the weight of $c$ also preserves the weight of $\lambda c$.

Our main notion is the following \emph{unweighted} version of sparsifiability. For $S\subseteq[m]$, write
\[
    C|_S=\{c|_S:c\in C\}.
\]

\begin{definition}\label{def:unweighted-sparsify-intro}
For $C\subseteq\R_{\geq 0}^m$ and $\eps>0$, the \emph{unweighted sparsifiability} of $C$ with parameter $\eps$, denoted $\US(C,\eps)$, is
\[
    \US(C,\eps)
    =
    \max_{S\subseteq[m]}
    \min
    \left\{
        |T| :
        T\subseteq S,\ 
        \exists\, w:T\to\R_{\geq 0}
        \text{ such that }
        \forall c\in C|_S,\ 
        \sum_{i\in T} w(i)c_i
        \in
        (1\pm\eps)\sum_{i\in S} c_i
    \right\}.
\]
\end{definition}

Thus, $\US(C,\eps)$ measures the worst-case sparsifier size over all coordinate restrictions of $C$. This is the standard robustness requirement in sparsification: a code should be considered sparsifiable only if none of its coordinate projections hides a much harder sparsification instance.

The adjective ``unweighted'' distinguishes this setting from the corresponding \emph{weighted} version, where the input coordinates may already come with prescribed weights.

\begin{definition}\label{def:weighted-sparsify-intro}
For $C\subseteq\R_{\geq 0}^m$, $\eps>0$, and an input weight function $w':[m]\to\R_{\geq 0}$, a \emph{$(1\pm\eps)$-sparsifier} of the weighted code $(C,w')$ is a set $T\subseteq[m]$, together with weights $w:T\to\R_{\geq 0}$, such that for every codeword $c\in C$,
\[
    \langle c|_T,w\rangle
    \in
    (1\pm\eps)\cdot \langle c,w'\rangle .
\]

The \emph{weighted sparsifiability} of $C$ with parameter $\eps$, denoted $\WS(C,\eps)$, is
\[
    \WS(C,\eps)
    =
    \max_{w':[m]\to\R_{\geq 0}}
    \min
    \left\{
        |T| :
        T\subseteq[m],\
        \exists\, w:T\to\R_{\geq 0}
        \text{ such that } (T,w)
        \text{ is a } (1\pm\eps)\text{-sparsifier of }(C,w')
    \right\}.
\]
\end{definition}

We will also use a randomized variant, which captures the sampling-and-reweighting schemes that underlie most randomized sparsification results.

\begin{definition}\label{def:random-sparsify-intro}
For $C\subseteq\R_{\geq 0}^m$ and $\eps>0$, a \emph{$(1\pm\eps)$ random sparsifier} of $C$ is a distribution $\Delta$ over weight functions $w:[m]\to\R_{\geq 0}$ such that:
\begin{enumerate}
    \item for every coordinate $i\in[m]$, the sampling scheme is \emph{unbiased}:
    \[
        \E_{w\sim\Delta} w_i = 1;
    \]
    \item with probability at least $1-1/m$ over $w\sim\Delta$, the sampled weight function $w$ is a $(1\pm\eps)$-sparsifier of $C$, i.e.,
    \[
        \sum_{i=1}^m w(i)c_i
        \in
        (1\pm\eps)\sum_{i=1}^m c_i
        \qquad
        \text{for every } c\in C.
    \]
\end{enumerate}
The size of the random sparsifier is
\[
    \E_{w\sim\Delta}|\Supp(w)|,
\]
the expected number of coordinates receiving nonzero weight.

The \emph{random sparsifiability} of $C$ with parameter $\eps$ is
\[
    \mathrm{RS}(C,\eps)
    =
    \max_{S\subseteq[m]}
    \min_{\Delta:\ (1\pm\eps)\text{ random sparsifier of }C|_S}
    \E_{w\sim\Delta}|\Supp(w)|.
\]
\end{definition}

The unbiasedness condition says that every coordinate has expected weight $1$ under the sampling scheme. This is the usual normalization in randomized sparsification: coordinates may be sampled with different probabilities and reweighted differently, but their expected contribution remains unchanged. The vast majority of sampling-based sparsification results fit into this framework. A notable exception is the deterministic spectral sparsifier of Batson, Spielman, and Srivastava~\cite{BSS09}, and results derived from it.

It is important to emphasize that a random sparsifier is not a uniform sampling scheme. Even in the common independent-sampling paradigm, the sampling probabilities may be highly nonuniform and chosen using structural information about the instance: edge strengths in cut sparsification, effective resistances in spectral graph sparsification, leverage scores in matrix sparsification, or sensitivities in more general settings. In such schemes, coordinate $i$ is typically sampled with probability $p_i$ and assigned weight $1/p_i$ if retained, which satisfies the unbiasedness condition
\[
    \E_{w\sim\Delta} w_i = 1.
\]
Our definition is broader still: the distribution $\Delta$ may be any distribution over nonnegative weight functions $w:[m]\to\R_{\geq 0}$. Thus, the sampled coordinates need not be independent, and the final weights may be chosen as an arbitrary function of the sampled set and of the code $C$, provided only that the coordinate-wise unbiasedness condition holds.

\subsubsection{A Global Classification}

We now state our global classification theorem. The central parameter is a real-valued analogue of non-redundancy, which we call \emph{continuous-valued non-redundancy} and denote by $\CVNRD$. Informally, $\CVNRD$ measures the largest block-diagonal obstruction inside a real-valued code: each block behaves like a complete code over two well-separated values, while the interaction between different blocks is small.

For ease of notation, for disjoint sets $A_1,\dots,A_p\subseteq[m]$, we write
\[
    A_{\neq i}=\bigcup_{j\neq i} A_j,
    \qquad
    \wt(c|_{A_{\neq i}})=\sum_{\ell\in A_{\neq i}} c_\ell .
\]

That is, $\wt(c|_{A_{\neq i}})$ is the total weight of $c$ on the selected coordinate sets (viewed as blocks) other than $A_i$.

\begin{definition}\label{def:CVNRDparameterIntro}
    Let $C \subseteq \R_{\geq 0}^m$, and let $\eps,\chi,\rho\in(0,1]$. The \emph{continuous-valued non-redundancy} of $C$ with parameters $(\eps,\chi,\rho)$, denoted $\CVNRD(C,\eps,\chi,\rho)$, is the maximum integer $\ell$ for which there exist subcodes $C_1,\dots,C_p\subseteq C$ and disjoint sets $A_1,\dots,A_p\subseteq[m]$ satisfying the following conditions:
    \begin{enumerate}
        \item \emph{(Witness size)} 
        \[
            \ell=\sum_{i=1}^p |A_i|.
        \]

        \item \emph{(Singleton edge case)}
        For each $i\in[p]$, if $|A_i|=1$, then $C_i=\{c\}$ consists of a single codeword, and if $A_i=\{j\}$, we set
        \[
            b_2^{(i)}=c_j>0.
        \]

        \item \emph{(Complete shattering condition)}
        For each $i\in[p]$ with $|A_i|\ge2$, there exist values
        $b_1^{(i)},b_2^{(i)}\in\R_{\geq 0}$ with
        \[
            b_1^{(i)} < (1-\eps)b_2^{(i)}
        \]
        such that for every subset $B\subseteq A_i$, there is a codeword $c\in C_i$ satisfying
        \[
            c_j\in b_1^{(i)}\cdot[1,1+\eps\chi]
            \quad\text{for all } j\in B,
        \]
        and
        \[
            c_j\in b_2^{(i)}\cdot[1,1+\eps\chi]
            \quad\text{for all } j\in A_i\setminus B.
        \]

        \item \emph{(Small cumulative off-block weight)}
        For every $i\in[p]$ and every $c\in C_i$,
        \[
            \wt(c|_{A_{\neq i}})
            \le
            \rho\cdot b_2^{(i)}\cdot |A_i|.
        \]

        \item \emph{(Small entrywise off-block weight)}
        For every $i\in[p]$ and every $c\in C_i$,
        \[
            \max_{j\in A_{\neq i}} c_j
            \le
            \rho\cdot b_2^{(i)}.
        \]
    \end{enumerate}
\end{definition}

Let us unpack this definition. A $\CVNRD$ witness should be thought of as an approximately \emph{block-diagonal} family of subcodes. Within each selected coordinate set $A_i$, the subcode $C_i$ behaves like a \emph{complete block}: for every subset $B\subseteq A_i$, some codeword realizes the ``low'' value on $B$ and the ``high'' value on $A_i\setminus B$. This is the real-valued analogue of a Boolean block containing all of $\{0,1\}^{A_i}$.

The parameters $\eps$ and $\chi$ control two different aspects of this completeness condition. The parameter $\eps$ controls the required separation between the low and high levels: the condition
\[
    b_1^{(i)} < (1-\eps)b_2^{(i)}
\]
ensures that the two levels are sufficiently distinct. For example, if the two intended levels are $1$ and $2$, then this separation condition holds for any $\eps<1/2$, but fails for $\eps\ge 1/2$.

The parameter $\chi$ controls how tightly the realized entries must cluster around these two levels. For instance, suppose the high level is $b_2^{(i)}=2$. Then entries equal to $2$ and $2.02$ can both be treated as realizations of the same high level only if
\[
    2.02 \in 2\cdot[1,1+\eps\chi],
\]
equivalently, only if $\eps\chi\ge 0.01$. Thus, smaller $\chi$ makes the notion of a complete block more rigid. Finally, $\rho$ controls how block-diagonal the witness is: smaller $\rho$ forces the off-block entries to be small, both in total weight and coordinate by coordinate, relative to the scale $b_2^{(i)}$ of the block.

In typical applications, $\eps$ is tied to the target sparsification accuracy, while $\chi$ and $\rho$ are chosen to be $\mathrm{poly}(\eps,1/\log m)$. With this notation, our main theorem is the following.

\begin{theorem}\label{thm:CVNRDClassifyIntro}
    Let $C \subseteq \R_{\geq 0}^{m}$ and $\eps\in(0,1)$ be given, and set $\eps'=\frac{\eps^4}{10^{13}\log^8(m)}$.
    Then
    \[
    \US(C,\eps)
    \leq
    \min_{\chi\in(0,1],\,\rho\in(0,1]}
    O\left(
        \frac{\CVNRD(C,\eps',\chi,\rho)}{\chi^6\rho}
        \cdot
        \mathrm{poly}(\log(m),\eps^{-1})
    \right).
    \]
    Conversely,
    \[
    \mathrm{RS}(C,\eps)
    \geq
    \Omega\left(
        \eps\cdot
        \CVNRD\left(
            C,
            10\eps,
            \chi=\frac{1}{100\log^2(m)},
            \rho=\frac{\eps}{100\log^2(m)}
        \right)
    \right).
    \]
\end{theorem}

In words, \cref{thm:CVNRDClassifyIntro} shows that the sparsifiability of an arbitrary real-valued code is controlled by its continuous-valued non-redundancy. The upper bound says that every code can be sparsified to within polynomial factors in $\eps^{-1}$, $\chi^{-1}$, $\rho^{-1}$, and $\log m$ of its $\CVNRD$. The lower bound says that $\CVNRD$ is also an  obstruction: \emph{any unbiased distribution} over reweighted coordinate subsets that sparsifies the code must have expected support size proportional to $\CVNRD$, up to the stated parameters.

This is a strong lower bound against randomized sparsification algorithms. As discussed above, the term ``random'' here does not mean uniform or independent sampling; it allows arbitrary nonuniform, correlated, and adaptively weighted sampling-and-reweighting schemes, subject only to coordinate-wise unbiasedness. The remaining gap is that, for arbitrary real-valued codes, our lower bound is proved for this broad randomized model rather than for all reweighted coordinate subsets. Equivalently, it is not known in full generality whether $\US(C,\eps)$ must always be at least $\Omega(\mathrm{RS}(C,\eps))$. We conjecture that this gap is only an artifact of the proof, and that $\CVNRD$ is a lower bound for all sparsifiers. In the bounded-aspect-ratio settings considered later (in \cref{ssec:intro-stronger-def}), this issue disappears: the analogous structural parameters give lower bounds against \emph{all} sparsifiers.

A key feature of \cref{thm:CVNRDClassifyIntro} is that it pays no factor depending on $\log |C|$. This is essential: the code $C$ may be infinite, as in spectral graph sparsification, and even finite real-valued obstructions may require exponentially many codewords. In this sense, the theorem extends the central message of~\cite{brakensiek2025redundancy}: sparsifiability is governed not by the number of codewords, but by the largest structured obstruction contained in the code.

Indeed, \cref{thm:CVNRDClassifyIntro} already essentially recovers the main theorem of~\cite{brakensiek2025redundancy} for $\zo$-valued codes, albeit with worse dependence on $\eps^{-1}$ and $\log m$. To see why, let $C\subseteq\zo^m$, and consider a largest witness
\[
    A_1,\dots,A_p\subseteq[m],
    \qquad
    C_1,\dots,C_p\subseteq C
\]
to $\CVNRD(C,\eps',\chi,\rho)$ with $\rho<1$. Since the entries of $C$ are only $0$ and $1$, the two levels in each complete block must correspond to $0$ and $1$: informally, $b_1^{(i)}\approx 0$ and $b_2^{(i)}\approx 1$. Moreover, the entrywise off-block condition gives
\[
    c_j \le \rho b_2^{(i)} < 1
    \qquad
    \text{for all } j\in A_{\neq i}
\]
and every $c\in C_i$. Since the entries are Boolean, this forces $c_j=0$ on all off-block coordinates. Thus, in the Boolean setting, every $\CVNRD$ witness is actually a truly block-diagonal collection of complete Boolean blocks.

Concretely, the resulting matrix of codewords and coordinates has the form
\[
\begin{bmatrix}
    \{0,1\}^{A_1} & 0 & 0 & \dots \\
    0 & \{0,1\}^{A_2} & 0 & \dots \\
    0 & 0 & \{0,1\}^{A_3} & \dots \\
    \vdots & \vdots & \vdots & \ddots
\end{bmatrix},
\]
where rows are grouped by $A_1,A_2,A_3,\dots$, and columns are grouped by the corresponding subcodes $C_1,C_2,C_3,\dots$. From such a matrix, one can extract an identity submatrix of size
\[
    \sum_i |A_i|=\CVNRD(C,\eps',\chi,\rho):
\]
inside each complete block $\{0,1\}^{A_i}$, simply choose the codewords that form the identity on $A_i$. Therefore, for Boolean codes,
\[
    \NRD(C)\geq \CVNRD(C,\eps',\chi,\rho).
\]
Applying \cref{thm:CVNRDClassifyIntro} then yields sparsifiers of size
\[
    O\!\left(\NRD(C)\cdot \mathrm{poly}(\eps^{-1},\log m)\right),
\]
recovering the qualitative content of~\cite{brakensiek2025redundancy}.

As a byproduct, this gives a new proof that non-redundancy characterizes sparsifiability of $\zo$-valued codes, without using Gilmer's entropy method. For clarity and sharper parameters, we also include in \cref{sec:zoManipulations} a simpler self-contained proof specialized to $\zo$-valued codes, matching the bounds of~\cite{brakensiek2025redundancy}.

Together, \cref{def:CVNRDparameterIntro} and \cref{thm:CVNRDClassifyIntro} identify the structural obstruction governing sparsifiability in arbitrary real-valued codes. We next show that this perspective also leads to new sparsification results. We first discuss applications to stronger forms of graph and hypergraph sparsification, and then turn to bounded-aspect-ratio codes, which yield sharper characterizations and applications to submodular functions and valued constraint satisfaction problems.

\subsubsection{Applications to Stronger Spectral Graph and Hypergraph Sparsification}\label{ssec:intro-app-spectral}

We first illustrate the power of the global classification theorem through graph sparsification. Recall that the \emph{spectral code} of a graph $G=(V,E)$ is the code
$C_{G,2}\subseteq\R_{\geq 0}^E$ with one codeword $c^{(x)}$ for each vector $x\in\R^V$, where
\[
    c^{(x)}_{(u,v)}=(x_u-x_v)^2 .
\]
The weight of this codeword is the usual graph energy, that is, $   \sum_{e\in E} c^{(x)}_e   = \sum_{e=(u,v)\in E} (x_u-x_v)^2 $.
Thus, a sparsifier of $C_{G,2}$ is exactly a \emph{spectral graph sparsifier}: a small reweighted subset of edges preserving this quadratic form for every $x\in\R^V$.

One can therefore ask whether \cref{thm:CVNRDClassifyIntro} recovers the celebrated near-linear-size \emph{spectral graph sparsification} theorem, namely the existence of sparsifiers with $O\!\left(n\cdot \mathrm{poly}(\log n,\eps^{-1})\right)$ edges.
As we explain in the technical overview \cref{sec:techOverviewGraphs}, this is indeed the case. 
More interestingly, the $\CVNRD$ perspective avoids the usual spectral-sparsification toolkit, namely, effective resistances, leverage scores, eigenvalues, and matrix concentration, and instead reduces the existence of spectral sparsifiers to an elementary combinatorial statement about graphs.

The same viewpoint applies beyond quadratic energies. For any $t\ge 1$, define the \emph{$t$-spectral code} $C_{G,t}\subseteq\R_{\geq 0}^E$ to have one codeword $c^{(x)}$ for each vector $x\in\R^V$, with
\[
    c^{(x)}_{(u,v)} = |x_u-x_v|^t .
\]
The weight of $c^{(x)}$ is then the \emph{$t$-spectral energy} of $x$ on $G$:
\[
    \sum_{e\in E} c^{(x)}_e
    =
    \sum_{e=(u,v)\in E}|x_u-x_v|^t .
\]
Because \cref{thm:CVNRDClassifyIntro} applies to arbitrary real-valued codes, the same combinatorial argument used for $t=2$ also yields sparsifiers for all $t\ge 1$.

\begin{theorem}
    Let $G=(V,E)$ be a graph on $n$ vertices, let $\eps\in(0,1)$, and let $t\ge 1$. Then there is a subset of edges $E'\subseteq E$, together with weights $w:E'\to\R_{\geq 0}$, such that for every $x\in\R^V$,
    \[
        \sum_{e=(u,v)\in E'} w(e)\cdot |x_u-x_v|^t
        \in
        (1\pm\eps)\cdot
        \sum_{e=(u,v)\in E}|x_u-x_v|^t .
    \]
    Moreover,
    \[
        |E'| = O\!\left(n\cdot \mathrm{poly}(\log^t n,\eps^{-t})\right).
    \]
\end{theorem}

Prior to our work, the best known bound for this problem was  $\widetilde{O}\!\left(\frac{\max(n,n^{t/2})}{\eps^2}\right)$
from Jambulapati, Lee, Liu, and Sidford~\cite{JLLS23}. The main difficulty is that existing approaches rely either on matrix concentration, which is tailored to quadratic forms and therefore to the case $t=2$, or on generic chaining,
which again relies on a ``smooth'' structure in the sparsified functions. Our framework makes no such assumptions. Once the relevant $\CVNRD$ obstruction is ruled out by elementary graph combinatorics, the sparsifier follows from the general theorem.

The same methods also extend to the substantially more difficult setting of \emph{spectral hypergraph sparsification}. Given a hypergraph $H=(V,E)$, the usual spectral energy of an assignment $x\in\R^V$ is $ \sum_{e\in E}\max_{u,v\in e}(x_u-x_v)^2 $.
More generally, for $t\ge 1$, we define the \emph{$t$-spectral energy} of $H$ by $ \sum_{e\in E}\max_{u,v\in e}|x_u-x_v|^t $.
Just as in graphs, our techniques yield near-linear-size sparsifiers for this energy in hypergraphs.

\begin{theorem}
    Let $H=(V,E)$ be a hypergraph on $n$ vertices and $m$ hyperedges, let $\eps\in(0,1)$, and let $t\ge 1$. Then there is a subset of hyperedges $E'\subseteq E$, together with weights $w:E'\to\R_{\geq 0}$, such that for every $x\in\R^V$,
    \[
        \sum_{e\in E'} w(e)\cdot \max_{u,v\in e}|x_u-x_v|^t
        \in
        (1\pm\eps)\cdot
        \sum_{e\in E}\max_{u,v\in e}|x_u-x_v|^t .
    \]
    Moreover,
    \[
        |E'| = O\!\left(n\cdot \mathrm{poly}(\log^t m,\eps^{-t})\right).
    \]
\end{theorem}

This theorem makes no assumptions on the sizes of the hyperedges. Despite the difficulty of constructing near-linear-size spectral hypergraph sparsifiers in prior work~\cite{JLS22, Lee23}, once \cref{thm:CVNRDClassifyIntro} is available, the task again reduces to ruling out large $\CVNRD$ witnesses. In the case $t=2$, this gives the first proof of spectral hypergraph sparsification that does not use chaining. For $t>2$, it gives the first near-linear-size sparsifiers of this kind. We discuss the combinatorial argument underlying these applications in \cref{sec:techOverviewGraphs}.

\subsubsection{Stronger Classifications With Bounded Aspect Ratio}\label{ssec:intro-stronger-def}

In many applications, the codes we wish to sparsify are not arbitrary subsets of $\R_{\geq 0}^m$ but have more structure such as a \emph{bounded-aspect-ratio}: every nonzero entry lies in a bounded range. We capture this by considering codes $C\subseteq(\{0\}\cup[1,k])^m$, where $k$ is typically constant. In this bounded-aspect-ratio setting, our general classification theorem admits a simpler and tighter form.

\begin{definition}\label{def:continuousRVNRDintro}
    Let $C \subseteq \big(\{0\}\cup[1,k]\big)^m$. The bounded-aspect ratio continuous non-redundancy of $C$ with parameter $\eps$, denoted by
     $\mathrm{BACNRD}(C,\eps)$ is the largest integer $\ell$ for which there exist disjoint\footnote{Note that, although we emphasize disjointness here, this disjointness is implied by the $4$th condition below.} $A_1,\dots,A_p\subseteq[m]$ and subcodes $C_1,\dots,C_p\subseteq C$ such that:
    \begin{enumerate}
    \item (Witness Size) $\sum_{i=1}^p|A_i|=\ell$.
        \item (Singleton Edge Case) If $|A_i|=1$, then $(C_i)|_{A_i}\neq\{0\}$.
        \item (Complete Block Shattering) If $|A_i|\ge2$, there exists $\gamma\in [0,k]^{A_i}$ such that for every $B\subseteq A_i$ there is $c\in C_i$ with
        \[
            \text{for } b\in B:\ c_b\ge \gamma_b+\eps,\qquad
            \text{for } b\in A_i\setminus B:\ c_b\le \gamma_b-\eps.
        \]
        \item (Zeros Off-Block) For every $i\neq j$, $(C_i)|_{A_j}=0^{A_j}$.
    \end{enumerate}
\end{definition}

The witness size, singleton case, and shattering condition parallel the corresponding conditions in \cref{def:CVNRDparameterIntro}. The main simplification is the off-block condition: here the off-block entries are required to be exactly zero, rather than merely small in total weight and coordinate-wise magnitude as in $\CVNRD$.

With this cleaner obstruction in hand, we obtain a sharper characterization of sparsifiability in the bounded-aspect-ratio regime. Up to the displayed dependence on $\eps$ and polylogarithmic factors, $\mathrm{BACNRD}$ governs the optimal sparsifier size.

\begin{theorem}\label{thm:continuousRVNRDintro}
    Let $C \subseteq \left ( \{0\} \cup [1,k] \right )^m$ with $k = O(1)$. Then,
    \begin{enumerate}
        \item There exists $\eps'$ such that $\eps' = \Omega(\eps)$ and $\US(C, \eps') \geq \Omega(\eps \cdot \mathrm{BACNRD}(C, \eps))$. 
        \item For $\eps' = \frac{\eps} {256\log(m)}$, $\US(C, \eps) \leq \widetilde{O}(\mathrm{BACNRD}(C, \eps') / \eps^4)$.
    \end{enumerate}
\end{theorem}

Thus, unlike in the fully general real-valued setting, the lower bound here applies to arbitrary sparsifiers, not only to randomized sampling-and-reweighting schemes. This is the sense in which bounded aspect ratio gives a stronger classification.

When the input code is equipped with weights, there is an analogous parameter, the \emph{bounded-aspect-ratio continuous-valued chain length}, which governs \emph{weighted sparsifiability}. We defer these results to \cref{sec:chainLength}.

\subsubsection{Applications of Bounded-Aspect Ratio Sparsification}\label{sssec:intro-app-bounded}

Finally, as applications of the bounded-aspect-ratio theory, we obtain sharper bounds for sparsifying sums of submodular functions and a first general characterization of the worst-case sparsifiability of \emph{valued constraint satisfaction problems (VCSPs)}.\footnote{Equivalently, these may be viewed as hypergraphs equipped with arbitrary real-valued splitting functions.}

\paragraph{Submodular Sparsification.}

In \emph{submodular sparsification}, one is given submodular functions $f_1,\dots,f_m:2^{[n]}\to\R_{\geq 0}$. Recall that submodularity means that, for every $i\in[m]$ and all sets $A,B\subseteq[n]$,
$f_i(A)+f_i(B)\geq f_i(A\cup B)+f_i(A\cap B)$. The goal is to find an ideally small set $T\subseteq[m]$, together with weights $w:T\to\R_{\geq 0}$, such that for every $S\subseteq[n]$,
\[
    \sum_{i\in T} w(i)f_i(S)
    \in
    (1\pm\eps)\sum_{i\in[m]} f_i(S).
\]
Despite attention in numerous works~\cite{KZ23, KK23, JLLS23, KPS24b}, the best known sparsifiers for arbitrary submodular functions have size $O(n^3/\eps^2)$, due to Kenneth and Krauthgamer~\cite{KK23}. On the other hand, the best known lower bounds are $\Omega(n^2)$, and it has remained open to pin down the complexity in this range. Using our bounded-aspect-ratio framework, we essentially settle this question for \emph{bounded} submodular functions.

\begin{theorem}\label{thm:intro-submodular}
    Let $k=O(1)$ be an integer, and let $f_1,\dots,f_m:2^{[n]}\to\{0,1,\dots,k\}$ be submodular functions. Then there is a set $T\subseteq[m]$, together with weights $w:T\to\R_{\geq 0}$, such that for every $S\subseteq[n]$,
    \[
        \sum_{i\in T} w(i)f_i(S)
        \in
        (1\pm\eps)\sum_{i\in[m]} f_i(S),
    \]
    and $|T|=\widetilde{O}(n^2/\eps^4)$.
\end{theorem}

The lower bound of $\Omega(n^2)$ already holds for $\zo$-valued submodular functions, via directed cut functions; see~\cite{KK23} for further discussion. Thus, \cref{thm:intro-submodular} gives further evidence that $n^2$ is perhaps the right sparsifier size for submodular functions. Extending these techniques from bounded to arbitrary nonnegative submodular functions, potentially using the full $\CVNRD$ classification theorem, remains a promising open direction.

\paragraph{Valued CSP Sparsification}

Our results also give tight characterizations for the sparsifiability of \emph{valued CSPs}. Fix a predicate $P:\Sigma^r\to\R_{\geq 0}$ and a universe of variables $x_1,\dots,x_n\in\Sigma$. An instance $\Psi$ consists of constraints indexed by $i\in[m]$, where the $i$th constraint applies $P$ to a tuple $T_i\in\binom{[n]}{r}$ of variables and may be assigned a weight $w_i\geq0$. For an assignment $x\in\Sigma^n$, the value of the instance is
$\Psi(x)=\sum_{i\in[m]} w_i P(x|_{T_i})$. A sparsifier is a reweighted sub-instance $\widehat{\Psi}$ such that, for every $x\in\Sigma^n$,
\[
    \widehat{\Psi}(x)
    \in
    (1\pm\eps)\Psi(x),
\]
while retaining as few constraints as possible. This is exactly \emph{general hypergraph sparsification}: variables are vertices, constraints are hyperedges, the predicate $P$ is the splitting function, and the alphabet $\Sigma$ specifies the number of parts in the partition; the standard cut case corresponds to $|\Sigma|=2$.

Following prior work on CSP sparsification~\cite{KK15, FK17, BZ20, khanna2024code, khanna2025efficient, brakensiek2025redundancy, brakensiek2025tight}, we fix a predicate $P$ of constant arity over a constant-size alphabet and ask for the worst-case size of a $(1\pm\eps)$-sparsifier over all $n$-variable instances. We denote this quantity by $\mathrm{SPR}(P,n,\eps)$. While valued CSPs have appeared previously in real-valued sparsification~\cite{BZ20, brakensiek2025redundancy}, concrete structural characterizations were known only for random instances~\cite{brakensiek2025tight}.

Most prior structural results focus on Boolean predicates. For $P:\Sigma^r\to\{0,1\}$, CSP sparsification coincides with $\zo$-valued code sparsification: each assignment $x\in\Sigma^n$ maps to the codeword $c(x)\in\zo^m$ with $c(x)_i=\mathbf{1}[P(x|_{T_i})=1]$. For real-valued predicates, our framework suggests the following analogue of non-redundancy.

\begin{definition}
    Let $P:\Sigma^r\to\mathbb{R}_{\ge0}$ and consider all $n$-variable instances $\Psi$ with constraints on $T_1,\dots,T_m$. 
    Define $\mathrm{DNRD}(P,n)$ as the maximum $\ell$ such that there exist disjoint $A_1,\dots,A_p\subseteq[m]$ and assignment families $X_1,\dots,X_p\subseteq\Sigma^n$ with:
    \begin{enumerate}
    \item \emph{(Witness Size)} $\sum_{i=1}^p|A_i|=\ell$.
            \item (\emph{Singletons}) If $|A_i|=1$, then there exists $x\in X_i$ with $P(x|_{T_\ell})\neq 0$ for the unique $\ell\in A_i$.
        \item (\emph{Block-diagonal}) For every $i\neq j$, for every $x\in X_i$ and every $\ell\in A_j$, $P(x|_{T_\ell})=0$.

        \item (\emph{Completeness}) If $|A_i|\ge2$, there exist $b_0\neq b_1\in\mathrm{Im}(P)$ such that for every $v\in\{b_0,b_1\}^{A_i}$ there exists $x\in X_i$ with $P(x|_{T_\ell})=v_\ell$ for all $\ell\in A_i$.
    \end{enumerate}
\end{definition}

Indeed, by \cref{thm:continuousRVNRDintro}, we can immediately conclude that $\mathrm{DNRD}(P,n)$ is essentially \emph{equal} (up to polylogarithmic factors) to the worst-case sparsifiability of VCSPs with predicate $P$.

\begin{corollary}
        For every predicate $P:\Sigma^r\to\mathbb{R}_{\ge0}$ and every $\eps\in(0,\kappa)$ for some $\kappa=\Omega(1)$,
        \[
            \Omega\!\big(\mathrm{DNRD}(P,n)\big)\ \le\ \mathrm{SPR}(P,n,\eps)\ \le\ \widetilde{O}\!\big(\mathrm{DNRD}(P,n)/\eps^{4}\big).
        \]
\end{corollary}

For constant $\eps$, this corollary provides an exact analog of~\cite{brakensiek2025redundancy}: it reduces the study of VCSP sparsifier size to understanding the redundancy induced by the predicate $P$. 
In particular, when complete blocks are absent, VCSP sparsifiability coincides with that of the corresponding $\{0,1\}$-valued CSP (and code). 

Interestingly, as we show in \cref{sec:valuedCSPs}, these complete blocks arise even with very simple predicates. As an example, if we consider the predicate $P: \zo^2 \rightarrow \{0, 1, 2 \}$ introduced in the work of \cite{brakensiek2025tight}, where:
\[
P(00) = 0, \quad P(01) = 0, \quad P(10) = 1, \quad P(11) = 2,
\]
we in fact show that $\mathrm{DNRD}(P, n) = \Omega(n^2)$, as we are able to construct $\Omega(n)$ disjoint, complete blocks of size $\Omega(n)$. For comparison, if one looks only at the ``zero vs. non-zero'' structure of $P$, one gets the predicate $\hat{P}: \zo^2 \rightarrow \zo$ such that 
\[
\hat{P}(00) = 0, \quad \hat{P}(01) = 0, \quad \hat{P}(10) = 1, \quad \hat{P}(11) = 1.
\]
Such a predicate is known to have $\NRD(\hat{P}, n) = O(n)$! This shows that, despite potentially seeming unnatural, these notions of ``complete blocks'' are in fact inherent to the study of VCSPs. 

Finally, chain length and non-redundancy have played a central role in the study of $\{0,1\}$-valued CSPs, with deep connections to kernelization, query complexity, and the broader structural theory of VCSPs~\cite{bessiere2020chain, chen2020best, lagerkvist2020sparsification, carbonnel2022redundancy, brakensiek2025richness, cohen2006complexity, cohen2013algebraic, kolmogorov2015power, thapper2016complexity, kolmogorov2017complexity, carbonnel2018complexity}. 
We view extending these connections to valued CSPs through our real-valued analogs as a compelling direction for future research, continuing the rich line of structural work on CSPs.

In the remainder of the introduction, we highlight some of the techniques that underlie our characterization theorems. 

\subsection{Technical Overview}

\subsubsection{Building Stronger Graph Sparsifiers}\label{sec:techOverviewGraphs}

\paragraph{Setting Up the Basics}

To help build intuition and understanding of the exact implications of \cref{thm:CVNRDClassifyIntro}, we start by showing how this theorem can be used to easily derive the existence of near-linear size spectral graph sparsifiers. In order to do so, we recall the equivalent ``code'' formulation of the spectrum of a graph $G = (V, E)$. Indeed, we define $C_G \subseteq \R_{\geq 0}^E$, which has one codeword for each vector $x \in \R^V$. For such a vector $x$, the corresponding codeword $c^{(x)} \in \R_{\geq 0}^E$ is designed such that $c^{(x)}_{(u,v)} = (x_u - x_v)^2$. As we saw, the weight of a codeword $c^{(x)}$ is 
$\sum_{e \in E} c^{(x)}_{e} = \sum_{e = (u,v) \in E} (x_u - x_v)^2.$

For this code $C_G$, we can then understand its $\CVNRD$. Recalling \cref{def:CVNRDparameterIntro} and translating this into the language of $C_G$, this will be the largest collection of disjoint sets of \emph{edges} $E_1, E_2, \dots E_p \subseteq E$, along with sets of \emph{vertex potential assignments} (because these are in bijection with the codewords of $C_G$) $\Psi_1, \dots \Psi_p \subseteq \R^V$ such that:
\begin{enumerate}
        \item (Single Element Edge Case) For each $i \in [p]$, if $|E_i| = 1$, then there is a single assignment $x \in \Psi_i$, and we set $b_2^{(i)} = (x_u - x_v)^2 > 0$ for the single edge $e \in E_i$.
        \item (Complete Shattering Condition) For each $i \in [p]$, if $|E_i| \geq 2$, there exists $b_1^{(i)}, b_2^{(i)} \in \R$, $b_1^{(i)} < b_2^{(i)} (1 - \eps)$ such that for every $B \subseteq E_i$, there is an assignment $x \in \Psi_i$ such that for $e = (u,v) \in B$, $(x_u - x_v)^2 \in b_1^{(i)} \cdot[1, 1 + \eps \cdot \chi]$ and for $e = (u,v) \in E_i - B$, $(x_u - x_v)^2\in b_2^{(i)} \cdot [1, 1 + \eps \cdot \chi]$.
        \item (Small Cumulative Off-block Weight) For every $i \in [p]$, and for every $x \in \Psi_i$, \[
        \sum_{e = (u,v) \in E_{\neq i}} (x_u -x_v)^2 \leq \rho \cdot b_2^{(i)} \cdot |E_i|.
        \]
        \item (Small Entrywise Off-block Weight) For every $i \in [p]$, and for every $x \in \Psi_i$, \[
        \max_{e = (u,v) \in E_{\neq i}} (x_u -x_v)^2 \leq \rho \cdot b_2^{(i)}.
        \]
    \end{enumerate}

For a collection of edges $E_1, \dots E_k$ which satisfies these conditions, the $\CVNRD$ is then defined to be $\sum_i |E_i|$. In order to show that graphs admit small spectral sparsifiers, we must then show that \emph{there cannot be} any large set of edges which satisfies the above conditions. Note that we are free to choose the parameters $\chi, \rho$ to our advantage (although within reason, as we pay polynomial factors in their inverse in the sparsifier size). Thus, we will seek to show: 

\begin{theorem}\label{thm:boundCVNRDGraphintro}
    Let $G = (V, E)$ be a graph on $n$ vertices and $m$ edges. Let $\eps > 0$. Then, \[
    \CVNRD \left (C_G, \eps, \chi =\frac{1}{100 \log^2(m)}, \rho = \frac{\eps^2}{(8 \log(n))^2 \cdot2} \right ) \leq O( n).\]
\end{theorem}

Note that if we show this theorem, then we immediately obtain $(1 \pm \eps)$ spectral sparsifiers of size $O(n \cdot \mathrm{poly}(\log(n), \eps^{-1}))$ via \cref{thm:CVNRDClassifyIntro}.

\paragraph{Understanding the Witness Structure}

The key observation we make is that because of the complete shattering condition, each block of edges $E_i$ \emph{cannot} have any short cycles. In fact, we show that $E_i$ cannot have any cycles of length $\leq \log(n)$; by standard girth-density tradeoffs, this immediately implies that each block of edges $E_i$ contains $O(n)$ edges.

To understand why there cannot be any short cycles in a block of edges $E_i$, we proceed by contradiction. So let us suppose that there is a cycle $C \subseteq E_i$ with length $ \ell \leq \log(n)$. And for now, let us suppose that $\ell$ is odd. Then, the complete shattering condition guarantees us that we can find an assignment $x \in \R^V$ such that \emph{for every single edge} $e = (u,v) \in C$, $(x_u - x_v)^2 \in b_2^{(i)} \cdot [1, 1 + \eps \cdot \chi]$. For simplicity, let us ignore this $\eps \cdot \chi$ error term; in the full version of the proof this error term is chosen to be sufficiently small so that it is immaterial for cycles of size $\leq \log(n)$. Thus, this implies that there is an assignment $x \in \R^V$ such that for every edge $(u,v)$ in the cycle, $|x_u - x_v| = \sqrt{b_2^{(i)}}$. We can write this as $x_u - x_v = \sigma_{u,v} \cdot \sqrt{b_2^{(i)}}$, where $\sigma_{u,v} \in \{\pm 1\}$. Now, we order the edges in the cycle, and sum their energies: \[
\left | \sum_{e = (u,v) \in C} (x_u - x_v) \right |= \left | \sum_{e = (u,v) \in C} \sigma_{u,v}\cdot \sqrt{b_2^{(i)}} \right | \geq \sqrt{b_2^{(i)}}.
\]
In particular, this last inequality follows because there are an \emph{odd} number of edges in the cycle, and so $\sum_{e = (u,v) \in C} \sigma_{u,v}$ \emph{cannot} be $0$; it must be either $\leq -1$ or $\geq 1$. But this then reveals a contradiction, as by definition, if one sums over the edges of the cycle, it must be the case that $\sum_{e = (u,v) \in C} (x_u - x_v) = 0$ (each vertex appears once in the positive direction, and once in the negative direction). 

It may seem that this above approach does not rule out cycles of even length. This is true in the sense that we \emph{cannot} use the same type of shattering assignment to rule out these kinds of cycles. Instead, we pick a ``distinguished'' edge $e^* \in C$, and consider an assignment $x \in \R^V$ such that for every edge $e = (u,v) \in C - e^*$, $(x_u - x_v)^2 = b_2^{(i)}$, and for the last edge $e^* = (u^*, v^*)$,  $(x_{u^*} - x_{v^*})^2 = b_1^{(i)}$. Repeating the above argument shows that $
\left | \sum_{e = (u,v) \in C - e^*} (x_u - x_v) \right | \geq \sqrt{b_2^{(i)}}$. All that's left is to observe that
\[
\left | \sum_{e = (u,v) \in C } (x_u - x_v) \right | \geq \left | \sum_{e = (u,v) \in C - e^*} (x_u - x_v) \right | - |x_{u^*} - x_{v^*}| \geq \sqrt{b_2^{(i)}} - \sqrt{b_1^{(i)}} > \frac{\eps \sqrt{b_2^{(i)}}}{2} > 0,
\]
which again yields a contradiction. In this final inequality, we have used the fact that $b_1^{(i)}$ and $b_2^{(i)}$ are separated from one another by a factor of $(1-\eps)$. Together, this shows that there can be no short cycles \emph{within} each block of edges $E_i$.

\paragraph{Understanding the Off-Block Structure}

The preceding paragraphs show that \emph{within} each block $E_i$, there can be no short cycles. But, we also know that if the combined set $E_1 \cup E_2 \cup \dots \cup E_p$ has $\geq 100n$ edges (for instance), then there must be some cycle $C$ of length $\leq \log(n)$ in the \emph{union} of these blocks of edges. Moreover, because no single block contains short cycles, it must be the case that $C$ uses edges from at least two different blocks. WLOG, we will assume that these blocks are $E_1$ and $E_2$. Now, let us look at the edges in $E_1 \cap C$; as above, the full proof will condition on whether $|E_1 \cap C|$ is odd or even, but for the brief exposition here, we consider when $|E_1 \cap C|$ is odd. Just as above, we then consider the assignment $x \in \R^V$ (which is guaranteed to exist because of the complete shattering condition) which gives every edge $e = (u,v) \in E_1 \cap C$ energy $(x_u - x_v)^2 = b_2^{(i)}$. Just as above, because there is an \emph{odd} number of edges in $E_1 \cap C$, this means that $\left | \sum_{e = (u,v) \in C \cap E_1} (x_u - x_v) \right | \geq \sqrt{b_2^{(1)}}.$

Now, for this same assignment $x \in \R^V$, we look at the ``off-block'' energy (i.e., the energy coming from edges in $E_{\neq 1}$). As above, we know that in any cycle $C$, $\sum_{e = (u,v) \in C} (x_u - x_v) = 0$; because the length of the cycle is $\leq \log(n)$, this means there must be some edge $(u,v) \in C \setminus E_1$ such that $|x_u - x_v| \geq \frac{\sqrt{b_2^{(1)}}}{\log(n)}$. This immediately gives a contradiction though; indeed, the energy on this edge is then $\geq \frac{b_2^{(1)}}{\log^2(n)}$, which violates the small entrywise off-block weight condition! Indeed, by our choice of $\rho$, we would require that $\forall e = (u,v) \in E_{\neq 1}$, $(x_u - x_v)^2 \leq \frac{\eps^2}{128 \log^2(n)} \cdot b_2^{(1)}$. Note that the dependence of $\rho$ on $\eps$ is not necessary in this case when $|E_1 \cap C|$ is odd, but does become necessary when $|E_1 \cap C|$ is even, as then our expression looks like $\left | \sum_{e = (u,v) \in C \cap E_1} (x_u - x_v) \right | \geq \eps\cdot \sqrt{ b_2^{(1)}}.$

This then concludes \cref{thm:boundCVNRDGraphintro}; if there are too many edges in the $\CVNRD$ witness, this implies there is a short cycle in the graph, which in turn must use edges from multiple different blocks. By choosing a careful shattering assignment from one of these blocks, this forces some edge in the rest of the cycle to pick up too much energy, thereby violating the off-block condition.

\paragraph{Generalizing to Higher Powers and Hypergraphs}

Importantly, nothing in the above argument relies on the specific choice of energy being $(x_u - x_v)^2$; in general, we may consider the energy on an edge to be $|x_u - x_v|^t$ for some $t \geq 1$. The same argument will work for showing that in each individual block of edges $E_i$ that there are no short cycles. Thus, when we consider a witness to $\CVNRD$ with (say) $\geq 100 n$ many edges, we can invoke the same argument above to show that there are cycles of length $< \log(n)$ which thereby \emph{must} use edges from multiple different blocks (say that $E_1$ is one such block). As above, we then repeat the argument to show that, by exploiting the complete shattering condition, we can create a potential difference between all of the non-$E_1$ edges. Note that this will now look like $\left | \sum_{e = (u,v) \in C \cap E_1} (x_u - x_v) \right | \geq \eps \cdot \left ( b_2^{(1)} \right)^{1/t}.$ Importantly, this means there must be some edge in $E_{\neq 1}$ which receives potential difference $\geq \frac{\eps \cdot \left ( b_2^{(1)} \right)^{1/t}}{\log(n)}$, and thus receives energy $\geq \frac{\eps^t \cdot b_2^{(1)}}{\log^t(n)}$. By choosing $\rho$ to now depend on $\eps^t$, this then violates the maximum entrywise off-block energy condition. 

The generalization to hypergraphs is more subtle, so we only very briefly give the intuition. The key insight turns out to be that \emph{any hypergraph} which satisfies the conditions of \cref{def:CVNRDparameterIntro} can be \emph{turned into a graph} satisfying the same conditions, while only paying a $\log^2(n)$ factor in the size of the witness. This immediately bounds the $\CVNRD$ of a hypergraph's energy by $O(n \log^2(n))$; indeed, if the size of a witness is any larger, this yields a $\CVNRD$ witness in a graph of size $\omega(n)$, which then violates the bounds provided above. 

Proving this simulation statement (i.e., that any hypergraph is effectively simulating a graph) relies on an invocation of the so-called ``graph-theoretic Sauer-Shelah Lemma'' due to Cesa-Bianchi and Haussler \cite{cesa1998graph}. For more context, the complete shattering condition in a hypergraph says that for a block of hyperedges $E_i$, there exists $b_1^{(i)}, b_2^{(i)} \in \R$, $b_1^{(i)} < b_2^{(i)} (1 - \eps)$ such that for every $B \subseteq E_i$, there is an assignment $x \in \Psi_i$ such that for $e \in B$, $\max_{u,v \in e}(x_u - x_v)^2 \in b_1^{(i)} \cdot[1, 1 + \eps \cdot \chi]$ and for $e \in E_i - B$, $\max_{u,v}(x_u - x_v)^2\in b_2^{(i)} \cdot [1, 1 + \eps \cdot \chi]$. Whenever $\max_{u,v}(x_u - x_v)^2\in b_2^{(i)} \cdot [1, 1 + \eps \cdot \chi]$, there must be a ``representative pair'' of vertices $u,v \in e$ which are contributing the energy $(x_u - x_v)^2\approx b_2^{(i)}$. For each hyperedge, and each assignment $x \in \R^V$, we thus associate a label $\in \binom{n}{2}$ corresponding to which ``representative pair'' witnesses this high energy. The graph-theoretic Sauer-Shelah Lemma \cite{cesa1998graph} then shows that one can find a \emph{slightly smaller} subset of hyperedges $E'_i \subseteq E_i$ such that, whenever a hyperedge $e$ receives energy $\max_{u,v \in e}(x_u - x_v)^2\approx b_2^{(i)}$, \emph{it is always the same pair} $(u,v)$ which is the ``representative pair''. This is then formally translated into a graph sparsification problem, where each hyperedge is replaced with its representative pair. 

\subsubsection{Building $\NRD$-Size Sparsifiers for $\zo$-Valued Codes Without Gilmer}\label{sec:techOverviewZeroOne}

\paragraph{Background}

The preceding discussion gives intuition about how to use $\CVNRD$ to understand the sparsifiability of real-valued structures. In this section, we now give intuition about how we actually prove our sparsifiability results. To do so, we start by re-deriving a result of Brakensiek and Guruswami \cite{brakensiek2025redundancy} which characterizes the sparsifiability of $\zo$-valued codes by their $\NRD$, but importantly, this proof does \emph{not} use Gilmer's entropy method \cite{gilmer2022constant}. Recall that, given a code $C \subseteq \zo^m$, the \emph{non-redundancy} of $C$, denoted $\NRD(C)$, is the largest collection of indices $a_1, \dots a_k \in [m]$ along with codewords $c^{(1)}, \dots c^{(k)}$ such that for $i \in [k]$ $c^{(i)}_{a_i} = 1$, and $c^{(i)}_{a_{\neq i}} = 0$. Equivalently, if one lines up the codewords in $C$ as columns of an $m \times |C|$ matrix, $\NRD(C)$ is the largest identity submatrix (up to permutation of rows and columns).

The key result which drives sparsification in \cite{brakensiek2025redundancy} is essentially the following (which we have edited for clarity):
\begin{theorem}[\cite{brakensiek2025redundancy} Decomposition Theorem]\label{thm:decompositionTheoremIntro}
    Let $C \subseteq \zo^m$ consist of only codewords of Hamming weight $\in [d/2, d]$. Then, for any $\eps > 0$, there exists a set $S \subseteq [m]$ of size $|S| \leq O(\NRD(C) \cdot \log^4(m) / \eps^2)$ such that 
    \[
|C|_{\bar{S}}| \leq 2^{\eps^2 \cdot d / 10000}.
    \]
\end{theorem}

In words, this theorem is stating that, for any code $C$ with codewords of weight $\in [d/2, d]$, there is a \emph{small} set of coordinates $S$ (whose size scales only with the non-redundancy of $C$), such that, after peeling off the set $S$, the number of codewords in $C$ drastically decreases. In particular, after peeling off this set of coordinates $S$, one can then \emph{randomly sample} the remaining coordinates in $\bar{S}$ at rate $1/2$ to get a set $T \subseteq [m] - S$ with the guarantee that, with high probability over the random sampling, for every codeword $c \in C$,
\[
\sum_{i \in S} c_i + 2 \cdot \sum_{i \in T} c_i \in (1 \pm \eps) \cdot \sum_{i \in [m]} c_i.
\]
Thus, this effectively reduces the number of coordinates in the code by a factor of $2$, and repeated application is used to ultimately produce sparsifiers of size $\widetilde{O}(\NRD(C) / \eps^2)$.

The proof of this decomposition theorem in \cite{brakensiek2025redundancy} relies on an intricate analysis using  Gilmer's entropy method \cite{gilmer2022constant}, the Sauer-Shelah lemma, along with an exploitation of the duality between so-called $\theta$-covers and $\theta$-sparse distributions over a code $C$. In what follows, we take an entirely different approach and show that the only essential ingredient from this list is the Sauer-Shelah lemma.

Towards this end, we recall this lemma:

\begin{lemma}[Sauer-Shelah Lemma, \cite{sauer1972density} \cite{shelah1972combinatorial}.]\label{lem:SSintro}
    Let $C \subseteq \zo^m$. Then, there exists a set $A \subseteq [m]$, $|A| \geq \frac{\log(|C|)}{\log(m+1)}$ such that $C|_{A} = \zo^A$.
\end{lemma}

Intuitively, this shows that whenever a code $C$ has a lot of codewords, \emph{it must be the case} that there is some ``complete subcode'' hiding inside the code. This is denoted by $\zo^A$: meaning that we have all possible patterns of $0$'s and $1$'s on those coordinates. Importantly, the size of $A$ scales with the number of codewords in $C$.

\paragraph{Our Approach to Proving \cref{thm:decompositionTheoremIntro}}

Our approach to proving \cref{thm:decompositionTheoremIntro} is direct; we let $C, d, \eps$ be given as in the statement of their theorem, and we let $S$ be the \emph{smallest possible set} such that $|C_{\bar{S}}| \leq 2^{\eps^2 d / 10000}$. Our goal is then to show that if $S$ is large, then there must be a similarly large \emph{witness to} $\NRD$ inside the code $C$.

As a first observation, we can see that if $S$ is large, then there is a large set of ``disjoint complete subcodes'' that we can find inside the code $C$. Indeed, let us consider an iterative algorithm which peels off complete subcodes. To start, we initialize a set $R = \emptyset \subseteq [m]$. Naturally then, it should be the case that $|R| < |S|$. Because $S$ is the \emph{smallest possible set} such that $|C_{\bar{S}}| \leq 2^{\eps^2 d / 10000}$, it will also be the case that $|C_{\bar{R}}| >2^{\eps^2 d / 10000} $. Now, we can invoke \cref{lem:SSintro}. This shows that there is some subset $A_1 \subseteq [m] - R$ of size $|A_1| \geq \frac{\log(|C_{\bar{R}}|)}{\log(m+1)} \approx \frac{\eps^2 d}{\log(m)}$ (for the convenience of the reader, we will avoid carrying around constant factors), such that $C|_{A_1} = \zo^{A_1}$.

After we find such a subcode, we will now say that these coordinates $A_1$ are effectively ``used up'', and so we add them to the set $R$. Now, assuming it is still true that $|R| < |S|$, we again repeat the observation that $|C_{\bar{R}}| >2^{\eps^2 d / 10000} $ and peel off another set $A_2$, $|A_2| \approx \frac{\eps^2 d}{\log(m)}$ such that $C|_{A_2} = \zo^{A_2}$. Moreover, it must be the case that $A_2$ is \emph{completely disjoint} from $A_1$, as we had already added $A_1$ to $R$ at this point. Note that we can keep repeating this procedure until $|R| \geq |S|$. At this point, we will have recovered sets $A_1, \dots A_k$ such that $\sum_{i \in [k]} |A_i| \geq |S|$, for each $A_i$, $C|_{A_i} = \zo^{A_i}$, and $|A_i| \approx \frac{\eps^2 d}{\log(m)}$. For convenience, we will let $C_1 \subseteq C$ denote a minimal set of codewords which witnesses $(C_1)|_{A_1} = \zo^{A_1}$, and similarly denote $C_i$ as the collection of codewords for $A_i$. Although the resulting collection of codewords is not block-diagonal, we do know that, because every codeword in $C$ had weight $\leq d$, for each codeword $c \in C_i$, $\wt(c|_{A_{\neq i}}) \leq d$.

\begin{figure}
    \centering
    \begin{tikzpicture}[
  font=\small,
  diag/.style={fill=blue!14, draw=black, line width=0.45pt},
  off/.style={fill=gray!9},
  grid/.style={draw=black, line width=0.45pt},
  lab/.style={font=\tiny, align=center},
  star/.style={font=\large, text=gray!55},
  ell/.style={font=\Large, text=gray!70},
  braceLab/.style={font=\scriptsize, align=center}
]

\coordinate (O) at (0,0);
\def\xzero{0}
\def\xone{2.05}
\def\xtwo{4.10}
\def\xthree{6.15}
\def\xfour{6.75}
\def\xfive{8.80}

\def\yzero{0}
\def\yone{-1.12}
\def\ytwo{-2.24}
\def\ythree{-3.36}
\def\yfour{-4.00}
\def\yfive{-5.12}


\fill[off] (\xzero,\yzero) rectangle (\xfive,\yfive);
\fill[diag] (\xzero,\yzero) rectangle (\xone,\yone);
\fill[diag] (\xone,\yone) rectangle (\xtwo,\ytwo);
\fill[diag] (\xtwo,\ytwo) rectangle (\xthree,\ythree);
\fill[diag] (\xfour,\yfour) rectangle (\xfive,\yfive);

\draw[grid, line width=0.65pt] (\xzero,\yzero) rectangle (\xfive,\yfive);
\foreach \x in {\xone,\xtwo,\xthree,\xfour} {
  \draw[grid] (\x,\yzero) -- (\x,\yfive);
}
\foreach \y in {\yone,\ytwo,\ythree,\yfour} {
  \draw[grid] (\xzero,\y) -- (\xfive,\y);
}

\node[align=center] at (1.025,0.38) {$C_1$};
\node[align=center] at (3.075,0.38) {$C_2$};
\node[align=center] at (5.125,0.38) {$C_3$};
\node[align=center] at (7.375,0.38) {$C_k$};


\node[anchor=east] at (-0.10,-0.56) {$A_1$};
\node[anchor=east] at (-0.10,-1.68) {$A_2$};
\node[anchor=east] at (-0.10,-2.80) {$A_3$};
\node[anchor=east] at (-0.10,-4.56) {$A_k$};

Left brace for witness subcollections.
\draw[decorate, decoration={brace, amplitude=5pt, mirror}]
  (-0.62,\yone) -- (-0.62,\yfive)
  node[midway, left=6pt, braceLab]
  {off-block weight $\leq d$};

\node[lab] at (1.025,-0.56)
  {$C_1|_{A_1}$\\[-2pt]$=\zo^{A_1}$\\[-1pt]{\tiny all patterns}};
\node[lab] at (3.075,-1.68)
  {$C_2|_{A_2}$\\[-2pt]$=\zo^{A_2}$\\[-1pt]{\tiny all patterns}};
\node[lab] at (5.125,-2.80)
  {$C_3|_{A_3}$\\[-2pt]$=\zo^{A_3}$\\[-1pt]{\tiny all patterns}};
\node[lab] at (7.675,-4.56)
  {$C_k|_{A_k}$\\[-2pt]$=\zo^{A_k}$\\[-1pt]{\tiny all patterns}};

\foreach \x/\y in {
  3.075/-0.56,5.125/-0.56,7.575/-0.56,
  1.025/-1.68,5.125/-1.68,7.575/-1.68,
  1.025/-2.80,3.075/-2.80,7.575/-2.80,
  1.025/-4.56,3.075/-4.56,5.125/-4.56}
  {\node[star] at (\x,\y) {$\ast$};}

\node[ell] at (6.45,-0.56) {$\cdots$};
\node[ell] at (6.45,-1.68) {$\cdots$};
\node[ell] at (6.45,-2.80) {$\cdots$};
\node[ell] at (1.025,-3.68) {$\vdots$};
\node[ell] at (3.075,-3.68) {$\vdots$};
\node[ell] at (5.125,-3.68) {$\vdots$};
\node[ell] at (6.45,-3.68) {$\ddots$};
\node[ell] at (7.575,-3.68) {$\vdots$};


\end{tikzpicture}
    \caption{The resulting subcode collection after peeling dense subcodes.}
    \label{fig:figure1}
\end{figure}

The result of this sequence of operations is a collection of subcodes and sets $A_1, \dots A_k$ which has the form of \cref{fig:figure1}. To prove \cref{thm:decompositionTheoremIntro}, our goal simply to show that this collection of subcodes \emph{contains} a large identity matrix somewhere! Our general guiding principle is that we will continue to process this collection of subcodes until it becomes \emph{truly} block-diagonal.

\paragraph{Reducing Off-Diagonal Weight}

Our first step towards making this collection more block-diagonal is a simple sampling scheme which reduces the off-block weight. Rather than a condition saying for each codeword $c \in C_i$, $\wt(c|_{A_{\neq i}}) \leq d$, we want something stronger; namely that for each codeword $c \in C_i$, $\wt(c|_{A_{\neq i}}) \leq \frac{\eps^2 d}{\log^3(m)}$.

To do this, we subsample the \emph{blocks} $A_1, \dots A_k$ at rate $\frac{\eps^2}{10\log^3(m)}$. We let $T \subseteq [k]$ denote the indices of the blocks which survive sampling, and we assume for simplicity that $|T| = \frac{k \eps^2}{10 \log^3(m)}$. Importantly, for an index $i \in T$, we can now look at a codeword $c \in C_i$. In expectation, we see that 
\[
\E_{T} \left [\sum_{j \neq i \in T}\wt(c|_{A_j}) \right] \leq \wt(c|_{A_{\neq i}}) \cdot \frac{\eps^2}{10\log^3(m)} \leq \frac{\eps^ 2d}{10 \log^3(m)},
\]
where we have used the fact that each coordinate in $A_{\neq i}$ is kept in $T$ with probability $\frac{\eps^2}{10\log^3(m)} $. We say that a codeword $c \in C_i: i \in T$ is \emph{good} if $\sum_{j \neq i \in T}\wt(c|_{A_j}) \leq \frac{\eps^2 d}{ \log^3(m)}$. Importantly, for a fixed $c \in C_i$, $\Pr_T[c \text{ is good} | i \in T] \geq \frac{9}{10}$, by a simple Markov bound.

Thus, $\E_{T}[\left | \{c \in C_i: c\text{ is good}\} \right | | i \in T] \geq \frac{9}{10} \cdot |C_i|$. Using this, we then say that a block $A_i$ is good (with respect to $T$), if $\left | \{c \in C_i: c\text{ is good}\} \right | \geq \frac{|C_i|}{2}$. Again by a simple Markov bound, we see that $A_i$ is \emph{good} with probability $\geq 4/5$ conditioned on $i$ being sampled. All together, this means that when we sample $|T| = \frac{k \eps^2}{10 \log^3(m)}$ many blocks, \emph{in expectation}, at least $4/5$ of these blocks will be good; letting $T'$ denote this subset of good blocks, we then know that there must \emph{exist} some collection of blocks $A_i: i \in T'$ such that each such block $A_i$ is good. We then let $C'_i \subseteq C_i$ denote the corresponding codewords that \emph{are good} for this set of blocks $T'$.

So, what do we get from this entire procedure? The result is a collection of $\geq \frac{k \eps^2}{20 \log^3(m)}$ many blocks, $A_i: i \in T'$, along with sets of codewords $C'_i \subseteq C_i: i \in T'$ such that:
\begin{enumerate}
    \item Each $C'_i$ is not necessarily a complete subcode anymore, but each $C'_i$ \emph{still contains a lot of codewords}. Indeed, $|C'_i| \geq |C_i|/2 = 2^{|A_i|-1}$.
    \item The off-diagonal weight is now significantly reduced: for each $i \in T'$, $c \in C'_i$, $\sum_{j \neq i \in T'}\wt(c|_{A_j}) \leq \frac{\eps^2 d}{ \log^3(m)}$.
\end{enumerate}

For simplicity, going forward, we will refer to $A_i : i \in T'$ as $A_1, \dots A_{k'}$ for $k' = |T'|$. We summarize this resulting picture in \cref{fig:figure2}.

\begin{figure}
    \centering
\begin{tikzpicture}[
  font=\small,
  diag/.style={fill=blue!14, draw=black, line width=0.45pt},
  off/.style={fill=gray!9},
  grid/.style={draw=black, line width=0.45pt},
  lab/.style={font=\tiny, align=center},
  star/.style={font=\large, text=gray!55},
  ell/.style={font=\Large, text=gray!70},
  braceLab/.style={font=\scriptsize, align=center}
]

\coordinate (O) at (0,0);
\def\xzero{0}
\def\xone{2.05}
\def\xtwo{4.10}
\def\xthree{6.15}
\def\xfour{6.75}
\def\xfive{8.80}

\def\yzero{0}
\def\yone{-1.12}
\def\ytwo{-2.24}
\def\ythree{-3.36}
\def\yfour{-4.00}
\def\yfive{-5.12}


\fill[off] (\xzero,\yzero) rectangle (\xfive,\yfive);
\fill[diag] (\xzero,\yzero) rectangle (\xone,\yone);
\fill[diag] (\xone,\yone) rectangle (\xtwo,\ytwo);
\fill[diag] (\xtwo,\ytwo) rectangle (\xthree,\ythree);
\fill[diag] (\xfour,\yfour) rectangle (\xfive,\yfive);

\draw[grid, line width=0.65pt] (\xzero,\yzero) rectangle (\xfive,\yfive);
\foreach \x in {\xone,\xtwo,\xthree,\xfour} {
  \draw[grid] (\x,\yzero) -- (\x,\yfive);
}
\foreach \y in {\yone,\ytwo,\ythree,\yfour} {
  \draw[grid] (\xzero,\y) -- (\xfive,\y);
}

\node[align=center] at (1.025,0.38) {$C'_1$};
\node[align=center] at (3.075,0.38) {$C'_2$};
\node[align=center] at (5.125,0.38) {$C'_3$};
\node[align=center] at (7.375,0.38) {$C'_{k'}$};


\node[anchor=east] at (-0.10,-0.56) {$A_1$};
\node[anchor=east] at (-0.10,-1.68) {$A_2$};
\node[anchor=east] at (-0.10,-2.80) {$A_3$};
\node[anchor=east] at (-0.10,-4.56) {$A_{k'}$};

Left brace for witness subcollections.
\draw[decorate, decoration={brace, amplitude=5pt, mirror}]
  (-0.62,\yone) -- (-0.62,\yfive)
  node[midway, left=6pt, braceLab]
  {off-block weight \\$\leq \frac{\eps^2 d}{\log^3(m)}$};

\node[lab] at (1.025,-0.56)
  {$|C'_1|_{A_1}|$\\ \\ [-2pt]$\geq 2^{|A_1|-1}$};
\node[lab] at (3.075,-1.68)
  {$|C'_2|_{A_2}|$\\ \\ [-2pt]$\geq 2^{|A_2|-1}$};
\node[lab] at (5.125,-2.80)
  {$|C'_3|_{A_3}|$\\ \\ [-2pt]$\geq 2^{|A_3|-1}$};
\node[lab] at (7.675,-4.56)
  {$|C'_{k'}|_{A_{k'}}|$\\ \\ [-2pt]$\geq 2^{|A_{k'}|-1}$};

\foreach \x/\y in {
  3.075/-0.56,5.125/-0.56,7.575/-0.56,
  1.025/-1.68,5.125/-1.68,7.575/-1.68,
  1.025/-2.80,3.075/-2.80,7.575/-2.80,
  1.025/-4.56,3.075/-4.56,5.125/-4.56}
  {\node[star] at (\x,\y) {$\ast$};}

\node[ell] at (6.45,-0.56) {$\cdots$};
\node[ell] at (6.45,-1.68) {$\cdots$};
\node[ell] at (6.45,-2.80) {$\cdots$};
\node[ell] at (1.025,-3.68) {$\vdots$};
\node[ell] at (3.075,-3.68) {$\vdots$};
\node[ell] at (5.125,-3.68) {$\vdots$};
\node[ell] at (6.45,-3.68) {$\ddots$};
\node[ell] at (7.575,-3.68) {$\vdots$};


\end{tikzpicture}
    \caption{The resulting subcode collection after off-block weight reduction.}
    \label{fig:figure2}
\end{figure}

\paragraph{Off-Block Weight Grouping}

The resulting subcodes in \cref{fig:figure2} are still far from perfect; recall that our goal was to get a \emph{truly} block-diagonal collection, and so far all we have achieved is an approximately block-diagonal collection with slightly smaller off-diagonal weight. Note that we cannot afford to keep subsampling blocks, as each time we do this, we \emph{reduce} the number of blocks. Because our goal was to find a \emph{large} identity submatrix, each subsampling operation further limits the size of the submatrix we will find. 

Nevertheless, this reduced off-diagonal weight will be incredibly valuable to us; let us focus our attention on a single subcode $C'_1$. We know that the number of $\zo$-patterns realized in $A_1$ is at least $\left | (C'_1)|_{A_1} \right | \geq 2^{|A_1| - 1} \geq 2^{\eps^2 d / \log(m) - 1}$. Now, let us count the number of $\zo$-patterns realized in $(C'_1)|_{A_{> 1}}$. The key observation here is that, because the off-diagonal weight is bounded, the maximum possible number of $\zo$-patterns in the off-diagonal is bounded by $\binom{m}{\leq \frac{\eps^2 d}{\log^3(m)}} \leq 2^{\eps^2 d/\log^2(m)}$ which is \emph{significantly smaller} than the number of patterns in $(C'_1)|_{A_1}$.

More formally, for a $\zo$-pattern $z \in \zo^{A_2 \cup A_3 \cup \dots A_{k'}}$ on the off-diagonal coordinates, we let $\mathrm{Class}(z, C'_1) = \{c \in C'_1: c|_{A_2 \cup \dots A_{k'} } = z \}$. By a pigeonhole principle, we then see that there exists a choice of $z$ such that 
\[
|\mathrm{Class}(z, C'_1)| \geq \frac{2^{|A_1| - 1}}{2^{\eps^2 d/\log^2(m)}} \geq 2^{|A_1| - 1 - \eps^2 d/\log^2(m)} \geq 2^{|A_1|/2}.
\]
For such a $z$, we let $C''_1$ denote $\mathrm{Class}(z, C'_1)$, and define $C''_2, \dots C''_{k'}$ similarly. To summarize, we thus have a collection of $\geq \frac{k \eps^2}{20 \log^3(m)}$ many blocks, $A_1, \dots A_{k'}$, along with sets of codewords $C''_i$ such that:
\begin{enumerate}
    \item Each $C''_i$ is not necessarily a complete subcode, but each $C''_i$ \emph{still contains a lot of codewords}. Indeed, $|C''_i| \geq 2^{|A_i| /2}$.
    \item The off-diagonal coordinates are \emph{shared} for each $C''_i$. I.e., there exists $z \in \zo^{\bigcup_{j \neq i A_j}}$, $\wt(z) \leq \frac{\eps^2 d}{\log^3(m)}$, and for all $c \in C''_i$, $c|_{\bigcup_{j \neq i A_j}} = z$.
\end{enumerate}

\begin{figure}
    \centering
\begin{tikzpicture}[
  font=\small,
  diag/.style={fill=blue!14, draw=black, line width=0.45pt},
  off/.style={fill=gray!9},
  grid/.style={draw=black, line width=0.45pt},
  lab/.style={font=\tiny, align=center},
  star/.style={font=\large, text=gray!55},
  ell/.style={font=\Large, text=gray!70},
  braceLab/.style={font=\scriptsize, align=center}
]

\coordinate (O) at (0,0);
\def\xzero{0}
\def\xone{2.05}
\def\xtwo{4.10}
\def\xthree{6.15}
\def\xfour{6.75}
\def\xfive{8.80}

\def\yzero{0}
\def\yone{-1.12}
\def\ytwo{-2.24}
\def\ythree{-3.36}
\def\yfour{-4.00}
\def\yfive{-5.12}


\fill[off] (\xzero,\yzero) rectangle (\xfive,\yfive);
\fill[diag] (\xzero,\yzero) rectangle (\xone,\yone);
\fill[diag] (\xone,\yone) rectangle (\xtwo,\ytwo);
\fill[diag] (\xtwo,\ytwo) rectangle (\xthree,\ythree);
\fill[diag] (\xfour,\yfour) rectangle (\xfive,\yfive);

\draw[grid, line width=0.65pt] (\xzero,\yzero) rectangle (\xfive,\yfive);
\foreach \x in {\xone,\xtwo,\xthree,\xfour} {
  \draw[grid] (\x,\yzero) -- (\x,\yfive);
}
\foreach \y in {\yone,\ytwo,\ythree,\yfour} {
  \draw[grid] (\xzero,\y) -- (\xfive,\y);
}

\node[align=center] at (1.025,0.38) {$C''_1$};
\node[align=center] at (3.075,0.38) {$C''_2$};
\node[align=center] at (5.125,0.38) {$C''_3$};
\node[align=center] at (7.375,0.38) {$C''_{k'}$};


\node[anchor=east] at (-0.10,-0.56) {$A_1$};
\node[anchor=east] at (-0.10,-1.68) {$A_2$};
\node[anchor=east] at (-0.10,-2.80) {$A_3$};
\node[anchor=east] at (-0.10,-4.56) {$A_{k'}$};

Left brace for witness subcollections.
\draw[decorate, decoration={brace, amplitude=5pt, mirror}]
  (-0.62,\yone) -- (-0.62,\yfive)
  node[midway, left=6pt, braceLab]
  {off-block weight \\$\leq \frac{\eps^2 d}{\log^3(m)},$ \\
  shared off-diagonal \\pattern};

\node[lab] at (1.025,-0.56)
  {$|C''_1|_{A_1}|$\\ \\ [-2pt]$\geq 2^{|A_1|/2}$};
\node[lab] at (3.075,-1.68)
  {$|C''_2|_{A_2}|$\\ \\ [-2pt]$\geq 2^{|A_2|/2}$};
\node[lab] at (5.125,-2.80)
  {$|C''_3|_{A_3}|$\\ \\ [-2pt]$\geq 2^{|A_3|/2}$};
\node[lab] at (7.675,-4.56)
  {$|C''_{k'}|_{A_{k'}}|$\\ \\ [-2pt]$\geq 2^{|A_{k'}|/2}$};
\node[lab] at (1.025,-1.68)
  {$0000000000$\\ $0000000000$\\ $0000000000$\\$1111111111$};
\node[lab] at (1.025,-2.80)
  {$1111111111$\\ $0000000000$\\ $1111111111$\\$0000000000$};
\node[lab] at (1.025,-4.56)
  {$0000000000$\\ $1111111111$\\ $0000000000$\\$0000000000$};

\foreach \x/\y in {
  3.075/-0.56,5.125/-0.56,7.575/-0.56,
  5.125/-1.68,7.575/-1.68,
  3.075/-2.80,7.575/-2.80,
  3.075/-4.56,5.125/-4.56}
  {\node[star] at (\x,\y) {$\ast$};}

\node[ell] at (6.45,-0.56) {$\cdots$};
\node[ell] at (6.45,-1.68) {$\cdots$};
\node[ell] at (6.45,-2.80) {$\cdots$};
\node[ell] at (1.025,-3.68) {$\vdots$};
\node[ell] at (3.075,-3.68) {$\vdots$};
\node[ell] at (5.125,-3.68) {$\vdots$};
\node[ell] at (6.45,-3.68) {$\ddots$};
\node[ell] at (7.575,-3.68) {$\vdots$};


\end{tikzpicture}
    \caption{The resulting subcode collection after finding shared off-diagonal support. Note that the off-diagonal is only filled in for $C''_1$ for clarity, but will be true for all codes $C''_i$.}
    \label{fig:figure3}
\end{figure}

This is summarized in \cref{fig:figure3}.

\paragraph{Making the Collection Truly Block-Diagonal}

With a collection of subcodes in the form of \cref{fig:figure3}, there is now a simple procedure to make the collection block-diagonal. The starting point for the procedure is a basic greedy algorithm which first turns the collection into a \emph{block upper-triangular} collection of codes.

To start, we initialize a set $Q = \emptyset$. In the first iteration, we look at the set $A_1$; if $|A_1 \cap Q| \leq \frac{|A_1|}{10}$ (which will be the case), then we decide to \emph{commit} and keep the set $A_1$ in our collection. However, our goal is to make an upper-triangular matrix; thus because we have committed to keep $A_1$, we must \emph{remove} all of the non-zero coordinates in the codewords $C''_1$. Formally, we set $Q \leftarrow Q \cup \Supp((C''_1)_{A_{> 1}})$. This is effectively the ``banned'' set of coordinates which will be removed from our witness. 

We then proceed to the set $A_2$, again we check whether $|A_2 \cap Q| \leq \frac{|A_2|}{10}$; if so, then we commit to keeping $A'_2 = A_2 \setminus Q$, and again update $Q$ as $Q \leftarrow Q \cup \Supp((C''_2)_{A_{> 2}})$. Otherwise, if $|A_2 \cap Q| > \frac{|A_2|}{10}$, we simply throw the set $A_2$ away. In either case, we then proceed to look at the set $A_3$.

A simple analysis will show that, in the course of this procedure, very few sets are thrown away, and thus the result from this procedure is
a collection of $k'' \geq \frac{k \eps^2}{40 \log^3(m)}$ many blocks, $A'_1, \dots A'_{k''}$, along with sets of codewords $C''_i$ such that:
\begin{enumerate}
    \item Each $C''_i$ is not necessarily a complete subcode, but each $C''_i$ \emph{still contains many patterns}. Indeed, $|(C''_i)|_{A'_i}| \geq \frac{2^{|A_i| /2}}{2^{|A_i|/10}} \geq 2^{0.4 |A_2|}$. Note that the small reduction is due to the fact that $A'_i = A_i \setminus Q$, where $Q \cap A_i$ can be as large as $\leq |A_i|/10$. Naturally, this can then reduce the number of different $\zo$-patterns in $(C''_i)|_{A'_i}$ by a factor of $2^{|A_i|/10}$.
    \item The above-diagonal coordinates are \emph{shared} for each $C''_i$. I.e., there exists $z \in \zo^{\bigcup_{j < i A_j}}$, $\wt(z) \leq \frac{\eps^2 d}{\log^3(m)}$, and for all $c \in C''_i$, $c|_{\bigcup_{j < i A_j}} = z$.
    \item The below-diagonal coordinates are all $0$ for every codeword $c \in C''_i$.
\end{enumerate}

\begin{figure}
    \centering
\begin{tikzpicture}[
  font=\small,
  diag/.style={fill=blue!14, draw=black, line width=0.45pt},
  off/.style={fill=gray!9},
  grid/.style={draw=black, line width=0.45pt},
  lab/.style={font=\tiny, align=center},
  star/.style={font=\large, text=gray!55},
  ell/.style={font=\Large, text=gray!70},
  braceLab/.style={font=\scriptsize, align=center}
]

\coordinate (O) at (0,0);
\def\xzero{0}
\def\xone{2.05}
\def\xtwo{4.10}
\def\xthree{6.15}
\def\xfour{6.75}
\def\xfive{8.80}

\def\yzero{0}
\def\yone{-1.12}
\def\ytwo{-2.24}
\def\ythree{-3.36}
\def\yfour{-4.00}
\def\yfive{-5.12}


\fill[off] (\xzero,\yzero) rectangle (\xfive,\yfive);
\fill[diag] (\xzero,\yzero) rectangle (\xone,\yone);
\fill[diag] (\xone,\yone) rectangle (\xtwo,\ytwo);
\fill[diag] (\xtwo,\ytwo) rectangle (\xthree,\ythree);
\fill[diag] (\xfour,\yfour) rectangle (\xfive,\yfive);

\draw[grid, line width=0.65pt] (\xzero,\yzero) rectangle (\xfive,\yfive);
\foreach \x in {\xone,\xtwo,\xthree,\xfour} {
  \draw[grid] (\x,\yzero) -- (\x,\yfive);
}
\foreach \y in {\yone,\ytwo,\ythree,\yfour} {
  \draw[grid] (\xzero,\y) -- (\xfive,\y);
}

\node[align=center] at (1.025,0.38) {$C''_1$};
\node[align=center] at (3.075,0.38) {$C''_2$};
\node[align=center] at (5.125,0.38) {$C''_3$};
\node[align=center] at (7.375,0.38) {$C''_{k''}$};


\node[anchor=east] at (-0.10,-0.56) {$A'_1$};
\node[anchor=east] at (-0.10,-1.68) {$A'_2$};
\node[anchor=east] at (-0.10,-2.80) {$A'_3$};
\node[anchor=east] at (-0.10,-4.56) {$A'_{k''}$};

Left brace for witness subcollections.
\draw[decorate, decoration={brace, amplitude=5pt, mirror}]
  (-0.7,\yone) -- (-0.7,\yfive)
  node[midway, left=6pt, braceLab]
  {off-block weight \\$\leq \frac{\eps^2 d}{\log^3(m)},$ \\
  shared above-diagonal \\pattern \\
  zero below-diagonal};

\node[lab] at (1.025,-0.56)
  {$|C''_1|_{A'_1}|$\\ \\ [-2pt]$\geq 2^{0.4|A_1|}$};
\node[lab] at (3.075,-1.68)
  {$|C''_2|_{A'_2}|$\\ \\ [-2pt]$\geq 2^{0.4|A_2|}$};
\node[lab] at (5.125,-2.80)
  {$|C''_3|_{A'_3}|$\\ \\ [-2pt]$\geq 2^{0.4|A_3|}$};
\node[lab] at (7.675,-4.56)
  {$|C''_{k''}|_{A'_{k''}}|$\\ \\ [-2pt]$\geq 2^{0.4|A_{k''}|}$};
\node[lab] at (1.025,-1.68)
  {$0000000000$\\ $0000000000$\\ $0000000000$\\$0000000000$};
\node[lab] at (1.025,-2.80)
  {$0000000000$\\ $0000000000$\\ $0000000000$\\$0000000000$};
\node[lab] at (1.025,-4.56)
  {$0000000000$\\ $0000000000$\\ $0000000000$\\$0000000000$};
\node[lab] at (3.075,-0.56)
  {$0000000000$\\ $1111111111$\\ $0000000000$\\$1111111111$};
\node[lab] at (3.075,-2.80)
  {$0000000000$\\ $0000000000$\\ $0000000000$\\$0000000000$};
\node[lab] at (3.075,-4.56)
  {$0000000000$\\ $0000000000$\\ $0000000000$\\$0000000000$};
\node[lab] at (5.125,-0.56)
  {$1111111111$\\ $0000000000$\\ $0000000000$\\$1111111111$};
\node[lab] at (5.125,-1.68)
  {$1111111111$\\ $0000000000$\\ $0000000000$\\$0000000000$};
\node[lab] at (5.125,-4.56)
  {$0000000000$\\ $0000000000$\\ $0000000000$\\$0000000000$};

\foreach \x/\y in {
  7.575/-0.56,
  7.575/-1.68,
  7.575/-2.80}
  {\node[star] at (\x,\y) {$\ast$};}

\node[ell] at (6.45,-0.56) {$\cdots$};
\node[ell] at (6.45,-1.68) {$\cdots$};
\node[ell] at (6.45,-2.80) {$\cdots$};
\node[ell] at (1.025,-3.68) {$\vdots$};
\node[ell] at (3.075,-3.68) {$\vdots$};
\node[ell] at (5.125,-3.68) {$\vdots$};
\node[ell] at (6.45,-3.68) {$\ddots$};
\node[ell] at (7.575,-3.68) {$\vdots$};


\end{tikzpicture}
    \caption{The collection of subcodes after being made block upper-triangular.}
    \label{fig:figure4}
\end{figure}

This picture is summarized in \cref{fig:figure4}.

Finally, by a simple permutation of coordinates and labels, we can invoke this exact same procedure on the block-upper triangular collection of codewords to make them \emph{truly} block-diagonal (i.e., by making the above diagonal coordinates now the below diagonal coordinates). For clarity, we will refer to the resulting sets that we keep as $\hat{A}_1, \dots \hat{A}_{\hat{k}}$, and the subcodes as $\hat{C}_1, \dots \hat{C}_{\hat{k}}$, with the understanding that $\hat{k} \geq \frac{k \eps^2}{80 \log^3(m)}$, and $|(\hat{C}_i)|_{\hat{A}_i}| \geq 2^{0.3 |A_i|}$. This picture is summarized in \cref{fig:figure5}.

\begin{figure}
    \centering
\begin{tikzpicture}[
  font=\small,
  diag/.style={fill=blue!14, draw=black, line width=0.45pt},
  off/.style={fill=gray!9},
  grid/.style={draw=black, line width=0.45pt},
  lab/.style={font=\tiny, align=center},
  star/.style={font=\large, text=gray!55},
  ell/.style={font=\Large, text=gray!70},
  braceLab/.style={font=\scriptsize, align=center}
]

\coordinate (O) at (0,0);
\def\xzero{0}
\def\xone{2.05}
\def\xtwo{4.10}
\def\xthree{6.15}
\def\xfour{6.75}
\def\xfive{8.80}

\def\yzero{0}
\def\yone{-1.12}
\def\ytwo{-2.24}
\def\ythree{-3.36}
\def\yfour{-4.00}
\def\yfive{-5.12}


\fill[off] (\xzero,\yzero) rectangle (\xfive,\yfive);
\fill[diag] (\xzero,\yzero) rectangle (\xone,\yone);
\fill[diag] (\xone,\yone) rectangle (\xtwo,\ytwo);
\fill[diag] (\xtwo,\ytwo) rectangle (\xthree,\ythree);
\fill[diag] (\xfour,\yfour) rectangle (\xfive,\yfive);

\draw[grid, line width=0.65pt] (\xzero,\yzero) rectangle (\xfive,\yfive);
\foreach \x in {\xone,\xtwo,\xthree,\xfour} {
  \draw[grid] (\x,\yzero) -- (\x,\yfive);
}
\foreach \y in {\yone,\ytwo,\ythree,\yfour} {
  \draw[grid] (\xzero,\y) -- (\xfive,\y);
}

\node[align=center] at (1.025,0.38) {$\hat{C}_1$};
\node[align=center] at (3.075,0.38) {$\hat{C}_2$};
\node[align=center] at (5.125,0.38) {$\hat{C}_3$};
\node[align=center] at (7.375,0.38) {$\hat{C}_{\hat{k}}$};


\node[anchor=east] at (-0.10,-0.56) {$\hat{A}_1$};
\node[anchor=east] at (-0.10,-1.68) {$\hat{A}_2$};
\node[anchor=east] at (-0.10,-2.80) {$\hat{A}_3$};
\node[anchor=east] at (-0.10,-4.56) {$\hat{A}_{\hat{k}}$};

Left brace for witness subcollections.
\draw[decorate, decoration={brace, amplitude=5pt, mirror}]
  (-0.7,\yone) -- (-0.7,\yfive)
  node[midway, left=6pt, braceLab]
  {zeros off-diagonal};

\node[lab] at (1.025,-0.56)
  {$|\hat{C}_1|_{\hat{A}_1}|$\\ \\ [-2pt]$\geq 2^{0.3|A_1|}$};
\node[lab] at (3.075,-1.68)
  {$|\hat{C}_2|_{\hat{A}_2}|$\\ \\ [-2pt]$\geq 2^{0.3|A_2|}$};
\node[lab] at (5.125,-2.80)
  {$|\hat{C}_3|_{\hat{A}_3}|$\\ \\ [-2pt]$\geq 2^{0.3|A_3|}$};
\node[lab] at (7.675,-4.56)
  {$|\hat{C}_{\hat{k}}|_{\hat{A}_{\hat{k}}}|$\\ \\ [-2pt]$\geq 2^{0.3|A_{\hat{k}}|}$};
\node[lab] at (1.025,-1.68)
  {$0000000000$\\ $0000000000$\\ $0000000000$\\$0000000000$};
\node[lab] at (1.025,-2.80)
  {$0000000000$\\ $0000000000$\\ $0000000000$\\$0000000000$};
\node[lab] at (1.025,-4.56)
  {$0000000000$\\ $0000000000$\\ $0000000000$\\$0000000000$};
\node[lab] at (3.075,-0.56)
  {$0000000000$\\ $0000000000$\\ $0000000000$\\$0000000000$};
\node[lab] at (3.075,-2.80)
  {$0000000000$\\ $0000000000$\\ $0000000000$\\$0000000000$};
\node[lab] at (3.075,-4.56)
  {$0000000000$\\ $0000000000$\\ $0000000000$\\$0000000000$};
\node[lab] at (5.125,-0.56)
  {$0000000000$\\ $0000000000$\\ $0000000000$\\$0000000000$};
\node[lab] at (5.125,-1.68)
  {$0000000000$\\ $0000000000$\\ $0000000000$\\$0000000000$};
\node[lab] at (5.125,-4.56)
  {$0000000000$\\ $0000000000$\\ $0000000000$\\$0000000000$};

\node[lab] at (7.575,-0.56)
  {$0000000000$\\ $0000000000$\\ $0000000000$\\$0000000000$};
\node[lab] at (7.575,-1.68)
  {$0000000000$\\ $0000000000$\\ $0000000000$\\$0000000000$};
\node[lab] at (7.575,-2.80)
  {$0000000000$\\ $0000000000$\\ $0000000000$\\$0000000000$};


\node[ell] at (6.45,-0.56) {$\cdots$};
\node[ell] at (6.45,-1.68) {$\cdots$};
\node[ell] at (6.45,-2.80) {$\cdots$};
\node[ell] at (1.025,-3.68) {$\vdots$};
\node[ell] at (3.075,-3.68) {$\vdots$};
\node[ell] at (5.125,-3.68) {$\vdots$};
\node[ell] at (6.45,-3.68) {$\ddots$};
\node[ell] at (7.575,-3.68) {$\vdots$};


\end{tikzpicture}
    \caption{The collection of subcodes after being made block diagonal.}
    \label{fig:figure5}
\end{figure}

\paragraph{One Final Sauer-Shelah}

Once the subcodes are in the form of \cref{fig:figure5}, we can now do one final invocation of Sauer-Shelah. In particular, because each $|(\hat{C}_i)|_{\hat{A}_i}| \geq 2^{0.3 |A_i|}$, this intuitively means that each block realizes many different $\zo$-patterns. \cref{lem:SSintro} implies that whenever this is the case, then there is a \emph{large} complete subcode inside this set!

So, for each such block $(\hat{C}_i)|_{\hat{A}_i}$, we invoke \cref{lem:SSintro}. This guarantees that we find $\widetilde{A}_i \subseteq \hat{A}_i$ such that $(\hat{C}_i)|_{\widetilde{A}_i} = \zo^{\widetilde{A}_i}$. Moreover, because $|(\hat{C}_i)|_{\hat{A}_i}|$ is large, we also know that $|\widetilde{A}_i| = \Omega \left (  \frac{|A_i|}{\log(m)} \right )$! So, we can replace each of the blocks along the diagonal of \cref{fig:figure5} with a complete subcode of nearly comparable size! The resulting collection is shown in \cref{fig:figure6}.

\begin{figure}[h]
    \centering
\begin{tikzpicture}[
  font=\small,
  diag/.style={fill=blue!14, draw=black, line width=0.45pt},
  off/.style={fill=gray!9},
  grid/.style={draw=black, line width=0.45pt},
  lab/.style={font=\tiny, align=center},
  star/.style={font=\large, text=gray!55},
  ell/.style={font=\Large, text=gray!70},
  braceLab/.style={font=\scriptsize, align=center}
]

\coordinate (O) at (0,0);
\def\xzero{0}
\def\xone{2.05}
\def\xtwo{4.10}
\def\xthree{6.15}
\def\xfour{6.75}
\def\xfive{8.80}

\def\yzero{0}
\def\yone{-1.12}
\def\ytwo{-2.24}
\def\ythree{-3.36}
\def\yfour{-4.00}
\def\yfive{-5.12}


\fill[off] (\xzero,\yzero) rectangle (\xfive,\yfive);
\fill[diag] (\xzero,\yzero) rectangle (\xone,\yone);
\fill[diag] (\xone,\yone) rectangle (\xtwo,\ytwo);
\fill[diag] (\xtwo,\ytwo) rectangle (\xthree,\ythree);
\fill[diag] (\xfour,\yfour) rectangle (\xfive,\yfive);

\draw[grid, line width=0.65pt] (\xzero,\yzero) rectangle (\xfive,\yfive);
\foreach \x in {\xone,\xtwo,\xthree,\xfour} {
  \draw[grid] (\x,\yzero) -- (\x,\yfive);
}
\foreach \y in {\yone,\ytwo,\ythree,\yfour} {
  \draw[grid] (\xzero,\y) -- (\xfive,\y);
}

\node[align=center] at (1.025,0.38) {$\hat{C}_1$};
\node[align=center] at (3.075,0.38) {$\hat{C}_2$};
\node[align=center] at (5.125,0.38) {$\hat{C}_3$};
\node[align=center] at (7.375,0.38) {$\hat{C}_{\hat{k}}$};


\node[anchor=east] at (-0.10,-0.56) {$\widetilde{A}_1$};
\node[anchor=east] at (-0.10,-1.68) {$\widetilde{A}_2$};
\node[anchor=east] at (-0.10,-2.80) {$\widetilde{A}_3$};
\node[anchor=east] at (-0.10,-4.56) {$\widetilde{A}_{\hat{k}}$};


\node[lab] at (1.025,-0.56)
  {$|\hat{C}_1|_{\widetilde{A}_1}|$\\ \\ [-2pt]$= \zo^{\widetilde{A}_1}$};
\node[lab] at (3.075,-1.68)
  {$|\hat{C}_2|_{\widetilde{A}_2}|$\\ \\ [-2pt]$= \zo^{\widetilde{A}_2}$};
\node[lab] at (5.125,-2.80)
  {$|\hat{C}_3|_{\widetilde{A}_3}|$\\ \\ [-2pt]$= \zo^{\widetilde{A}_3}$};
\node[lab] at (7.675,-4.56)
  {$|\hat{C}_{\hat{k}}|_{\widetilde{A}_{\hat{k}}}|$\\ \\ [-2pt]$= \zo^{\widetilde{A}_{\hat{k}}}$};
\node[lab] at (1.025,-1.68)
  {$0000000000$\\ $0000000000$\\ $0000000000$\\$0000000000$};
\node[lab] at (1.025,-2.80)
  {$0000000000$\\ $0000000000$\\ $0000000000$\\$0000000000$};
\node[lab] at (1.025,-4.56)
  {$0000000000$\\ $0000000000$\\ $0000000000$\\$0000000000$};
\node[lab] at (3.075,-0.56)
  {$0000000000$\\ $0000000000$\\ $0000000000$\\$0000000000$};
\node[lab] at (3.075,-2.80)
  {$0000000000$\\ $0000000000$\\ $0000000000$\\$0000000000$};
\node[lab] at (3.075,-4.56)
  {$0000000000$\\ $0000000000$\\ $0000000000$\\$0000000000$};
\node[lab] at (5.125,-0.56)
  {$0000000000$\\ $0000000000$\\ $0000000000$\\$0000000000$};
\node[lab] at (5.125,-1.68)
  {$0000000000$\\ $0000000000$\\ $0000000000$\\$0000000000$};
\node[lab] at (5.125,-4.56)
  {$0000000000$\\ $0000000000$\\ $0000000000$\\$0000000000$};

\node[lab] at (7.575,-0.56)
  {$0000000000$\\ $0000000000$\\ $0000000000$\\$0000000000$};
\node[lab] at (7.575,-1.68)
  {$0000000000$\\ $0000000000$\\ $0000000000$\\$0000000000$};
\node[lab] at (7.575,-2.80)
  {$0000000000$\\ $0000000000$\\ $0000000000$\\$0000000000$};


\node[ell] at (6.45,-0.56) {$\cdots$};
\node[ell] at (6.45,-1.68) {$\cdots$};
\node[ell] at (6.45,-2.80) {$\cdots$};
\node[ell] at (1.025,-3.68) {$\vdots$};
\node[ell] at (3.075,-3.68) {$\vdots$};
\node[ell] at (5.125,-3.68) {$\vdots$};
\node[ell] at (6.45,-3.68) {$\ddots$};
\node[ell] at (7.575,-3.68) {$\vdots$};


\end{tikzpicture}
    \caption{The collection of subcodes after the final invocation of Sauer-Shelah.}
    \label{fig:figure6}
\end{figure}

We can now simply read off a diagonal matrix from this collection. The size of the resulting diagonal matrix will be \[
\sum_{i \in [\hat{k}]} |\widetilde{A}_i| \geq \hat{k} \cdot |\widetilde{A}_1| = \Omega \left ( \frac{k \cdot |A_i| \eps^2}{\log^4(m)} \right ) = \Omega \left ( \frac{|S| \eps^2}{\log^4(m)} \right ),
\]
where the final equality uses the fact that $k \cdot |A_i| \geq |S|$. Hence, $\NRD(C) \geq \Omega \left ( \frac{|S| \eps^2}{\log^4(m)} \right )$, and so $|S| = O \left ( \frac{\NRD(C)\log^4(m)}{\eps^2} \right )$. This concludes the proof of \cref{thm:decompositionTheoremIntro}.

\subsubsection{Going Beyond $\zo$-Valued Codes to Codes with Bounded Aspect Ratio}

There is one important manner in which the proof strategy above is actually \emph{stronger} than that of Brakensiek and Guruswami \cite{brakensiek2025redundancy}. Namely, rather than simply returning an identity matrix, the actual witness that is uncovered is of the form of \cref{fig:figure6}, where each block is a \emph{complete} sub-code of large size (proportional to the parameter $d$). This turns out to be the key strengthening that is required when dealing with codes beyond $\zo$-values. 

To illustrate this, let us consider a code $C \subseteq \{0,1,2\}^m$, and let us suppose that every codeword $c \in C$ has Hamming weight in the range $[d, 2d]$ for some integer $d$. It is not hard to show that, just as \cref{thm:decompositionTheoremIntro} suffices for building sparsifiers in the $\zo$-valued setting, the key for building sparsifiers in the $\{0,1,2\}$-valued setting is to consider \emph{the smallest set} $S \subseteq [m]$ such that $|C|_{\bar{S}}| \leq 2^{\eps^2 d / 10000}$.

From such a starting point, we can now repeat the procedure outlined in \cref{sec:techOverviewZeroOne}: we initialize a set $R = \emptyset \subseteq [m]$. Because $S$ is the \emph{smallest possible set} such that $|C_{\bar{S}}| \leq 2^{\eps^2 d / 10000}$, it will also be the case that $|C_{\bar{R}}| >2^{\eps^2 d / 10000} $. In \cref{sec:techOverviewZeroOne}, we then invoke \cref{lem:SSintro}, which shows that there is some subset $A_1 \subseteq [m] - R$ of large size such that $C|_{A_1} = \zo^{A_1}$. Unfortunately, this is a statement specialized to $\zo$-valued codes. In this more general setting, the appropriate analog turns out to be the generalization of Sauer-Shelah to \emph{Natarajan dimension}:

\begin{lemma}[Sauer-Shelah Lemma for Discrete Alphabets \cite{shalev2014understanding}.]\label{lem:NatarajanIntro}
    Let $C \subseteq \{0,1,2\}^m$. Then, there exists a set $A \subseteq [m]$, $|A| = \Omega\left (\frac{\log(|C|)}{\log(m)} \right )$ along with two symbols $a \neq b \in \{0,1,2\}$ such that $\{a, b\}^A \subseteq C|_{A}$.
\end{lemma}

With this in hand, we can pick up where we were; \cref{lem:NatarajanIntro} shows that there is some subset $A_1 \subseteq [m] - R$ of size $\Omega\left (\frac{\eps^2 d}{\log(m)} \right )$ along with symbols $a_1, b_1$ such that $\{a_1, b_1\}^{A_1} \subseteq C|_{A_1}$. 

After we find this subcode, we will now say that these coordinates in $A_1$ are effectively ``used up'', and so we add them to the set $R$. Now, again assuming it is still true that $|R| < |S|$, we repeat this observation, peeling off another set $A_2$, $|A_2| \approx \frac{\eps^2 d}{\log(m)}$ such that $\{a_2, b_2\}^{A_2} \subseteq C|_{A_2} $ for some choice of symbols $a_2, b_2$. Once again, we then update $R$ and repeat this sequence of operations until $|R| \geq |S|$. The resulting set of subcodes then looks like \cref{fig:figure7}.

\begin{figure}[h]
    \centering
\begin{tikzpicture}[
  font=\small,
  diag/.style={fill=blue!14, draw=black, line width=0.45pt},
  off/.style={fill=gray!9},
  grid/.style={draw=black, line width=0.45pt},
  lab/.style={font=\tiny, align=center},
  star/.style={font=\large, text=gray!55},
  ell/.style={font=\Large, text=gray!70},
  braceLab/.style={font=\scriptsize, align=center}
]

\coordinate (O) at (0,0);
\def\xzero{0}
\def\xone{2.05}
\def\xtwo{4.10}
\def\xthree{6.15}
\def\xfour{6.75}
\def\xfive{8.80}

\def\yzero{0}
\def\yone{-1.12}
\def\ytwo{-2.24}
\def\ythree{-3.36}
\def\yfour{-4.00}
\def\yfive{-5.12}


\fill[off] (\xzero,\yzero) rectangle (\xfive,\yfive);
\fill[diag] (\xzero,\yzero) rectangle (\xone,\yone);
\fill[diag] (\xone,\yone) rectangle (\xtwo,\ytwo);
\fill[diag] (\xtwo,\ytwo) rectangle (\xthree,\ythree);
\fill[diag] (\xfour,\yfour) rectangle (\xfive,\yfive);

\draw[grid, line width=0.65pt] (\xzero,\yzero) rectangle (\xfive,\yfive);
\foreach \x in {\xone,\xtwo,\xthree,\xfour} {
  \draw[grid] (\x,\yzero) -- (\x,\yfive);
}
\foreach \y in {\yone,\ytwo,\ythree,\yfour} {
  \draw[grid] (\xzero,\y) -- (\xfive,\y);
}

\node[align=center] at (1.025,0.38) {$C_1$};
\node[align=center] at (3.075,0.38) {$C_2$};
\node[align=center] at (5.125,0.38) {$C_3$};
\node[align=center] at (7.375,0.38) {$C_k$};


\node[anchor=east] at (-0.10,-0.56) {$A_1$};
\node[anchor=east] at (-0.10,-1.68) {$A_2$};
\node[anchor=east] at (-0.10,-2.80) {$A_3$};
\node[anchor=east] at (-0.10,-4.56) {$A_k$};

Left brace for witness subcollections.
\draw[decorate, decoration={brace, amplitude=5pt, mirror}]
  (-0.62,\yone) -- (-0.62,\yfive)
  node[midway, left=6pt, braceLab]
  {off-block weight $\leq d$};

\node[lab] at (1.025,-0.56)
  {$C_1|_{A_1}$\\[-2pt]$=\{1,2\}^{A_1}$\\[-1pt]{\tiny all patterns}};
\node[lab] at (3.075,-1.68)
  {$C_2|_{A_2}$\\[-2pt]$=\{1,2\}^{A_2}$\\[-1pt]{\tiny all patterns}};
\node[lab] at (5.125,-2.80)
  {$C_3|_{A_3}$\\[-2pt]$=\{0,2\}^{A_3}$\\[-1pt]{\tiny all patterns}};
\node[lab] at (7.675,-4.56)
  {$C_k|_{A_k}$\\[-2pt]$=\{0, 1\}^{A_k}$\\[-1pt]{\tiny all patterns}};

\foreach \x/\y in {
  3.075/-0.56,5.125/-0.56,7.575/-0.56,
  1.025/-1.68,5.125/-1.68,7.575/-1.68,
  1.025/-2.80,3.075/-2.80,7.575/-2.80,
  1.025/-4.56,3.075/-4.56,5.125/-4.56}
  {\node[star] at (\x,\y) {$\ast$};}

\node[ell] at (6.45,-0.56) {$\cdots$};
\node[ell] at (6.45,-1.68) {$\cdots$};
\node[ell] at (6.45,-2.80) {$\cdots$};
\node[ell] at (1.025,-3.68) {$\vdots$};
\node[ell] at (3.075,-3.68) {$\vdots$};
\node[ell] at (5.125,-3.68) {$\vdots$};
\node[ell] at (6.45,-3.68) {$\ddots$};
\node[ell] at (7.575,-3.68) {$\vdots$};


\end{tikzpicture}
    \caption{The collection of initial subcodes in the multi-valued setting.}
    \label{fig:figure7}
\end{figure}

Note that we can then perform the exact same set of simplifying operations that was performed in \cref{sec:techOverviewZeroOne}. This results in a collection of codes $\widetilde{C}_1, \dots \widetilde{C}_{\widetilde{k}} \subseteq C$, along with $\widetilde{A}_1 \subseteq A_1, \dots \widetilde{A}_{\widetilde{k}} \subseteq A_{\widetilde{k}}$ such that:
\begin{enumerate}
    \item $\sum_{i \in [\widetilde{k}]} |\widetilde{A}_i| = \Omega \left ( \frac{|S| \eps^2}{\log^4(m)} \right )$.
    \item For each $i \in [\widetilde{k}]$, there are symbols $a \neq b$ such that $(\widetilde{C}_i)|_{\widetilde{A}_i} = \{a,b\}^{\widetilde{A}_i}$.
    \item For each $c \in \widetilde{C}_i$, $c|_{\widetilde{A}_{\neq i}} = 0$.
\end{enumerate}

In particular, this then means that we find a $\mathrm{BACNRD}$ witness of size $\Omega \left ( \frac{|S| \eps^2}{\log^4(m)} \right )$, or equivalently, that $|S| \leq \widetilde{O}(\mathrm{BACNRD}(C) / \eps^2)$. 

To recap, this means that we can peel off a set of $ \widetilde{O}(\mathrm{BACNRD}(C) / \eps^2)$ many coordinates from our starting code $C$ such that, after peeling, we can afford to sample at rate $1/2$ while still producing a sparsifier. Carefully repeating this procedure then yields sparsifiers of total size $\widetilde{O}(\mathrm{BACNRD}(C) / \eps^2)$. This effectively yields the characterization of \cref{thm:continuousRVNRDintro}, at least in the case where the input code has discretely separated values. 

\subsubsection{Going to Continuous Values With Bounded Aspect Ratio}\label{sec:techOverviewBAC}

Generalizing to continuous codes is more subtle than simply the discrete alphabet setting as we are no longer guaranteed that there is \emph{any} choice of a set $S$ such that $|C_{\bar{S}}| \leq 2^{\eps^2 d / 10000}$. Indeed, in the continuous setting, the number of codewords can even be infinite. It turns out that rather than studying the number of codewords in the code, the right notion to work with is the \emph{size of the smallest $\ell_\infty$-cover} of a subcode. To understand these notions, for the rest of this section we focus on codes $C \subseteq \left (\{0 \} \cup [1, 2] \right )^m$.

\begin{definition}
  Let $C \subseteq \left (\{0 \} \cup [1, 2] \right )^m$ and $\eps>0$. An \emph{$\eps$-cover} is a set $C'\subseteq C \subseteq \left (\{0 \} \cup [1, 2] \right )^m$ such that for every $c\in C$ there exists $c'\in C'$ such that, for all $i\in[m]$,
    \[
        c_i \in [c'_i \pm \eps]
    \]
    Let $|\mathrm{Cover}(C,\eps)|$ denote the size of the smallest such cover.\footnote{When $C$ is $\{0\}$ (or empty), we set $|\mathrm{Cover}(C,\eps)|=0$.}
\end{definition}

Thus, the appropriate analog of the smallest set $S$ which reduces the number of codewords in $C$ is to peel the smallest set which \emph{reduces the size} of the smallest cover on the resulting code; i.e., the smallest set $S$ such that $|\mathrm{Cover}(C|_{\bar{S}}, \eps)| \leq 2^{\eps^2 d / 10000}$. Just as before, where having a small number of codewords suffices for random sampling to build sparsifiers, having a small cover also turns out to be sufficient.

However, just as above, we then also need an analog of Sauer-Shelah in this setting. If it is the case that a code $C$ satisfies $|\mathrm{Cover}(C|_{\bar{S}}, \eps)| \geq 2^{\eps^2 d / 10000}$, can we find a large, complete subcode? The answer again turns out to be yes, using a result of Alon, Ben-David, Cesa-Bianchi, and Haussler \cite{AlonBCH97} which relates cover sizes to fat-shattering dimension:

\begin{definition}
    For a code $C \subseteq \left (\{0 \} \cup [1, 2] \right )^m$, the fat-shattering dimension of $C$ with parameter $\eps > 0$, denoted by $\mathrm{fdim}(C, \eps)$ is the largest $\ell \leq m$ such that there exists $A \subseteq [m], |A| = \ell$ along with a vector $\gamma \in [0,2]^A$, such that for every $B \subseteq A$, there is a vector $v \in C$ such that for $i \in B$, $v_i \geq \gamma_i + \eps$, and for $i \in A - B$, $v_i \leq \gamma_i - \eps$.
\end{definition}

With this definition, we can then introduce the appropriate analog of Sauer-Shelah:

\begin{lemma}[Lemma 3.5 of \cite{AlonBCH97}]\label{lem:coverCountingBoundIntro}
    Let $C \subseteq [0,2]^m$ be a code, let $\eps > 0$, and let $d = \mathrm{fdim}(C, \eps/4)$. Then, 
    \[
    |\mathrm{Cover}(C, \eps)| \leq 2 \cdot \left ( \frac{16m}{\eps^2} \right )^{d\log(4em/d \eps)}.
    \]
\end{lemma}

In particular, when $|\mathrm{Cover}(C, \eps)|$ is large, this immediately implies that $\mathrm{fdim}(C, \eps/4)$ is large as well! Plugging this into the same framework as discussed in \cref{sec:techOverviewZeroOne} then essentially leads to the statement of \cref{thm:continuousRVNRDintro}.

\subsubsection{Going to Unbounded Aspect Ratios}

Going from codes $C \subseteq (\{0\} \cup [1,2])^m$ to $C  \subseteq \R_{\geq 0}^m$ adds several more layers of subtlety to designing sparsifiers. This is reflected by the more careful definition of $\CVNRD$, as \cref{def:CVNRDparameterIntro} already deviates from the definition we used in \cref{sec:techOverviewBAC}. \cref{def:CVNRDparameterIntro} does not keep the ``perfectly block diagonal'' structure that we saw in the $\zo$-setting and the bounded aspect ratio setting, and this is inherent: E.g., if one takes the resulting structure from \cref{fig:figure7} and replaces every $0$ with a small value, say $1/m^2$, then the sparsifiability of the code is entirely unchanged. However, there is now no way to ensure that the off-diagonal entries are in fact identically $0$.

The reasons for this modification are several-fold, as it turns out that characterizing real-valued sparsifiability with \emph{unbounded} aspect ratios is significantly more intricate than the bounded aspect ratio setting. We summarize some of the key barriers below:
\begin{enumerate}
    \item The first key subtlety is that when we try to follow the same path laid out in \cref{sec:techOverviewBAC}, our goal is to have a statement which says that if $|\mathrm{Cover}(C|_{\bar{R}}, \eps)| \geq 2^{\eps^2 d / 10000}$, then it must be the case that there is some large, ``complete'' shattering code which is contained inside of $C$. Unfortunately, such a statement is already not correct. The true formulation of \cref{lem:coverCountingBoundIntro} has a dependence in the exponent which looks like $d \log(2em \cdot \frac{c_{\mathrm{max}}}{c_{\min}})$, where $c_{\max}$ is the maximum weight symbol which ever appears in a codeword, and $c_{\min}$ is the smallest weight symbol which appears in a codeword. In this setting with unbounded aspect ratios, this ratio $\frac{c_{\mathrm{max}}}{c_{\min}}$ can already be too large, thereby worsening the resulting \emph{lower bounds} that we get for $d$ (and thus, the size of our shattering witnesses).
    \item Even worse, it is no longer clear that our target bound on the cover sizes is even useful.  Indeed, suppose we had a guarantee that $|\mathrm{Cover}(C|_{\bar{R}}, \eps)| < 2^{\eps^2 d / 10000}$. Our goal would now be to argue that random sampling at rate $1/2$ will preserve all codeword weights. But, this argument crucially relies on Chernoff-like concentration. Just as above, the concentration probabilities one gets from Chernoff \emph{also} have a dependence on the aspect ratio $\frac{c_{\mathrm{max}}}{c_{\min}}$. This effect is so strong that if  $\frac{c_{\mathrm{max}}}{c_{\min}} \approx m$, then we even lose the guarantee that \emph{a single codeword} has its weight preserved with high probability, much less that \emph{all} codewords have their weight preserved.
    \item Not only this, but in the discussion in \cref{sec:techOverviewZeroOne} and \cref{sec:techOverviewBAC}, we assumed that we were given a code $C$ where every codeword had approximately the same Hamming weight. Implicitly, we were using the fact that if codewords have the same Hamming weight, then they also (at least approximately) have the same \emph{actual weight}. Such a statement is no longer true in the $\R_{\geq 0}^m$ setting!
\end{enumerate}

In our ultimate proof, we thus proceed by carefully analyzing a sequence of carefully defined auxiliary codes which \emph{now do} have bounded aspect ratios. Roughly speaking, we focus on symbols in some range $[a, a(1 + \eps)]$, with the goal of saying that, for every codeword which has a significant portion of its weight coming from symbols in this range, the contribution coming from these symbols is preserved to a $(1 \pm \eps)$ factor when we do random sampling. This forms the foundation for producing our $\CVNRD$ sparsifiers, and alleviates the aforementioned issues: bounded aspect ratios remove the extra factors in  \cref{lem:coverCountingBoundIntro}, they improve the concentration when invoking Chernoff, and they allow for a simpler relationship between the Hamming weight of a codeword and its actual weight.

Unfortunately, this approach raises other issues: it is true that we can now prove relationships between the sparsifiability of our original code and the $\CVNRD$ of these \emph{auxiliary codes}. But can we guarantee that these $\CVNRD$ witnesses from the auxiliary codes can be translated into $\CVNRD$ witnesses of our \emph{original code}? By again relying on the careful definition of our auxiliary codes, it turns out that the answer is yes. This proof uses a more careful invocation of the off-diagonal weight reduction technique used in \cref{sec:techOverviewZeroOne} to reduce the weight of each codeword in its off-diagonal coordinates as much as possible, along with careful codeword weight counting which shows that if the $\CVNRD$ witnesses from the auxiliary codes are ever \emph{too small}, then they can be effectively ignored as barriers to sparsification.

There are still several more obstacles dealing with e.g., the ``tightness'' of the grouping of the coordinates in \cref{def:CVNRDparameterIntro} (see the $\chi$-parameter in the ``Complete Block Shattering'' condition), and how to formally prove that our sampling scheme produces sparsifiers, but we invite the interested reader to \cref{sec:continuousUnbounded} for the complete description.

\section{Preliminaries}

\subsection{Sparsification Definitions}

Central to our work is the notion of a \emph{sparsifier}:

\begin{definition}
    For $C\subseteq \R_{\geq 0}^m$ and real $\epsilon > 0$, a set of weights $w: [m] \rightarrow \R_{\geq 0}$ is said to be a $(1 \pm \eps)$-sparsifier of $C$ if, for every codeword $c \in C$:
    \[
    \langle w, c \rangle \in (1 \pm \eps) \cdot \left ( \sum_{i = 1}^m c_i\right ).
    \]
    The size of the sparsifier is the number of non-zero weights, i.e., $|\Supp(w)|$. We let $\mathrm{SPR}(C, \eps)$ denote the \emph{smallest size} of any $(1 \pm \eps)$ sparsifier of $C$.
\end{definition}

\begin{remark}
Note that throughout this work, when we deal with codes $C \subseteq \R_{\geq 0}^m$, we will assume that $C$ contains all \emph{scalings} of codewords $c \in C$. I.e., if $c \in C$ then $\lambda \cdot c \in C$, for any $\lambda > 0$. This is without loss of generality, as if $\langle w, c \rangle \in (1 \pm \eps) \cdot \left ( \sum_{i = 1}^m c_i\right )$, then $\langle w, \lambda \cdot c \rangle \in (1 \pm \eps) \cdot \left ( \sum_{i = 1}^m \lambda \cdot c_i\right ).$
\end{remark}

For a code $C$, we define its unweighted sparsifiability as:

\begin{definition}\label{def:unwtd-sparsification}
  For $C\subseteq \R_{\geq 0}^m$ and real $\epsilon > 0$, the \emph{unweighted sparsifiability} $\US(C,\eps)$ is the maximum over all sets $S \subseteq [m]$, of the minimum size $(1 \pm \eps)$ sparsifier of $C|_S = \{ c|_S: c \in C\}$. I.e.,
  \[
  \US(C,\eps) = \max_{S \subseteq [m]} \mathrm{SPR}(C|_S, \eps).
  \]
\end{definition}

In particular, note that in the above definition we take the maximum over all coordinate projections of the code $C$; in particular, there may exist a \emph{subset} of coordinates on which the code is very ``unsparsifiable'', even if the entire code as a whole is sparsifiable. 

We now also introduce the notion of \emph{random sampling sparsifiability}:

\begin{definition}\label{def:randomSparsifier}
      For $C\subseteq \R_{\geq 0}^m$ and real $\epsilon > 0$, a $(1 \pm \eps)$ random sparsifier of $C$ with parameter $\eps$ is a distribution $\Delta$ over functions $w: [m] \rightarrow \R_{\geq 0}$ such that for every $i \in [m], \E_{w \sim \Delta}w_i = 1$, and with probability $1 - 1 / m$ over $w \sim \Delta$, for every $c \in C$,
  \[
  \sum_{i = 1}^m w(i) \cdot c_i \in (1 \pm \eps)\sum_{i = 1}^m c_i.
  \]
  The size of the random sparsifier is $\E_{w \sim \Delta} |\Supp(w)|$, which measures the number of non-zero weights in the sampled $w$.  We let 
      \[
      \mathrm{RSPR}(C, \eps) = \min_{\Delta: \Delta \text{ is a } (1\pm \eps) \text{ random sparsifier}} \E_{w \sim \Delta} |\Supp(w)|
      \]
      denote the optimal size $(1 \pm \eps)$ random sparsifier of $C$.
\end{definition}

\begin{definition}\label{def:random-sparsification}
  For $C\subseteq \R_{\geq 0}^m$ and real $\epsilon > 0$, the \emph{random sparsifiability} $\mathrm{RS}(C,\eps)$ is the maximum over all sets $S \subseteq [m]$, of the minimum size $(1 \pm \eps)$ random sparsifier of $C$. I.e.,
  \[
  \mathrm{RS}(C,\eps) = \max_{S \subseteq [m]} \mathrm{RSPR}(C|_S, \eps).
  \]
\end{definition}

\subsection{Non-Redundancy in $\zo$-Valued Codes}

We first recall the notion of non-redundancy:

\begin{definition}\label{def:NRD}
    Given a code $C \subseteq \zo^m$, the \emph{non-redundancy} of $C$, denoted $\NRD(C)$ is the largest integer $\ell$ such that there are indices $a_1, \dots a_{\ell} \in [m]$ and codewords $c^{(1)}, \dots c^{(\ell)} \in C$ such that:
    \begin{enumerate}
        \item For all $i \in [\ell]$, $c^{(i)}_{a_i} = 1$.
        \item For all $i \neq j \in [\ell]$, $c^{(i)}_{a_j} = 0$.
    \end{enumerate}
\end{definition}

We will also frequently make use of the $\zo$-version of a $\R$-valued code:

\begin{definition}\label{def:hatCodes}
    Let $C \subseteq \R^m$, and let $\varphi: \R \rightarrow \zo$ be the map which sends $0 \rightarrow 0$, $\{\R - 0\} \rightarrow 1$. We let $\hat{C} = \{\varphi(c): c \in C \}$ be the version of $C$ where all non-zero values are replaced with $1$'s.

    For every $\hat{c} \in \hat{C}$, we also let $\Class(\hat{c}) = \{c \in C: \varphi(c) = \hat{c} \}$.
\end{definition}

Throughout the paper, for a code $C \subseteq \zo^m$, we will use the notation that $C_{[L, U]}$ contains all codewords $c \in C$ of Hamming weight in the range $[L, U]$. Likewise, we will use $C_{\leq U}$ to mean $C_{[0, U]}$. This is particularly useful when using the following theorem from \cite{brakensiek2025redundancy}:

\begin{theorem}[Theorem 4.16 in \cite{brakensiek2025redundancy}]\label{thm:NRDdecomposition}
    Let $C \subseteq \zo^m$, let $d \in \Z^+$ be an integer, and let $\lambda \in \R^{+}$. Then, there exists $I \subseteq [m]$ of size at most $2 \lambda \NRD(C) \log(4m)$ such that $C_{[d/2, d]}|_{\bar{I}}$ has at most $m \cdot 2^{3d \log^2(2m) / \lambda}$ many codewords.
\end{theorem}

Likewise, we will also use the following simple claim for bounding the support size of codes with small codeword weights:

\begin{claim}[Lemma 3.3 of \cite{brakensiek2025redundancy}]\label{clm:boundSupportSizeSmalld}
    For all $C \subseteq \zo^m$ and all $d \in \{0, 1, \dots m\}$,
    \[
    |\Supp(C_{\leq d})| \leq d \cdot \NRD(C).
    \]
\end{claim}

\subsection{VC Dimension, Sauer-Shelah, and Discrete Relatives}

We start by recapping several useful notions that we will require. The first is that of VC-dimension:

\begin{definition}\label{def:VCDimension}
    Let $C \subseteq \zo^m$ be an arbitrary code. The VC-dimension of $C$ (denoted by $\mathrm{VCDim}(C)$ is the largest integer $\ell$ such that there exists $A \subseteq [m]$ such that $|A| = \ell$ and $\zo^A \subseteq C|_A$.
\end{definition}

The works of \cite{sauer1972density, shelah1972combinatorial} showed the following theorem relating the size of a code with its VC-dimension:

\begin{theorem}[Sauer-Shelah \cite{sauer1972density}, \cite{shelah1972combinatorial}]\label{thm:sauerShelah}
    Let $C \subseteq \zo^m$ be an arbitrary code. Then,
    \[
    |C| \leq (m+1)^{\mathrm{VCDim}(C)}.
    \]
\end{theorem}

We will also frequently make use of an extension of VC-dimension to alphabets of size larger than $2$, using so-called Natarajan dimension:

\begin{definition}
    For a code $C \subseteq \{a_1, a_2, \dots a_k \}^m$, we say that the \emph{Natarajan dimension} of $C$, denoted by $\mathrm{NDim}(C)$ is the largest size of a set $T \subseteq [m]$ such that:
    \begin{enumerate}
        \item There exists $f_0, f_1: T \rightarrow \{a_1, a_2, \dots a_k \}$ such that $\forall i \in T: f_0(i) \neq f_1(i)$.
        \item $  \prod_{i \in T} \{ f_0(i), f_1(i)\} \subseteq C|_T $.
    \end{enumerate}
\end{definition}

Importantly, we have the following characterization of Natarajan dimension:

\begin{lemma}[Lemma 29.4 in \cite{shalev2014understanding}, originally due to \cite{natarajan1989learning}]\label{lem:NDimLB}
    Let $C \subseteq \{a_1, a_2, \dots a_k \}^m$. Then, 
    \[
    |C| \leq (mk^2)^{\mathrm{NDim}(C)}.
    \]

In particular, 
\[
\mathrm{NDim}(C) \geq \frac{\log|C|}{\log(mk^2)}.
\]
\end{lemma}

This also gives us the following corollary:

\begin{corollary}\label{cor:completeSubmatrixLB}
    Let $C \subseteq \{0, a_1, a_2, \dots a_k \}^m$. Then, there exists $T \subseteq [m]$, $|T| \geq \frac{\log|C|}{\log(mk^2) \cdot \binom{k}{2}}$, along with symbols $b_0 \neq b_1 \in \{0, a_1, a_2, \dots a_k \}$ such that 
    \[
    \{b_0, b_1\}^{|T|} \subseteq C|_{T}.
    \]
\end{corollary}

Now, we prove the following theorem:

\begin{theorem}\label{thm:specificWitness}
    Let $C \subseteq \{0, 1, 2, \dots k \}^m$ be such that for every $B \subseteq [m]$, there is $c \in C$ such that:
    \begin{enumerate}
        \item For $j \in B$: $c_j = 0$.
        \item For $j \in [m] - B$: $c_j \in \{1, 2, 3, \dots k \}$.
    \end{enumerate}
    Then, there is a choice of $A \subseteq [m]$ such that $|A| = \Omega \left ( \frac{m}{k^3 \cdot \log(mk^2)} \right )$ and a choice of $b \in [k]$ such that $\{0, b\}^{A} \subseteq C|_A$.
\end{theorem}

\begin{proof}
    We let $\delta > 0$ we a parameter that we select later. To start, let $C' \subseteq C$ be the set of all codewords of Hamming weight $\delta m$: i.e., $C' = \{ c \in C: \mathrm{Ham}(c) = \delta m\}$. Note that the number of such codewords is $|C'| \geq \binom{m}{\delta m} \geq 2^{H(\delta)m}$.

    Now, we define the notion of a \emph{filter} $\varphi: [m] \rightarrow \{0\} \times [k]$. In particular, for a fixed $\varphi$, we say that 
    \[
    \varphi(C') = \{c \in C': \forall i \in [m], c_i \in \varphi(i)\}.
    \]
    This is effectively a coordinate-wise constraint that in the $i$th coordinate, a codeword must take value $\in \varphi(i)$.

    Now, let us fix a codeword $c \in C'$, and select a filter \emph{uniformly} at random. We see then that 
    \[
    \Pr[c \in \varphi(C')] = \left ( \frac{1}{k}\right )^{\delta m},
    \]
    as it is both necessary and sufficient that, for every coordinate $i \in \Supp(c)$, $\varphi(i) = \{0, c_i \}$. Because $\varphi$ is chosen uniformly at random, this then occurs with probability $\left ( \frac{1}{k}\right )^{|\Supp(c)|} = \left ( \frac{1}{k}\right )^{\delta m}$.

    Thus, we see that 
    \[
     \E_{\varphi} \left [\left |\varphi(C') \right | \right] = |C'| \cdot \left ( \frac{1}{k}\right )^{\delta m} = 2^{H(\delta)m (1 - o(1))} \cdot 2^{-\delta m \log(k)}.
    \]
    In particular, there must \emph{exist} a choice of $\varphi$ for which 
    \[
    \left |\varphi(C') \right | \geq 2^{H(\delta)m} \cdot 2^{-\delta m \log(k)}.
    \]
    Now, setting $\delta = \frac{1}{2k}$, we see then that 
    \[
    \left |\varphi(C') \right | \geq 2^{\frac{\log(2k)}{2k} m} \cdot 2^{- \frac{\log(k)}{2k}m} = 2^{\frac{m}{2k}}.
    \]

    For this $\varphi(C')$, we then invoke \cref{cor:completeSubmatrixLB}. This guarantees us a set of coordinates $A \subseteq [m]$, 
    \[
    |A| \geq \frac{\log|\varphi(C')|}{\log(mk^2) \binom{k}{2}} \geq \frac{m}{2k \cdot \log(mk^2) \binom{k}{2}},
    \]
    along with symbols $b_0, b_1$ such that $\{b_0, b_1\}^A \subseteq \varphi(C')|_A$. The key fact now is that in each coordinate $i \in A$, we know that $\varphi(C')_i \in \{0, b\}$, for a \emph{single choice} of symbol $b \in \{1, 2, \dots k \}$. In particular, this means that it \emph{cannot} be the case that both of $b_0, b_1$ are non-zero symbols.

    This conclude the proof, as we have shown that for some symbol $b \neq 0$, we have $\{0, b\}^A \subseteq \varphi(C')|_A \subseteq C|_A$, with $|A| \geq \frac{m}{2k \cdot \log(mk^2) \binom{k}{2}}$.
\end{proof}

\subsection{Continuous Notions of VC-Dimension}

Because we eventually work with sparsifying \emph{continuous} codes $C \subseteq [0,1]^m$, we also need continuous analogs of VC-dimension. For this, we use of the notion of fat-shattering dimension and cover sizes:

\begin{definition}\label{def:fatShatteringDim}
    For a code $C \subseteq [1,k]^m$, the fat-shattering dimension of $C$ with parameter $\eps > 0$, denoted by $\mathrm{fdim}(C, \eps)$ is the largest $\ell \leq m$ such that there exists $A \subseteq [m], |A| = \ell$ along with a vector $\gamma \in [1,k]^A$, such that for every $B \subseteq A$, there is a vector $v \in C$ such that for $i \in B$, $v_i \geq \gamma_i + \eps$, and for $i \in A - B$, $v_i \leq \gamma_i - \eps$.
\end{definition}

\begin{definition}\label{def:cover}
    Let $C \subseteq \left ( \{0 \} \cup [1,k] \right)^m$ be a code. For a parameter $\eps > 0$, we say that $C' \subseteq \left ( \{0 \} \cup [1,k] \right)^m$ is an $\eps$-cover of $C$ if, for every $c \in C$ there exists a $c' \in C'$ such that for every $i \in [m]$, 
    \[
    c_i \in [c'_i \cdot (1 \pm \eps)].
    \]
    We use $|\mathrm{Cover}(C, \eps)|$ to denote the size of the \emph{smallest} $\eps$-cover of $C$. When $C$ contains only the all-zeros codeword (or no codewords), for simplicity we say the size of the smallest cover is $0$. 
\end{definition}

The reason that we will be interested in covers is that sparsifying a cover suffices for sparsifying the original code:

\begin{claim}\label{clm:coverSparsifier}
    Let $C \subseteq [0,1]^m$ be a code, let $\eps' > 0$, and let $C'$ be an $\eps'/4$-cover of $C$. Then, if a set of coordinates $S \subseteq [m]$ along with weights $\hat{w}_S$ is a $(1 \pm \eps'/4)$-sparsifier of $C'$, then $S, \hat{w}_S$ is also a $(1 \pm \eps')$ sparsifier of $C$.
\end{claim}

\begin{proof}
    Indeed, consider any codeword $c \in C$, and let $c'$ denote the corresponding covering codeword in $C'$. Then, we know that for every $i \in [m]$, $c_i \in (1 \pm \eps'/4) \cdot c'_i$. Thus, it must be the case that $\wt(c) \in (1 \pm \eps'/4) \wt(c')$, and that $\widetilde{\wt}(c|_S) \in (1 \pm \eps'/4) \widetilde{\wt}(c'|_S)$. All together then, we see that 
    \[
    \widetilde{\wt}(c|_S) \in (1 \pm \eps'/4) \cdot \widetilde{\wt}(c'|_S) \in (1 \pm \eps'/4)^2 \cdot \wt(c') \in (1 \pm \eps'/4)^3 \cdot \wt(c) \in (1 \pm \eps') \wt(c),
    \]
    as we desire.
\end{proof}

We will also frequently use the following  ``submodularity'' property of cover sizes:

\begin{claim}\label{clm:coverUnion}
    Let $C = C^{(1)} \cup C^{(2)} \cup \dots \cup C^{(p)}$, and let $\eps > 0$. Then, 
    \[
    |\mathrm{Cover}(C, \eps)| \leq \sum_{i = 1}^p |\mathrm{Cover}(C^{(i)}, \eps|.
    \]
\end{claim}

\begin{proof}
    This follows by taking the cover of $C$ to be the union of the covers of $C^{(1)}, \dots C^{(p)}$. Clearly, every codeword in $C$ is then covered.
\end{proof}

Importantly, there is the following relationship between fat-shattering dimension and cover sizes:

\begin{lemma}[Lemma 3.5 of \cite{AlonBCH97}]\label{lem:coverCountingBound}
    Let $C \subseteq [1,k]^m$ be a code, let $\eps > 0$, and let $d = \mathrm{fdim}(C, \eps/4)$. Then, 
    \[
    |\mathrm{Cover}(C, \eps)| \leq 2 \cdot \left ( \frac{4mk^2}{\eps^2} \right )^{d\log(2emk/d \eps)}.
    \]
\end{lemma}

\subsection{Concentration Bound}

We will frequently make use of the following Chernoff bound. 
\begin{claim}\label{clm:chernoffBound}[See, for instance, Lemma 2.1 in \cite{SY19}]
    Let $X_1, \dots X_{\ell}$ be random variables such that each $X_i \in [0, R]$, and let $S_{\ell} = \sum_{i = 1}^{\ell} X_i$. Then, for any $t > 0$, 
    \[
    \Pr \left [ |S_n - \E[S_n]| \geq t\right ] \leq 2e^{-\frac{2t^2}{\ell \cdot R^2}}.
    \]
\end{claim}

\section{Warming-Up: Characterizing $\zo^m$-Sparsifiability With Sauer-Shelah}\label{sec:zoManipulations}

In this section, we show how to (essentially) re-derive Theorem 4.16 of \cite{brakensiek2025redundancy} using a purely combinatorial argument (notably, without using Gilmer's entropy method). Instead, the only non-trivial tool we use is \cref{thm:sauerShelah}. To this end, we first re-state Theorem 4.16 of \cite{brakensiek2025redundancy}:\footnote{Note that we are not concerned with exactly achieving the same logarithmic factors - the restatement below has added a few such factors.}

\begin{theorem}[Theorem 4.16 of \cite{brakensiek2025redundancy}]\label{thm:brakensiekNRD}
    For any $C \subseteq \zo^m$, integer $d \geq 1$, and real $\lambda \geq 1$, there exists $I \subseteq [m]$ of size at most $\frac{400 \lambda \cdot \NRD(C)}{\log(2m)}$ such that $(C_{\leq d})_{\bar{I}}$ has at most $m \cdot 2^{100d \log^4(2m) / \lambda}$ many codewords. 
\end{theorem}

\subsection{Finding ``Dense'' Subcodes}

    To prove this, we define $S \subseteq [m]$ to be the \emph{smallest} set of coordinates such that $(C_{\leq d})_{\bar{S}}$ has at most $m \cdot 2^{100d \log^4(2m) / \lambda}$ many codewords. Our goal will be to show that if $S$ is large then we can in fact recover a witness to $\NRD$ which is proportionally large.

As our starting point, we make the following observation:

\begin{claim}\label{clm:DenseSets}
    For any $C \subseteq \zo^m$, integer $d \geq 1$, and real $\lambda \geq 1$, let $S \subseteq [m]$ to be the \emph{smallest} set of coordinates such that $(C_{\leq d})_{\bar{S}}$ has at most $m \cdot 2^{100d \log^4(2m) / \lambda}$ many codewords. Then, there exist disjoint $A_1, \dots A_{\ell} \subseteq [m]$ along with $C_1, \dots C_{\ell} \subseteq C_{\leq d}$ such that:
    \begin{enumerate}
        \item $|A_i| = \frac{100d\log^3(2m)}{\lambda}$.
        \item $(C_i)|_{A_i} = \zo^{A_i}$. Moreover, there is exactly one codeword in $C_i$ for each pattern in $A_i$.
        \item For all $c \in C_i$, $\wt(c) \leq d$.
        \item $\ell \cdot \frac{100d\log^3(2m)}{\lambda} \geq |S|$.
    \end{enumerate}
\end{claim}

\begin{proof}
    Let $S$ be defined as above. 
    In particular, this means that for any set $R \subseteq [m]$ such that $|R| < |S|$, it must be the case that $(C_{\leq d})_{\bar{R}}$ has $>m \cdot 2^{100d \log^4(2m) / \lambda}$ many codewords. Now, we invoke \cref{thm:sauerShelah}: this theorem implies that there exists a set $A \subseteq \bar{R}$ such that $|A| \geq \frac{\log(m \cdot 2^{100d \log^4(2m) / \lambda})}{\log(m+1)} \geq \frac{100d\log^3(2m)}{\lambda}$ and $\zo^A \subseteq (C_{\leq d})_A$.

    This motivates the following procedure: we start with $R_1 = \emptyset$, and initialize our counter to be $i = 1$. Now, while $|R_i| < |S|$, we know that there \emph{exists} a set $A_i \subseteq \bar{R_i}$ such that $\zo^{A_i} \subseteq (C_{\leq d})_{A_i}$ and $|A_i| \geq \frac{100d\log^3(2m)}{\lambda}$ (WLOG we can assume equality here by truncating the set $A_i$ if it is larger). We let $C_i \subseteq C_{\leq d}$ denote a minimal set of codewords such that $(C_i)|_{A_i} = \zo^{A_i}$, and then set $R_{i+1} = R_i \cup A_i$ (intuitively, $R_{i+1}$ is the set of coordinates which have already peeled by the $i+1$st iteration). 

    At the termination of this procedure, we have recovered a sequence of disjoint $A_1, \dots A_{\ell} \subseteq [m]$ along with $C_1, \dots C_{\ell} \subseteq C_{\leq d}$ such that:
    \begin{enumerate}
        \item $|A_i| = \frac{100d\log^3(2m)}{\lambda}$.
        \item $(C_i)|_{A_i} = \zo^{A_i}$.
        \item For all $c \in C_i$, $\wt(c) \leq d$.
        \item $|R_{\ell+1}| = \ell \cdot \frac{100d\log^3(2m)}{\lambda} \geq |S|$.
    \end{enumerate}

Note that by the minimality of $C_i$, it will also be the case that $|C_i| = 2^{|A_i|}$. I.e., that there is exactly one codeword in $C_i$ for each pattern of $0$'s and $1$'s in $A_i$.
\end{proof}

    Recall that ultimately, our goal is to bound the size of $S$ in terms of the non-redundancy of $C$. To do this, we will show that we can find a \emph{large} $\NRD$ witness inside this collection $A_1, \dots A_{\ell}, C_1, \dots C_{\ell}$. In the following sections, we prove this in a sequence of abstract steps, and then show how we can use this to conclude the theorem.

\subsection{Key Lemma Statement}

The key lemma we will prove is the following:

\begin{lemma}[Processing $\zo$-Valued Codes]\label{lem:generalProcessing}
    Let $C \subseteq \zo^m$, $A_1, \dots A_p \subseteq [m]$, $C_1, \dots C_p \subseteq C$, $d, \tau \geq 40$ be such that:
    \begin{enumerate}
    \item For every $i \in [p]$, $|A_i| = \tau$.
        \item For every $i \in [p]$, $(C_i)|_{A_i} = \zo^{A_i}$. 
        \item For every $i \in [p]$, for every $c \in C_i$, $\sum_{j \neq i \in [p]} \wt(c|_{A_j}) \leq d$.
    \end{enumerate}

    Then, for any $\eta \geq 1$ such that $\frac{10d}{\eta} \leq \frac{\tau}{100 \log(m)}$, there exists a set $\hat{T} \subseteq [p]$ along with $\hat{A}_i \subseteq A_i: i \in T'$, and $\hat{C}_i \subseteq C_i: i \in [p]$ such that:
    \begin{enumerate}
        \item For every $i \in \hat{T}$, $(\hat{C}_i)|_{\hat{A}_i} = \zo^{\hat{A}_i}$. 
        \item For every $i \in \hat{T}$, for every $c \in \hat{C}_i$, $\sum_{j \neq i \in \hat{T}} \wt(c|_{\hat{A}_j}) = 0$. 
        \item $|\hat{A}_i| \geq \frac{\tau}{10 \log(m)}$.
        \item $|\hat{T}| \geq \frac{p}{4 \eta}$.
    \end{enumerate}
\end{lemma}

To see why it is useful, we prove \cref{thm:brakensiekNRD} using only \cref{lem:generalProcessing} and \cref{clm:DenseSets}. 

\begin{proof}[Proof of \cref{thm:brakensiekNRD}.]
Note that if $\lambda \leq \log^3(2m)$, then the claim is trivial, as then $m \cdot 2^{100d\log^4(2m)/\lambda} \leq \binom{m}{ \leq d}$, and the claim holds vacuously. Thus, we may assume that $\lambda \geq \log^3(2m)$. We also assume WLOG that $|C_i| = 2^{|A_i|}$; i.e., that it is a minimal set which witnesses all $\zo$ patterns on $A_i$.

    In this case, we first invoke \cref{clm:DenseSets}. Given our code $C \subseteq \zo^m$, $d \geq 1$ and $\lambda \geq 1$, this returns $A_1, \dots A_{\ell} \subseteq [m]$ along with $C_1, \dots C_{\ell} \subseteq C_{\leq d}$ such that:
    \begin{enumerate}
        \item $|A_i| = \frac{100d\log^3(2m)}{\lambda} := \tau$.
        \item $(C_i)|_{A_i} = \zo^{A_i}$.
        \item For all $c \in C_i$, $\wt(c) \leq d$.
        \item $\sum_{i = 1}^{\ell} |A_i| = |R_{\ell+1}| = \ell \cdot \frac{100d\log^3(2m)}{\lambda} \geq |S|$.
    \end{enumerate}

    Note that if $\tau \leq 40$, then $\lambda \geq 10 d \log^3(2m)$. As in \cite{brakensiek2025redundancy}, there is a trivial procedure which peels off $\leq \NRD(C) \cdot d \leq \NRD(C) \cdot \lambda$ many coordinates which removes all codewords in $C_{\leq d}$. Thus, we may assume that $\tau \geq 40$.

    Then, we invoke \cref{lem:generalProcessing} on these subcodes with \[
    \eta =  \frac{1000 d \log(m)}{\tau} =  \frac{1000 d \log(m)}{100 d \log^3(2m) / \lambda}=  \frac{10\lambda \log(m)}{\log^3(2m)}.
    \]
    This returns 
    $\hat{T} \subseteq [\ell]$ along with $\hat{A}_i \subseteq A_i: i \in T'$, and $\hat{C}_i \subseteq C_i: i \in [p]$ such that:
    \begin{enumerate}
        \item For every $i \in \hat{T}$, $(\hat{C}_i)|_{\hat{A}_i} = \zo^{\hat{A}_i}$. 
        \item For every $i \in \hat{T}$, for every $c \in \hat{C}_i$, $\sum_{j \neq i \in \hat{T}} \wt(c|_{\hat{A}_j}) = 0$. 
        \item $|\hat{A}_i| \geq \frac{\tau}{10 \log(m)}$.
        \item $|\hat{T}| \geq \frac{\ell}{4 \eta}$.
    \end{enumerate}

    In particular, we can observe that this block-diagonal collection of subcodes \emph{contains} a diagonal sub-matrix of  size $\geq \frac{\ell}{4 \eta} \cdot \frac{\tau}{10 \log(m)} \geq \frac{|S|}{40 \eta \log(m)} \geq \frac{\log(2m)|S|}{400 \lambda }$.

    In particular, this then means that $\NRD(C) \geq \frac{\log(2m)|S|}{400 \lambda }$, and so $|S| \leq \frac{400 \lambda \cdot \NRD(C)}{\log(2m)}$.
\end{proof}

In the remainder of the section, we prove \cref{lem:generalProcessing}.

\subsection{Finding Nice Subcodes}

As our first step towards proving it, we rely on the following claim:

\begin{claim}\label{clm:ReduceOffDiagWeightGeneral}
    Let $A_1, \dots A_{\ell} \subseteq [m]$ along with $C_1, \dots C_{\ell} \subseteq \zo^m$ be such that:
    \begin{enumerate}
        \item $|A_i| = \tau \geq 40$.
        \item $(C_i)|_{A_i} = \zo^{A_i}$.
        \item For all $c \in C_i$, $\wt(c) \leq d$
    \end{enumerate}
    Then, for any $\eta \geq 1$, there exists $\hat{T} \subseteq [\ell]$, along with $C'_i \subseteq C_i: i \in \hat{T}$ such that:
    \begin{enumerate}
    \item $|\hat{T}| = \frac{\ell}{2 \eta}$.
    \item For $i \in \hat{T}$, $|A_i| = \tau$.
    \item For $i \in \hat{T}$, $\left | (C'_i)|_{A_i} \right | \geq 2^{9|A_i|/10}$.
    \item For $i \in \hat{T}$, $c \in C'_i$, $\sum_{j \neq i \in \hat{T}}\wt(c|_{A_j}) \leq \frac{10d}{\eta}$.
\end{enumerate}
\end{claim}

In words, this claim is showing that we can prune our collection of blocks (and slightly prune the number of codewords in each of the surviving blocks) such that now codewords have a bounded \emph{off-block} weight.

\begin{proof}
We assume WLOG that $|C_i| = 2^{|A_i|}$; i.e., that it is a minimal set which witnesses all $\zo$ patterns on $A_i$.

    Consider constructing $T$ by \emph{randomly sampling} exactly $\frac{\ell}{\eta}$ indices of $[\ell]$. Now, for an index $i \in T$, we say that $A_i$ is a \emph{good} block if, for $\geq \frac{1}{10}$ of the codewords $c \in C_i$, it is the case that $\sum_{j \neq i \in T} \wt(c|_{A_{j}}) \leq \frac{10d}{\eta}$. Observe that for a codeword $c \in C_i$, \[
\E_T \left [\sum_{j \neq i \in T} \wt(c|_{A_{j}}) \right ] \leq \wt(c) \cdot \frac{1}{\eta}.
\]
Because $\wt(c) \leq d$, for each codeword $c \in C_i$, we see that $\sum_{j \neq i \in T} \wt(c|_{A_{j}}) \leq \frac{10d}{\eta}$ with probability $\geq 9/10$ over $T$. So, this means that 
\[
\E_T \left [ \left | \left \{ c \in C_i: \sum_{j \neq i \in T} \wt(c|_{A_{j}}) \leq \frac{10d}{\eta} \right \} \right | \right ] \geq \frac{9}{10} \cdot |C_i|.
\]
By another Markov bound, this means that 
\[
\Pr_T\left [ \left | \left \{ c \in C_i: \sum_{j \neq i \in T}\wt(c|_{A_{j}}) \leq \frac{10d}{\eta} \right \} \right | \geq \frac{|C_i|}{10}\right ] \geq 1/2.
\]
In particular, this then means that $\E_{T}[ |\{i \in T: C_i \text{ is good}\} |] \geq \frac{|T|}{2} = \frac{\ell}{2 \eta}$, and so such a set $T$ must exist.

Letting $\hat{T} \subseteq T$ denote this subset of indices which is good, and letting $C'_i \subseteq C_i$ denote the set of good codewords, we then see that:
\begin{enumerate}
    \item $|\hat{T}| = \frac{\ell}{2 \eta}$.
    \item For $i \in \hat{T}$, $|A_i| = \tau$.
    \item For $i \in \hat{T}$, $\left | (C'_i)|_{A_i} \right | \geq \frac{2^{|A_i|}}{10} \geq 2^{9|A_i|/10}$.\footnote{Here, we are using that $\tau \geq 40$.}
    \item For $i \in \hat{T}$, $c \in C'_i$, $\sum_{j \neq i \in \hat{T}}\wt(c|_{A_j}) \leq \frac{10d}{\eta}$.
\end{enumerate}
\end{proof}

Now, we further process the above set of sub-codes into a form where the off-diagonal entries of each block exactly match:

\begin{claim}\label{clm:offDiagonalSharedGeneral}
Let $\hat{T} \subseteq [\ell]$, $A_i \subseteq [m], C_i \subseteq \zo^m: i \in \hat{T}$, $\eta, \tau$ be defined such that:
    \begin{enumerate}
    \item $|\hat{T}| = \frac{\ell}{2 \eta}$.
    \item For $i \in \hat{T}$, $|A_i| = \tau \geq 40$.
    \item For $i \in \hat{T}$, $\left | (C_i)|_{A_i} \right | \geq 2^{9|A_i|/10}$.
    \item For $i \in \hat{T}$, $c \in C_i$, $\sum_{j \neq i \in \hat{T}}\wt(c|_{A_j}) \leq \frac{10d}{\eta}$.
    \item $\frac{10d}{\eta} \leq \frac{|A_i|}{100 \log(m)}$.
\end{enumerate}

Then, there exists $\hat{C}_i \subseteq C_i: i \in \hat{T}$ such that:
   \begin{enumerate}
    \item $|\hat{T}| = \frac{\ell}{2 \eta}$.
    \item For $i \in \hat{T}$, $|A_i| = \tau$.
    \item For $i \in \hat{T}$, $\left | (\hat{C}_i)|_{A_i} \right | \geq 2^{|A_i|/2}$.
    \item For $i \in \hat{T}$, there exists $z \in \zo^{\bigcup_{j \neq i \in \hat{T}}A_j}$, $\wt(z) \leq \frac{10d}{\eta}$ such that for all $c \in \hat{C}_i$, $(c)|_{\bigcup_{j \neq i \in \hat{T}}A_j} = z$.
    \item $\frac{10d}{\eta} \leq \frac{|A_i|}{100 \log(m)}$.
\end{enumerate}
\end{claim}

\begin{proof}
    Let us fix our attention on a single block $A_i, C_i$ for $i \in \hat{T}$. Now, for a codeword $c \in C_i$, we look at $(c)|_{\bigcup_{j \neq i \in \hat{T}}A_j} \in \zo^{\bigcup_{j \neq i \in \hat{T}}A_j}$; namely, this is the projection of the codeword $c$ to the coordinates of $A_j: j \neq i, j \in \hat{T}$. Note that, by the supposed conditions, we know that $\wt \left ( (c)|_{\bigcup_{j \neq i \in \hat{T}}A_j}\right ) \leq \frac{10d}{\eta }\leq \frac{|A_i|}{100\log(m)}$. Thus, there are at most $\binom{m}{\leq\frac{|A_i|}{100\log(m)}} \leq 2^{|A_i|/100}$ many possibilities for this projected codeword $(c)|_{\bigcup_{j \neq i \in \hat{T}}A_j}$. Now, for $z \in \zo^{\bigcup_{j \neq i \in \hat{T}}A_j}$, we let $\mathrm{Class}(z, C_i) = \{ c \in C_i:(c)|_{\bigcup_{j \neq i \in \hat{T}}A_j} = z \}$. By pigeonhole principle, we know that there \emph{exists} a $z$ such that 
    \[
    |\mathrm{Class}(z, C_i)| \geq \frac{|C_i|}{2^{|A_i|/100}} \geq \frac{2^{9|A_i|/10}}{2^{|A_i|/100}}\geq2^{|A_i|/2},
    \]
    where we have used that $|C_i| \geq 2^{9|A_i|/10}$. Define $\hat{C}_i \subseteq C_i$ to be $\mathrm{Class}(z, C_i)$ for some $z$ such that $|\mathrm{Class}(z, C_i)| \geq 2^{|A_i|/2}$. This then concludes the claim.
    \end{proof}

\subsection{Making the Subcodes Block-Diagonal}

In this section, we provide a sequence of claims that shows that we can process any collection of codewords of the form of \cref{clm:offDiagonalSharedGeneral} into a collection of block-diagonal subcodes. Our first step in doing this is to turn our collection of codewords into an \emph{upper-triangular collection}. We make this formal below:

\begin{claim}\label{clm:processNiceToUpperTriangularGeneral}
    Let $C \subseteq \zo^m$ be a code, and let $\hat{T} \subseteq [\ell]$, $\eta, \tau$, 
    $A_i \subseteq [m], \hat{C}_i \subseteq C: i \in \hat{T}$ be defined as above such that:
   \begin{enumerate}
    \item $|\hat{T}| = \frac{\ell}{2 \eta}$.
    \item For $i \in \hat{T}$, $|A_i| = \tau \geq 10$.
    \item For $i \in \hat{T}$, $\left | (\hat{C}_i)|_{A_i} \right | \geq 2^{|A_i|/2}$.
    \item For $i \in \hat{T}$, there exists $z \in \zo^{\bigcup_{j \neq i \in \hat{T}}A_j}$, $\wt(z) \leq \frac{10d}{\eta}$ such that for all $c \in \hat{C}_i$, $(c)|_{\bigcup_{j \neq i \in \hat{T}}A_j} = z$.
    \item $\frac{10d}{\eta} \leq \frac{|A_i|}{100 \log(m)}$.
\end{enumerate}
Then, there exists $T' \subseteq \hat{T}$, $\widetilde{A}_i \subseteq A_i, {i \in T'}$ such that:
\begin{enumerate}
    \item $|T'| \geq 0.9 \cdot |\hat{T}| \geq \frac{0.9\ell}{2\eta}$.
    \item For $i \in T'$, $|\widetilde{A}_i| \geq \frac{9}{10} \cdot |A_i| \geq 0.9 \tau$.
    \item For $i \in T'$, $|(\hat{C}_i)|_{\widetilde{A}_i}| \geq 2^{0.4|A_i|}$.
    \item For every $i \in T'$, and every $c \in \hat{C}_i$, 
    \[
    \sum_{j > i \in T'} \wt(c|_{\widetilde{A}_i}) = 0.
    \]
    \item $\frac{10d}{\eta} \leq \frac{|A_i|}{100 \log(m)}$.
\end{enumerate}
\end{claim}

\begin{proof}
    We prove this via a greedy algorithm:

    \begin{algorithm}[H]
    \caption{ProcessNiceIntoUpperTriangular$(A_i, \hat{C}_i)_{i \in \hat{T}}$}
    Initialize $F = \emptyset$. \tcp{The set of rows being removed.}
	Initialize $T' = \emptyset$. \tcp{The set of retained indices.}
	\For{$i \in \hat{T}$}{
    \If{$\left | A_i \setminus F \right | \geq \frac{9|A_i|}{10}$}{
    $F \leftarrow F \cup (\supp(z^{(i)}) \cap \bigcup_{j > i, j \in \hat{T}} A_j)$.\label{line:restrictFutureincGeneral} \\
    $T' \leftarrow T' \cup \{i\}$. \\
    }
    }
    \Return${F, (A_i \setminus F, \hat{C}_i)_{i \in T'}}.$
    \end{algorithm}

We now make several observations about the above algorithm:
\begin{enumerate}
    \item For every $i \in T'$, we claim that $|A_i \setminus F| \geq \frac{9 |A_i|}{10}$ (using the final returned set $F$). This follows because when $i$ is added to $T'$, $|A_i \setminus F| \geq \frac{9 |A_i|}{10}$, and from this point, all future iterations can only add indices to $F$ that correspond to sets $A_{>i}$ (see \cref{line:restrictFutureincGeneral}).
    \item For every $i \in T'$, and every $c \in \hat{C}_i$, 
    \[
    \sum_{j > i \in T'} \wt(c|_{A_j \setminus F}) = 0.
    \]
    This follows by construction: after the index $i$ is added to $T'$, we also add the support of $z^{(i)} \cap A_j$ to $F$ for all $j > i \in \hat{T}$. For every codeword $c \in \hat{C}_i$, by our hypothesis, it was the case that $(c)|_{\bigcup_{j \neq i \in \hat{T}}A_j} = z^{(i)}$, and this support is exactly removed by the set $F$. 
    \item Now, we lower bound $|(\hat{C}_i)|_{A_i \setminus F}|$ for $i \in T'$. For this, we observe that by the above, $|A_i \setminus F| \geq \frac{9 |A_i|}{10}$, and hence \[
    |(\hat{C}_i)|_{A_i \setminus F}| \geq \frac{\left | (\hat{C}_i)|_{A_i} \right |}{2^{|A_i|/10}} \geq 2^{0.4 |A_i|}.
    \]
    Here, we have used that removing $|A_i|/10$ many coordinates can only reduce the number of codewords by a factor of $2^{|A_i|/10}$.
    \item Finally, we lower bound $|T'|$. For this, we let $t = |T' \cap \hat{T}|$. In particular, this means that $|\hat{T}| - t$ many indices are \emph{not} chosen to be in the set $T'$. But, for every such index $i$ which is not chosen to be in the set $T'$, it must be the case that $\left | A_i \setminus F \right | < \frac{9|A_i|}{10}$, or equivalently, that 
    \[
    |A_i \cap F| \geq \frac{|A_i|}{10}.
    \]
    Thus, it must be the case that 
    \[
    |F| \geq \frac{|A_i|}{10} \cdot (|\hat{T}| - t).
    \]
    But, at the same time, we know that 
    \[
    |F| \leq t \cdot \frac{10d}{\eta} \leq t \cdot \frac{|A_i|}{100 \log(m)},
    \]
    as the only case in which $F$ increments is when we add $\supp(z^{(i)})$ to $F$ after committing to keep a block.
    Together, this implies that 
    \[
    \frac{|A_i|}{10} \cdot (|\hat{T}| - t) \leq t \cdot \frac{|A_i|}{100 \log(m)},
    \]
    so 
    \[
    t \geq |\hat{T}| \cdot \left ( \frac{1/10}{1/10 + 1 / 100 \log(m)} \right ) \geq |\hat{T}| \cdot \frac{9}{10}.
    \]
\end{enumerate}

Defining $\widetilde{A}_i = A_i \setminus F$, we obtain our desired claim. 
\end{proof}

Finally, we show how to transform the set of upper-triangular subcodes into a truly block-diagonal set of subcodes. 

\begin{claim}\label{clm:processUTtoDiagGeneral}
    Let $C \subseteq \zo^m$ be a code, and let
    $ \widetilde{A}_i \subseteq A_i \subseteq [m], \hat{C}_i \subseteq C: i \in T', \ell, \eta$ be defined as above such that:
\begin{enumerate}
    \item $|T'| \geq \frac{0.9\ell}{2\eta}$.
    \item $|A_i| = \tau \geq 10$.
    \item For $i \in T'$, $|\widetilde{A}_i| \geq 0.9 \tau$.
    \item For $i \in T'$, $|(\hat{C}_i)|_{\widetilde{A}_i}| \geq 2^{0.4|A_i|}$.
    \item For $i \in T'$, there exists $z \in \zo^{\bigcup_{j \neq i \in T'}\widetilde{A}_j}$, $\wt(z) \leq \frac{10d}{\eta}$ such that for all $c \in \hat{C}_i$, $(c)|_{\bigcup_{j \neq i \in T'}\widetilde{A}_j} = z$.
    \item For every $i \in T'$, and every $c \in \hat{C}_i$, 
    \[
    \sum_{j > i \in T'} \wt(c|_{\widetilde{A}_i}) = 0.
    \]
    \item $\frac{10d}{\eta} \leq \frac{\tau}{100 \log(m)}$.
\end{enumerate}
Then, there exists $T'' \subseteq T'$, $A''_i \subseteq \widetilde{A}_i: i \in T''$ such that:

\begin{enumerate}
    \item $|T''| \geq \frac{0.7\ell}{2 \eta}$.
    \item For $i \in T''$, $|A''_i| \geq 0.8 \tau$.
    \item For $i \in T''$, $|(\hat{C}_i)|_{A''_i}| \geq 2^{0.3|A_i|}$.
    \item For every $i \in T''$, and every $c \in \hat{C}_i$, 
    \[
    \sum_{j \neq i \in T''} \wt(c|_{A''_j}) = 0.
    \]
\end{enumerate}
\end{claim}

\begin{proof}
    We re-use essentially the same algorithm as before, though we swap the order of the for-loop.

 \begin{algorithm}[H]
    \caption{ProcessUpperTriangularIntoDiagonal$(\widetilde{A}_i, \hat{C}_i)_{i \in T''}$}
    Initialize $F = \emptyset$. \tcp{The set of rows being removed.}
	Initialize $T'' = \emptyset$. \tcp{The set of retained indices.}
	\For{$i \in \mathrm{Reverse}(T')$}{
    \If{$\left | \widetilde{A}_i \setminus F \right | \geq \frac{9|\widetilde{A}_i|}{10}$}{
    $F \leftarrow F \cup (\supp(z^{(i)}) \cap \bigcup_{j< i, j \in T'} \widetilde{A}_j)$.\label{line:restrictFutureinc2General} \\
    $T'' \leftarrow T''\cup \{i\}$. \\
    }
    }
    \Return${F, (\widetilde{A}_i \setminus F, \hat{C}_i)_{i \in T''}}.$
    \end{algorithm}

We make mirroring observations about this algorithm:
\begin{enumerate}
    \item For every $i \in T''$, we claim that $|\widetilde{A}_i \setminus F| \geq \frac{9 |\widetilde{A}_i|}{10}$ (using the final returned set $F$). This follows because when $i$ is added to $T''$, $|\widetilde{A}_i \setminus F| \geq \frac{9 |\widetilde{A}_i|}{10}$, and from this point, all future iterations can only add indices to $F$ that correspond to sets $\widetilde{A}_{<i}$ (see \cref{line:restrictFutureinc2General}).
    \item For every $i \in T''$, and every $c \in \hat{C}_i$, 
    \[
    \sum_{j \neq i \in T''} \wt(c|_{\widetilde{A}_j \setminus F}) = 0.
    \]
    This follows by construction: after the index $i$ is added to $T''$, we also add the support of $z^{(i)} \cap \widetilde{A}_j$ to $F$ for all $j < i \in T'$. For every codeword $c \in \hat{C}_i$, by our hypothesis, it was the case that $(c)|_{\bigcup_{j \neq i \in \hat{T}}A_j} = z^{(i)}$, and this support is exactly removed by the set $F$. We also use the fact that, by our assumption, it is already the case that for $i \in T'$, and every $c \in \hat{C}_i$, 
    \[
    \sum_{j > i \in T'} \wt(c|_{\widetilde{A}_j}) = 0,
    \]
    and thus, so too $\sum_{j > i \in T'} \wt(c|_{\widetilde{A}_j \setminus F}) = 0$.
    Thus, 
    \[
    \sum_{j \neq i \in T''} \wt(c|_{\widetilde{A}_j \setminus F}) = \sum_{j > i \in T''} \wt(c|_{\widetilde{A}_j \setminus F}) + \sum_{j < i \in T''} \wt(c|_{\widetilde{A}_j \setminus F}) = 0.
    \]
    \item Now, we lower bound $|(\hat{C}_i)|_{\widetilde{A}_i \setminus F}|$ for $i \in T''$. For this, we observe that by the above, $|\widetilde{A}_i \setminus F| \geq \frac{9 |\widetilde{A}_i|}{10}$, and hence \[
    |(\hat{C}_i)|_{\widetilde{A}_i \setminus F}| \geq \frac{\left | (\hat{C}_i)|_{\widetilde{A}_i} \right |}{2^{|\widetilde{A}_i|/10}} \geq 2^{0.3 |A_i|}.
    \]
    \item Finally, we lower bound $|T''|$. For this, we let $t = |T'' \cap T'|$. In particular, this means that $|T'| - t$ many indices are \emph{not} chosen to be in the set $T''$. But, for every such index $i$ which is not chosen to be in the set $T''$, it must be the case that $\left | \widetilde{A}_i \setminus F \right | < \frac{9|\widetilde{A}_i|}{10}$, or equivalently, that 
    \[
    |\widetilde{A}_i \cap F| \geq \frac{|\widetilde{A}_i|}{10}.
    \]
    Thus, it must be the case that 
    \[
    |F| \geq \frac{|\widetilde{A}_i|}{10} \cdot (|T'| - t).
    \]
    But, at the same time, we know that 
    \[
    |F| \leq t \cdot \frac{10d}{\eta} \leq t \cdot \frac{\tau}{100 \log(m)} \leq t\cdot \frac{1.2 |\widetilde{A}_i|}{100 \log(m)}
    \]
    as the only case in which $F$ increases in size is when we add $\supp(z^{(i)})$ to $F$ after committing to keep a block.
    Together, this implies that 
    \[
    \frac{|\widetilde{A}_i|}{10} \cdot (|T'| - t) \leq t\cdot \frac{1.2 |\widetilde{A}_i|}{100 \log(m)},
    \]
    so 
    \[
    t \geq |T'| \cdot \left ( \frac{1/10}{1/10 + 1.2/100 \log(m)} \right ) \geq |T'| \cdot \frac{8}{10}.
    \]
\end{enumerate}

Defining $A''_i = \widetilde{A}_i \setminus F$, we obtain our desired claim. 
\end{proof}

\subsection{Concluding the Lemma}

Now, we use all of the above claims to conclude our desired lemma.

\begin{proof}[Proof of \cref{lem:generalProcessing}.]

Let $C$ be as given in the statement of \cref{lem:generalProcessing}. I.e., such that
$C \subseteq \zo^m$, $A_1, \dots A_p \subseteq [m]$, $C_1, \dots C_p \subseteq C$, $d, \tau$ and:
    \begin{enumerate}
    \item For every $i \in [p]$, $|A_i| = \tau$.
        \item For every $i \in [p]$, $(C_i)|_{A_i} = \zo^{A_i}$.
        \item For every $i \in [p]$, for every $c \in C_i$, $\sum_{j \in [p]} \wt(c|_{A_j}) \leq d$.
    \end{enumerate}

As our first step, we invoke \cref{clm:ReduceOffDiagWeightGeneral} with the prescribed choice of $\eta$. This returns $T' \subseteq [\ell]$, along with $C'_i \subseteq C_i$ for $i \in T'$ such that:

\begin{enumerate}
    \item $|T'| = \frac{p}{2 \eta}$.
    \item For $i \in T'$, $|A_i| = \tau$.
    \item For $i \in T'$, $\left | (C'_i)|_{A_i} \right | \geq 2^{9|A_i|/10}$.
    \item For $i \in T'$, $c \in C'_i$, $\sum_{j \neq i \in \hat{T}}\wt(c|_{A_j}) \leq \frac{10d}{\eta}$.
\end{enumerate}

Next, we invoke \cref{clm:offDiagonalSharedGeneral} on $T', A_i, C'_i$. Note that by our conditions on $\eta$, $\frac{10d}{\eta} \leq \frac{\tau}{100 \log(m)}$, and hence we satisfy the conditions of \cref{clm:offDiagonalSharedGeneral}. This returns $C''_i \subseteq C_i: i \in T'$ such that:
\begin{enumerate}
    \item $|T'| = \frac{p}{2 \eta}$.
    \item For $i \in T'$, $|A_i| = \tau$.
    \item For $i \in T'$, $\left | (C''_i)|_{A_i} \right | \geq 2^{|A_i|/2}$.
    \item For $i \in T'$, there exists $z \in \zo^{\bigcup_{j \neq i \in T'}A_j}$, $\wt(z) \leq \frac{10d}{\eta}$ such that for all $c \in C''_i$, $(c)|_{\bigcup_{j \neq i \in T'}A_j} = z$.
    \item $\frac{10d}{\eta} \leq \frac{\tau}{100 \log(m)}$.
\end{enumerate}

Next, we invoke \cref{clm:processNiceToUpperTriangularGeneral} on $T', A_i, C''_i$. This returns $T''' \subseteq T'$, $A'''_i \subseteq A_i: i \in T'''$ such that:

\begin{enumerate}
    \item $|T'''| \geq 0.9\cdot |T'| \geq \frac{0.9p}{2\eta}$.
    \item For $i \in T'''$, $|A'''_i| \geq \frac{9}{10} \cdot |A_i| \geq 0.9 \tau$.
    \item For $i \in T'''$, $|(C''_i)|_{A'''_i}| \geq 2^{0.4|A_i|}$.
    \item For every $i \in T'''$, and every $c \in C''_i$, 
    \[
    \sum_{j > i \in T'''} \wt(c|_{A'''_i}) = 0.
    \]
    \item $\frac{10d}{\eta} \leq \frac{|A_i|}{100 \log(m)}$.
\end{enumerate}

Then, we invoke \cref{clm:processUTtoDiagGeneral} on $T''', A'''_i, C''_i$. This returns $T'''' \subseteq T'''$, $A''''_i \subseteq A'''_i: i \in T''''$ such that: 

\begin{enumerate}
    \item $|T''''| \geq \frac{0.7p}{2 \eta}$.
    \item For $i \in T''''$, $|A''''_i| \geq 0.8 \tau$.
    \item For $i \in T''''$, $|(C''_i)|_{A''''_i}| \geq 2^{0.3|A_i|}$.
    \item For every $i \in T''''$, and every $c \in C''_i$, 
    \[
    \sum_{j \neq i \in T''''} \wt(c|_{A''''_j}) = 0.
    \]
\end{enumerate}

Finally, we invoke \cref{thm:sauerShelah} on each subcode $(C''_i)|_{A''''_i}: i \in T''''$. This returns $\hat{C}_i \subseteq C''_i$ and $\hat{A}_i \subseteq A''''_i$ such that $(\hat{C}_i)_{\hat{A_i}} = \zo^{\hat{A}_i}$, and $|\hat{A}_i| \geq \frac{0.3 |A_i|}{\log(m)}$. Because of condition 4 above, we know \[
    \sum_{j \neq i \in T''''} \wt(c|_{\hat{A}_j}) = 0.
    \]

Letting $\hat{T} \leftarrow T''''$, we obtain our desired claim. 
\end{proof}

\section{Characterizing Sparsifiability in Continuous Codes with Bounded Aspect Ratio}\label{sec:nonRedundancy}

In this section, we prove our sparsification theorem for unweighted \emph{continuous}-valued codes with bounded aspect ratio, and show its rough almost-equivalence to the real-valued non-redundancy. We will focus on codes $C \subseteq \left ( \{0\} \cup [1,k] \right )^m$. Recall that here we define the non-redundancy as follows:

\begin{definition}\label{def:continuousRVNRD}
    Let $C \subseteq \left ( \{0\} \cup [1,k] \right )^m$. The bounded-aspect continuous non-redundancy of $C$ with parameter $\eps$, denoted by $\BACNRD(C, \eps)$, is the maximum integer $\ell$ such that there exists $A_1, \dots A_p \subseteq [m]$ and $C_1, \dots C_p \subseteq C$, with $\sum_{i = 1}^p |A_i| = \ell$, and:
    \begin{enumerate}
        \item For every $i \in [p]$, if $|A_i| = 1$, then $(C_i)|_{A_i} \neq \{ 0\}$.
        \item  For every $i \in [p]$ if $|A_i| \geq 2$, then there is a choice of vector $\gamma \in [1,k]^{A_i}$, such that for every $B \subseteq A_i$, there is a codeword $c \in C_i$ such that for $b \in B$, $v_b \geq \gamma_b + \eps$, and for $b \in A_i - B$, $v_b \leq \gamma_b - \eps$.
        \item For every $j \neq i$, $(C_i)|_{A_j} = 0^{A_j}$.
    \end{enumerate}
\end{definition}

For this definition, we have the following characterization:

\begin{theorem}\label{thm:RVNRDcontinuous}
    Let $C \subseteq \left ( \{0\} \cup [1,k] \right )^m$ with $k = O(1)$. Then,
    \begin{enumerate}
        \item There exists $\eps'$ such that $\eps' = \Omega(\eps)$ and $\US(C, \eps') = \Omega(\eps \cdot \BACNRD(C, \eps))$. \label{item:continuousRVNRDLB}
        \item For $\eps' = \eps / 256\log(m)$, $\US(C, \eps) = \widetilde{O}(\BACNRD(C, \eps') / \eps^4)$.\label{item:continuousRVNRDUB}
    \end{enumerate}
\end{theorem}
\subsection{Basic Properties of $\BACNRD$}

\begin{remark}\label{rmk:RVNRDvsNRD}
    Immediately, we can observe that $\BACNRD(C) \geq \NRD(\hat{C})$, as any witness to a $\NRD$ lower bound in $\hat{C}$ is also a valid witness to $\BACNRD$.
\end{remark}

We also have the following convenient inheritance property:

\begin{remark}\label{rmk:inheritance}
    Let $C \subseteq \R^m$ be a code. Let $S \subseteq [m]$, and let $C' \subseteq C$. Then $\BACNRD(C'|_S) \leq \BACNRD(C)$. Essentially, the non-redundancy of the parent instance is always at least as that of any sub-instance. 
\end{remark}

\subsection{Lower Bound}

To start, we prove the lower bound of \cref{thm:RVNRDcontinuous}.

\begin{proof}[Proof of \cref{thm:RVNRDcontinuous}, \cref{item:continuousRVNRDLB}.]
    Indeed, let $C, \eps$ be given as in \cref{thm:RVNRDcontinuous}, and let us consider a single piece of the witness of $\BACNRD(C, \eps)$, which we denote (WLOG) by $C_1 \subseteq C$ and $A_1 \subseteq [m]$. Let $S \subseteq [m]$ denote a $(1 \pm \eps/4k)$ sparsifier of the $\BACNRD(C, \eps)$ witness, with corresponding weights $\widetilde{w}_S$. Our goal is to show that $|S \cap A_1| = \Omega( \eps \cdot |A_1|)$. Repeating the same argument for all $A_i: i \in [p]$ will then yield that 
    \[
    |S| = \sum_{i \in [p]} |A_i \cap S| = \sum_{i \in [p]} \Omega(\eps |A_i|) = \Omega(\eps \cdot \BACNRD(C, \eps)).
    \]

    So, let $\ell_1 = |A_1|$, and let us suppose for the sake of contradiction that $|S \cap A_1| \leq \frac{\eps \ell_1}{100k}$. Note that if $\ell_1 = 1$, then there is a codeword $v \in C_1$ such that $v_{A_1} \neq 0$, and otherwise $v = 0$. Not keeping a row from $A_1$ will immediately imply that the sparsifier gives weight $0$ to $v$, which is a contradiction. We focus the remainder of our attention on the case that $\ell_1 \geq 2$.
    
    By \cref{def:continuousRVNRD}, we know that there is a vector $\gamma \in [1,k]^{A_1}$ such that:
    \begin{enumerate}
        \item For $B = S \cap A_1$, there is a $v \in C_1$ such that for $b \in B$, $v_b \geq \gamma_b + \eps$, and for $b \in A_1 - B$, $v_b \leq \gamma_b - \eps$.
        \item For $B = A_1 - A_1 \cap S$, there is a $v' \in C_1$ such that for $b \in B$, $v'_b \geq \gamma_b + \eps$, and for $b \in A_1 - B$, $v'_b \leq \gamma_b - \eps$.
    \end{enumerate}

    Stated another way, there are vectors $v, v' \in C_1$ such that:
    \begin{enumerate}
        \item For $b \in S$, $v_b \geq \gamma_b + \eps$, and $v'_b \leq \gamma_b - \eps$.
        \item For $b \in A_1 - S$, $v_b \leq \gamma_b - \eps$ and $v'_b \geq \gamma_b + \eps$.
    \end{enumerate}

    As an immediate consequence, we claim that the sparsifier reports a weight for $v$ which is at least $(1 + \eps/k)$ times the weight of $v'$. Indeed, the weights reported by the sparsifier are $\widetilde{\wt}(v|_S) = \sum_{i \in S} v_i \cdot \widetilde{w}_i$ and $\widetilde{\wt}(v'|_S) = \sum_{i \in S} v'_i \cdot \widetilde{w}_i$. But, for every $i \in S$, we know that $v_i \geq \gamma_i + \eps$, while $v'_i \leq \gamma_i - \eps$. Thus, for $i \in S$
    \[
    \frac{v_i}{v'_i} \geq \frac{k}{k - 2\eps} \geq 1 + \eps/k,
    \]
    and so it must also be the case that $\widetilde{\wt}(v|_S) \geq (1 + \eps/k) \widetilde{\wt}(v'|_S)$ (note that here we are assuming that the sparsifier only assigns non-negative weights).

    However, we claim that it is also the case that the original codewords $v, v'$ satisfy $\wt(v) \leq \wt(v')$. Indeed, to see why, we have that 
    \[
    \wt(v) \leq |S\cap A_1| \cdot k + \sum_{i \in A_1 - S} \gamma_i - \eps
    \]
    and 
    \[
     \wt(v') \geq \sum_{i \in A_1 - S} \gamma_i + \eps,
    \]
where here we have used that the entire weight of $v, v'$ is contributed by coordinates in $A_1$ (thus we do not sum over coordinates from other parts). 
Thus, 
    \[
    \wt(v') - \wt(v) \geq \sum_{i \in A_1 - S} (\gamma_i + \eps) - |S\cap A_1| \cdot k - \sum_{i \in A_1 - S} (\gamma_i - \eps) \geq 2 \eps \cdot |A_1 - S| - |S\cap A_1| \cdot k
    \]
    \[
    \geq 2\eps (1 - \eps/100k) \cdot \ell_1 - k \cdot \eps \ell_1/100k \geq \eps \ell_1 - \eps \ell_1/100 \geq 0.
    \]

    However, this now yields a contradiction, as we assumed that $S$ constituted a $(1 \pm \eps/4k)$ sparsifier. But this would imply that $\wt(v') \leq (1 + \eps/4k) \widetilde{\wt}(v'|_S)$, and $\wt(v) \geq (1 - \eps/4k) \widetilde{\wt}(v|_S)$, and thus 
    \[
    \frac{\wt(v)}{\wt(v')} \geq \frac{(1 - \eps/4k) \widetilde{\wt}(v|_S)}{(1 + \eps/4k) \widetilde{\wt}(v'|_S)} \geq \frac{(1 - \eps/4k)}{(1 + \eps/4k)} \cdot (1 + \eps/k) > 1,
    \]
    which contradicts the fact that $\wt(v) \leq \wt(v')$. This concludes the proof. 
\end{proof}

Now, we proceed to a proof of the upper bound. 

\subsection{Upper Bound}

\subsubsection{Constructing the Sparsifier}

We first present the algorithm constructing the sparsifier and show that this indeed returns a $(1 \pm \eps)$ sparsifier of our starting code. Then, in the subsequent sections, we will prove an upper bound on the size of this sparsifier in terms of the $\BACNRD$.

For a code $C \subseteq \left ( \{0\} \cup [1,k] \right )^m$, we let $\hat{C}$ denote the $\zo$-version of the code, where every non-zero symbol gets mapped to $1$. To start, we let $d$ range over $\{1, 2, 4, 8, \dots m \}$. Now, for each choice of $d$, we consider the code $\hat{C}_{[d/2, d]}$. Because this is a code over $\zo^m$, we can then directly invoke \cref{thm:NRDdecomposition} to conclude that for any $\lambda$, there is a set $I_d$ of size $\leq 2 \lambda \NRD(C) \log(4m)$ such that $\hat{C}_{[d/2, d]}|_{\bar{I}}$ has $\leq m \cdot 2^{3d \log^2(2m) / \lambda}$ many codewords. 

For our use case, we set $\lambda = 10000 \cdot \log^4(m)k^2 /\eps'^2$, for $\eps' = \eps / 16\log(m)$, where $\eps$ is our ultimate desired accuracy for our sparsifier. Note that when $d$ is very small, $d \leq 10000k^2 \log^3(m) / \eps'^2$, we invoke \cref{clm:boundSupportSizeSmalld} to in fact remove the \emph{entire support} of $\hat{C}_{[d/2, d]}$ (i.e., $I_d = \hat{C}_{[d/2, d]}$). In this way, we can ensure that we do not ever have to worry about preserving the smallest weight codewords, as their supports are explicitly taken care of. 

For the other choices of $d$, after removing $I_1, \dots I_m$ from $\hat{C}$, our code then has sufficiently few codewords such that sampling the coordinates at rate $1/2$ would yield a $(1 \pm \eps')$ sparsifier of $\hat{C}$. The only remaining roadblock is that our goal is to sparsify $C$, \emph{not just} $\hat{C}$.

To this end, we also consider the codes $C_{[d/2, d]} = \{\Class(\hat{c}): \hat{c} \in \hat{C}_{[d/2, d]} \}$, where these codes are now defined over $(\{0\} \cup [1,k])^m$. For our final piece of notation, we let $S_{d} \subset [m]$ denote the \emph{smallest} set of coordinates such that \[
\left | \mathrm{Cover} \left ((C_{[d/2, d]})|_{\bar{S_d} \cap \bar{I_d}}, \eps'/4 \right )\right | \leq 2^{\eps'^2 d / 10000k^2 \log(m)},
\]
i.e., the smallest set of coordinates whose removal (in addition to the removal of $I_d$), reduces the size of the smallest cover of $C_{[d/2, d]}$ to be $\leq 2^{\eps'^2 d / 10000k^2\log(m)}$. 

We let $L_d = |S_d|$. Finally, we let $\widetilde{S} \subseteq [m]$ denote a random sample of all coordinates in $[m] - \bigcup_{d \in \{1, 2, 4, \dots m \}} (S_d \cup I_d)$ with rate $1/2$. Our key claim is the following:

\begin{claim}\label{clm:singleRoundNRDSparsifyContinuous}
    Let $C \subseteq (\{0\} \cup [1,k])^m$ be an unweighted code, and let $S_d, I_d: d \in \{1, 2, 4, \dots m \}$ and $\widetilde{S}$ be given as above. Then, with high probability over the random sampling of $\widetilde{S}$, the code $C$ with weights $1$ for every coordinate in $S_d, I_d: d \in \{1, 2, 4, \dots m \}$, weight $2$ for every coordinate in $\widetilde{S}$, and weight $0$ for all remaining coordinates:
    \begin{enumerate}
        \item is a $(1 \pm \eps')$ sparsifier of $C$.
        \item retains $\leq \sum_{d \in \{1, 2, 4, \dots m \}} (|S_d| + |I_d|) + \frac{2m}{3}$ coordinates.
    \end{enumerate}
\end{claim}

\begin{proof}
    We start by proving the second point, as it is essentially trivial. Indeed, we keep all coordinates in $I_d, S_d$, yielding the first term. All remaining coordinates are subsampled at rate $1/2$, and so by a Chernoff bound, with probability $1 - 2^{-\Omega(m)}$, at most $\frac{2m}{3}$ of those coordinates survive. 

    Next, we prove the first point. We consider a fixed $d \in \{1, 2, 4, 8, \dots m\}$ and prove that all codewords in $C_{[d/2, d]}$ have their weight preserved. Recall that, by construction, all codewords of weight $\leq 10000 \log^3(m) k^2 / \eps'^2$ have their supports retained explicitly, and thus their weights are preserved to a factor of exactly $1$. Thus, we must only focus on codewords in $C_{[d/2, d]}$ for $d > 10000 \log^3(m) k^2 / \eps'^2$. So, fix any such $d$ and $C_{[d/2, d]}$.
    
    After removing $I_d$ and $S_d$, we know that 
    \[
   \left | \mathrm{Cover} \left ((C_{[d/2, d]})|_{\bar{S_d} \cap \bar{I_d}}, \eps'/4 \right )\right | \leq 2^{\eps'^2 d / 10000k^2 \log(m)}.
    \]
    In particular, for each codeword $c' \in  \mathrm{Cover} \left ((C_{[d/2, d]})|_{\bar{S_d} \cap \bar{I_d}}, \eps'/4 \right )$, we now sample the coordinates of $c'$ (which we denote by $\Supp(c')$) at rate $1/2$ if the coordinate is not in $\bigcup_{d} S_d \cup I_d$, assigning weight $2$ if it survives, and otherwise we sample at rate $1$ (assigning weight $1$ if it survives). For each coordinate $j \in \Supp(c')$, we let $X_j$ denote the random variable corresponding to the weight contributed by coordinate $j$; i.e., if $j$ is sampled w.p. $1/2$, $X_j = 2 \cdot c'_j$ and if $j$ is not sampled, then $X_j = 0$. Importantly, $\E[\sum_{j \in \Supp(c')} X_j] = \wt(c')$.

    By \cref{clm:chernoffBound}, we also know that 
    \[
    \Pr \left [ \left | \sum_{j \in \Supp(c')} X_j - \wt(c') \right |  \geq \eps' d/8\right ] \leq 2e^{- \frac{(\eps' d/8)^2}{|\Supp(c')| \cdot (2k)^2}} \leq 2e^{- \frac{\eps'^2 d}{256k^2}},
    \]
    where here we used that each $X_i \in [0,2k]$, and that $|\Supp(c')| \leq d$. Finally, recall that we are only focusing on the case when $d \geq \frac{10000\log^3(m)k^2}{\eps'^2}$. In this parameter regime,
    \[
    \leq \Pr[\exists c' \in \mathrm{Cover} \left ((C_{[d/2, d]})|_{\bar{S_d} \cap \bar{I_d}}, \eps'/4 \right ): c' \text{ not preserved to additive }\eps'(d/8)]
    \]
    \[
    \leq 2^{\eps'^2 d / 10000k^2 \log(m)} \cdot 2e^{- \frac{\eps'^2 d}{256k^2}} \leq 2^{-\eps'^2 d / 512k^2} \leq 1 / \mathrm{poly}(m).
    \]
    From here, we can then take a union bound over all $O(\log(m))$ choices of $d$, so we see that 
    \[
    \Pr[\exists d, \exists c' \in \mathrm{Cover} \left ((C_{[d/2, d]})|_{\bar{S_d} \cap \bar{I_d}}, \eps'/4 \right ): c' \text{ not preserved to additive }\eps'(d/8)] \leq 1 / \mathrm{poly}(m).
    \]

    Importantly, for every codeword $c \in C_{[d/2, d]}$, we can now break $c$ into two parts: $c|_{I_d \cup S_d}$ and $c|_{\bar{I_d} \cap \bar{S_d}}$. In the first part, $c|_{I_d \cup S_d}$ is preserved exactly, while in the second part, the preceding argument shows that the covering codeword $c'|_{\bar{I_d} \cap \bar{S_d}}$ of $c|_{\bar{I_d} \cap \bar{S_d}}$ satisfies that $c'|_{\bar{I_d} \cap \bar{S_d}}$ has its weight preserved to an additive $\eps' d/8$ error. Importantly then, the codeword $c''$ which is equal to $c'|_{\bar{I_d} \cap \bar{S_d}}$ on $\bar{I_d} \cap \bar{S_d}$ and equal to $c$ in $I_d \cup S_d$ is an $\eps'/4$-covering codeword of $c$ and has its weight preserved under the sparsifier to additive error $\eps' d/8$. Because we are dealing with codewords in $C_{[d/2, d]}$, it must also be the case that $\wt(c'') \geq d/2$. Thus, the sparsifier preserves the weight of $c''$ to error $(1 \pm \eps'/4)$. Thus, we can invoke \cref{clm:coverSparsifier} to conclude that the sparsifier preserves the weight of $c$ to a $(1 \pm \eps')$ factor as well. 
    
This concludes the correctness of the sparsifier.
\end{proof}

The above procedure suffices for a single level of sparsification. Ultimately however, we repeat the above procedure $O(\log(m))$ times, to continue sparsifying the coordinates in $\widetilde{S}$. We use the algorithm below to summarize:

\begin{algorithm}[H]
    \caption{SparsifyBACNRD$(C, \eps' , m)$}\label{alg:sparsifyRVNRDContinuous}
    $\ell = 1$. \\
    $C^{(\ell)} = C$. \\
    \While{$\ell \leq \log_{3/2}(m)$}{
    \For{$d \in \{1, 2, 4, \dots m \}$}{
    \If{$d < \frac{10000 \log^3(m)k^2}{\eps'^2}$}{
Let $I_d^{(\ell)} = \Supp(C^{(\ell)}_{[d/2, d]})$.
}
\Else{
    Let $I_d^{(\ell)}$ be defined as in \cref{thm:NRDdecomposition} for $\hat{C}^{(\ell)}_{[d/2, d]}$, with parameter $\lambda = 10000 \log^4(m)k^2 / \eps'^2$. \\
    Let $S_d^{(\ell)}$ be the smallest set of coordinates such that
$\left | \mathrm{Cover} \left ((C_{[d/2, d]})|_{\overline{S^{(\ell)}_d} \cap \overline{I^{(\ell)}_d}}, \eps'/4 \right )\right | \leq 2^{\eps'^2 d / 10000 k^2\log(m)}$. \\
}
    }
    Let $\widetilde{S}^{(\ell)}$ be the set of re-weighted coordinates of size $\leq 2m/3$ guaranteed by \cref{clm:singleRoundNRDSparsifyContinuous}. \\
    Let $C^{(\ell+1)} = C|_{\widetilde{S}^{(\ell)}}$.\\
    $\ell \leftarrow \ell + 1$.
    }
    \Return{$\bigcup_{j = 1}^{\ell} 2^{j-1} \cdot \bigcup_{d \in \{1, 2, 4, \dots m\}} S^{(j)}_d \cup I^{(j)}_d$.}
\end{algorithm}

It remains only to bound the size of the resulting sparsifier and to show that the resulting sparsifier indeed preserves the weight of every codeword to a factor of $(1 \pm \eps)$. We start by a simple bound on the size of the sparsifier:

\begin{claim}\label{clm:RVNRDsparsifierSizeBasicContinuous}
    \cref{alg:sparsifyRVNRDContinuous} returns a sparsifier of size 
    \[
    \leq \sum_{j = 1}^{\ell} \sum_{d \in \{1, 2, 4, \dots, m\}} |S^{(j)}_d| + |I^{(j)}_d| \leq 4 \log^2(m) \cdot \max_{j \in [\log_{3/2}(m)]} \max_{d \in \{1, 2, 4, \dots m\}} \max  \left ( |S^{(j)}_d|, |I^{(j)}_d| \right ).
    \]
\end{claim}

\begin{proof}
    This follows trivially by the form of the returned sparsifier, and that there are $\leq 4 \log^2(m)$ terms in the sum. 
\end{proof}

\begin{claim}\label{clm:RVNRDsparsifyCorrectnessContinuous}
    \cref{alg:sparsifyRVNRDContinuous} returns a $(1 \pm \eps)$ sparsifier of $C$.
\end{claim}

\begin{proof}
    We prove the claim inductively. Observe that via \cref{clm:singleRoundNRDSparsifyContinuous}, $2 \cdot \widetilde{S}^{(1)} \cup \bigcup_{d \in \{1, 2, 4, \dots m\}} S^{(1)}_d \cup I^{(1)}_d$ is a $(1 \pm \eps')$ sparsifier of $C$. By the same claim, we also know that $2 \cdot \widetilde{S}^{(2)} \cup \bigcup_{d \in \{1, 2, 4, \dots m\}} S^{(2)}_d \cup I^{(2)}_d$ sparsifier of $C^{(2)}$. By composition then, we see that $\bigcup_{d \in \{1, 2, 4, \dots m\}} S^{(1)}_d \cup I^{(1)}_d \cup 2 \cdot \bigcup_{d \in \{1, 2, 4, \dots m\}} S^{(2)}_d \cup I^{(2)}_d \cup 4 \cdot \widetilde{S}^{(2)}$ is a $(1 \pm \eps')^2$ sparsifier of $C$.

    Inductively, we claim that $2^{i} \cdot \widetilde{S}^{(i)} \cup \bigcup_{j = 1}^{i} 2^{j-1} \cdot \bigcup_{d \in \{1, 2, 4, \dots m\}} S^{(j)}_d \cup I^{(j)}_d $ is a $(1 \pm \eps')^{i}$ sparsifier of $C$. Then, because $2 \cdot \widetilde{S}^{(i+1)}\bigcup_{d \in \{1, 2, 4, \dots m\}} S^{(i+1)}_d \cup I^{(i+1)}_d$ is a $(1 \pm \eps')$ sparsifier of $C^{(i+1)}$, by composition, we obtain that $2^{i+1} \cdot \widetilde{S}^{(i+1)} \cup \bigcup_{j = 1}^{i+1} 2^{j-1} \cdot \bigcup_{d \in \{1, 2, 4, \dots m\}} S^{(j)}_d \cup I^{(j)}_d $ is a $(1 \pm \eps')^{i+1}$ sparsifier of $C$.

    Finally, this means that $\bigcup_{r = 1}^{\log_{3/2}(m)} 2^r \cdot S^{(r)} \cup 2^{\log_{3/2}(m) + 1} \cdot C^{(\log_{3/2}(m))}$ is a $(1 \pm \eps')^{\log_{3/2}(m)}$ sparsifier of $C$. But by \cref{clm:singleRoundNRDSparsifyContinuous}, we know that after so many rounds of resparsification, $C^{(\log_{3/2}(m))}$ has no support (it is empty - as in each round it decreases by a factor of $2/3$), and thus $\bigcup_{r = 1}^{\log_{3/2}(m)} 2^r \cdot S^{(r)}$ is a $(1 \pm \eps')^{\log_{3/2}(m)} = (1 \pm \eps)$ sparsifier of $C$, where the final equality holds because $\eps' \leq \eps/2 \log(m)$.
\end{proof}

The primary difficulty comes in showing that these sets $S^{(j)}_d, I^{(j)}_d$ are bounded in size. We discuss this more in the following subsection.

\subsubsection{Bounding the Size of Peeled Sets by Upper-Triangular Sub-Codes }

In this section, we prove the following key lemma:

\begin{lemma}\label{lem:constructUpperTriangular}
	Let $C \subseteq (\{0\} \cup [1, k])^m$ be a code, let $\eps' > 0$, let $d \geq  10000 k^2 \log^3(m) / \eps'^2$, and let $L_d$ denote the size of the smallest set $T \subseteq [m]$ such that 
	\[
	\left | \mathrm{Cover}(C|_{\bar{T}}, \eps'/4) \right | \leq 2^{\eps'^2 d / 10000k^2 \log(m)}.
	\]
	Further, suppose that every codeword $c \in \hat{C}$ satisfies $\mathrm{wt}(c) \leq d$ and that $|\hat{C}| \leq m \cdot 2^{3d \log^2(2m) \eps'^2/ 10000 k^2 \log^4(m)}$. Then, one can find sets of coordinates $A_1, \dots A_p$, and sets of codewords $C_1, \dots C_p$ such that:
	\begin{enumerate}
		\item For every $i \in [p]$, there exists a $\gamma \in [1,k]^{A_i}$ such that, for every $B \subseteq A_i$ there exists a codeword $c \in C_i$ such that, for $b \in B$, $c_b \geq \gamma_b + \eps'/16$ and for $b \in A_i - B$, $c_b \leq \gamma_b - \eps'/16$.
		\item For $i < j \in [p]$, $C_i |_{A_j} = 0$.
		\item Each $|A_i| = \Omega(\eps'^2 d)$ and $\sum_{i = 1}^p |A_i| = \widetilde{\Omega}(\eps'^2 \cdot L_d)$.
		\item For $i \in [p]$, there exists $\hat{c} \in \hat{C}$ such that $C_i \subseteq  \Class(\hat{c})$.
	\end{enumerate}
\end{lemma}

Intuitively, this lemma states that whenever we have to remove a lot of coordinates (rows) from our code in order to decrease the size of the cover, this is because there is a large upper-triangular matrix of a specific form contained in our code. Ultimately, we will show that we can even further modify this upper-triangular matrix to get it in our desired \emph{block-diagonal} form.

Now, to prove this lemma, we introduce an algorithm for producing a set of coordinates and codes of the desired form:

\begin{algorithm}[H]
	\caption{ConstructUpperTriangularWitness$(C, d,L_d, \eps')$}\label{alg:constructUTWitness}
	Initialize $R_d = \emptyset$. \tcp{The set of rows being removed.}
	Initialize $Q_d = \emptyset$. \tcp{The rows witnessing the upper-triangular structure.} 
	Initialize $C_d = \emptyset$. \tcp{The set of columns witnessing the U-T structure.}
	Initialize $i = 1$. \\
	\While{$|R_{d}| < L_d$}{
		Let $\hat{c}_i$ be a codeword in $\hat{C}$ such that $|\mathrm{Cover}(\Class(\hat{c}_i)|_{[m] - R_{d}}, \eps'/4)| \geq 2^{\eps'^2 d / 20000k^2 \log(m)}$.  \label{line:bigClass}\\
		Let $Y_i \subseteq \Supp((\hat{c}_i)|_{[m] - R_{d}})$ be a set of size $\widetilde{\Omega}(\eps'^2 d)$ such that there exists $\gamma \in [1, k]^{Y_i}$ such that for every $B \subseteq Y_i$, there is a $v \in \Class(\hat{c}_i)|_{[m] - R_{d}}$ such that for $b \in B$, $v_b \geq \gamma_b + \eps'/16$ and for $b \in Y_i - B$, $v_b \leq \gamma_b - \eps'/16$. Let $Z_i$ be the subset of these $2^{|Y_i|}$ many codewords in $ \Class(\hat{c}_i)$ which witness this property. \label{line:bigFDim}\\
		$Q_{d} \leftarrow Q_{d} \cup Y_i$. \\
		$R_{d} \leftarrow R_{d} \cup \Supp(\hat{c}_i)$. \\
		$C_{d} \leftarrow C_{d} \cup Z_{i}$. \\
		$i \leftarrow i + 1$.\\
	}
	\Return{$Q_{d} = (Y_1, \dots Y_i), C_{d} = (Z_1, \dots Z_i)$.}
\end{algorithm}

First, we prove that every line in the above algorithm is indeed possible:

\begin{claim}\label{clm:bigClass}
	Under the conditions of \cref{lem:constructUpperTriangular}, there is always a $\hat{c}_i$ as required in \cref{line:bigClass}.
\end{claim}

\begin{proof}
To start, because $|R_d| < L_d$, we know that 
\[
\left | \mathrm{Cover}(C|_{\bar{T}}, \eps'/4) \right | > 2^{\eps'^2 d / 10000k^2 \log(m)}.
\]
Now, because $|\hat{C}| \leq m \cdot 2^{3d \log^2(2m) \eps'^2/ 10000 \log^4(m)}$, and $C|_{\bar{T}}= \bigcup_{\hat{c} \in \hat{C}} \Class(\hat{c})|_{\bar{T}}$, by \cref{clm:coverUnion}, it must be the case that there exists $\hat{c} \in \hat{C}$ such that 
\[
|\mathrm{Cover}(\Class(\hat{c})|_{[m] - R_{d}}, \eps'/4)| \geq \frac{2^{\eps'^2 d / 10000k^2 \log(m)}}{m \cdot 2^{3d \log^2(2m) \eps'^2/ 10000k^2 \log^4(m)}}.
\]
Now, using the fact that $d \geq 10000 k^2 \log^3(m) / \eps'^2$, we see that $3d \log^2(2m) \eps'^2/ 10000k^2 \log^4(m) \leq 12  \eps'^2 d / 10000 \log^2(m)$. Likewise, by our lower bound on $d$, we also know that $\log(m) \leq \eps'^2 d / 10000k^2 \log^2(m)$. Thus, together, we see that 
\[
\frac{2^{\eps'^2 d / 10000k^2 \log(m)}}{m \cdot 2^{3d \log^2(2m) \eps'^2/ 10000k^2 \log^4(m)}} \geq \frac{2^{\eps'^2 d / 10000k^2 \log(m)}}{2^{13\eps'^2 d / 10000k^2 \log^2(m)}} \geq 2^{(1 - 13 / \log(m)) \cdot \eps'^2 d / 10000k^2 \log(m)}
\]
\[
\geq2^{ \eps'^2 d / 20000k^2 \log(m)},
\]
as we desire. 
\end{proof}

Likewise, we show that we can find sets $Y_i, Z_i$ in accordance with the above algorithm:

\begin{claim}\label{clm:bigFDim}
    Under the conditions of \cref{lem:constructUpperTriangular} and \cref{line:bigClass}, there always exists $Y_i, Z_i$ as required in \cref{line:bigFDim}.
\end{claim}

\begin{proof}
    By \cref{clm:bigClass}, we know that we can find $\hat{c}_i$ such that 
    \[
    |\mathrm{Cover}(\Class(\hat{c}_i)|_{[m] - R_{d}}, \eps'/4)| \geq 2^{\eps'^2 d / 20000k^2 \log(m)}.
    \]
    Now, by \cref{lem:coverCountingBound}, we know that, for any code $W \subseteq (\{0\} \cup [1,k])^m$ that 
    \[
    |\mathrm{Cover}(W, \eps'/4)| \leq 2 \cdot \left ( \frac{64mk^2}{\eps'^2} \right )^{\mathrm{fdim}(W, \eps'/16)\log(8emk/\mathrm{fdim}(W, \eps'/16) \eps')}.
    \]
    By setting $W = \Class(\hat{c}_i)|_{[m] - R_{d}}$ and using our lower bound on the cover size, this implies that 
    \[
    2 \cdot \left ( \frac{64mk^2}{\eps'^2} \right )^{\mathrm{fdim}(W, \eps'/16)\log(8emk/\mathrm{fdim}(W, \eps'/16) \eps')} \geq 2^{\eps'^2 d / 20000k^2 \log(m)}.
    \]
    Taking the logs of both sides yields that 
    \[
    \mathrm{fdim}(W, \eps'/16)\log(8emk/\mathrm{fdim}(W, \eps'/16) \eps') \cdot \log(64mk^2 / \eps'^2) + 1 \geq \eps'^2 d / 20000k^2 \log(m).
    \]
    Thus, we have that 
    \[
    \mathrm{fdim}(W, \eps'/16) \geq \frac{\eps'^2 d / 20000k^2 \log(m) - 1}{\log(8emk/\mathrm{fdim}(W, \eps'/16) \eps') \cdot \log(64mk^2 / \eps'^2)} = \widetilde{\Omega}(\eps'^2 d),
    \]
    where we are using that $k = O(1)$, that  $\mathrm{fdim}(W, \eps'/16) \leq m$, and that $\eps'^2 d / 20000k^2 \log(m) = \Omega(\log(m))$, hence the $-1$ in the numerator is negligible. 

    In total, this means that for the code $W = \Class(\hat{c}_i)|_{[m] - R_{d}}$, the fat-shattering dimension with parameter $\eps'/16$ is at least $\widetilde{\Omega}(\eps'^2 d)$, which, as per \cref{def:fatShatteringDim}, exactly implies that we can find a set $Y_i \subseteq [m] - R_{d}$, and $\gamma \in [1, k]^{Y_i}$ such that for every $B \subseteq Y_i$, there is a $v \in \Class(\hat{c}_i)|_{[m] - R_{d}}$ such that for $b \in B$, $v_b \geq \gamma_b + \eps'/16$ and for $b \in Y_i - B$, $v_b \leq \gamma_b - \eps'/16$. Letting $Z_i$ denote these vectors which realize the fat-shattering dimension then yields the claim. 
\end{proof}

Now, we lower bound the size of $Q_d$ that is returned by our algorithm. 

\begin{claim}\label{clm:witnessSize}
    Under the conditions of \cref{lem:constructUpperTriangular}, when \cref{alg:constructUTWitness} terminates, $|Q_{d}| = \widetilde{\Omega}(\eps'^2 L_d)$ and $i \geq \frac{L_d}{ d}$.
\end{claim}
\begin{proof}
    Note that in every iteration, $R_{d}$ increases in size by at most $d$, as we assume that every codeword has support size at most $d$. At the same time, in every iteration, $Q_{d}$ increases in size by $\widetilde{\Omega}(\eps'^2 d)$. Thus, the number of iterations until $R_{d}$ exceeds $L_d$ is at least $\frac{L_d}{ d}$, and in this same number of iterations $Q_{d}$ reaches size at least 
    \[
    |Q_{d}| \geq \frac{L}{ d} \cdot \widetilde{\Omega}(\eps'^2 d )=\widetilde{\Omega}(\eps'^2 L_d).
    \]
\end{proof}

\begin{claim}\label{clm:upperTriangular}
    Under the conditions of \cref{lem:constructUpperTriangular}, suppose \cref{alg:constructUTWitness} returns $Q_d = (Y_1, \dots Y_i)$ and $C_d = (Z_1, \dots Z_i)$. Then, for every $a > b \in [i]$ it is the case that $(Z_b)|_{Y_{a}} = 0$.
\end{claim}

\begin{proof}
    Observe that in iteration $i$, it is always the case that $Y_i \subseteq [m] - R_d$, where $R_d$ is the current set of ``banned rows''. Now, consider the iteration $b$ in which $Z_b$ is recovered as a set of codewords. At the end of this iteration, $R_d \leftarrow R_d \cup \Supp(\hat{c}_b)$, where $\hat{c}_b$ is the codeword in $\hat{C}$ such that $Z_b \subseteq \Class(\hat{c}_b)_{[m] - R_d}$. In particular, this ensures that \emph{for all future iterations} $a > b$, it is the case that $\Supp(\hat{c}_b) = \Supp(Z_b) \subseteq R_d$ (these rows will never be recovered in the future), and thus, there is no coordinate in $Y_a$ for which $(Z_b)_a \neq 0$.
\end{proof}

Now, we are ready to conclude our proof of \cref{lem:constructUpperTriangular}:

\begin{proof}[Proof of \cref{lem:constructUpperTriangular}.]
Given a code of the form guaranteed by \cref{lem:constructUpperTriangular}, we run \cref{alg:constructUTWitness}. By \cref{clm:bigClass} and \cref{clm:bigFDim}, we know that we can always run this algorithm to completion (i.e., that every line with an existential argument is valid). 

Now, we claim that the returned sets   $Q_d = (Y_1, \dots Y_i)$ and $C_d = (Z_1, \dots Z_i)$ satisfy our desired conditions in  \cref{lem:constructUpperTriangular}. For clarity, we let $i = p$, and let $Y_1, \dots Y_p = A_1, \dots A_p$, and likewise let $Z_1, \dots Z_p = C_1, \dots C_p$ to match the notation of \cref{lem:constructUpperTriangular}. Now, we prove that each condition of  \cref{lem:constructUpperTriangular} is satisfied:

\begin{enumerate}
    \item First, we show that for every $j \in [p]$ there exists a $\gamma \in [1,k]^{A_j}$ such that, for every $B \subseteq A_j$ there exists a codeword $c \in C_j$ such that, for $b \in B$, $c_b \geq \gamma_b + \eps'/16$ and for $b \in A_j - B$, $c_b \leq \gamma_b - \eps'/16$. This follows by definition, as in \cref{line:bigFDim}, we select $Y_j = A_j$ and $Z_j = C_j$ exactly in such a manner to ensure that we have this property (and by \cref{clm:bigFDim}, we know that this is possible). 
    \item Next, we show that for every $a < b \in [p]$, $C_a |_{A_b} = 0$. This follows exactly by \cref{clm:upperTriangular}.
    \item Next, we show that $|A_i| = \widetilde{\Omega}(\eps'^2 d)$ and $\sum_{i = 1}^p |A_i| = \widetilde{\Omega}(\eps'^2 \cdot L_d)$. This follows by \cref{clm:witnessSize} and \cref{clm:bigFDim}.
    \item Lastly, we show that for $a \in [p]$, there exists $\hat{c} \in \hat{C}$ such that $C_i \subseteq  \Class(\hat{c})$. This again follows by construction (see \cref{line:bigFDim}).
\end{enumerate}

This concludes the proof. 
\end{proof}

While this lemma already suffices for establishing the existence of a large upper-triangular matrix whenever $L_d$ is large, our ultimate goal is to show the existence of a large \emph{block diagonal matrix}. We show how this can be achieved in the following subsection. 

\subsubsection{Processing Upper-Triangular Sub-Codes into Block-Diagonal Sub-Codes}

In this section, we show the following lemma:

\begin{lemma}\label{lem:constructBlockDiagonal}
    Let $C, k, \eps', d, L_d$ be given as in \cref{lem:constructUpperTriangular}, and let $A_1, \dots A_p, C_1, \dots C_p$ be the upper-triangular witness as guaranteed by \cref{lem:constructUpperTriangular} (in particular, satisfying conditions 1-4 of \cref{lem:constructUpperTriangular}). Then, one can find $A'_1 \subseteq A_1, \dots A'_p \subseteq A_p, C'_1\subseteq C_1, \dots C'_p \subseteq C_p$ such that:
    \begin{enumerate}
        \item $\sum_{i = 1}^p |A'_i| = \widetilde{\Omega}(\eps'^4 L_d)$.
        \item For every $i \neq j \in [p]$, $(C'_i)_{A'_j} = 0$. 
        \item For every $i \in [p]$, there exists $\gamma \in [1,k]^{A'_i}$ such that for every $B \subseteq A'_i$ there exists a codeword $c \in C'_i$ such that, for $b \in B$, $c_b \geq \gamma_b + \eps'/16$ and for $b \in A'_i-B $, $c_b \leq \gamma_b - \eps'/16$. 
    \end{enumerate}
\end{lemma}

Importantly, the key modification is in the second point: now we require that \emph{every off diagonal} block is identically $0$. 

To prove this lemma, we introduce a new algorithm. Note that, as per \cref{lem:constructUpperTriangular}, we know that $|A_i| = \widetilde{\Omega}(\eps'^2 \cdot d)$, where the $\widetilde{\Omega}(\cdot)$ hides poly-logarithmic factors in $m / \eps'$. Thus, going forward, we let $\kappa$ denote a constant such that $|A_i| \geq \frac{\eps'^2 \cdot d}{\log(m / \eps')^{\kappa}}$.

\begin{algorithm}[H]
\caption{ConstructBlockDiagonalWitness$(A_1, \dots A_p, C_1, \dots C_p)$}\label{alg:constructBlockDiagonalwitness}
Initialize $R = \emptyset$. \tcp{The set of rows being removed.}
Initialize $j = p$. \\
Initialize $r =1$. \\
\While{$j \geq 1$}{
\If{$|A_j \setminus R| \geq \frac{\eps'^2 d}{2 \log(m / \eps')^{\kappa}}$}{
Let $A'_{r} = A_j \setminus R $. \label{line:disjointFromR}\\
$R \leftarrow R \cup \Supp(C_j)$.\label{line:incrementR} \\
$C'_r = C_j$. \\
}
\Else{
$A'_{r}  = \emptyset$. \\
$C'_r = \emptyset$.
}
$r \leftarrow r + 1$.\\
$j \leftarrow j -1$.\\
}
\Return{$(A'_1, \dots A'_p), (C'_1, \dots C'_p)$.}
\end{algorithm}

Now, we have several claims regarding the above algorithm. First, we establish that the resulting codewords and rows constitute a \emph{block-diagonal} form:

\begin{claim}\label{clm:BlockDiagonalWitnessValid}
    Let $A_1, \dots A_p, C_1, \dots C_p$ be of the form guaranteed by \cref{lem:constructUpperTriangular}. Then, \cref{alg:constructBlockDiagonalwitness} produces $(A'_1, \dots A'_p), (C'_1, \dots C'_p)$ such that:
    \begin{enumerate}
        \item For every $i \neq j \in [p]$, $(C'_i)|_{A'_j} = 0$. 
        \item For every $i \in [p]$, there exists $\gamma \in [1,k]^{A'_i}$ such that for every $B \subseteq A'_i$ there exists a codeword $c \in C'_i$ such that, for $b \in B$, $c_b \geq \gamma_b + \eps'/16$ and for $b \in A'_i - B$, $c_b \leq \gamma_b - \eps'/16$. 
    \end{enumerate}
\end{claim}

\begin{proof}
    We start by showing the second point. Note that if $A'_i = \emptyset, C'_i = \emptyset$, then the second point trivially holds. We instead focus on the case when $A'_i \neq \emptyset$. In this case, we know that $C'_i = C_i$ as returned by \cref{lem:constructUpperTriangular}. We also know that for every $i \in [p]$, there exists $\gamma \in [1,k]^{A_i}$ such that for every $B \subseteq A_i$ there exists a codeword $c \in C_i$ such that, for $b \in B$, $c_b \geq \gamma_b + \eps'/16$ and for $b \in A'_i - B$, $c_b \leq \gamma_b - \eps'/16$. Thus, when restricting to $A'_i \subseteq A_i$, the same condition must hold true.

    Next, we show that the first condition holds. So, consider any $i \neq j \in [p]$. If $A'_j = \emptyset$ or $C'_i = \emptyset$, then the condition holds vacuously. So, let us instead consider $A'_j$ and $C'_i$ which are not empty. Now, there are two cases for us to consider, depending on whether $i < j$ or $j < i$:
    \begin{enumerate}
        \item If $i < j$, this means that $C'_i, A'_i$ were processed before $C'_j, A'_j$ in \cref{alg:constructBlockDiagonalwitness}. However, this also means that after $C'_{i}$ was stored, the entire support of $C'_{i}$ was added to the set $R$ (by construction of \cref{line:incrementR}). Because all future $A'_{>i} \cap R = \emptyset$ (in \cref{line:disjointFromR}, we set $A'_r = A_j \setminus R$), it must be the case that $A'_{j}$ does not contain any rows from the support of $C'_{i}$, and thus $C'_{i}|_{A'_{j}} = 0$.
        \item Next, we suppose that $j < i$. This means that $C'_{i} = C_{p-i}$ and $A'_{j} \subseteq A_{p-j}$ for $C_{p-i}$, $A_{p-j}$ as returned by \cref{lem:constructUpperTriangular}. Importantly, because $p-j > p-i$, we then know that $C'_{i}|_{A_{p-j}} = C_{p-i}|_{A_{p-j}} = 0$, and because $A'_{j} \subseteq A_{p-j}$, then we know that $C'_{i}|_{A'_j} = 0$ as well. 
    \end{enumerate} 
    This concludes the proof, as we have shown both of the above conditions. 
\end{proof}

It remains only to lower bound the size of the block-diagonal sub-code that we recover:

\begin{claim}\label{clm:blockDiagonalSize}
     Let $A_1, \dots A_p, C_1, \dots C_p$ be of the form guaranteed by \cref{lem:constructUpperTriangular}. Then, \cref{alg:constructBlockDiagonalwitness} produces $(A'_1, \dots A'_p), (C'_1, \dots C'_p)$ such that $\sum_{i = 1}^p |A'_i| = \widetilde{\Omega}(\eps'^4 L_d)$.
\end{claim}

\begin{proof}
    Observe that a set $A'_i$ is non-empty if and only if $|A_{p-i} \setminus R| \geq \frac{\eps'^2 d}{2 \log(m / \eps')^{\kappa}}$. So, let $q$ denote the number of non-empty sets $A'_i$ that we store. We will lower bound this value $q$.

    In particular, let the number of \emph{empty} sets $A'_i$ be denoted by $p -q$. Because each $A_{p-i} $ satisfies $|A_{p-i}| \geq \frac{\eps'^2 d}{\log(m / \eps')^k}$, it must be the case that $|R| \geq (p-q) \cdot \frac{\eps'^2 d}{2\log(m / \eps')^k}$, as for every empty set $A'_i$, $R$ must contain at least $\frac{\eps'^2 d}{\log(m / \eps')^k} - \frac{\eps'^2 d}{2\log(m / \eps')^k} \geq \frac{\eps'^2 d}{2\log(m / \eps')^k}$ many rows from $A_{p-i}$. 
    
    However, if $|R| \geq (p-q) \cdot \frac{\eps'^2 d}{2\log(m / \eps')^k}$, then it must be the case that at least $\frac{|R|}{d} \geq (p-q) \cdot \frac{\eps'^2 d}{2d\log(m / \eps')^k} = \frac{(p-q) \eps'^2}{2\log(m / \eps')^k}$ distinct sets $A'_i$ are non-empty, as the only time $R$ increases in size is in \cref{line:incrementR}, and this line is only reached when the current set $A'_i$ is non-empty. Further, by \cref{lem:constructUpperTriangular}, we know that $\Supp(A_j) \leq d$, and so whenever we increment $R$, its size increases by at most $d$. Thus, this implies that 
    \[
    q \geq \frac{|R|}{d} \geq \frac{(p-q) \eps'^2}{2\log(m / \eps')^k}. 
    \]
    Rearranging, this implies that 
    \[
    q \cdot \left ( 1 + \frac{\eps'^2}{2 \log(m / \eps')^{\kappa}}\right ) \geq \frac{p\eps'^2}{2 \log(m / \eps')^{\kappa}},
    \]
    which implies that $q \geq \frac{p\eps'^2}{4 \log(m / \eps')^{\kappa}}$.

    Finally, every non-empty set $A'_i$ is of size at least $\frac{\eps'^2d}{2 \log(m / \eps')^{\kappa}}$, so this implies that 
    \[
    \sum_{i = 1}^p |A'_i| \geq q \cdot \frac{\eps'^2d}{2 \log(m / \eps')^{\kappa}} \geq \frac{p\eps'^2}{4 \log(m / \eps')^{\kappa}} \cdot \frac{\eps'^2d}{2 \log(m / \eps')^{\kappa}} \geq \frac{L_d}{d} \cdot \frac{\eps'^4 d}{\mathrm{polylog}(m/\eps')}
    \]
    \[
    = \widetilde{\Omega}(L_d \cdot \eps'^4),
    \]
    where we have used that the number of sets $A_1, \dots A_p$ as returned by \cref{alg:constructUTWitness} satisfies $p \geq L_d / d$ as per \cref{clm:witnessSize}.
\end{proof}

We now conclude the proof of \cref{lem:constructBlockDiagonal}:

\begin{proof}[Proof of \cref{lem:constructBlockDiagonal}.]

The first condition follows from \cref{clm:blockDiagonalSize}. The remaining two conditions follow from \cref{clm:BlockDiagonalWitnessValid}.
\end{proof}

Now, in the next subsection, we show how \cref{lem:constructBlockDiagonal} suffices for upper-bounding the size of our returned sparsifier. 

\subsubsection{Bounding Sparsifier Size by $\BACNRD$}

In this section, we prove the following key lemma:

\begin{lemma}\label{lem:sizeBound}
    Let $C \subseteq  \left ( \{0 \} \cup [1,k] \right )^m$ be a code, let $\eps > 0$, and let $S \subseteq [m], \hat{w}_S$ be a $(1 \pm \eps)$ sparsifier as returned by \cref{alg:sparsifyRVNRDContinuous}. Then $|S| = \widetilde{O}(\BACNRD(C, \eps'/16) / \eps^4)$.
\end{lemma}

\begin{proof}
    By \cref{clm:RVNRDsparsifierSizeBasicContinuous}, we already know that 
    \[
    |S| \leq \sum_{j = 1}^{\ell} \sum_{d \in \{1, 2, 4, \dots, m\}} |S^{(j)}_d| + |I^{(j)}_d| \leq 4 \log^2(m) \cdot \max_{j \in [\log_{3/2}(m)]} \max_{d \in \{1, 2, 4, \dots m\}} \max  \left ( |S^{(j)}_d|, |I^{(j)}_d| \right ),
    \]
    where $S^{(j)}_d, I^{(j)}_d$ are the sets constructed by \cref{alg:sparsifyRVNRDContinuous}. Now, we bound the maximum size of any such set $S^{(j)}_d, I^{(j)}_d$. Immediately, when $d \leq 10000k^2 \log^3(m) / \eps'^2$, we know that we can remove \emph{all} codewords of weight $[d/2, d]$ removing only $\widetilde{O}(\NRD(\hat{C^{(j)}}) // \eps^2) \leq \widetilde{O}(\NRD(\hat{C}) // \eps^2)$ many rows (by \cref{clm:boundSupportSizeSmalld}). Then, using the relationship between $\NRD$ and $\BACNRD$ (\cref{rmk:RVNRDvsNRD}), we can also see that $|I^{(j)}_d| = \widetilde{O}(\BACNRD(C, \eps'/16) / \eps^2)$. Note that for $d \leq 10000k^2 \log^3(m) / \eps'^2$, it is the case that $S^{(j)}_d = \emptyset$, and so these sets are trivially bounded in size. 

    All that remains then is to bound the size of $S^{(j)}_d, I^{(j)}_d$ when $d > 10000k^2 \log^3(m) / \eps'^2$. First, recall that $I^{(j)}_d$ is defined using \cref{thm:NRDdecomposition} with parameter $\lambda = 10000k^2 \log^4(m) / \eps'^2$. Thus, we know that 
    \[
    |I^{(j)}_d| \leq 2 \cdot (10000k^2 \log^4(m) / \eps'^2) \cdot \NRD(\hat{C^{(j)}}) \cdot \log(4m) = \widetilde{O}(\BACNRD(C, \eps'/16) / \eps'^2),
    \]
    where we have again used the relationship between $\NRD$ and $\BACNRD$ (\cref{rmk:RVNRDvsNRD}). Finally, it remains only to bound the size of $S^{(j)}_d$. Importantly, recall that Let $S_d^{(j)}$ is the smallest set of coordinates such that
$\left | \mathrm{Cover} \left ((C^{(j)}_{[d/2, d]})|_{\overline{S^{(j)}_d} \cap \overline{I_d^{(j)}}}, \eps'/4 \right )\right | \leq 2^{\eps'^2 d / 10000 k^2\log(m)}$. Thus, if we let $L_d = |S^{(j)}_d|$, we can invoke \cref{lem:constructBlockDiagonal} using the code $(C^{(j)}_{[d/2, d]})|_{\overline{I_d^{(j)}}}$ (along with the same choices of $d, \eps'$, and $k$), using the key fact that for any set of size smaller than $L_d$, the size of any cover must be sufficiently large. In particular, this guarantees that we can
find sets of coordinates $A'_1, \dots A'_p \subseteq [m], C'_1, \dots C'_p \subseteq (C^{(j)}_{[d/2, d]})|_{\overline{I_d^{(j)}}}$ such that:
    \begin{enumerate}
        \item $\sum_{i = 1}^p |A'_i| = \widetilde{\Omega}(\eps'^4 L_d)$.
        \item For every $i \neq j \in [p]$, $(C'_i)_{A'_j} = 0$. 
        \item For every $i \in [p]$, there exists $\gamma \in [1,k]^{A'_i}$ such that for every $B \subseteq A'_i$ there exists a codeword $c \in C'_i$ such that, for $b \in B$, $c_b \geq \gamma_b + \eps'/16$ and for $b \in A'_i - B$, $c_b \leq \gamma_b - \eps'/16$. 
    \end{enumerate}
The above coordinates and subcodes exactly satisfy the conditions of \cref{def:continuousRVNRD}, and therefore show that $\BACNRD((C^{(j)}_{[d/2, d]})|_{\overline{I_d^{(j)}}}, \eps'/16) = \widetilde{\Omega}(\eps'^4 L_d) = \widetilde{\Omega}(\eps'^4 |S_d^{(j)}|)$. Now, because $(C^{(j)}_{[d/2, d]})|_{\overline{I_d^{(j)}}}$ is a restriction of the rows and codewords of $C$, it must also therefore be the case (by \cref{rmk:inheritance}) that 
\[
\BACNRD(C, \eps'/16) \geq \BACNRD((C^{(j)}_{[d/2, d]})|_{\overline{I_d^{(j)}}}, \eps'/16), 
\]
and thus we also have 
\[
\BACNRD(C, \eps'/16) = \widetilde{\Omega}(\eps'^4 |S_d^{(j)}|),
\]
which implies that 
\[
|S_d^{(j)}| = \widetilde{O}(\BACNRD(C, \eps'/16) / \eps'^4).
\]

So, we have shown that all sets $I_d^{(j)}, S^{(j)}_d$ are bounded in size by $\widetilde{O}(\BACNRD(C, \eps'/16) / \eps'^4)$. Thus, using our starting inequality, we know that our sparsifier size is bounded by 
\[
|S| \leq 4 \log^2(m) \cdot \widetilde{O}(\BACNRD(C, \eps'/16) / \eps'^4) = \widetilde{O}(\BACNRD(C, \eps'/16) / \eps'^4),
\]
as we desire.
\end{proof}

Now, we conclude with a proof of \cref{thm:RVNRDcontinuous}, \cref{item:continuousRVNRDUB}.

\begin{proof}[Proof of \cref{thm:RVNRDcontinuous}, \cref{item:continuousRVNRDUB}]
    Indeed, for a code $C \subseteq (\{0\} \cup [1,k])^m$ and $\eps > 0$, construct the sparsifier $S \subseteq [m]$, with weights $\hat{w}_S$ as per \cref{alg:sparsifyRVNRDContinuous}, using $\eps' = \eps/16\log(m)$. By \cref{clm:RVNRDsparsifyCorrectnessContinuous}, we know that this algorithm returns $S , \hat{w}_S$ which constitutes a $(1 \pm \eps)$-sparsifier of $C$. Likewise, by \cref{lem:sizeBound}, we know that the size of the returned sparsifier is bounded by 
    \[
    |S| = \widetilde{O}(\BACNRD(C, \eps'/16) / \eps'^4) =\widetilde{O}(\BACNRD(C, \eps/256\log(m)) / \eps^4),
    \]
    thus concluding the proof of the theorem. 
\end{proof}

\section{Sparsifying Discrete Codes with Bounded Aspect Ratio}\label{sec:discreteCodes}

In this section, we strengthen the previous results for the specific setting of sparsifying codes $C \subseteq \{0, a_1, \dots a_k\}^m$, where each $a_i \geq 1$, $a_i = O(1)$, and $|a_i - a_j| = \Omega(1)$. For this setting, we have the following notion of discrete-valued non-redundancy:

\begin{definition}\label{def:discreteRVNRD}
    Let $C \subseteq \{0, a_1, \dots a_k\}^m$. The bounded aspect ratio, discrete-valued non-redundancy of $C$, denoted by $\BADNRD(C)$ is the maximum integer $\ell$ such that there exists $A_1, \dots A_p \subseteq [m]$ and $C_1, \dots C_p \subseteq C$, with $\sum_{i = 1}^p |A_i| = \ell$, and:
    \begin{enumerate}
    \item For every $i\in[p]$, if $|A_i| = 1$, then $(C_i)|_{A_i} \in \{a_1, \dots a_k\}$.
        \item For every $i \in [p]$, if $|A_i| \geq 2$, then $(C_i)|_{A_i} = \{b_0, b_1\}^{|A_i|}$ for some $\{b_0, b_1\} \subseteq \{0, a_1, \dots a_k\}$.
        \item For every $j \neq i$, $(C_i)|_{A_j} = 0^{A_j}$.
    \end{enumerate}
\end{definition}

In this section, we show the following theorem characterizing unweighted sparsifiability of discrete-valued codes in terms of their $\BADNRD$:

\begin{theorem}\label{thm:DiscreteDomainRVNRD}
    Let $C \subseteq \{0, a_1, a_2, \dots a_k \}^m$ such that for each $i \in [k], a_i \geq 1$, $a_i = O(1)$, and for all $i \neq j \in [k]$, $|a_i - a_j| = \Omega(1)$. Then:
    \begin{enumerate}
        \item There is a choice of $\eps = \Omega(1)$ such that $\US(C, \eps) = \Omega(\BADNRD(C))$. \label{item:DiscreteDomainRVNRDLB}
        \item $\US(C, \eps) = \widetilde{O}(\BADNRD(C) / \eps^4)$. \label{item:DiscreteDomainRVNRDUB}
    \end{enumerate}
\end{theorem}

Note that this assumption that $a_i \geq 1$ is without loss of generality, as if any $a_i \leq 1$, we can simply re-scale all of the $a_i$'s without altering sparsifiability.

\subsection{Discrete-Valued Non-Redundancy Characterizations}

In this section, we prove both parts of \cref{thm:DiscreteDomainRVNRD}.

\subsubsection{Proof of Lower Bound}

The lower bound follows from showing that any $\BADNRD(C)$ witness also constitutes an $\BACNRD(C, \eps)$ witness, for some $\eps = \Omega(1)$.

\begin{proof}[Proof of \cref{thm:DiscreteDomainRVNRD}, \cref{item:DiscreteDomainRVNRDLB}.]

Indeed, let the code $C$ be given as in the theorem statement, and let $C_1, \dots C_p, A_1, \dots A_p$ denote the maximum size $\BADNRD$ witness (using \cref{def:discreteRVNRD}). Now, let $\eps = \min_{i \neq j \in [k] } |a_i - a_j|$. By definition then, we know that $C_1, \dots C_p, A_1, \dots A_p$ constitutes an $\BACNRD(C, \eps/2)$ witness, (now in the notion of \cref{def:continuousRVNRD}), as in each subcode $C_i$, we know that $C_i|_{A_i} = \{b_0, b_1\}^{A_i}$, and thus we can use the vector $\gamma = ((b_0 + b_1) / 2)^{A_i}$. In particular, this means that $\BACNRD(C, \eps/2) \geq \BADNRD(C)$. We can then invoke \cref{thm:RVNRDcontinuous}, \cref{item:continuousRVNRDLB}, to conclude that there exists an $\eps' = \Omega(\eps) = \Omega(1)$ such that the $\eps'$ unweighted sparsifiability of $C$ is $\Omega(\eps \cdot \BACNRD(C, \eps/2)) = \Omega(\BADNRD(C))$. Note that it is possible that $((b_0 + b_1) / 2) \in [1/2, 1]$, and thus may not \emph{exactly} fit into the framework of $\mathrm{BACNRD}$; in this case, one may simply double all of the entries of the code $C$, and invoke the $\mathrm{BACNRD}$ lower bound. 
\end{proof}

\subsubsection{Proof of Upper Bound}

As above, the upper bound follows from observing that each complete sub-code (in the notion of \cref{def:continuousRVNRD}) must still contain a large witness to the $\BADNRD$. We present this argument below.

\begin{proof}[Proof of \cref{thm:DiscreteDomainRVNRD}, \cref{item:DiscreteDomainRVNRDUB}.]
Indeed, let $\eps$ be given. Then, for any subset of the coordinates of $C$ (which we will still refer to as $C$), and for $\eps' = \eps / 256\log(m)$, we can sparsify $C$ to $\widetilde{O}(\BACNRD(C, \eps') / \eps^4)$ many coordinates.

To prove the theorem at hand, we claim that $\BACNRD(C, \eps') = \widetilde{O}(\BADNRD(C))$. This implies as an immediate corollary that for any subset of coordinates of $C$, we can sparsify it to $\widetilde{O}(\BADNRD(C) / \eps^4)$ many coordinates, as we desire. 

To see why $\BACNRD(C, \eps') = \widetilde{O}(\BADNRD(C))$, let us consider any subcode $C_i$ with support $A_i$ as returned by \cref{def:continuousRVNRD}. If $|A_i| = 1$, then there must already be some vector $v \in C_i$ such that $v|_{A_i} \neq 0$. Otherwise, all we will use is that $\left |C_i|_{A_i} \right| \geq 2^{|A_i|}$, which is due to the fact that for every choice of $B \subseteq A_i$, there is a different codeword $v \in (C_i)|_{A_i}$. Now, we can simply invoke \cref{cor:completeSubmatrixLB} to conclude that there must exist symbols $b_0 \neq b_1 \in \{0, a_1, \dots a_k\}$, along with a set $A'_i \subseteq A_i$ such that $|A'_i| = \widetilde{\Omega}(|A_i|)$ and $\{b_0, b_1\}^{A'_i} \subseteq (C_i)|_{A'_i}$. Note that this is already in the form of an $\BADNRD$ witness. Thus, we see that 
\[
\BADNRD(C) \geq \sum_{i = 1}^p |A'_i| = \sum_{i = 1}^p \widetilde{\Omega}(|A_i|) = \widetilde{\Omega}(\BACNRD(C, \eps')),
\]
as we desire.
\end{proof}

\section{Sparsifying Continuous Unbounded Codes}\label{sec:continuousUnbounded}

In this section, we consider sparsifying a code $C \subseteq \R_{\geq 0}^m$.
\begin{remark}\label{rmk:WLOGscale}
    Because sparsification is unchanged under scaling, we assume without loss of generality that for every vector $v \in C$, $\max_{i \in [m]} v_i = m^3$. This is because, for any codeword $c \in C$, sparsifying any \emph{scaling} of the codeword $\alpha \cdot c$ for $\alpha > 0$ is \emph{equivalent} to sparsifying $c$. 
\end{remark}
Note that the choice of having the upper end of the interval at $m^3$ is intentional; during the proof, we will end up ``thresholding'' around the value $1$ as it is sufficiently smaller than $m^3$.

Towards this end, we first define the notion of \emph{continuous-valued non-redundancy} which we will show governs the sparsifiability of these codes:

\begin{definition}\label{def:CVNRD}
    Let $C \subseteq \{[0, m^3] \}^m$, and let $\eps > 0$. $\CVNRD(C, \eps)$ is defined as the maximum $\ell$ such that there exists $C_1, \dots C_p \subseteq C$ along with $A_1, \dots A_p \subseteq [m]$ such that:
    \begin{enumerate}
        \item $\ell = \sum_{i =1}^p |A_i|$.
                \item For each $i \in [p]$, if $|A_i| = 1$, $C_i$ is of size $1$, and we set $b_2^{(i)} = c_j > 0$ for the single codeword $c \in C_i$ and single coordinate $j \in A_i$.
        \item For each $i \in [p]$, if $|A_i| \geq 2$, there exists $b_1^{(i)}, b_2^{(i)} \in \R$, $b_1^{(i)} < b_2^{(i)} (1 - \eps)$ such that for every $B \subseteq A_i$, there is a $c \in C_i$ such that for $j \in B$, $c_j \in b_1^{(i)} \cdot[1, 1 + \eps / 100\log^2(m)]$ and for $j \in A_i - B$, $c_j \in b_2^{(i)} \cdot [1, 1 + \eps /100 \log^2(m)]$.
        \item For every $i \in [p]$, and for every $c \in C_i$, $\wt(c|_{A_{\neq i}}) \leq \frac{\eps b_2^{(i)} \cdot |A_i|}{100 \log^2(m)}$.
        \item For every $i \in [p]$ and every $c \in C_i$, $\max_{j \in A_{\neq i}} c_j \leq \frac{b_2^{(i)}}{100\log(m)}$.
    \end{enumerate}
\end{definition}

Note that we will also occasionally use an even more parameterized version of $\CVNRD$:

\begin{definition}\label{def:CVNRDparameter}
    Let $C \subseteq \{[0, m^3] \}^m$, and let $\eps, \chi, \rho > 0$. $\CVNRD(C, \eps, \chi, \rho)$ is defined as the maximum $\ell$ such that there exists $C_1, \dots C_p \subseteq C$ along with $A_1, \dots A_p \subseteq [m]$ such that:
    \begin{enumerate}
        \item $\ell = \sum_{i =1}^p |A_i|$.
                \item For each $i \in [p]$, if $|A_i| = 1$, $C_i$ is of size $1$, and we set $b_2^{(i)} = c_j > 0$ for the single codeword $c \in C_i$ and single coordinate $j \in A_i$.
        \item For each $i \in [p]$, if $|A_i| \geq 2$, there exists $b_1^{(i)}, b_2^{(i)} \in \R$, $b_1^{(i)} < b_2^{(i)} (1 - \eps)$ such that for every $B \subseteq A_i$, there is a $c \in C_i$ such that for $j \in B$, $c_j \in b_1^{(i)} \cdot[1, 1 + \eps \cdot \chi]$ and for $j \in A_i - B$, $c_j \in b_2^{(i)} \cdot [1, 1 + \eps \cdot \chi]$.
        \item For every $i \in [p]$, and for every $c \in C_i$, $\wt(c|_{A_{\neq i}}) \leq \rho \cdot b_2^{(i)} \cdot |A_i|$.
        \item For every $i \in [p]$ and every $c \in C_i$, $\max_{j \in A_{\neq i}} c_j \leq \rho \cdot b_2^{(i)}$.
    \end{enumerate}
\end{definition}

When we do not include $\chi, \rho$ in the parameterization of $\CVNRD$, we are thus effectively taking $\chi = \frac{1}{100\log^2(m)}, \rho = \frac{\eps}{100 \log^2(m)}$.

In this section, we prove the following theorem:

\begin{theorem}
    Let $C \subseteq [0,m^3]^{m}, \eps >0, \chi > 0, \rho > 0$ be given. Let $\eps' = \frac{\eps^4}{10^{13}\log^{8}(m)}$. Then,
    \[
    \mathrm{SPR}(C, \eps) = O\left ( \frac{\CVNRD(C, \eps', \chi, \rho)}{\chi^6 \rho}\cdot \mathrm{poly}(\log(m), \eps^{-1})\right ).
    \]
    Taking $\chi = \frac{1}{100\log^2(m)}, \rho = \frac{\eps}{100 \log^2(m)}$, we also obtain that 
    \[
    \mathrm{SPR}(C, \eps) = O\left ( \CVNRD(C, \eps')\cdot \mathrm{poly}(\log(m), \eps^{-1})\right ).
    \]
\end{theorem}

\subsection{Warm-Up: $\CVNRD$ Gives Random Sampling Sparsification Lower Bounds}

To begin building intuition for the behavior of $\CVNRD$, we start by showing that $\CVNRD$ constitutes a lower bound for any scheme which produces sparsifiers by random sampling: 

\begin{theorem}
    Let $C \subseteq \R_{\geq 0}^m$, and let $\eps > 0$. Then,
    \[
    \mathrm{RS}(C, \eps) = \Omega(\eps \cdot \CVNRD(C, 10\eps)).
    \]
\end{theorem}

\begin{proof}
    Indeed, let $C$ be given, and define $\eps' = 10 \eps$. Let $C_1, \dots C_p \subseteq C$ along with $A_1, \dots A_p \subseteq [m]$ denote the maximal size $\CVNRD(C, \eps')$ witness of $C$ such that:
    \begin{enumerate}
        \item $\CVNRD(C, \eps') = \ell = \sum_{i =1}^p |A_i|$.
        \item For each $i \in [p]$, if $|A_i| \geq 2$, there exists $b_1^{(i)}, b_2^{(i)} \in \R$, $b_1^{(i)} < b_2^{(i)} (1 - \eps')$ such that for every $B \subseteq A_i$, there is a $c \in C_i$ such that for $j \in B$, $c_j \in b_1^{(i)} \cdot[1, 1 + \eps' / 100\log^2(m)]$ and for $j \in A_i - B$, $c_j \in b_2^{(i)} \cdot [1, 1 + \eps' /100 \log^2(m)]$.
        \item For each $i \in [p]$, if $|A_i| = 1$, $C'_i$ is of size $1$, and we set $b_2^{(i)} = c_j > 0$ for the single codeword $c \in C'_i$ and single coordinate $j \in A_i$.
        \item For every $i \in [p]$, and for every $c \in C_i$, $\wt(c|_{A_{\neq i}}) \leq \frac{\eps' b_2^{(i)} \cdot |A_i|}{100 \log^2(m)}$.
    \end{enumerate}

    Define $S = A_1 \cup \dots \cup A_p$, and let $C' = C|_S$. We assume WLOG that $S$ is the first $|S|$ many coordinates of $C$, and let $m' = |S|$.
    We will show that $\mathrm{RSPR}(C', \eps) = \Omega(\eps \cdot \CVNRD(C, \eps'))$, thereby obtaining our desired theorem (see \cref{def:random-sparsification} and \cref{def:randomSparsifier}).

    To see this, let us suppose for the sake of contradiction that there is a probability distribution $\Delta$ over $w:[m] \rightarrow \R_{\geq 0}$ such that $\E_{w \sim \Delta} |\Supp(w)| \leq \frac{\eps' m'}{200}$, and yet $\Delta$ is a $(1 \pm \eps)$ random sparsifier of $C'$. Immediately, this implies that there must exist a block $A_i: i \in [p]$ such that $\E_{w \sim \Delta} |\Supp(w) \cap A_i| \leq \frac{\eps' |A_i|}{200}$. Thus, when we do random sampling, we see (by a simple Markov bound) that with probability $\geq 1/2$, fewer than $\frac{\eps'|A_i|}{100}$ many coordinates from $A_i$ are selected. Conditioned on this event, we let $T_i$ denote the set of sampled coordinates in $A_i$, and we let $w$ denote the corresponding weights.

    We can quickly observe that if $|A_i| = 1$, then we are easily done, as no coordinates from $A_i$ are in the sparsifier. The only remaining weight for the codeword $c \in C_i$ comes from $\wt(c|_{A_{\neq i}}) \leq \frac{\eps' \wt(c)}{100 \log^2(m)}$, and because the sampling scheme for weights is unbiased, the expected weight placed on coordinates $A_{\neq i}$ is $\frac{\eps' \wt(c)}{100 \log^2(m)}$. By a simple Markov bound, this means that after sampling, with probability $\geq 3/4$, the reported weight for codeword $c$ is $\leq \frac{4\eps' \wt(c)}{100 \log^2(m)} \notin (1\pm \eps) \cdot \wt(c)$.

    Otherwise, we focus on the situation when $|A_i| \geq 2$. In this case, we consider the two codewords $c \in C_i$:\begin{enumerate}
        \item $c$ such that, for every $j \in T_i$,$c_j \in b_2^{(i)} \cdot [1, 1 + \eps' /100 \log^2(m)]$, and for $j \in A_i - T_i$, $c_j \in b_1^{(i)} \cdot [1, 1 + \eps' /100 \log^2(m)]$.
        \item $c'$ such that, for every $j \in A_i - T_i$,$c'_j \in b_2^{(i)} \cdot [1, 1 + \eps' /100 \log^2(m)]$, and for $j \in T_i$, $c'_j \in b_1^{(i)} \cdot [1, 1 + \eps' /100 \log^2(m)]$.
    \end{enumerate}

    Importantly, we then see that the weights assigned to these codewords from the sparsifier (in the block $A_i$) satisfy
    \begin{align}\label{eq:viewFromSparsifier}
    \langle w|_{T_i}, c|_{T_i} \rangle \geq \frac{b_2^{(i)}}{b_1^{(i)} \cdot (1 + \eps' / 100 \log^2(m))} \cdot  \langle w|_{T_i}, c'|_{T_i} \rangle,
    \end{align}
    whereas originally, 
    \begin{align*}
    \sum_{j \in A_i} c_j \leq |T_i| \cdot b_2^{(i)} \cdot (1 + \eps' /100 \log^2(m)) + (|A_i| - |T_i|) \cdot b_1^{(i)} \cdot (1 + \eps' /100 \log^2(m)),
    \end{align*}
    and 
    \begin{align*}
    \sum_{j \in A_i} c'_j \geq |T_i| \cdot b_1^{(i)}  + (|A_i| - |T_i|) \cdot b_2^{(i)}.
    \end{align*}

    Now, recall that in the original code, the ``off-block'' weight of each codeword was bounded (condition 4 above). Thus, we can see that 
    \begin{align}\label{eq:origWeightc}
    \sum_{j \in [m']} c_j \leq |T_i| \cdot b_2^{(i)} \cdot (1 + \eps' /100 \log^2(m)) + (|A_i| - |T_i|) \cdot b_1^{(i)} \cdot (1 + \eps' /100 \log^2(m)) + \frac{\eps' b_2^{(i)} \cdot |A_i|}{100 \log^2(m)}.
    \end{align}
    and 
    \begin{align}\label{eq:origWeightcprime}
    \sum_{j \in [m']} c'_j \geq |T_i| \cdot b_1^{(i)}  + (|A_i| - |T_i|) \cdot b_2^{(i)}.
    \end{align}
Using the fact that $b_1^{(i)} < b_2^{(i)} (1 - \eps')$, and that $|T_i| \leq \frac{\eps'}{100} \cdot |A_i|$,
\begin{align}\label{eq:origWeightRatio}
\frac{\sum_{j \in [m']} c'_j}{\sum_{j \in [m']} c_j} \geq \frac{|T_i| \cdot b_1^{(i)}  + (|A_i| - |T_i|) \cdot b_2^{(i)}}{|T_i| \cdot b_2^{(i)} \cdot (1 + \eps' /100 \log^2(m)) + (|A_i| - |T_i|) \cdot b_1^{(i)} \cdot (1 + \eps' /100 \log^2(m)) + \frac{\eps' b_2^{(i)} \cdot |A_i|}{100 \log^2(m)}}
\end{align}
\[
\geq \frac{(|A_i| - |T_i|) }{|T_i|\cdot (1 + \eps' /100 \log^2(m)) + (|A_i| - |T_i|) \cdot (1 - \eps') \cdot (1 + \eps' /100 \log^2(m)) + \frac{\eps'\cdot |A_i|}{100 \log^2(m)}}
\]
\[
\geq \frac{(1 - \eps' / 100) }{\frac{\eps'}{100}\cdot (1 + \eps' /100 \log^2(m)) + (1 - 99\eps'/100)  + \frac{\eps'}{100 \log^2(m)}}
\]
\[
\geq \frac{(1 - \eps' / 100) }{(1 - 96 \eps'/ 100)} \geq (1 + 9\eps' / 10).
\]

    Importantly, because our sampling scheme is unbiased, we also know that in the supposed \emph{sparsifier}, it must also be the case that (with probability $\geq 3/4$ by a Markov bound) 
    \[
    \langle w_{T \neq i}, c'|_{T \neq i} \rangle \leq 4 \cdot \frac{\eps' b_2^{(i)} \cdot |A_i|}{100 \log^2(m)}.
    \]
    We can then see that
    \[
     \langle w|_{T}, c'|_{T} \rangle \leq \langle w|_{T_i}, c'|_{T_i} \rangle + 4 \cdot \frac{\eps' b_2^{(i)} \cdot |A_i|}{100 \log^2(m)}.
     \]
     Thus, 
     \[
     \langle w|_{T_i}, c'|_{T_i} \rangle \geq \langle w|_{T}, c'|_{T} \rangle - 4 \cdot \frac{\eps' b_2^{(i)} \cdot |A_i|}{100 \log^2(m)}.
     \]
     If we assume that $\langle w|_{T}, c'|_{T} \rangle$ is indeed a $(1 \pm \eps'/10)$ sparsifier of $c'$, then using \cref{eq:origWeightcprime}, which states that $
     \sum_{j \in [m']} c'_j \geq |T_i| \cdot b_1^{(i)}  + (|A_i| - |T_i|) \cdot b_2^{(i)} \geq (1 - \eps' / 100) \cdot |A_i| \cdot b_2^{(i)}$, we have 
     \[
     \langle w|_{T_i}, c'|_{T_i} \rangle \geq \langle w|_{T}, c'|_{T} \rangle - 4 \cdot \frac{\eps' b_2^{(i)} \cdot |A_i|}{100 \log^2(m)} \geq \langle w|_{T}, c'|_{T} \rangle - \frac{4 \eps'}{100 \log^2(m)} \cdot (1 - \eps'/100) \cdot \sum_{j \in [m']} c'_j.
     \]
     \[
     \geq (1 -\eps'/10) \cdot \sum_{j \in [m']} c'_j - (\eps'/100) \cdot \sum_{j \in [m']} c'_j \geq (1 - \eps'/5) \cdot \sum_{j \in [m']} c'_j.
     \]
     Then, using \cref{eq:viewFromSparsifier}, we see that 
     \[
     \langle w|_{T_i}, c|_{T_i} \rangle \geq \frac{b_2^{(i)}}{b_1^{(i)} \cdot (1 + \eps' / 100 \log^2(m))} \cdot  \langle w|_{T_i}, c'|_{T_i} \rangle \geq \frac{1}{1-\eps'} \cdot \frac{1}{1 + \eps / 100 \log^2(m)} \cdot (1 - \eps'/5) \cdot \sum_{j \in [m']} c'_j
     \]
     \begin{align}\label{eq:lowerboundSparsifierweightc}
     \geq (1 + 7 \eps' / 10) \cdot \sum_{j \in [m']} c'_j.
     \end{align}
     Finally then, we see that using \cref{eq:lowerboundSparsifierweightc} and our assumption that $w$ preserves the weight of $c'$ to a $(1 \pm \eps'/10)$ factor, that 
     \begin{align}\label{eq:sparsifierWeightRatio}
     \frac{\langle w|_T, c|_T \rangle}{\langle w|_T, c'|_T \rangle} \geq \frac{\langle w|_{T'}, c|_{T_i} \rangle}{\langle w|_T, c'|_T \rangle} \geq \frac{(1 + 7 \eps' / 10) \cdot \sum_{j \in [m']} c'_j}{(1 + \eps'/10) \cdot \sum_{j \in [m']} c'_j} \geq (1 + 5\eps'/10).
     \end{align}

     But, this is now a contradiction: indeed by \cref{eq:origWeightRatio},
     \[
     \frac{\sum_{j \in [m']} c'_j}{\sum_{j \in [m']} c_j} \geq (1 + 9\eps' / 10),
     \]
     whereas by \cref{eq:sparsifierWeightRatio}, 
     \[
     \frac{\langle w|_T, c|_T \rangle}{\langle w|_T, c'|_T \rangle} \geq (1 + 5\eps'/10).
     \]
     Together, these imply that 
     \[
     \frac{\langle w|_T, c|_T \rangle}{\langle w|_T, c'|_T \rangle} \cdot \frac{\sum_{j \in [m']} c'_j}{\sum_{j \in [m']} c_j} \geq (1 + 14\eps'/10),
     \]
     which means that either (1) $\frac{\langle w|_T, c|_T \rangle}{\sum_{j \in [m']} c_j} \geq (1 + 6\eps'/10)$ or that (2) $\frac{\sum_{j \in [m']} c'_j}{\langle w|_T, c'|_T \rangle} \geq (1 + 6\eps'/10)$, both of which violate our assumption that $w$ is a $(1 \pm \eps'/10)$ sparsifier of $C$. Note that we conditioned on two events: (a) that we do not sample more than $\frac{\eps' |A_i|}{100}$ coordinates from $A_i$ (which occurs with probability $\geq 1/2$), and (b) that $\langle w_{T \neq i}, c'|_{T \neq i} \rangle \leq 4 \cdot \frac{\eps' b_2^{(i)} \cdot |A_i|}{100 \log^2(m)}$, which occurs with probability $\geq 3/4$ by a Markov bound. Thus, we know that with probability $\geq 1/2$, the sampling scheme with probabilities $q$ is \emph{not} a $(1 \pm \eps)$ random sparsifier of $C$, thus concluding the theorem.
\end{proof}

\subsection{Notation}

To prove our main theorems in this section, we first introduce some new pieces of notation. 

\begin{definition}
    For a codeword $c \in [0,m^3]^{m}$, a value $a \in [\frac{1}{1+\eps}, m^3]$, and $\eps > 0$, we use $c^{(a, \eps)}$ to refer to the codeword  $\in \{[a \cdot (1 + \eps)^{-1}, a \cdot (1 + \eps)^2]\}^m$ where for $i \in [m]$, 
    \[
    c^{(a, \eps)}_i = \min(\max(c_i, a \cdot (1 + \eps)^{-1}), a\cdot (1 + \eps)^2).
    \]
    Intuitively, any symbol in $c$ that is less than $a \cdot (1 + \eps)^{-1}$ gets rounded up to $a \cdot (1 + \eps)^{-1}$, and any symbol greater than $a \cdot (1 + \eps)^2$ gets rounded down to $a \cdot (1 + \eps)^2$.
\end{definition}

\begin{definition}\label{def:roundingCodes}
    For $C \subseteq [0,m^3]^{m}$, a value $a \in [1, m^3]$ and $\eps > 0$, and a value $\beta \in \Z$, we use $C_{a, \beta, \eps}$ to refer to a code $\subseteq \{[a \cdot (1 + \eps)^{-1}, a \cdot (1 + \eps)^2]\}^m$ such that:
    \[
    C_{a, \beta, \eps} = \{c^{(a, \eps)}: c \in C, |\{i\in [m]: c_i \in [a, (1 + \eps)a ]\}| \in [\beta, 2\beta] \}.
    \]
    In words, this is the code which includes only codewords $c$ in $C$ which have between $\beta$ and $2\beta$ many symbols with value $[a, (1 + \eps)a ]$, and then for each such codeword, replaces it with the $c^{(a, \eps)}$ version. 
\end{definition}

Now, we partition the codewords in $C_{a, \beta, \eps}$:

\begin{definition}\label{def:boundary}
    For $\eta \in \left [\frac{200 \log^4(m)}{\eps^4}\right] $, and a codeword $c \in C \subseteq [0,m^3]^{m}$, we define the $\eta$-boundary weight of $c$ to be: 
    \[
     \sum_{a \in \{1, (1+\eps), (1 + \eps)^2, \dots \} } \sum_{i \in [m]} c_i \cdot \mathbf{1}\left [c_i \in \left [a\cdot \left (1 + (\eta-1/2) \cdot \frac{\eps^4}{\mathrm{200\log^4(m)}} \right), a\cdot \left (1 + (\eta+1/2) \cdot \frac{\eps^4}{\mathrm{200\log^4(m)}} \right ) \right ] \right).
    \]
\end{definition}

We have the following claim:

\begin{claim}\label{clm:partitionBoundaryWeight}
    For each codeword $c \in C \subseteq [0,m^3]^{m}$, there exists a choice of $\eta$ such that the $\eta$-boundary weight of $c$ is $\leq \frac{\eps^4\wt(c)}{200 \log^4(m)}$.
\end{claim}

\begin{proof}
    This follows because the sum over all choices of $\eta \in \left [\frac{200 \log^4(m)}{\eps^4}\right] $ exactly counts the weight of $c$. So, there must exist a choice of $\eta$ such that $\eta$-boundary weight of $c$ is $\leq \frac{\eps^4\wt(c)}{200 \log^4(m)}$.
\end{proof}

\begin{remark}\label{rmk:smallBoundary}
Going forward, we partition the codewords $c \in C$ by the value of $\eta$ which achieves \cref{clm:partitionBoundaryWeight}. We will avoid adding a subscript of $\eta$ to the code $C$, and will pay a factor of $\frac{200 \log^4(m)}{\eps^4\wt(c)}$ at the end to take a union over sparsifying all of these sets of codewords. We assume WLOG that $\eta = 1$.
\end{remark}

\begin{definition}
    For a codeword $c \in [0,m^3]^m$, a value $a \in [0,m^3]$ and a value $\eps > 0$, we let:
    \begin{enumerate}
        \item $\wt_{a, \eps}(c) = \sum_{i = 1}^m c_i \cdot \mathbf{1}[c_i \in a \cdot[1, (1 + \eps)]]$.
        \item $\#_{a, \eps}(c) = \sum_{i = 1}^m\mathbf{1}[c_i \in a \cdot[1, (1 + \eps)]]$.
        \item More generally, for a symbol $b \in [0,m^3]$ $\wt_{[a,b]}(c) = \sum_{i = 1}^m c_i \cdot \mathbf{1}[c_i \in [a,b]]$ and $\#_{[a,b]}(c) = \sum_{i = 1}^m\mathbf{1}[c_i \in [a,b]]$.
    \end{enumerate}
\end{definition}

For this section, we will use \emph{additive covers} instead of multiplicative covers.

\begin{definition}\label{def:additiveCover}
    Let $C \subseteq \left ( [0, m^3]\right)^m$ be a code. For a parameter $\eps > 0$, we say that $C' \subseteq \left ( [0, m^3] \right)^m$ is an $\eps$-cover of $C$ if, for every $c \in C$ there exists a $c' \in C'$ such that for every $i \in [m]$, 
    \[
    c_i \in [c'_i \pm \eps].
    \]
    We use $|\mathrm{Cover}(C, \eps)|$ to denote the size of the \emph{smallest} $\eps$-cover of $C$.
\end{definition}

The appropriate translation of continuous Sauer-Shelah for this setting is the following:

\begin{lemma}[Additive Cover Sauer-Shelah, see Lemma 3.5 in \cite{AlonBCH97}]\label{lem:additiveABCH}
Let $C \subseteq [a, b]^m$. For $\eps \in (0,1)$ let $d = \mathrm{fdim}(C, \eps/4)$. Then,
\[
|\mathrm{Cover}(C, \eps)| \leq 2 \cdot \left ( \frac{4m(b-a)^2}{\eps^2}\right )^{d\log(2en(b-a) / d\eps)}.
\]
\end{lemma}

\subsection{Defining the Peeled Sets}

With the above preliminaries established, we are now ready to proceed to a discussion of how we build our sparsifier. As before, our goal, given a code $C \subseteq [0,m^3]^{m}$ and parameter $\eps > 0$, is to peel off a small set of coordinates $S \subseteq [m]$, such that when we sample the remaining coordinates $[m] - S$ at rate $1/2$ (and give these sampled coordinates weight $2$), the resulting code is a $(1 \pm \eps)$ sparsifier of $C$. 

To do this, let us fix a value of $a \in \ \left\{\frac{1}{1 + \eps}, 1, (1 + \eps), \dots \frac{m}{(1 + \eps)}, m \right \}$, and a value of $\beta \in \{1, 2, 4, 8, \dots m \}$. 
Now, observe that, for a codeword $c \in C$, if $\wt_{a, \eps}(c) \leq \frac{\eps^2 \cdot \wt(c)}{100 \log(m)}$, then the codeword $c$ is effectively already taken care of: indeed, under random sampling at rate $1/2$, the maximum additive error introduced by symbols in $c$ with value $\in a \cdot[1, 1 + \eps]$ is at most $\wt_{a, \eps}(c)$, which is a negligible portion of the weight of $c$. Thus, we instead focus (WLOG) only on codewords $c \in C$ where $\wt_{a, \eps}(c) \geq \frac{\eps^2 \cdot \wt(c)}{100 \log(m)}$. 

In particular, when we focus our attention on $C_{a, \beta, \eps}$, this means that we can assume (without loss of generality) that for every codeword $c^{(a, \beta, \eps)} \in C_{a, \beta, \eps}$, that the corresponding codeword $c \in C$ satisfied $\wt_a(c) \geq \frac{\eps^2 \cdot \wt(c)}{100 \log(m)}$, and so in particular, 
\[
a \cdot (1 +\eps) \cdot 2 \beta \geq \wt_a(c) \geq \frac{\eps^2 \cdot \wt(c)}{100 \log(m)}. 
\]
Equivalently, \begin{align}
    a \cdot \beta \geq \frac{\eps^2 \wt(c)}{400 \log(m)}. \label{eq:contributionFromaSymbols}
\end{align}

Now, let us consider the set $C_{a, \beta, \eps}$ (note that there is an implicit dependence here on $\eta$ - see \cref{rmk:smallBoundary}). We let $S_{a, \beta, \eps} \subseteq [m]$ denote the \emph{smallest} set of coordinates that can be removed such that 
\begin{align}\label{eq:defOfS}
 \left |\mathrm{Cover} \left ( \left ( C_{a, \beta, \eps} \right ) _{\overline{S_{a, \beta, \eps}}}, \frac{a\eps^4}{1600 \log^4(m)}\right )\right| \leq 2^{\eps^2 \beta / 1000 \log^4(m)}.
\end{align}
There is a minor edge case here: if $\beta$ is too small, the exponent may not be well-defined. Thus, we declare that if $\beta \leq \frac{10000 \log^5(m)}{\eps^2}$, then we will simply remove all $i \in [m]$ such that there is $c$ in $C_{a, \beta, \eps}$ such that $c_i \in [a, a(1+\eps)]$, and add these coordinates to $S_{a, \beta, \eps}$.

As our first step, we show that, after peeling the above coordinates, random sampling at rate $1/2$ succeeds in preserving codeword weights:

\begin{claim}\label{clm:rightNumberSymbols}
     Let $a, \beta, \eps, C$ be given, and let $C_{a, \beta, \eps}$ and $S_{a, \beta, \eps}$ be defined as above. Let $C' \subseteq C$ denote the set of all codewords which are mapped to $C_{a, \beta, \eps}$ under \cref{def:roundingCodes}. Now, let $S \subseteq [m]$ denote a set of coordinates which is removed such that $S_{a, \beta, \eps} \subseteq S$. Let $T$ denote a random sample at rate $1/2$ of the coordinates $[m] - S$. Then, with probability $1 - 2^{-\frac{\eps^2 \beta}{200 \log^4(m)}}$, for every codeword $c \in C'$, it is the case that
\[
    \left | \#_{\left [a, a(1 + \eps) \right ]}(c|_{T}) - \frac{\#_{\left [a, a(1 + \eps) \right ]}(c|_{\overline{S}})}{2} \right | \leq  \frac{\eps \beta}{8 \log^2(m)}.
    \]
\end{claim}

\begin{proof}
    After removing the coordinates $S$, we know that 
    \[
    \left |\mathrm{Cover} \left ( \left ( C_{a, \beta, \eps} \right ) _{\overline{S}}, \frac{a\eps^4}{1600 \log^4(m)}\right )\right| \leq 2^{\eps^2 \beta / 1000 \log^4(m)}.
    \]

    So, let us consider the code $D :=  \mathrm{Cover} \left ( \left ( C_{a, \beta, \eps} \right ) _{\overline{S}}, \frac{a\eps^4}{1600 \log^4(m)}\right ) \subseteq \left [\frac{a}{1 + \eps}, a \cdot (1 + \eps)^2 \right ]^{[m] - S}$. First, we observe that for every codeword $v \in D$, $\#_{\left [a(1 + \frac{\eps^4}{800 \log^4(m)}), a(1 + \eps - \frac{\eps^4}{800 \log^4(m)}) \right ]} (v) \leq 2 \beta$. Note that this is because $D$ is an $\frac{a\eps^4}{1600 \log^4(m)}$ cover of $C_{a, \beta, \eps}$, where every codeword has $\leq 2 \beta$ symbols in the range $[a, (1 + \eps) a]$. If some codeword $v \in D$ had$\#_{\left [a(1 + \frac{\eps^4}{800 \log^4(m)}), a(1 + \eps - \frac{\eps^4}{800 \log^4(m)}) \right ]} (v)  > 2 \beta$, this would then imply that some codeword in $C_{a, \beta, \eps}$ has $> 2 \beta$ symbols in the range \[
    \left (a\left (1 + \frac{\eps^4}{800 \log^4(m)} \right ) - \frac{a\eps^4}{1600 \log^4(m)}, a\left(1 + \eps - \frac{\eps^4}{800 \log^4(m)} \right ) + \frac{a\eps^4}{1600 \log^4(m)} \right ) \in [a, a \cdot (1 + \eps)].
    \]

    Next, we claim that, for a single $v \in D$, a simple Hoeffding bound implies that 
    \[
    \Pr_{T}\left [\left | \#_{\left [a(1 + \frac{\eps^4}{800 \log^4(m)}), a(1 + \eps - \frac{\eps^4}{800 \log^4(m)}) \right ]}(v|_{T}) - \frac{\#_{\left [a(1 + \frac{\eps^4}{800 \log^4(m)}), a(1 + \eps - \frac{\eps^4}{800 \log^4(m)}) \right ]}(v|_{\overline{S}})}{2} \right | \geq t \right ] 
    \]
    \[
    \leq 2e^{-2t^2 / 2\beta} \leq 2e^{-t^2 / \beta}.
    \]
    Setting $t = \frac{\eps \beta}{10 \log^2(m)}$, we see that 
    \[
     \Pr_{T}\left [\left | \#_{\left [a(1 + \frac{\eps^4}{800 \log^4(m)}), a(1 + \eps - \frac{\eps^4}{800 \log^4(m)}) \right ]}(v|_{T}) - \frac{\#_{\left [a(1 + \frac{\eps^4}{800 \log^4(m)}), a(1 + \eps - \frac{\eps^4}{800 \log^4(m)}) \right ]}(v|_{\overline{S}})}{2} \right | \geq  \frac{\eps \beta}{10 \log^2(m)}\right ]  
     \]
     \[
     \leq 2e^{-\frac{\eps^2 \beta}{100 \log^4(m)}}.
    \]

    Now, taking a union bound over \emph{all} $\leq 2^{\eps^2 \beta / 1000 \log^4(m)}$ codewords $v \in D$, we see that with probability $1 - 2^{-\frac{\eps^2 \beta}{200 \log^4(m)}}$, for every $v \in D$,
    \begin{align}\label{eq:preserveSupportCoverVectors}
    \left | \#_{\left [a(1 + \frac{\eps^4}{800 \log^4(m)}), a(1 + \eps - \frac{\eps^4}{800 \log^4(m)}) \right ]}(v|_{T}) - \frac{\#_{\left [a(1 + \frac{\eps^4}{800 \log^4(m)}), a(1 + \eps - \frac{\eps^4}{800 \log^4(m)}) \right ]}(v|_{\overline{S}})}{2} \right | \leq  \frac{\eps \beta}{10 \log^2(m)}.
    \end{align}

    Lastly, we consider any codeword $c \in C'$. 
    Because $D$ is an $\frac{a \eps^4}{1600 \log^4(m)}$ cover of $C_{a, \beta, \eps}$, it must be the case that there is a codeword $v \in D$ such that for $i \in [m]$, whenever $c_i \in [a, a(1 + \eps)]$, $v_i \in \left [c_i - \frac{a \eps^4}{1600 \log^4(m)}, c_i + \frac{a \eps^4}{1600 \log^4(m)} \right ]$. We fix such a codeword $v \in D$ for $c$. We let $I \subseteq [m] - S$ denote the set of indices such that $v_i \in \left [ a(1 + \frac{\eps^4}{800 \log^4(m)}), a(1 + \eps - \frac{\eps^4}{800 \log^4(m)}) \right ]$, and we thus see:
    \begin{enumerate}
        \item Whenever $c_i \geq a(1 + \frac{\eps^4}{800 \log^4(m)}) + \frac{a \eps^4}{1600 \log^4(m)}$ or $c_i \leq a(1 + \eps - \frac{\eps^4}{800 \log^4(m)})-\frac{a \eps^4}{1600 \log^4(m)}$, then $i \in I$.
        \item Whenever $c_i \leq a$ or $c_i \geq a(1 + \eps)$, then $i \notin I$.
        \item Whenever $c_i \in [a, a(1 + \frac{\eps^4}{800 \log^4(m)}) + \frac{a \eps^4}{1600 \log^4(m)}]$ or $c_i \in [a(1 + \eps - \frac{\eps^4}{800 \log^4(m)})-\frac{a \eps^4}{1600 \log^4(m)}, a(1 + \eps)]$, the index $i$ \emph{may or may not} be in $I$. We call this set of coordinates $I_{\mathrm{undetermined}}$.
    \end{enumerate}

    Importantly however, we now invoke \cref{rmk:smallBoundary} to bound the size of $I_{\mathrm{undetermined}}$:

    \begin{claim}
        $|I_{\mathrm{undetermined}}| \leq \frac{2 \eps^2 \beta}{\log^3(m)}$.
    \end{claim}

    \begin{proof}
        By \cref{rmk:smallBoundary} and \cref{def:boundary}, we know that the fraction of $c$'s weight that is in symbols in the range 
        \[
        \left [a, a + \frac{a\eps^4}{400 \log^4(m)} \right ] \cup \left [a(1 + \eps) \left (1 - \frac{\eps^4}{400 \log^4(m)} \right ), a(1 + \eps) \right ]
        \]
        is $\leq \frac{\eps^4}{200 \log^4(m)}$. In particular, 
        \[
        \left [a, a(1 + \frac{\eps^4}{800 \log^4(m)}) + \frac{a \eps^4}{1600 \log^4(m)} \right ] \cup \left [a(1 + \eps - \frac{\eps^4}{800 \log^4(m)})-\frac{a \eps^4}{1600 \log^4(m)}, a(1 + \eps) \right ]
        \]
        \[
        \subseteq \left [a, a + \frac{a\eps^4}{400 \log^4(m)} \right ] \cup \left [a(1 + \eps) \left (1 - \frac{\eps^4}{400 \log^4(m)} \right ), a(1 + \eps) \right ],
        \]
        this means that $\sum_{i \in I_{\mathrm{undetermined}}} c_i \leq \frac{\eps^4\wt(c)}{200 \log^4(m)}$.

        Because for $i \in I_{\mathrm{undetermined}}$, $c_i \geq a$, we then see that 
        \[
        |I_{\mathrm{undetermined}}| \leq \frac{\eps^4\wt(c)}{200a \log^4(m)}.
        \]
        Plugging in \cref{eq:contributionFromaSymbols}, we then see that 
        \[
        |I_{\mathrm{undetermined}}| \leq \frac{\eps^4}{200 \log^4(m)} \cdot \frac{400 \log(m) \beta}{\eps^2}\leq \frac{2 \eps^2 \beta}{\log^3(m)}.
        \]
    \end{proof}

    To conclude, we know that $\#_{[a, a(1 + \eps)]}(c) = |I \cup I_{\mathrm{undetermined}}|$. By \cref{eq:preserveSupportCoverVectors}, we know that 
    \[
    \left | \frac{|I|}{2} - |I \cap T| \right | \leq \frac{\eps \beta}{10 \log^2(m)}.
    \]
    Thus, we see that 
    \[
    \left | \#_{[a, a(1 + \eps)]}(c|_T) - \frac{\#_{[a, a(1 + \eps)]}(c|_{\overline{S}})}{2}\right | = \left | \frac{|I| + |I_{\mathrm{undetermined}} \setminus I|}{2} -  |I \cap T| - |(I_{\mathrm{undetermined}} \setminus I) \cap T|\right | 
    \]
    \[
    \leq \left | \frac{|I|}{2} - |I \cap T| \right | + \left |\frac{|I_{\mathrm{undetermined}} \setminus I|}{2} -  |(I_{\mathrm{undetermined}} \setminus I) \cap T| \right | 
    \]
    \[
    \leq \frac{\eps \beta}{10 \log^2(m)} + \frac{2\eps^3 \beta}{\log^3(m)} \leq \frac{\eps \beta}{8 \log^2(m)}.
    \]
\end{proof}

Now, we show that \cref{clm:rightNumberSymbols} implies that each codeword has its weight preserved:

\begin{claim}\label{clm:sparsifierCorrectness}
    Let $C$ be given, and let \[
    S = \bigcup_{a \in \{\frac{1}{1 + \eps}, 1, 1+\eps, \dots m\}, \beta, \eta} S_{\eta, a, \beta, \eps},
    \]
    for $S_{\eta, a, \beta, \eps}$ as defined in \cref{eq:defOfS}. Let $T$ be a random sample of the coordinates in $[m] - S$ at rate $1/2$.
    Then, with probability $\geq 1/2$, for every codeword $c \in C$, it is the case that
    \[
    \left | \wt(c) - (\wt(c|_S) + 2 \cdot \wt(c|_T))\right |  \leq 5 \eps \wt(c).
    \]
    In particular, there exists a choice of set $T$ such that the above holds.
\end{claim}

\begin{proof}
    To start, our probability bound holds by taking a union bound over \cref{clm:rightNumberSymbols} for every $a \in \{\frac{1}{1 + \eps}, 1, 1+\eps, \dots m\}, \beta \in \{1, 2, 3, \dots m\}, \eta \in \left [ \frac{200 \log^4(m)}{\eps^4} \right ]$. In particular, recall that in our construction, whenever $\beta \leq \frac{10000 \log^5(m)}{\eps^2}$, then we simply remove all $i \in [m]$ such that there is $c$ in $C_{a, \beta, \eps}$ such that $c_i \in [a, a(1+\eps)]$, and add these coordinates to $S_{a, \beta, \eps}$. Thus, \cref{clm:rightNumberSymbols} holds deterministically for $\beta \leq \frac{10000 \log^5(m)}{\eps^2}$. When $\beta > \frac{10000 \log^5(m)}{\eps^2}$, \cref{clm:rightNumberSymbols} holds with probability $\geq 1 - m^{50}$, which is sufficient for taking a union bound over all the choices of $a, \beta, \eta$.

    Now, it remains to show the correctness. For this, we consider any codeword $c \in C$. We can write 
    \[
    \wt(c) = \sum_{i \in [m]} c_i \cdot \mathbf{1}\left [c_i \leq \frac{1}{1+\eps} \right] + \sum_{a\in \{\frac{1}{1 + \eps}, 1, 1+\eps, \dots m \}} \sum_{i = 1}^m c_i \cdot \mathbf{1}[(c_i) \in [a, a \cdot (1 + \eps))].
    \]
    Our goal is to bound 
    \[
    \left | \wt(c) - (\wt(c|_S) + 2 \cdot \wt(c|_T))\right | \leq \sum_{i \in [m]} c_i \cdot \mathbf{1}\left [c_i \leq \frac{1}{1+\eps} \right] +
    \]
    \[
\sum_{a\in \{\frac{1}{1 + \eps}, 1, 1+\eps, \dots m \}} \left | \sum_{i \in [m]  - S} c_i \cdot \mathbf{1}[(c_i) \in [a, a \cdot (1 + \eps))] - 2 \cdot \sum_{i \in T} c_i \cdot \mathbf{1}[(c_i) \in [a, a \cdot (1 + \eps))] \right |.
    \]

    For our codeword $c$, we fix a symbol $a$, and let $\beta$ denote the number such that the number of occurrences of the symbols in the range $[a, a(1 +\eps))$ in $c$ is in the range $[\beta, 2 \beta)$. Now, we have several cases:
    \begin{enumerate}
        \item Suppose that $\wt_{a, \eps}(c) < \frac{\eps^2 \cdot \wt(c)}{100 \log(m)}$. Then, after sampling, it must be the case that
        \[
        \left | \sum_{i \in [m]  - S} c_i \cdot \mathbf{1}[(c_i) \in [a, a \cdot (1 + \eps))] - 2 \cdot \sum_{i \in T} c_i \cdot \mathbf{1}[(c_i) \in [a, a \cdot (1 + \eps))] \right | \leq \frac{\eps^2 \cdot \wt(c)}{100 \log(m)}.
        \]
        \item Otherwise, suppose that $\wt_{a, \eps}(c) \geq \frac{\eps^2 \cdot \wt(c)}{100 \log(m)}$. Then, by \cref{clm:rightNumberSymbols}, we know that \[
        \left | \#_{\left [a, a(1 + \eps) \right ]}(c|_{T}) - \frac{\#_{\left [a, a(1 + \eps) \right ]}(c|_{\overline{S}})}{2} \right | \leq  \frac{\eps \beta}{8 \log^2(m)}.
        \] 
        So, we see that 
        \[
        \left | \sum_{i \in [m]  - S} c_i \cdot \mathbf{1}[(c_i) \in [a, a \cdot (1 + \eps))] - 2 \cdot \sum_{i \in T} c_i \cdot \mathbf{1}[(c_i) \in [a, a \cdot (1 + \eps))] \right |
        \]
        \[
        \leq \left | \sum_{i \in [m]  - S} c_i \cdot \mathbf{1}[(c_i) \in [a, a \cdot (1 + \eps))] - \sum_{i \in [m]  - S} a\cdot \mathbf{1}[(c_i) \in [a, a \cdot (1 + \eps))] \right | \]
        \[
        + 2 \left | \sum_{i \in T} c_i \cdot \mathbf{1}[(c_i) \in [a, a \cdot (1 + \eps))] - \sum_{i \in T} a\cdot \mathbf{1}[(c_i) \in [a, a \cdot (1 + \eps))] \right |
        \]
        \[
        +\left | \sum_{i \in [m]  - S} a \cdot \mathbf{1}[(c_i) \in [a, a \cdot (1 + \eps))] - 2 \cdot \sum_{i \in T} a \cdot \mathbf{1}[(c_i) \in [a, a \cdot (1 + \eps))] \right |
        \]
        \[
        \leq (\eps \cdot a) \cdot \left (\#_{[a, a \cdot (1 + \eps)]}(c|_{\bar{S}}) + 2\#_{[a, a \cdot (1 + \eps)]}(c|_{T}) \right ) + \frac{2\eps \beta}{8 \log^2(m)} \cdot a
        \]
        \[
        \leq (\eps \cdot a) \cdot (3\beta) + \frac{a\eps \beta}{4 \log^2(m)} \leq 4 \eps a \beta,
        \]
        where the second to last line uses the condition guaranteed by \cref{clm:rightNumberSymbols}.
    \end{enumerate}

    In total then, we see that 
\[
    \left | \wt(c) - (\wt(c|_S) + 2 \cdot \wt(c|_T))\right | \leq \sum_{i \in [m]} c_i \cdot \mathbf{1}\left [c_i \leq \frac{1}{1+\eps} \right] +
    \]
    \[
   \sum_{a\in \{\frac{1}{1 + \eps}, 1, 1+\eps, \dots m \}} \left | \sum_{i \in [m]  - S} c_i \cdot \mathbf{1}[(c_i) \in [a, a \cdot (1 + \eps))] - 2 \cdot \sum_{i \in T} c_i \cdot \mathbf{1}[(c_i) \in [a, a \cdot (1 + \eps))] \right |
    \]
    \[
    \leq \frac{\wt(c)}{m^2} + \frac{\eps^2 \cdot \wt(c)}{100 \log(m)} \cdot \log_{1 + \eps}(m^3) + \sum_{a\in \{\frac{1}{1 + \eps}, 1, 1+\eps, \dots m \}} 4 \eps a \beta
    \]
    \[
    \leq \frac{\wt(c)}{m^2} +\frac{\eps \wt(c)}{10} + 4 \eps \cdot \wt(c)\leq 5 \eps \cdot \wt(c).    \]
Here, we have used \cref{rmk:WLOGscale} in showing that $\sum_{i \in [m]} c_i \cdot \mathbf{1}\left [c_i \leq \frac{1}{1+\eps} \right] \leq m \leq \wt(c) / m^2$. This concludes the claim. 
\end{proof}

With the correctness of the sparsifier established, we now proceed to \emph{bound the size} of the set $S$ that we peel off. 

\subsection{Bounding the Size of Peeled Sets}

To start, we show that whenever $S_{a, \beta, \eps}$ is large, it must be the case that there is a collection of well-behaved subcodes that we can find inside $C$. We formalize this below:

\begin{claim}\label{clm:peelCVNRD}
    Let $a, \beta, \eps, C$ be given, and let $C_{a, \beta, \eps}$ and $S_{a, \beta, \eps}$ be defined as above. Then, letting $L_{a, \beta, \eps} = |S_{a, \beta, \eps}|$, there must exist $C_1, \dots C_p \subseteq C$, along with disjoint $A_1, \dots A_p \subseteq [m]$ such that:
    \begin{enumerate}
        \item $\sum_{i = 1}^p |A_i| \geq L_{a, \beta, \eps}$.
        \item For each $i \in [p]$, $|A_i| = \Theta \left (  \frac{\eps^2 \beta}{ \log^4(m) \log^2(m / \eps)} \right )$.
        \item For $i \in [p]$, there exists $\gamma \in [a \cdot (1 + \eps)^{-1}, a \cdot (1 + \eps)^2]^{A_i}$, such that for every $B \subseteq A_i$, there is a vector $c \in C_i$ such that for $j \in B$, $c_j \geq \gamma_j + \frac{a\eps^4}{16000 \log^4(m)}$, and for $j \in A_i - B$, $c_j \leq \gamma_j - \frac{a\eps^4}{16000 \log^4(m)}$.
        \item For $i \in [p]$, for every $c \in C_i$, $\wt(c|_{A_{\neq i}}) \leq O \left ( \frac{a \cdot |A_i|\log^5(m) \log^2(m / \eps)}{\eps^4}\right )$.
    \end{enumerate}
\end{claim}

\begin{proof}
    We consider the following iterative algorithm: we start with the code $C_{a, \beta, \eps}$ as promised above, we let $i = 1$, and let $R = \emptyset$. Now, while $|R| < L_{a, \beta, \eps}$, we know that it must be the case that 
    \[
    \left |\mathrm{Cover} \left ( \left ( C_{a, \beta, \eps} \right ) _{\overline{R}}, \frac{a\eps^4}{1600 \log^4(m)} \right )\right| > 2^{\eps^2 \beta / 1000 \log^4(m)}.
    \]
    So, letting $d = \mathrm{fdim}(C_{a, \beta, \eps}|_{\overline{R}}, \frac{a\eps^4}{16000 \log^4(m)})$, and invoking \cref{lem:additiveABCH}, we know that 
    \[
    2 \cdot \left ( \frac{4m (a \cdot (1 + \eps)^3)^2}{(\frac{a\eps^4}{16000 \log^4(m)})^2}\right )^{d \log(8 em a (1 + \eps)^4 1600\log^4(m) / da\eps^4)} \geq \left |\mathrm{Cover} \left ( \left ( C_{a, \beta, \eps} \right ) _{\overline{R}}, \frac{a\eps^4}{1600 \log^4(m)} \right )\right| 
    \]
    \[
    > 2^{\eps^2 \beta / 1000 \log^4(m)}.
    \]
    Taking the logarithm on both sides and using that $d \geq 1, (1 + \eps) \leq 2$, we obtain that 
    \[
    d \log(3 \cdot 10^6 \log^4(m) em / \eps^4) \cdot \log(10^{10} m\log^8(m) / \eps^8) > \frac{\eps^2 \beta}{1000 \log^4(m)} - 1,
    \]
    which in turn means that 
    \[
    d = \Omega \left (  \frac{\eps^2 \beta}{ \log^4(m) \log^2(m / \eps)} \right ).
    \]
    By the definition of fat-shattering dimension, this means that there is a subset of coordinates $A_i \subseteq [m] - R$ such that $|A_i| = d = \Omega \left (  \frac{\eps^2 \beta}{ \log^4(m) \log^2(m / \eps)} \right )$ along with a set of codewords $C^{(a, \beta, \eps)}_i \subseteq C^{(a, \beta, \eps)}$ such that there exists $\gamma \in [a \cdot (1 + \eps)^{-1}, a \cdot (1 + \eps)^2]^{A_i}$, such that for every $B \subseteq A_i$, there is a vector $c \in C^{(a, \beta, \eps)}_i$ such that for $j \in B$, $c_j \geq \gamma_j + \frac{a\eps^4}{16000 \log^4(m)}$, and for $j \in A_i - B$, $c_j \leq \gamma_j - \frac{a\eps^4}{16000 \log^4(m)}$. We let $C_i \subseteq C$ denote a corresponding set of codewords in $C$ which \emph{maps} to $C^{(a, \beta, \eps)}_i$ upon the transformation of \cref{def:roundingCodes}. Note that $C_i$ satisfies the exact same margin separation condition as $C^{(a, \beta, \eps)}_i$, as $C^{(a, \beta, \eps)}_i$ is obtained by \emph{truncating the values} (and thus undoing it can only separate codewords further). 

    Now, we update $R \leftarrow R \cup A_i$, and repeat the above procedure. Provided $|R| \leq L_{a, \beta, \eps}$, the above argument remains valid. Thus, at the termination of the above algorithm, it must be the case that $|R| = \sum_{i = 1}^p |A_i| \geq L_{a, \beta, \eps}$, that each $|A_i| = \Omega \left (  \frac{\eps^2 \beta}{ \log^4(m) \log^2(m / \eps)} \right )$, and that each $A_i$ satisfies the above margin separation condition with respect to $C_i$.

    For the final condition of the stated claim, we simply observe that for every $c^{(a, \beta, \eps)} \in C_{a, \beta, \eps}$, its corresponding $c \in C$ satisfies (by \cref{eq:contributionFromaSymbols})
    \[
    \wt(c|_{A_{\neq i}}) \leq \wt(c) \leq \frac{400 \log(m) \cdot a \cdot \beta}{\eps^2}.
    \]
    Because $|A_i| = \Omega \left (  \frac{\eps^2 \beta}{ \log^4(m) \log^2(m / \eps)} \right )$, this means that 
    \[
    a \cdot |A_i| = \Omega \left (  \frac{\eps^2 a\beta}{ \log^4(m) \log^2(m / \eps)} \right ),
    \]
    so 
    \[
    \wt(c|_{A_{\neq i}}) = O \left ( \frac{a \cdot |A_i|\log^5(m) \log^2(m / \eps)}{\eps^4}\right ).
    \]
    This concludes the claim. 
\end{proof}

\subsubsection{Some Useful Transformations on Blocks of Codewords}\label{sec:baseTransformations}

The previous claim is the starting point in a sequence of transformations that we will apply to slowly transform our subcodes into the form required by \cref{def:CVNRD}. We now present a series of related transformations which we will then ultimately apply to the subcodes guaranteed by \cref{clm:peelCVNRD}. We start with a simple claim which reduces the ``off-diagonal'' weight of the codewords in each subcode:

\begin{claim}\label{clm:continuousOffDiagWeightReduction}
    Let $a, \beta, \eps, \tau, d$ be given, along with $L$,  $C_1, \dots C_p \subseteq C$, and disjoint $A_1, \dots A_p \subseteq [m]$ such that:
    \begin{enumerate}
        \item $\sum_{i = 1}^p |A_i| \geq L$.
        \item For each $i \in [p]$, $|A_i| = \Theta \left (  \tau \right )$.
        \item For $i \in [p]$, there exists $\gamma \in [a \cdot (1 + \eps)^{-1}, a \cdot (1 + \eps)^2]^{A_i}$, such that for every $B \subseteq A_i$, there is a vector $c \in C_i$ such that for $j \in B$, $c_j \geq \gamma_j + \frac{a\eps^4}{16000 \log^4(m)}$, and for $j \in A_i - B$, $c_j \leq \gamma_j - \frac{a\eps^4}{16000 \log^4(m)}$.
        \item For $i \in [p]$, for every $c \in C_i$, $\wt(c|_{A_{\neq i}}) \leq d$.
    \end{enumerate}

    Then, for any $\rho \in (0,1)$, there exists a set $P \subseteq [p]$ along with subcodes $C'_i \subseteq C_i: i \in P$ such that:
    \begin{enumerate}
        \item $\sum_{i \in P} |A_i| = \Omega \left ( \rho \cdot L \right )$.
        \item For each $i \in P$, $|A_i| = \Theta \left (  \tau \right )$.
        \item For $i \in P$, there exists $\gamma \in [a \cdot (1 + \eps)^{-1}, a \cdot (1 + \eps)^2]^{A_i}$, such that for at least $2^{9|A_i|/10}$ choices of $B \subseteq A_i$, there is a vector $c \in C'_i$ such that for $j \in B$, $c_j \geq \gamma_j + \frac{a\eps^4}{16000 \log^4(m)}$, and for $j \in A_i - B$, $c_j \leq \gamma_j - \frac{a\eps^4}{16000 \log^4(m)}$.
        \item For $i \in P$, for every $c \in C'_i$, $\wt(c|_{A_{j \neq i, j \in {P}}}) \leq \rho \cdot d$.
    \end{enumerate}
\end{claim}

\begin{proof}
    Consider constructing $P'$ by \emph{randomly sampling} exactly $\rho \cdot p$ indices of $[p]$. Now, for an index $i \in P'$, we say that $A_i$ is a \emph{good} block if, for $\geq \frac{1}{10}$ of the codewords $c \in C_i$, it is the case that $\sum_{j \neq i \in P'} \wt(c|_{A_{j}}) \leq 10 \cdot \rho \cdot d$. Observe that for a codeword $c \in C_i$, \[
\E_{P'} \left [\sum_{j \neq i \in P'} \wt(c|_{A_{j}}) \right ] \leq \rho \cdot d.
\]
Thus, for each codeword $c \in C_i$, $\sum_{j \neq i \in P'} \wt(c|_{A_{j}}) \leq 10 \rho \cdot d$ with probability $\geq 9/10$ over $P$. So, this means that 
\[
\E_{P'} \left [ \left | \left \{ c \in C_i: \sum_{j \neq i \in P'} \wt(c|_{A_{j}}) \leq 10 \rho \cdot d \right \} \right | \right ] \geq \frac{9}{10} \cdot |C_i|.
\]
By another Markov bound, this means that 
\[
\Pr_{P'}\left [ \left | \left \{ c \in C_i: \sum_{j \neq i \in {P'}}\wt(c|_{A_{j}}) \leq 10 \rho \cdot d \right \} \right | \geq \frac{|C_i|}{10}\right ] \geq 1/2.
\]
In particular, this then means that $\E_{P'}[ |\{i \in P': C_i \text{ is good}\} |] \geq \frac{|P'|}{2} = \frac{\rho \cdot p}{2}$.

Letting $P \subseteq P'$ denote this subset of indices which is good, and letting $C'_i \subseteq C_i$ denote the corresponding set of good codewords, we then see that:
\begin{enumerate}
    \item $|P| \geq \frac{\rho \cdot p}{2}$.
    \item For each $i \in P$, $|A_i| = \Theta \left (  \tau \right )$ (this is because we have not altered the blocks $A_i$ that we keep, we have only deleted some entire blocks).
    \item For $i \in P$, there exists $\gamma \in [a \cdot (1 + \eps)^{-1}, a \cdot (1 + \eps)^2]^{A_i}$, such that for at least $2^{9|A_i|/10}$ choices of $B \subseteq A_i$, there is a vector $c \in C'_i$ such that for $j \in B$, $c_j \geq \gamma_j + \frac{a\eps^4}{16000 \log^4(m)}$, and for $j \in A_i - B$, $c_j \leq \gamma_j - \frac{a\eps^4}{16000 \log^4(m)}$. This follows because $|C'_i| \geq \frac{9}{10} \cdot |C_i|$.
    \item For $i \in P$, for every $c \in C'_i$, $\wt(c|_{A_{j \neq i, j \in {P}}}) \leq \rho \cdot d$.
    \item Because $|A_i| = \Theta(\tau)$, and $\sum_{i = 1}^[p] |A_i| \geq L$, we then know that $p = \Omega(|L| / \tau)$. Thus, \[
    \sum_{i \in P} |A_i| = \Omega( |P| \cdot \tau) = \Omega(\rho \cdot L).
    \]
\end{enumerate}
This concludes the claim. 
\end{proof}

Now, we make more claims which effectively ``clean up'' a witness of the above form. 

\begin{claim}\label{clm:makeCompleteBlock}
    Let $a,\eps, \tau$ be given, and let $A \subseteq [m], C \subseteq [0, m^3]^m$ be given such that:
    \begin{enumerate}
        \item $|A| = \Theta \left (  \tau \right )$.
        \item There exists $\gamma \in [a \cdot (1 + \eps)^{-1}, a \cdot (1 + \eps)^2]^{A}$, such that for $\geq 2^{9|A| /10}$ sets $B \subseteq A$, there is a vector $c \in C$ such that for $j \in B$, $c_j \geq \gamma_j + \frac{a\eps^4}{16000 \log^4(m)}$, and for $j \in A - B$, $c_j \leq \gamma_j - \frac{a\eps^4}{16000 \log^4(m)}$.
    \end{enumerate}

    Then, there exists a set $A' \subseteq A$, along with $C' \subseteq C$ such that:
    \begin{enumerate}
        \item $|A'| = \Theta \left (  \frac{\tau}{\log(m)} \right )$.
        \item There exists $\gamma \in [a \cdot (1 + \eps)^{-1}, a \cdot (1 + \eps)^2]^{A'}$, such that for every set $B \subseteq A'$, there is a vector $c \in C'$ such that for $j \in B$, $c_j \geq \gamma_j + \frac{a\eps^4}{16000 \log^4(m)}$, and for $j \in A' - B$, $c_j \leq \gamma_j - \frac{a\eps^4}{16000 \log^4(m)}$.
    \end{enumerate}
\end{claim}

\begin{proof}
    We consider the following transformation: we construct $\hat{C} = \{ \hat{c}: c \in C \}$, where $\hat{c}$ is the coordinate wise map which, on coordinate $j$, replaces $c_j$ with $0$ if $c_j \leq \gamma_j - \frac{a\eps^4}{16000 \log^4(m)}$, and replaces $c_j$ with $1$ if $c_j \geq \gamma_j + \frac{a\eps^4}{16000 \log^4(m)}$. In particular, we can then observe that $\hat{C} \subseteq \zo^{A}$, and that $|\hat{C}| \geq 2^{9|A|/10}$, as $C$ margin-shattered that many distinct choices of $B$. 

    Now, we can invoke \cref{thm:sauerShelah} to conclude that there must exist $A' \subseteq A$, along with $\hat{C}' \subseteq \hat{C}$ such that 
    \[
    \hat{C}'|_A = \zo^A.
    \]

    Undoing the $\hat{\cdot}$ map then returns a code of the prescribed form above. 
\end{proof}

The code of the above form is still not in our desired form: in particular, there is no \emph{single} choice of margin, as we still have a vector $\gamma$ where each entry can take on a different value. Below, we show that by a simply grouping argument, we can take care of this shortcoming. 

\begin{claim}\label{clm:makeSingleMargin}
    Let $a,\eps, \tau$ be given, and let $A \subseteq [m], C \subseteq [0, m^3]^m$ be given such that:
    \begin{enumerate}
        \item $|A| = \Theta \left (  \tau \right )$.
        \item There exists $\gamma \in [a \cdot (1 + \eps)^{-1}, a \cdot (1 + \eps)^2]^{A}$, such that for every set $B \subseteq A$, there is a vector $c \in C$ such that for $j \in B$, $c_j \geq \gamma_j + \frac{a\eps^4}{16000 \log^4(m)}$, and for $j \in A - B$, $c_j \leq \gamma_j - \frac{a\eps^4}{16000 \log^4(m)}$.
    \end{enumerate}

    Then, there are sets $A' \subseteq A, C' \subseteq C$, along with a value $b \in [a \cdot (1 + \eps)^{-1}, a \cdot (1 + \eps)^2]$, such that 
    \begin{enumerate}
        \item For every $B \subseteq A'$, there is a vector $c \in C$ such that for $j \in B$ $c_j \geq b + \frac{a\eps^4}{32000 \log^4(m)}$ and for $j \in A' - B$, $c_j \leq b - \frac{a\eps^4}{32000 \log^4(m)}$.
        \item $|A'| = \Omega \left ( \tau \cdot \frac{\eps^3}{\log^4(m)} \right )$.
    \end{enumerate}
\end{claim}

\begin{proof}

First, we set $\delta = \frac{a\eps^4}{16000 \log^4(m)} \cdot \frac{1}{10}$. We now discretize the interval $[a \cdot (1 + \eps)^{-1}, a \cdot (1 + \eps)^2]$ into chunks of size $\delta$. In particular, this breaks the interval $[a \cdot (1 + \eps)^{-1}, a \cdot (1 + \eps)^2]$ into 
    \[
    \leq \frac{a \cdot (1 + \eps)^2 - a \cdot (1 + \eps)^{-1}}{\delta} \leq \frac{4a\eps}{\delta}
    \]
    many intervals.  Thus, there must exist an interval, which we denote by $[b, b + \delta]$ such that at least a $\frac{\delta}{4a\eps}$ fraction of $\gamma$ coordinates are in the interval $[b, b + \delta]$. We denote this subset as $A' \subseteq A$. Plugging in our value of $\delta$ yields the desired size bound on $A'$. In particular, we then know that for every set $B \subseteq A'$, there is a vector $c \in C$ such that for $j \in B$, $c_j \geq \gamma_j + \frac{a\eps^4}{16000 \log^4(m)} \geq b + \frac{a\eps^4}{16000 \log^4(m)}$, and likewise for $j \in A - B$, $c_j \leq \gamma_j - \frac{a\eps^4}{16000 \log^4(m)} \leq b + \delta - \frac{a\eps^4}{16000 \log^4(m)} \leq b - \frac{a \eps^4}{32000 \log^4(m)}$.

So, for this choice of $b \in [a \cdot (1 + \eps)^{-1}, a \cdot (1 + \eps)^2]$, we see that for every set $B \subseteq A'$, there is a vector $c \in C$ such that for $j \in B$ $c_j \geq b + \frac{a\eps^4}{32000 \log^4(m)}$ and for $j \in A' - B$, $c_j \leq b - \frac{a\eps^4}{32000 \log^4(m)}$. Letting $C'$ denote exactly this minimal set of codewords which witness the shattering, we then have our desired claim. 
\end{proof}

Finally, we provide one more useful transformation: the above claims simply work with family of codewords which are margin separated. But, we do not know \emph{how} separated they are. In what follows, we show how to further process these codes such that codeword symbols only land in one of two ``small buckets''. We make this formal below:

\begin{claim}\label{clm:tightenWindow}
    Let $a,\eps, \tau$ be given, and let $A \subseteq [m], C \subseteq [0, m^3]^m$, $b \in [a \cdot (1 + \eps)^{-1}, a \cdot (1 + \eps)^2]$ be given such that:
    \begin{enumerate}
        \item $|A| = \Theta \left (  \tau \right )$.
        \item For every $B \subseteq A'$, there is a vector $c \in C$ such that for $j \in B$ $c_j \geq b + \frac{a\eps^4}{32000 \log^4(m)}$ and for $j \in A' - B$, $c_j \leq b - \frac{a\eps^4}{32000 \log^4(m)}$.
    \end{enumerate}

    Then, for any $\chi < 1$, there exist two values $b_1, b_2$, with $b_2 > b$, $b_2 > b_1 \left (1 + \frac{a\eps^4}{16000 \log^4(m)} \cdot \frac{1}{10} \right )^2$ along with a code $C' \subseteq C$ and coordinates $A' \subseteq A$ such that:
    \begin{enumerate}
        \item $|A'|  = \Omega\left (\chi^6 \cdot \eps^{24} \cdot \frac{|A|}{\log^{32}(m / \eps)} \right )$. 
        \item For every set $B \subseteq A'$, there exists a codeword $c \in C'$ such that for $j \in B$, we have $c_j \in \left [b_1, b_1 \cdot \left (1 + \frac{a\eps^4}{16000 \log^4(m)} \cdot \frac{1}{10} \cdot \chi \right ) \right ]$ and for $j \in A' - B$, $c_j \in \left [b_2, b_2 \cdot \left (1 + \frac{a\eps^4}{16000 \log^4(m)} \cdot \frac{1}{10} \cdot \chi \right ) \right ]$.
    \end{enumerate}
\end{claim}

\begin{proof}
    We set $\delta = \frac{\eps^4}{16000 \log^4(m)} \cdot \frac{1}{10} \cdot \chi$. We break the interval $[1, m^3]$ into powers of $1 + \delta$; more formally, we let $\phi: \R \rightarrow \Z$, where $\phi(x) =  \lceil \log_{1 + \delta} (x)\rceil$. Observe that this breaks the interval $[1, m^3]$ into 
    \[
    k \leq \log_{1 + \delta}(m^3) \leq \frac{3 \log(m)}{\log(1 + \delta)} = O( \log(m) / \delta)
    \]
    many buckets. 
    
    Now, we perform a two-step procedure: first, we create a code $C^{(\mathrm{lower})}$. This code relies on the function \[
    \phi^{\mathrm{lower}}(x) =  \begin{cases} \phi(x) \text{ if } x \leq b \\
    \perp \text{ if }x > b.
    \end{cases}
    \]
    In particular, we set $C^{(\mathrm{lower})} = \{ \phi^{\mathrm{lower}}(c): c \in C\}$, with the understanding that we apply $\phi^{\mathrm{lower}}$ in a coordinate-wise manner. Importantly, we now have the following property in $C^{(\mathrm{lower})}$: for every subset $B \subseteq A$, there exists a codeword $c^{(\mathrm{lower})} \in C^{(\mathrm{lower})}$ such that for $j \in B$ $c^{(\mathrm{lower})} \neq \perp$, and for $j \in A - B$ $c^{(\mathrm{lower})} = \perp$. Thus, we can now invoke \cref{thm:specificWitness} to conclude that there must be a symbol $t \neq \perp$, along with $A' \subseteq A$ and $C'^{\mathrm{lower}} \subseteq C^{(\mathrm{lower})}$ such that $|A'| = \Omega \left ( \frac{|A|}{(\log(m) / \delta)^3 \log(m  / \delta)}\right )$, and 
    \[
    C'^{\mathrm{lower}}|_{A'} = \{t, \perp \}^{A'}.
    \]
    Now, we let $C'$ be the result of \emph{undoing} this mapping. In particular, we then see that for codewords in $C'$ we have the property that, for every $B \subseteq A'$, there exists a $c' \in C'$ such that for $j \in B$ $(c')_j \in \left ((1 + \delta)^{t},(1 + \delta)^{t+1} \right)$, and for $j \in A' - B$, $(c')_j \geq b + \frac{a\eps^4}{32000 \log^4(m)}$. We let $b_1 = (1 + \delta)^{t}$.

    We now repeat this procedure with the map \[\phi^{\mathrm{upper}}(x) =  \begin{cases} \phi(x) \text{ if } x > b \\
    \perp \text{ if }x \leq b
    \end{cases}
    \]
    applied to the code $C'$. Repeating the above analysis, we find a code $C'' \subset C'$ along with a set of coordinates $A'' \subseteq A'$, \[
    |A''| = \Omega\left (\frac{|A'|}{ (\log(m) / \delta)^3 \log(m  / \delta) }\right ) = \Omega\left (\frac{|A|}{(\log(m) / \delta)^6 \log^2(m  / \delta)} \right ),
    \]
    and a symbol $b_2 > b$ such that, for every set $B \subseteq A''$, there is a codeword $c \in C''$ such that for $j \in B$, $c_j \in \left (b_1, b_1 ( 1 + \delta)\right)$ and for $j \in A'' - B$, $c_j \in \left (b_2, b_2 ( 1 + \delta)\right)$.

    Finally, we observe that $b_1 < b - \delta / \chi$ and $b_2 > b + \delta / \chi$: indeed, if $b_1 > b - \delta / \chi$, then there must be codewords such that $c_j  \in [(b-\delta), b]$. But, because $\delta < \frac{a \eps^4}{32000 \log^4(m)}$, there are no codewords who have values in this range. The same argument applied to the upper end will show that $b_2 >b + \delta  / \chi$.

Note that for cleanliness of the claim statement, we  now let $A' \leftarrow A''$ and $C' \leftarrow C''$.
\end{proof}

By combining these claims, we now have the following lemma:

\begin{lemma}\label{lem:combinedProcess}
    Let $a, \beta, \eps, \tau, d$ be given, along with $L$,  $C_1, \dots C_p \subseteq C$, and disjoint $A_1, \dots A_p \subseteq [m]$ such that:
    \begin{enumerate}
        \item $\sum_{i = 1}^p |A_i| \geq L$.
        \item For each $i \in [p]$, $|A_i| = \Theta \left (  \tau \right )$.
        \item For $i \in [p]$, there exists $\gamma \in [a \cdot (1 + \eps)^{-1}, a \cdot (1 + \eps)^2]^{A_i}$, such that for every $B \subseteq A_i$, there is a vector $c \in C_i$ such that for $j \in B$, $c_j \geq \gamma_j + \frac{a\eps^4}{16000 \log^4(m)}$, and for $j \in A_i - B$, $c_j \leq \gamma_j - \frac{a\eps^4}{16000 \log^4(m)}$.
        \item For $i \in [p]$, for every $c \in C_i$, $\wt(c|_{A_{\neq i}}) \leq d$.
    \end{enumerate}

    Then, for any $\rho, \chi \in (0,1)$, there exists a set $P \subseteq [p]$ along with subcodes $C'_i \subseteq C_i: i \in P$ and disjoint sets $A'_i \subseteq A_i: i \in P$ such that:
    \begin{enumerate}
        \item $|A'_i|   = \Omega\left (\chi^6 \cdot \eps^{27} \cdot \frac{\tau}{\log^{37}(m / \eps)} \right )$.
        \item For $i \in P$, for every $c \in C'_i$, $\wt(c|_{A'_{j \neq i, j \in {P}}}) \leq \rho \cdot d$.
        \item For every $i \in P$, there exists $b_1^{(i)}, b_2^{(i)}$ with $b_2^{(i)} > b^{(i)} \geq \frac{a}{2}$, $b_2^{(i)} > b_1^{(i)} \left (1 + \frac{\eps^4}{16000 \log^4(m)} \cdot \frac{1}{10} \right )^2$, such that for every set $B \subseteq A^{\mathrm{final}}_i$, there exists a codeword $c \in C^{\mathrm{final}}_i$ such that for $j \in B$, we have $c_j \in \left [b_1^{(i)}, b_1^{(i)} \cdot \left (1 + \frac{\eps^4}{16000 \log^4(m)} \cdot \frac{1}{10} \cdot \chi \right ) \right ]$ and for $j \in A' - B$, $c_j \in \left [b_2^{(i)}, b_2^{(i)} \cdot \left (1 + \frac{\eps^4}{16000 \log^4(m)} \cdot \frac{1}{10} \cdot \chi \right ) \right ]$.
        \item $\sum_{i \in P} |A'_i| = \Omega \left ( \frac{\rho \cdot L \cdot \chi^6 \cdot \eps^{27}}{\log^{37}(m / \eps)} \right )$.
    \end{enumerate}
\end{lemma}

\begin{proof}
    First, we apply \cref{clm:continuousOffDiagWeightReduction}. This yields a set $P \subseteq [p]$ along with subcodes $C'_i \subseteq C_i: i \in P$ such that:
    \begin{enumerate}
        \item $\sum_{i \in P} |A_i| = \Omega \left ( \rho \cdot L \right )$.
        \item For each $i \in P$, $|A_i| = \Theta \left (  \tau \right )$.
        \item For $i \in P$, there exists $\gamma \in [a \cdot (1 + \eps)^{-1}, a \cdot (1 + \eps)^2]^{A_i}$, such that for at least $2^{9|A_i|/10}$ choices of $B \subseteq A_i$, there is a vector $c \in C_i$ such that for $j \in B$, $c_j \geq \gamma_j + \frac{a\eps^4}{16000 \log^4(m)}$, and for $j \in A_i - B$, $c_j \leq \gamma_j - \frac{a\eps^4}{16000 \log^4(m)}$.
        \item For $i \in P$, for every $c \in C'_i$, $\wt(c|_{A_{j \neq i, j \in {P}}}) \leq \rho \cdot d$.
    \end{enumerate}

    Next, we apply \cref{clm:makeCompleteBlock} to each $A'_i, C'_i$ pair. This yields $A''_i \subseteq A'_i$, along with $C''_i \subseteq C'_i$ such that:
    \begin{enumerate}
        \item $|A''_i| = \Theta \left (  \frac{|A'_i|}{\log(m)} \right ) =\Theta \left (  \frac{\tau}{\log(m)} \right ) $.
        \item For $i \in P$, there exists $\gamma \in [a \cdot (1 + \eps)^{-1}, a \cdot (1 + \eps)^2]^{A''_i}$, such that for every set $B \subseteq A''_i$, there is a vector $c \in C''_i$ such that for $j \in B$, $c_j \geq \gamma_j + \frac{a\eps^4}{16000 \log^4(m)}$, and for $j \in A''_i - B$, $c_j \leq \gamma_j - \frac{a\eps^4}{16000 \log^4(m)}$.
    \end{enumerate}

    Next, we apply \cref{clm:makeSingleMargin} to $A''_i, C''_i: i \in P$. This yields $A'''_i \subseteq A''_i$, along with $C'''_i \subseteq C''_i$ such that: 
    \begin{enumerate}
        \item For $i \in P$, there is a value $b^{(i)}\in [a \cdot (1 + \eps)^{-1}, a \cdot (1 + \eps)^2]$, such that for every $B \subseteq A'''_i$, there is a vector $c \in C'''_i$ such that for $j \in B$ $c_j \geq b^{(i)} + \frac{a\eps^4}{32000 \log^4(m)}$ and for $j \in A' - B$, $c_j \leq b^{(i)} - \frac{a\eps^4}{32000 \log^4(m)}$.
        \item $|A'''_i| = \Omega \left ( |A''_i|\cdot \frac{\eps^3}{\log^4(m)} \right ) = \Omega \left ( \tau \cdot \frac{\eps^3}{\log^5(m)} \right )$.
    \end{enumerate}

    Finally, we apply \cref{clm:tightenWindow} to $C'''_i, A'''_i: i \in P$. This yields $C^{\mathrm{final}}_i \subseteq C'''_i, A^{\mathrm{final}}_i \subseteq A'''_i$ such that:
    \begin{enumerate}
        \item $|A^{\mathrm{final}}_i|  = \Omega\left (\chi^6 \cdot \eps^{24} \cdot \frac{|A'''_i|}{\log^{32}(m / \eps)} \right ) = \Omega\left (\chi^6 \cdot \eps^{27} \cdot \frac{\tau}{\log^{37}(m / \eps)} \right )$. 
        \item For every $i \in P$, there exists $b_1^{(i)}, b_2^{(i)}$ with $b_2^{(i)} > b^{(i)} \geq \frac{a}{2}$, $b_2^{(i)} > b_1^{(i)} \left (1 + \frac{\eps^4}{16000 \log^4(m)} \cdot \frac{1}{10} \right )^2$, such that for every set $B \subseteq A^{\mathrm{final}}_i$, there exists a codeword $c \in C^{\mathrm{final}}_i$ such that for $j \in B$, we have $c_j \in \left [b_1^{(i)}, b_1^{(i)} \cdot \left (1 + \frac{\eps^4}{16000 \log^4(m)} \cdot \frac{1}{10} \cdot \chi \right ) \right ]$ and for $j \in A' - B$, $c_j \in \left [b_2^{(i)}, b_2^{(i)} \cdot \left (1 + \frac{\eps^4}{16000 \log^4(m)} \cdot \frac{1}{10} \cdot \chi \right ) \right ]$.
    \end{enumerate}

    Finally, note that $\sum_{i \in P} |A^{\mathrm{final}}_i| = \Omega \left ( \frac{\rho \cdot L \cdot \chi^6 \cdot \eps^{27}}{\log^{37}(m / \eps)} \right )$, as $|P| = \Omega ( L / \tau)$, and we then multiply by our stated bound on $|A^{\mathrm{final}}_i|$. For cleanliness of the stated claim, we then let $A'_i \leftarrow A^{\mathrm{final}}_i, C'_i \leftarrow C^{\mathrm{final}}_i$.
\end{proof}

\subsubsection{Deriving the Characterization}

With the above transformations in hand, we are now ready to prove the following lemma:

\begin{lemma}\label{lem:CVNRDUpperBound}
    Let $a, \beta, \eps, C$ be given, and let $C_{a, \beta, \eps}$ and $S_{a, \beta, \eps}$ be defined as above. Then, letting $L := L_{a, \beta, \eps} = |S_{a, \beta, \eps}|$, and letting $\eps' = \frac{\eps^4}{160000 \log^4(m)}$ it must be the case that 
    \[
    L = O \left (  \frac{\CVNRD(C, \eps') \cdot \log^{191}(m / \eps)}{\eps^{156}}\right ).
    \]
\end{lemma}

\begin{proof}
    First, we invoke \cref{clm:peelCVNRD}. This implies that there must exist $C_1, \dots C_p \subseteq C$, along with disjoint $A_1, \dots A_p \subseteq [m]$ such that:
    \begin{enumerate}
        \item $\sum_{i = 1}^p |A_i| \geq L_{a, \beta, \eps}$.
        \item For each $i \in [p]$, $|A_i| = \Theta \left (  \frac{\eps^2 \beta}{ \log^4(m) \log^2(m / \eps)} \right )$.
        \item For $i \in [p]$, there exists $\gamma \in [a \cdot (1 + \eps)^{-1}, a \cdot (1 + \eps)^2]^{A_i}$, such that for every $B \subseteq A_i$, there is a vector $c \in C_i$ such that for $j \in B$, $c_j \geq \gamma_j + \frac{a\eps^4}{16000 \log^4(m)}$, and for $j \in A_i - B$, $c_j \leq \gamma_j - \frac{a\eps^4}{16000 \log^4(m)}$.
        \item For $i \in [p]$, for every $c \in C_i$, $\wt(c|_{A_{\neq i}}) \leq O \left ( \frac{a \cdot |A_i|\log^5(m) \log^2(m / \eps)}{\eps^4}\right )$.
    \end{enumerate}

    Now, we invoke \cref{lem:combinedProcess} on this collection of $A_i, C_i$ with $\chi = \left ( \frac{\eps^4}{ \log^4(m)} \right)^2$ and $\rho = \frac{\eps^{10}}{\log^{20}(m / \eps)} \cdot \frac{\chi^6 \eps^{27}}{\log^{37}(m / \eps)} = \frac{\eps^{10}}{\log^{20}(m / \eps)} \cdot \frac{\eps^{27}}{\log^{37}(m / \eps)} \cdot\left ( \frac{\eps^4}{ \log^4(m)} \right)^{12}  = \frac{ \eps^{85}}{\log^{105}(m / \eps)}$. This yields a set $P \subseteq [p]$ along with subcodes $C'_i \subseteq C_i: i \in P$ and disjoint sets $A'_i \subseteq A_i: i \in P$ such that: 
    \begin{enumerate}
        \item $|A'_i| = \Theta\left (\chi^6 \cdot \eps^{27} \cdot \frac{|A_i|}{\log^{37}(m / \eps)} \right ) = \Theta\left (\eps^{75} \cdot \frac{\beta}{\log^{85}(m / \eps)} \right )$.
        \item For $i \in P$, for every $c \in C'_i$, \[
        \wt(c|_{A'_{j \neq i, j \in {P}}}) \leq \rho \cdot  O \left ( \frac{a \cdot |A_i|\log^5(m) \log^2(m / \eps)}{\eps^4}\right ) = \frac{ \eps^{85}}{\log^{105}(m / \eps)} \cdot O \left ( \frac{a \cdot |A_i|\log^5(m) \log^2(m / \eps)}{\eps^4}\right )
        \]
        \[
        = \frac{ \eps^{85}}{\log^{105}(m / \eps)} \cdot O \left ( \frac{a \cdot \log^5(m) \log^2(m / \eps)}{\eps^4}\right ) \cdot O \left (  \frac{\eps^2 \beta}{ \log^4(m) \log^2(m / \eps)} \right )
        \]
        \begin{align}\label{eq:boundCodeWordWeight}
        = O \left ( \frac{\eps^{83} \cdot a \cdot \beta}{\log^{104}(m / \eps)}\right ).
        \end{align}
        
        \item For every $i \in P$, there exists $b_1^{(i)}, b_2^{(i)}$ with $b_2^{(i)} > b^{(i)} \geq \frac{a}{2}$, $b_2^{(i)} > b_1^{(i)} \left (1 + \frac{\eps^4}{16000 \log^4(m)} \cdot \frac{1}{10} \right )^2$, such that for every set $B \subseteq A'_i$, there exists a codeword $c \in C'_i$ such that for $j \in B$, we have $c_j \in \left [b_1^{(i)}, b_1^{(i)} \cdot \left (1 + \frac{\eps^4}{16000 \log^4(m)} \cdot \frac{1}{10} \cdot \chi \right ) \right ]$ and for $j \in A' - B$, $c_j \in \left [b_2^{(i)}, b_2^{(i)} \cdot \left (1 + \frac{\eps^4}{16000 \log^4(m)} \cdot \frac{1}{10} \cdot \chi \right ) \right ]$.
        \item \[
        \sum_{i \in P} |A'_i| = \Omega \left ( \frac{\rho \cdot L \cdot \chi^6 \cdot \eps^{27}}{\log^{37}(m / \eps)} \right ) \]
        \[
        =\Omega \left ( \frac{ L \cdot \eps^{27}}{\log^{37}(m / \eps)} \right ) \cdot \frac{ \eps^{85}}{\log^{105}(m / \eps)} \cdot \left ( \frac{\eps^4}{ \log^4(m)} \right)^{12} = \Omega \left ( \frac{L \cdot \eps^{156}}{\log^{190}( m / \eps)} \right ).
        \]
    \end{enumerate}

Importantly, we can observe that, for each block $A'_i$ it is the case that 
\[
|A'_i| \cdot b_2^{(i)} = \Omega (a \cdot |A'_i|) = \Omega \left ( \eps^{75} \cdot \frac{a \cdot \beta}{\log^{85}(m / \eps)}  \right ).
\]
At the same time, by \cref{eq:boundCodeWordWeight}, we see that for every $c \in C'_i$, 
\[
\wt(c|_{A'_{j \neq i, j \in {P}}}) \leq O \left ( \frac{\eps^{83} \cdot a \cdot \beta}{\log^{104}(m / \eps)}\right ).
\]

In particular, this then guarantees that 
\begin{align}\label{eq:codewordOffBlockSmall}
\wt(c|_{A'_{j \neq i, j \in {P}}}) \leq O \left ( \frac{\eps^8}{\log^{19}(m / \eps)} \right ) \cdot |A'_i| \cdot b_2^{(i)}.
\end{align}

Now, we can observe that this collection $A'_i, C'_i$ satisfies the requirements of \cref{def:CVNRD} with $\eps' = \frac{\eps^4}{160000 \log^4(m)}$: indeed, for every $i \in P$, there exists $b_1^{(i)}, b_2^{(i)}$ such that $b_1^{(i)} < b_2^{(i)} (1 - \eps')$, and for every $B \subseteq A'_i$, there is a $c \in C'_i$ such that for $j \in B$, $c_j \in b_1^{(i)} \cdot[1, 1 + \eps' \cdot \frac{\eps'}{\log^4(m)}] \in [1, 1 + \eps' \cdot \frac{1}{100 \log^2(m)}]$ and likewise, for $j \in A'_i - B$, $c_j \in b_2^{(i)} \cdot[1, 1 + \eps' \cdot \frac{\eps'}{\log^4(m)}] \in [1, 1 + \eps' \cdot \frac{1}{100 \log^2(m)}]$. 

As our final step, we now invoke \cref{lem:generalProcessing} on this collection $C'_i \subseteq C_i: i \in P$ and disjoint sets $A'_i \subseteq A_i: i \in P$. In particular, we construct auxiliary codes $\hat{C}_i: i \in P$, where, for a codeword $c \in C'_i$ we build the codeword $\hat{c}$ such that:
\[
\hat{c}_j = \begin{cases}
    1 \text{ if } j \in A'_i, c_j \in b_2^{(i)} \cdot [1, 1 + \eps' \cdot \frac{1}{100 \log^2(m)}] \\
    0 \text{ if } j \in A'_i, c_j \in b_1^{(i)} \cdot [1, 1 + \eps' \cdot \frac{1}{100 \log^2(m)}] \\
    1 \text{ if }j \in A'_{\neq i}, c_j \geq b_2^{(i)}/100 \log(m) \\
    0 \text{ if }j \in A'_{\neq i}, c_j < b_2^{(i)}/100 \log(m)
\end{cases}.
\] 
In particular, we then see that, in the language of \cref{lem:generalProcessing}, $\tau := |A'_i| = \Theta\left (\eps^{75} \cdot \frac{\beta}{\log^{85}(m / \eps)} \right )$, $(\hat{C}_i)|_{A'_i} = \zo^{A'_i}$, and for every $i \in P$, and every $\hat{c} \in \hat{C}_i$, \[
\sum_{j \neq \in P} \wt(\hat{c}|_{A'_j}) \leq \frac{\wt(c|_{A'_{j \neq i, j \in {P}}})}{b_2^{(i)}} \leq O \left ( \frac{\eps^8}{\log^{19}(m / \eps)} \right ) \cdot \tau.
\]
Thus, we can invoke \cref{lem:generalProcessing} with $\eta = 1$. This returns $P'' \subseteq P$ along with $A''_i \subseteq A_i: i \in P''$ and $\hat{C}''_i \subseteq \hat{C}_i$ such that:
\begin{enumerate}
    \item For every $i \in P''$, $(\hat{C}''_i)|_{A''_i} = \zo^{A''_i}$.
    \item For every $i \in P''$, for every $\hat{c} \in \hat{C}''_i$, $\sum_{j \neq i \in P''} \wt(\hat{c}|_{A''_j}) = 0$. 
    \item $|A''_i| \geq \frac{\tau}{10 \log(m)} \geq \Theta \left (\eps^{75} \cdot \frac{\beta}{\log^{86}(m / \eps)} \right )$.
    \item $|P''| \geq |P| / 4,$ and hence $\sum_{i \in P''} |A''_i| = \Omega \left ( |P| \cdot |A'_i| /\log(m) \right ) = \Omega \left ( \frac{L \cdot \eps^{156}}{\log^{191}( m / \eps)} \right )$.
\end{enumerate}

Finally, we invert the $\hat{(\cdot)}$ map that we performed. This yields $A''_i, C''_i, P''$ such that:
\begin{enumerate}
    \item $\sum_{i \in P''} |A''_i| = \Omega \left ( \frac{L \cdot \eps^{156}}{\log^{191}( m / \eps)} \right )$.
    \item For every $i \in P''$, there exists $b_1^{(i)}, b_2^{(i)}$ such that $b_1^{(i)} < b_2^{(i)} (1 - \eps')$, and for every $B \subseteq A''_i$, there is a $c \in C''_i$ such that for $j \in B$, $c_j \in b_1^{(i)} \cdot[1, 1 + \eps' \cdot \frac{\eps'}{\log^4(m)}] \in [1, 1 + \eps' \cdot \frac{1}{100 \log^2(m)}]$ and likewise, for $j \in A'_i - B$, $c_j \in b_2^{(i)} \cdot[1, 1 + \eps' \cdot \frac{\eps'}{\log^4(m)}] \in [1, 1 + \eps' \cdot \frac{1}{100 \log^2(m)}]$.
    \item Similarly, by \cref{eq:codewordOffBlockSmall}, we know that for $i \in P$, and for every $c \in C'_i$, \[
\wt(c|_{A'_{j \neq i, j \in {P}}}) \leq O \left ( \frac{\eps^8}{\log^{19}(m / \eps)} \right ) \cdot |A'_i| \cdot b_2^{(i)} \leq \frac{\eps' \cdot |A''_i| \cdot b_2^{(i)}}{100 \log^2(m)},
\]
as $|A''_i| = \Omega(|A'_i| /\log(m))$.
\item Because for every $i \in P''$, for every $\hat{c} \in \hat{C}''_i$, $\sum_{j \neq i \in P''} \wt(\hat{c}|_{A''_j}) = 0$, we know that $\max_{j \in A''_{\neq i}} c_j <\frac{b_2^{(i)}}{100 \log(m)}$.
\end{enumerate}

Thus, $\CVNRD(C, \eps') \geq \sum_{i \in P''} |A''_i| = \Omega \left ( \frac{L \cdot \eps^{156}}{\log^{191}( m / \eps)} \right )$. In particular, this implies that 
\[
L = O \left (  \frac{\CVNRD(C, \eps') \cdot \log^{191}(m / \eps)}{\eps^{156}}\right ).
\]
\end{proof}

Note that we can prove an analogous statement for \cref{def:CVNRDparameter}:

\begin{lemma}\label{lem:CVNRDUpperBoundParameter}
    Let $a, \beta, \eps, C, \chi, \rho$ be given, and let $C_{a, \beta, \eps}$ and $S_{a, \beta, \eps}$ be defined as above. Then, letting $L := L_{a, \beta, \eps} = |S_{a, \beta, \eps}|$, and letting $\eps' = \frac{\eps^4}{160000 \log^4(m)}$ it must be the case that 
    \[
L = O \left (  \frac{\CVNRD(C, \eps') \cdot \log^{191}(m / \eps)}{\rho \cdot \chi^6 \cdot \eps^{156}}\right ).
    \]
\end{lemma}

\begin{proof}
    First, we invoke \cref{clm:peelCVNRD}. This implies that there must exist $C_1, \dots C_p \subseteq C$, along with disjoint $A_1, \dots A_p \subseteq [m]$ such that:
    \begin{enumerate}
        \item $\sum_{i = 1}^p |A_i| \geq L_{a, \beta, \eps}$.
        \item For each $i \in [p]$, $|A_i| = \Theta \left (  \frac{\eps^2 \beta}{ \log^4(m) \log^2(m / \eps)} \right )$.
        \item For $i \in [p]$, there exists $\gamma \in [a \cdot (1 + \eps)^{-1}, a \cdot (1 + \eps)^2]^{A_i}$, such that for every $B \subseteq A_i$, there is a vector $c \in C_i$ such that for $j \in B$, $c_j \geq \gamma_j + \frac{a\eps^4}{16000 \log^4(m)}$, and for $j \in A_i - B$, $c_j \leq \gamma_j - \frac{a\eps^4}{16000 \log^4(m)}$.
        \item For $i \in [p]$, for every $c \in C_i$, $\wt(c|_{A_{\neq i}}) \leq O \left ( \frac{a \cdot |A_i|\log^5(m) \log^2(m / \eps)}{\eps^4}\right )$.
    \end{enumerate}

   Now, we invoke \cref{lem:combinedProcess} on this collection of $A_i, C_i$ with $\chi' = \chi \cdot \left ( \frac{\eps^4}{ \log^4(m)} \right)^2$ and $\rho' = \rho \cdot \frac{\eps^{10}}{\log^{20}(m / \eps)} \cdot \frac{(\chi')^6 \eps^{27}}{\log^{37}(m / \eps)} = \rho \chi^6 \cdot \frac{\eps^{10}}{\log^{20}(m / \eps)} \cdot \frac{\eps^{27}}{\log^{37}(m / \eps)} \cdot\left ( \frac{\eps^4}{ \log^4(m)} \right)^{12}  = \rho \chi^6 \cdot \frac{ \eps^{85}}{\log^{105}(m / \eps)}$. This yields a set $P \subseteq [p]$ along with subcodes $C'_i \subseteq C_i: i \in P$ and disjoint sets $A'_i \subseteq A_i: i \in P$ such that: 
    \begin{enumerate}
        \item $|A'_i| = \Omega\left ((\chi')^6 \cdot \eps^{27} \cdot \frac{|A_i|}{\log^{37}(m / \eps)} \right ) = \Omega\left (\chi^6 \cdot \eps^{75} \cdot \frac{\beta}{\log^{85}(m / \eps)} \right )$.
        \item For $i \in P$, for every $c \in C'_i$, \[
        \wt(c|_{A'_{j \neq i, j \in {P}}}) \leq \rho' \cdot  O \left ( \frac{a \cdot |A_i|\log^5(m) \log^2(m / \eps)}{\eps^4}\right ) = \rho \chi^6 \cdot \frac{ \eps^{85}}{\log^{105}(m / \eps)} \cdot O \left ( \frac{a \cdot |A_i|\log^5(m) \log^2(m / \eps)}{\eps^4}\right )
        \]
        \[
        = \rho \chi^6 \cdot \frac{ \eps^{85}}{\log^{105}(m / \eps)} \cdot O \left ( \frac{a \cdot \log^5(m) \log^2(m / \eps)}{\eps^4}\right ) \cdot O \left (  \frac{\eps^2 \beta}{ \log^4(m) \log^2(m / \eps)} \right )
        \]
        \begin{align}\label{eq:boundCodeWordWeightw}
        = O \left ( \rho \chi^6 \cdot \frac{\eps^{83} \cdot a \cdot \beta}{\log^{104}(m / \eps)}\right ).
        \end{align}
        
        \item For every $i \in P$, there exists $b_1^{(i)}, b_2^{(i)}$ with $b_2^{(i)} > b^{(i)} \geq \frac{a}{2}$, $b_2^{(i)} > b_1^{(i)} \left (1 + \frac{\eps^4}{16000 \log^4(m)} \cdot \frac{1}{10} \right )^2$, such that for every set $B \subseteq A'_i$, there exists a codeword $c \in C'_i$ such that for $j \in B$, we have $c_j \in \left [b_1^{(i)}, b_1^{(i)} \cdot \left (1 + \frac{\eps^4}{16000 \log^4(m)} \cdot \frac{1}{10} \cdot \chi' \right ) \right ]$ and for $j \in A' - B$, $c_j \in \left [b_2^{(i)}, b_2^{(i)} \cdot \left (1 + \frac{\eps^4}{16000 \log^4(m)} \cdot \frac{1}{10} \cdot \chi' \right ) \right ]$.
        \item \[
        \sum_{i \in P} |A'_i| = \Omega \left ( \frac{\rho' \cdot L \cdot (\chi')^6 \cdot \eps^{27}}{\log^{37}(m / \eps)} \right ) \]
        \[
        =\Omega \left ( \rho \cdot \chi^6 \cdot \frac{ L \cdot \eps^{27}}{\log^{37}(m / \eps)} \right ) \cdot \frac{ \eps^{85}}{\log^{105}(m / \eps)} \cdot \left ( \frac{\eps^4}{ \log^4(m)} \right)^{12} = \Omega \left (\rho \cdot \chi^6 \cdot  \frac{L \cdot \eps^{156}}{\log^{190}( m / \eps)} \right ).
        \]
    \end{enumerate}

Importantly, we can observe that, for each block $A'_i$ it is the case that 
\[
|A'_i| \cdot b_2^{(i)} = \Omega (a \cdot |A'_i|) = \Omega \left ( \chi^6 \cdot \eps^{75} \cdot \frac{a \cdot \beta}{\log^{85}(m / \eps)}  \right ).
\]
At the same time, by \cref{eq:boundCodeWordWeightw}, we see that for every $c \in C'_i$, 
\[
\wt(c|_{A'_{j \neq i, j \in {P}}}) \leq O \left ( \rho \chi^6 \frac{\eps^{83} \cdot a \cdot \beta}{\log^{104}(m / \eps)}\right ).
\]

In particular, this then guarantees that 
\begin{align}\label{eq:codewordOffBlockSmallw}
\wt(c|_{A'_{j \neq i, j \in {P}}}) \leq O \left ( \rho \cdot \frac{\eps^8}{\log^{19}(m / \eps)} \right ) \cdot |A'_i| \cdot b_2^{(i)}.
\end{align}

Now, we can observe that this collection $A'_i, C'_i$ satisfies the requirements of \cref{def:CVNRDparameter} with $\eps' = \frac{\eps^4}{160000 \log^4(m)}, \chi, \rho$: indeed, for every $i \in P$, there exists $b_1^{(i)}, b_2^{(i)}$ such that $b_1^{(i)} < b_2^{(i)} (1 - \eps')$, and for every $B \subseteq A'_i$, there is a $c \in C'_i$ such that for $j \in B$, $c_j \in b_1^{(i)} \cdot[1, 1 + \eps' \cdot \chi] $ and likewise, for $j \in A'_i - B$, $c_j \in b_2^{(i)} \cdot[1, 1 + \eps' \cdot \chi]$.

As our final step, we now invoke \cref{lem:generalProcessing} on this collection $C'_i \subseteq C_i: i \in P$ and disjoint sets $A'_i \subseteq A_i: i \in P$. In particular, we construct auxiliary codes $\hat{C}_i: i \in P$, where, for a codeword $c \in C'_i$ we build the codeword $\hat{c}$ such that:
\[
\hat{c}_j = \begin{cases}
    1 \text{ if } j \in A'_i, c_j \in b_2^{(i)} \cdot [1, 1 + \eps' \cdot \frac{1}{100 \log^2(m)}] \\
    0 \text{ if } j \in A'_i, c_j \in b_1^{(i)} \cdot [1, 1 + \eps' \cdot \frac{1}{100 \log^2(m)}] \\
    1 \text{ if }j \in A'_{\neq i}, c_j \geq \rho \cdot b_2^{(i)} \\
    0 \text{ if }j \in A'_{\neq i}, c_j < \rho \cdot b_2^{(i)}
\end{cases}.
\] 
In particular, we then see that, in the language of \cref{lem:generalProcessing}, $\tau := |A'_i| = \Theta\left (\chi^6 \cdot \eps^{75} \cdot \frac{\beta}{\log^{85}(m / \eps)} \right )$, $(\hat{C}_i)|_{A'_i} = \zo^{A'_i}$, and for every $i \in P$, and every $\hat{c} \in \hat{C}_i$, \[
\sum_{j \neq \in P} \wt(\hat{c}|_{A'_j}) \leq \frac{\wt(c|_{A'_{j \neq i, j \in {P}}})}{b_2^{(i)}} \leq O \left ( \rho \cdot \frac{\eps^8}{\log^{19}(m / \eps)} \right ) \cdot \tau.
\]
Thus, we can invoke \cref{lem:generalProcessing} with $\eta = 1$. This returns $P'' \subseteq P$ along with $A''_i \subseteq A_i: i \in P''$ and $\hat{C}''_i \subseteq \hat{C}_i$ such that:
\begin{enumerate}
    \item For every $i \in P''$, $(\hat{C}''_i)|_{A''_i} = \zo^{A''_i}$.
    \item For every $i \in P''$, for every $\hat{c} \in \hat{C}''_i$, $\sum_{j \neq i \in P''} \wt(\hat{c}|_{A''_j}) = 0$. 
    \item $|A''_i| \geq \frac{\tau}{10 \log(m)} \geq \Theta \left (\eps^{75} \cdot \frac{\beta}{\log^{86}(m / \eps)} \right )$.
    \item $|P''| \geq |P| / 4,$ and hence $\sum_{i \in P''} |A''_i| = \Omega \left ( |P| \cdot |A'_i| /\log(m) \right ) = \Omega \left (\rho \cdot \chi^6 \cdot \frac{L \cdot \eps^{156}}{\log^{191}( m / \eps)} \right )$.
\end{enumerate}

Finally, we invert the $\hat{(\cdot)}$ map that we performed. This yields $A''_i, C''_i, P''$ such that:
\begin{enumerate}
    \item $\sum_{i \in P''} |A''_i| = \Omega \left ( \rho \cdot \chi^6 \cdot \frac{L \cdot \eps^{156}}{\log^{191}( m / \eps)} \right )$.
    \item For every $i \in P''$, there exists $b_1^{(i)}, b_2^{(i)}$ such that $b_1^{(i)} < b_2^{(i)} (1 - \eps')$, and for every $B \subseteq A''_i$, there is a $c \in C''_i$ such that for $j \in B$, $c_j \in b_1^{(i)} \cdot[1, 1 + \eps' \cdot \frac{\eps'}{\log^4(m)}] \in [1, 1 + \eps' \cdot \frac{1}{100 \log^2(m)}]$ and likewise, for $j \in A'_i - B$, $c_j \in b_2^{(i)} \cdot[1, 1 + \eps' \cdot \frac{\eps'}{\log^4(m)}] \in [1, 1 + \eps' \cdot \frac{1}{100 \log^2(m)}]$.
    \item Similarly, by \cref{eq:codewordOffBlockSmall}, we know that for $i \in P$, and for every $c \in C'_i$, \[
\wt(c|_{A'_{j \neq i, j \in {P}}}) \leq O \left ( \rho \cdot \frac{\eps^8}{\log^{19}(m / \eps)} \right ) \cdot |A'_i| \cdot b_2^{(i)} \leq \rho \cdot |A''_i| \cdot b_2^{(i)},
\]
as $|A''_i| = \Omega(|A'_i| /\log(m))$.
\item Because for every $i \in P''$, for every $\hat{c} \in \hat{C}''_i$, $\sum_{j \neq i \in P''} \wt(\hat{c}|_{A''_j}) = 0$, we know that $\max_{j \in A''_{\neq i}} c_j < \rho \cdot b_2^{(i)}$.
\end{enumerate}

Thus, $\CVNRD(C, \eps') \geq \sum_{i \in P''} |A''_i| = \Omega \left ( \rho \cdot \chi^6 \frac{L \cdot \eps^{156}}{\log^{191}( m / \eps)} \right )$. In particular, this implies that 
\[
L = O \left (  \frac{\CVNRD(C, \eps') \cdot \log^{191}(m / \eps)}{\rho \cdot \chi^6 \cdot \eps^{156}}\right ).
\]
\end{proof}

\subsubsection{An Edge Case: When $\beta$ is Small}

Finally, we recall that in the construction of our sets $S_{a, \beta, \eps}$, which was that whenever $\beta \leq \frac{10000 \log^5(m)}{\eps^2}$, then we simply remove all $i \in [m]$ such that there is $c$ in $C_{a, \beta, \eps}$ such that $c_i \in [a, a(1+\eps)]$, and add these coordinates to $S_{a, \beta, \eps}$. 
Thus, we must also show that in this setting, $S_{a, \beta, \eps}$ is bounded by a function of $\CVNRD(C)$. Fortunately, this argument is significantly simpler than the argument above. 

Indeed, we have the following claim:

\begin{claim}\label{clm:findApproxDiag}
     Let $a, \beta, \eps, C$ be given, and let $C_{a, \beta, \eps}$ and $S_{a, \beta, \eps}$ be defined as above. Then, if $\beta \leq \frac{10000 \log^5(m)}{\eps^2}$, we can find disjoint $A_1, \dots A_p \subseteq [m]$ and $C_1, \dots C_p \subseteq C$ such that:
     \begin{enumerate}
         \item For $i \in [p]$, $|A_i| = 1$.
         \item For $i \in [p]$, the unique codeword $c \in C_i$, and coordinate $j \in A_i$, $c_j \in [a, a(1 + \eps)]$. 
         \item For $i \in [p]$, and the unique codeword $c \in C_i$, $\wt(c) \leq \frac{400 \log(m)}{\eps^2} \cdot  a \cdot \beta$.
         \item $p = |S_{a, \beta, \eps}|$.
     \end{enumerate}
\end{claim}

\begin{proof}
    Indeed, $S_{a, \beta, \eps}$ is the smallest set whose removal ensures that no codeword in $C_{a, \beta, \eps}$ has a symbol in the range $[a, a(1 +\eps)]$. So, this means that for every coordinate $j \in S_{a, \beta, \eps}$, there is a codeword $c \in C$ such that $c_j \in [a, a(1 +\eps)]$.

    So, we let our sets $A_1, \dots A_p$ be exactly the coordinates in $S_{a, \beta, \eps}$, and thus define $p = |S_{a, \beta, \eps}|$. We likewise define $C_i$ to be a single codeword which takes value $[a, a(1 + \eps)]$ in the coordinate in $A_i$.

    All that remains is to see is the third point above. This follows from \cref{eq:contributionFromaSymbols}.
\end{proof}

Now, we make the following additional claim:

\begin{claim}\label{clm:edgeCaseReduceWeight}
     Let $a, \beta, \eps, C$ be given, and let $C_{a, \beta, \eps}$ and $S_{a, \beta, \eps}$ be defined as above. Let $\beta \leq \frac{10000 \log^5(m)}{\eps^2}$.  Suppose we are given disjoint $A_1, \dots A_p \subseteq [m]$ and $C_1, \dots C_p \subseteq C$ such that:
     \begin{enumerate}
         \item For $i \in [p]$, $|A_i| = 1$.
         \item For $i \in [p]$, the unique codeword $c \in C_i$, and coordinate $j \in A_i$, $c_j \in [a, a(1 + \eps)]$. 
         \item For $i \in [p]$, and the unique codeword $c \in C_i$, $\wt(c) \leq \frac{400 \log(m)}{\eps^2} \cdot  a \cdot \beta$.
         \item $p = |S_{a, \beta, \eps}|$.
     \end{enumerate}

     Then, for any $\rho > 0$, we can find $P \subseteq [p]$ such that:
     \begin{enumerate}
         \item $|P| = \Omega \left ( \frac{\rho \eps^4}{\log^6(m)} \cdot p \right )$.
         \item For $i \in P$, the codeword $c \in C_i$, and the index $j \in A_i$, $c_j \in [a, a(1 + \eps)]$.
         \item For $A_i: i \in P$ and the corresponding codeword $c \in C_i$, $\sum_{i' \neq i \in P} \sum_{j \in A_{i'}} c_j \leq \frac{\rho}{10} \cdot a$
     \end{enumerate}
\end{claim}

\begin{proof}
    Consider sampling $P'$ from $[p]$ of size exactly $\frac{\rho}{100} \cdot \frac{\eps^2}{400 \log(m) \beta} \cdot p$. For a set $A_i: i \in P'$ and the corresponding codeword $c \in C_i$, we see that 
    \[
    \E_{P'}[\sum_{i' \neq i \in P'} \sum_{j \in A_{i'}} c_j] \leq \frac{\rho}{100} \cdot \frac{\eps^2}{400 \log(m) \beta} \cdot \wt(c) \leq \frac{\rho}{100} \cdot a.
    \]

    Thus, we say that $A_i: i \in P'$ is \emph{good}, if for the corresponding codeword $c \in C_i$, $\sum_{i' \neq i \in P'} \sum_{j \in A_{i'}} c_j \leq \frac{\rho}{10} \cdot a$, and we can observe that $\Pr[A_i \text{ is good}] \geq 9/10$ by a simple Markov bound. This then implies that, in expectation, the number of good indices $i \in P'$ is at least $9/10 \cdot |P'| \geq \frac{\rho}{20} \cdot \frac{\eps^2}{400 \log(m) \beta} \cdot p$. We let $P$ denote these good indices. Thus, we have that \[
    |P| =  \Omega \left ( \frac{\rho}{200} \cdot \frac{\eps^2}{400 \log(m) \beta} \cdot p \right ) = \Omega \left ( \frac{\rho \eps^2}{\log(m) \beta} \cdot p \right ) = \Omega \left ( \frac{\rho \eps^4}{\log^6(m)} \cdot p \right ).
    \]

    The condition that for $i \in P$, the codeword $c \in C_i$, and the index $j \in A_i$, $c_j \in [a, a(1 + \eps)]$ follows trivially. Finally, because we only retain the good $A_i$, we then also have that for $A_i: i \in P$ and the corresponding codeword $c \in C_i$, $\sum_{i' \neq i \in P} \sum_{j \in A_{i'}} c_j \leq \frac{\rho}{10} \cdot a$.
\end{proof}

With this, we get the following lemma:

\begin{lemma}\label{lem:CVNRDUpperBoundEdgeCase}
    Let $a, \beta, \eps, C$ be given with $\beta \leq \frac{10000 \log^5(m)}{\eps^2}$ and let $C_{a, \beta, \eps}$ and $S_{a, \beta, \eps}$ be defined as above. Then, letting $L := L_{a, \beta, \eps} = |S_{a, \beta, \eps}|$, and letting $\eps' = \frac{\eps^4}{160000 \log^4(m)}$ it must be the case that 
    \[
    L = O \left ( \CVNRD(C, \eps') \cdot \frac{\log^8(m)}{\eps^5}\right ).
    \]
\end{lemma}

\begin{proof}
    We first apply \cref{clm:findApproxDiag}. This yields disjoint $A_1, \dots A_p \subseteq [m]$ and $C_1, \dots C_p \subseteq C$ such that:
     \begin{enumerate}
         \item For $i \in [p]$, $|A_i| = 1$.
         \item For $i \in [p]$, the unique codeword $c \in C_i$, and coordinate $j \in A_i$, $c_j \in [a, a(1 + \eps)]$. 
         \item For $i \in [p]$, and the unique codeword $c \in C_i$, $\wt(c) \leq \frac{400 \log(m)}{\eps^2} \cdot  a \cdot \beta$.
         \item $p = |S_{a, \beta, \eps}|$.
     \end{enumerate}

     We then invoke \cref{clm:edgeCaseReduceWeight} with $\rho = \frac{\eps}{100 \log^2(m)}$. This yields a set $P \subseteq [p]$ such that:
     \begin{enumerate}
         \item $|P| = \Omega \left ( \frac{ \eps^5}{\log^8(m)} \cdot p \right ) = \Omega \left ( \frac{ \eps^5}{\log^8(m)} \cdot |S_{a, \beta, \eps}| \right )$.
         \item For $i \in P$, the codeword $c \in C_i$, and the index $j \in A_i$, $b_2^{(i)} = c_j \in [a, a(1 + \eps)]$.
         \item For $A_i: i \in P$ and the corresponding codeword $c \in C_i$, $\sum_{i' \neq i \in P} \sum_{j \in A_{i'}} c_j \leq \frac{\rho}{10} \cdot a \leq \frac{\eps}{100 \log^2(m)} \cdot b_2^{(i)}$.
         \item In particular, this last point immediately implies that for every $c \in C_i$, $\max_{i' \neq i \in P} \max_{j \in A_{i'}} c_j \leq \frac{\rho}{10} \cdot a \leq \frac{\eps}{100 \log^2(m)} \cdot b_2^{(i)}$.
     \end{enumerate}

     Because $|A_i| = 1$ for $i \in P$, this witness satisfies the definition of \cref{def:CVNRD}. Thus, 
     \[\CVNRD(C, \eps') \geq |P| =  \Omega \left ( \frac{ \eps^5}{\log^8(m)} \cdot |S_{a, \beta, \eps}| \right ).\]

     Thus, $|S_{a, \beta, \eps}| \leq O \left ( \CVNRD(C, \eps') \frac{\log^8(m)}{\eps^5}\right )$.
\end{proof}

We can also show an analogous version of the lemma for \cref{def:CVNRDparameter}:

\begin{lemma}\label{lem:CVNRDUpperBoundEdgeCaseParameter}
    Let $a, \beta, \eps, C$ be given with $\beta \leq \frac{10000 \log^5(m)}{\eps^2}$ and let $C_{a, \beta, \eps}$ and $S_{a, \beta, \eps}$ be defined as above. Let $\rho, \chi > 0$. Then, letting $L := L_{a, \beta, \eps} = |S_{a, \beta, \eps}|$, and letting $\eps' = \frac{\eps^4}{160000 \log^4(m)}$ it must be the case that 
    \[
    L = O \left ( \CVNRD(C, \eps', \chi, \rho) \cdot \frac{\log^6(m)}{\rho\cdot \eps^4}\right ).
    \]
\end{lemma}

\begin{proof}
    We first apply \cref{clm:findApproxDiag}. This yields disjoint $A_1, \dots A_p \subseteq [m]$ and $C_1, \dots C_p \subseteq C$ such that:
     \begin{enumerate}
         \item For $i \in [p]$, $|A_i| = 1$.
         \item For $i \in [p]$, the unique codeword $c \in C_i$, and coordinate $j \in A_i$, $c_j \in [a, a(1 + \eps)]$. 
         \item For $i \in [p]$, and the unique codeword $c \in C_i$, $\wt(c) \leq \frac{400 \log(m)}{\eps^2} \cdot  a \cdot \beta$.
         \item $p = |S_{a, \beta, \eps}|$.
     \end{enumerate}

     We then invoke \cref{clm:edgeCaseReduceWeight} with $\rho$. This yields a set $P \subseteq [p]$ such that:
     \begin{enumerate}
         \item $|P| = \Omega \left ( \rho \cdot \frac{ \eps^4}{\log^6(m)} \cdot p \right ) = \Omega \left ( \rho \cdot \frac{ \eps^4}{\log^6(m)} \cdot |S_{a, \beta, \eps}| \right )$.
         \item For $i \in P$, the codeword $c \in C_i$, and the index $j \in A_i$, $b_2^{(i)} = c_j \in [a, a(1 + \eps)]$.
         \item For $A_i: i \in P$ and the corresponding codeword $c \in C_i$, $\sum_{i' \neq i \in P} \sum_{j \in A_{i'}} c_j \leq \frac{\rho}{10} \cdot a \leq \rho \cdot b_2^{(i)}$.
         \item In particular, this last point immediately implies that for $A_i: i \in P$ and the corresponding codeword $c \in C_i$, $\max_{i' \neq i \in P} \max_{j \in A_{i'}} c_j \leq \frac{\rho}{10} \cdot a < \rho \cdot b_2^{(i)}$.
     \end{enumerate}

     Because $|A_i| = 1$ for $i \in P$, this witness satisfies the definition of \cref{def:CVNRD}. Thus, 
     \[\CVNRD(C, \eps', \chi, \rho) \geq |P| =  \Omega \left ( \rho \cdot \frac{ \eps^4}{\log^6(m)} \cdot |S_{a, \beta, \eps}| \right ).\]

     Thus, $|S_{a, \beta, \eps}| \leq O \left ( \CVNRD(C, \eps', \chi, \rho) \frac{\log^6(m)}{\rho\cdot \eps^4}\right )$.
\end{proof}

\subsubsection{Putting it Together}

Using the above claims, we now have the following:

\begin{lemma}\label{lem:peelSetContinuous}
    Let $C \subseteq [0, m^3]^m, \eps >0$ be given. Let $\eps' = \frac{\eps^4}{160000 \log^4(m)}$. Then, there exists a set $S \subseteq [m]$ of size $\leq \CVNRD(C, \eps') \cdot \mathrm{poly}(\log(m), \eps^{-1})$, along with a set $T \subseteq [m] - S$ of size $\leq \frac{2}{3} \cdot m$ such that, for every codeword $c \in C$,
    \[
    \sum_{i = 1}^m c_i \in (1 \pm 5\eps) \cdot \left ( \sum_{i \in S} c_i + \sum_{i \in T} 2 \cdot c_i \right ).
    \]
\end{lemma}

\begin{proof}
    We let the set $S$ be the union of all of the sets $S_{a, \beta, \eps} \subseteq [m]$ (see \cref{eq:defOfS}), for $a \in \{\frac{1}{1 + \eps}, 1, 1+\eps, \dots m \}$, $\beta \in \{1, 2, 4, 8, \dots m \}$, as well as all choices of $\eta$ (see \cref{rmk:smallBoundary}).

    In particular, we can see that there are $O(\log(m) / \eps)$ choices of $a$, $\log(m)$ choices of $\beta$, and $\frac{200 \log^4(m)}{\eps^4}$ choices of $\eta$. Thus, plugging in the bounds from \cref{lem:CVNRDUpperBoundEdgeCase} (for small $\beta$) and \cref{lem:CVNRDUpperBound} (for large $\beta$), we see that 
    \[
    |S| = O \left ( \frac{\log(m)}{\eps} \cdot \log(m) \cdot \frac{ \log^4(m)}{\eps^4} \cdot \CVNRD(C, \eps') \cdot \frac{\log^{191}(m / \eps)}{\eps^{156}}\right ) 
    \]
    \[
    = \CVNRD(C, \eps') \cdot \mathrm{poly}(\log(m), \eps^{-1}).
    \]

    Now, to construct the set $T$, we simply use \cref{clm:sparsifierCorrectness}. Indeed, this claim shows that when we choose the set $T$ via a random sample at rate $1/2$, we have that, with probability $\geq 1/2$, for every codeword $c \in C$,
    \begin{align}\label{eq:preserveCodewordWeight}
   \left ( \sum_{i \in S} c_i + \sum_{i \in T} 2 \cdot c_i \right ) \in (1 \pm 5\eps) \cdot\sum_{i = 1}^m c_i .
    \end{align}
    By a simple Chernoff bound over $m$, we then can see that with probability $\geq 1/4$, simultaneously \cref{eq:preserveCodewordWeight} holds and $|T| \leq 2m/3$. This then yields our desired lemma.
\end{proof}

We can show an analogous statement for the $\chi, \rho$ version of $\CVNRD$:

\begin{lemma}\label{lem:peelSetContinuousParam}
    Let $C \subseteq [0, m^3]^m, \eps >0, \rho  > 0, \chi > 0$ be given. Let $\eps' = \frac{\eps^4}{160000 \log^4(m)}$. Then, there exists a set $S \subseteq [m]$ of size $\leq \frac{\CVNRD(C, \eps', \chi, \rho)}{\chi^6 \rho} \cdot \mathrm{poly}(\log(m), \eps^{-1})$, along with a set $T \subseteq [m] - S$ of size $\leq \frac{2}{3} \cdot m$ such that, for every codeword $c \in C$,
    \[
    \sum_{i = 1}^m c_i \in (1 \pm 5\eps) \cdot \left ( \sum_{i \in S} c_i + \sum_{i \in T} 2 \cdot c_i \right ).
    \]
\end{lemma}

\begin{proof}
    We let the set $S$ be the union of all of the sets $S_{a, \beta, \eps} \subseteq [m]$ (see \cref{eq:defOfS}), for $a \in \{\frac{1}{1 + \eps}, 1, 1+\eps, \dots m \}$, $\beta \in \{1, 2, 4, 8, \dots m \}$, as well as all choices of $\eta$ (see \cref{rmk:smallBoundary}).

    In particular, we can see that there are $O(\log(m) / \eps)$ choices of $a$, $\log(m)$ choices of $\beta$, and $\frac{200 \log^4(m)}{\eps^4}$ choices of $\eta$. Thus, plugging in the bounds from \cref{lem:CVNRDUpperBoundEdgeCaseParameter} (for small $\beta$) and \cref{lem:CVNRDUpperBoundParameter} (for large $\beta$), we see that 
    \[
    |S| = O \left ( \frac{\log(m)}{\eps} \cdot \log(m) \cdot \frac{ \log^4(m)}{\eps^4} \cdot \CVNRD(C, \eps') \cdot \frac{\log^{191}(m / \eps)}{\eps^{156} \chi^6 \cdot \rho}\right ) 
    \]
    \[
    = \CVNRD(C, \eps') \cdot \mathrm{poly}(\log(m), \eps^{-1}).
    \]

    Now, to construct the set $T$, we simply use \cref{clm:sparsifierCorrectness}. Indeed, this claim shows that when we choose the set $T$ via a random sample at rate $1/2$, we have that, with probability $\geq 1/2$, for every codeword $c \in C$,
    \begin{align}\label{eq:preserveCodewordWeightParam}
   \left ( \sum_{i \in S} c_i + \sum_{i \in T} 2 \cdot c_i \right ) \in (1 \pm 5\eps) \cdot\sum_{i = 1}^m c_i .
    \end{align}
    By a simple Chernoff bound over $m$, we then can see that with probability $\geq 1/4$, simultaneously \cref{eq:preserveCodewordWeightParam} holds and $|T| \leq 2m/3$. This then yields our desired lemma.
\end{proof}

\subsection{Building the Sparsifier}

To build our sparsifier for $C$, we now adopt a simple re-sparsification approach as is typically used in these settings (see \cite{khanna2024code, khanna2025efficient, brakensiek2025redundancy} for instance). We outline this algorithm below:

\begin{algorithm}
    \caption{Sparsify$(C, \eps)$}\label{alg:SparsifyContinuous}
    $i = 0$. \\
    $T^{(0)} = [m]$. \\
    \While{$|\Supp(C|_{T^{(i)}})| > 0$}{
    Let $S^{(i)}, T^{(i+1)}$ be the sets as returned by \cref{lem:peelSetContinuous}. \\ 
    $i \leftarrow i+1$. \\
    }
    Let $d \leftarrow i$
    \Return{$d, (S^{(0)}, S^{(1)}, \dots S^{(d)}).$}
\end{algorithm}

First we bound the number of iterations performed in \cref{alg:SparsifyContinuous}:

\begin{claim}\label{clm:sparsifierRounds}
    When \cref{alg:SparsifyContinuous} is invoked on a code $C \subseteq [0, m^3]^m$ with parameter $\eps$, the algorithm runs for $d \leq 10 \log(m)$ many iterations.
\end{claim}

\begin{proof}
    Indeed, by \cref{lem:peelSetContinuous}, we know that $|T^{(i+1)}| \leq \frac{2}{3} \cdot |T^{(i)}|$. Thus, after $d = \log_{3/2}(m) \leq 10 \log(m)$ iterations, the set $T^{(d)}$ must be empty.  
\end{proof}

Now, we prove the accuracy of the aforementioned sparsifier:

\begin{claim}\label{clm:accuracyBoundcontinuous}
    When \cref{alg:SparsifyContinuous} is invoked on a code $C$ with parameter $\eps$, yielding sets $S^{(0)}, \dots S^{(d)} \subseteq [m]$, it must be the case that for every codeword $c \in C$,
    \[
    \left ( \sum_{j = 0}^d \sum_{i \in S^{(j)}} 2^j \cdot c_i \right ) \in (1 \pm 5 \eps)^{10 \log(m)} \cdot \left ( \sum_{i \in [m]} c_i\right ).
    \]
\end{claim}

\begin{proof}
    We prove the claim inductively. For the base case, observe that via \cref{lem:peelSetContinuous}, the weight function $w: S^{(0)} \rightarrow 1, w: T^{(1)} \rightarrow 2$ (and otherwise $w$ is $0$) is a $(1 \pm \eps)$ sparsifier of $C$. 

    Now, we claim inductively that for $d' < d$, and $j \in [d']$, the weight function $w^{(d')}: S^{(j)} \rightarrow 2^j, w^{(d')}: T^{(d'+1)} \rightarrow 2^{d'+1}$ and otherwise is $0$ is a $(1 \pm 5 \eps)^{d'}$ sparsifier of $C$. Next, by \cref{lem:peelSetContinuous}, we know that the weight function $w': S^{(d'+1)} \rightarrow 2^{d'+1}, w': T^{(d'+2)} \rightarrow 2^{d'+2}$ (and otherwise is $0$) must be a $(1 \pm \eps)$ sparsifier of $2^{d'+1} \cdot C|_{T^{(d'+1)}}$, equivalently the code with weight $2^{d'+1}$ on all coordinates in $T^{(d'+1)}$ (and otherwise $0$). 

    Thus, we can now compose these sparsifiers. Indeed, we construct the weight function $w^{(d'+1)}$ such that for $j \in [d'+1]$, $w^{(d'+1)}: S^{(j)} \rightarrow 2^j, w^{(d'+1)}: T^{(d'+2)} \rightarrow 2^{d'+2}$. This is exactly the result of replacing $w^{(d')}: T^{(d'+1)} \rightarrow 2^{d'+1}$ with $ w'$. \cref{lem:peelSetContinuous} guarantees that this worsens the accuracy of the sparsifier by at most a factor of $(1 \pm 5 \eps)$. Thus, we get that the weight function $w^{(d'+1)}$ is a $(1 \pm \eps)^{d'+1}$ sparsifier of $C$.

    By induction then (and using our bound on $d$ from \cref{clm:sparsifierRounds}), we then obtain the stated claim.
\end{proof}

Finally, we are now ready to conclude this section:

\begin{theorem}\label{thm:CVNRDSparsifier}
    Let $C \subseteq [0, m^3]^m, \eps >0$ be given. Let $\eps' = \frac{\eps^4}{10^{13}\log^{8}(m)}$. Then,
    \[
    \mathrm{SPR}(C, \eps) = O\left ( \CVNRD(C, \eps') \cdot \mathrm{poly}(\log(m), \eps^{-1})\right ).
    \]
\end{theorem}

\begin{proof}
    We simply invoke \cref{alg:SparsifyContinuous} on the code $C$ with parameter $\eps^* = \frac{\eps}{60 \log(m)}$. By \cref{clm:accuracyBoundcontinuous}, we see that the resulting set of weights is a $(1 \pm 5\eps^*)^{10 \log(m)} \in (1 \pm \frac{\eps}{12\log(m)})^{10 \log(m)} \in (1 \pm \eps)$ sparsifier of $C$.

    All that remains is to bound the size: for this, we observe that the size of the sparsifier produced by \cref{clm:accuracyBoundcontinuous} is at most the number of rounds of sparsification $O(\log(m))$ times the number of coordinates in each set $S^{(j)}$. \cref{lem:peelSetContinuous} guarantees that each set $S^{(j)}$ is of size $\leq \CVNRD(C, \eps' \cdot \mathrm{poly}(\log(m), \eps^{-1}))$, where $\eps' = \frac{(\eps^*)^4}{160000 \log^4(m)} \geq \frac{\eps^4}{10^{13}\log^{8}(m)}$. This yields the theorem. 
\end{proof}

We can likewise prove a version of the theorem with the $\chi, \rho$ dependence:

\begin{theorem}\label{thm:CVNRDSparsifierParam}
    Let $C \subseteq [0, m^3]^m, \eps >0, \chi > 0, \rho > 0$ be given. Let $\eps' = \frac{\eps^4}{10^{13}\log^{8}(m)}$. Then,
    \[
    \mathrm{SPR}(C, \eps) = O\left ( \frac{\CVNRD(C, \eps', \chi, \rho)}{\chi^6 \rho}\cdot \mathrm{poly}(\log(m), \eps^{-1})\right ).
    \]
\end{theorem}

\begin{proof}
    We simply invoke \cref{alg:SparsifyContinuous} on the code $C$ with parameter $\eps^* = \frac{\eps}{60 \log(m)}$. By \cref{clm:accuracyBoundcontinuous}, we see that the resulting set of weights is a $(1 \pm 5\eps^*)^{10 \log(m)} \in (1 \pm \frac{\eps}{12\log(m)})^{10 \log(m)} \in (1 \pm \eps)$ sparsifier of $C$.

    All that remains is to bound the size: for this, we observe that the size of the sparsifier produced by \cref{clm:accuracyBoundcontinuous} is at most the number of rounds of sparsification $O(\log(m))$ times the number of coordinates in each set $S^{(j)}$. \cref{lem:peelSetContinuousParam} guarantees that each set $S^{(j)}$ is of size $\leq \frac{\CVNRD(C, \eps', \chi, \rho)}{\chi^6 \rho}\cdot \mathrm{poly}(\log(m), \eps^{-1}))$, where $\eps' = \frac{(\eps^*)^4}{160000 \log^4(m)} \geq \frac{\eps^4}{10^{13}\log^{8}(m)}$. This yields the theorem. 
\end{proof}

\section{Applications to Sparsifying Valued CSPs}\label{sec:valuedCSPs}

As an application, we obtain the first general characterization of the worst-case sparsifiability of \emph{valued constraint satisfaction problems (VCSPs)}, or equivalently, of hypergraphs equipped with arbitrary real-valued splitting functions. Given a predicate $P: \Sigma^r \rightarrow \R^{\geq 0}$ and a universe of $n$ variables $x_1, \dots x_n \in \Sigma$, an instance $\Psi$ of a VCSP is then specified by creating $m$ constraints, each of which is the result of applying the predicate $P$ to some subset of $r$ variables (we denote these subsets by $T_1, \dots T_m \in \binom{[n]}{r}$), and re-weighted by a value $w_i$. For an assignment $x \in \Sigma^n$, the \emph{value} attained by the CSP is then exactly $\Psi(x) = \sum_{i \in [m]} w_i \cdot P(x|_{T_i})$. CSP sparsification then calls for a \emph{re-weighted} sub-CSP $\hat{\Psi}$ such that, for every $x \in \Sigma^n$, 
\[
\hat{\Psi}(x) \in (1 \pm \eps) \cdot \Psi(x),
\]
while $\hat{\Psi}$ retains as few constraints as possible. 
This exactly matches \emph{general hypergraph} sparsification: variables are vertices; each constraint induces a hyperedge with splitting function $P$; the alphabet $\Sigma$ specifies the number of parts of the partition (the standard cut case corresponds to $|\Sigma|=2$).

Following prior CSP sparsification work~\cite{KK15, FK17, BZ20, khanna2024code, khanna2025efficient, brakensiek2025redundancy, brakensiek2025tight, brakensiek2026classification}, we fix a predicate $P$ (with constant $|\Sigma|$ and arity $r$) and ask: what is the worst-case size of a $(1\pm\eps)$-sparsifier over all $n$-variable, unweighted instances? We denote this by $\mathrm{SPR}(P, n, \eps)$. 
While valued CSPs have previously been considered in real-valued sparsification~\cite{BZ20, brakensiek2025redundancy}, concrete structural characterizations were known only for \emph{random} instances~\cite{brakensiek2025tight}.

Most prior work has focused on Boolean predicates. 
For Boolean $P:\Sigma^r\to\{0,1\}$, CSP sparsification coincides with $\{0,1\}$-code sparsification by mapping each assignment $x$ to the codeword $c(x)\in\{0,1\}^m$ where $c(x)_i=\mathbf{1}[P(x|_{T_i})=1]$.
For general real-valued predicates, our framework naturally suggests an analogous structural parameter.

\begin{definition}\label{def:NRDforVCSPs}
    Let $P:\Sigma^r\to\mathbb{R}_{\ge0}$ and consider all $n$-variable instances $\Psi$ with constraints on $T_1,\dots,T_m$. 
    Define $\mathrm{DNRD}(P,n)$ as the maximum $\ell$ such that there exist disjoint $A_1,\dots,A_p\subseteq[m]$ and assignment families $X_1,\dots,X_p\subseteq\Sigma^n$ with $\sum_{i=1}^p|A_i|=\ell$ and:
    \begin{enumerate}
        \item (\emph{Block-diagonal}) For every $i\neq j$, for every $x\in X_i$ and every $\ell\in A_j$, $P(x|_{T_\ell})=0$.
        \item (\emph{Singletons}) If $|A_i|=1$, then there exists $x\in X_i$ with $P(x|_{T_\ell})\neq 0$ for the unique $\ell\in A_i$.
        \item (\emph{Completeness}) If $|A_i|\ge2$, there exist $b_0\neq b_1\in\mathrm{Im}(P)$ such that for every $v\in\{b_0,b_1\}^{A_i}$ there exists $x\in X_i$ with $P(x|_{T_\ell})=v_\ell$ for all $\ell\in A_i$.
    \end{enumerate}
\end{definition}

Indeed, by \Cref{thm:DiscreteDomainRVNRD}, we can immediately conclude that $\mathrm{DNRD}(P,n)$ is essentially \emph{equal} (up to polylogarithmic factors) to the worst-case sparsifiability of VCSPs with predicate $P$.

\begin{corollary}\label{cor:VCSPcharacterization}
        For every predicate $P:\Sigma^r\to\mathbb{R}_{\ge0}$ and every $\eps\in(0,\kappa)$ for some $\kappa=\Omega(1)$,
        \[
            \Omega\!\big(\mathrm{DNRD}(P,n)\big)\ \le\ \mathrm{SPR}(P,n,\eps)\ \le\ \widetilde{O}\!\big(\mathrm{DNRD}(P,n)/\eps^{4}\big).
        \]
\end{corollary}

For constant $\eps$, this corollary provides an exact analog of~\cite{brakensiek2025redundancy}: it reduces the study of VCSP sparsifier size to understanding the redundancy induced by the predicate $P$. 
In particular, when complete blocks are absent, VCSP sparsifiability coincides with that of the corresponding $\{0,1\}$-valued CSP (and code). 
An identical statement holds for weighted sparsifiability when $\mathrm{DNRD}$ is replaced by an analog with chain length.

Note that it may seem that the above definition is artificial, but it turns out that these ``witnesses of unsparsifiability'' that are suggested by \cref{def:NRDforVCSPs} are in fact inherent in the sparsification of VCSPs. To see this, let us study the predicate $P: \zo^2 \rightarrow \{0, 1, 2 \}$ such that:
\[
P(00) = 0, \quad P(01) = 0, \quad P(10) = 1, \quad P(11) = 2.
\]

We have the following claim:

\begin{claim}
    $\mathrm{DNRD}(P, n) = \Omega(n^2)$.
\end{claim}

\begin{proof}
    To see this, let us consider the complete bipartite instance on $n$ variables. In more detail, we let $L, R$ denote two sets of $n/2$ variables, and for each $u \in L$ and each $v \in R$, we add a constraint of the form $P(x_u, x_v)$.

    Now, for each $u \in L$, we consider the set of assignments $\Psi_u \in \zo^{L \cup R} = \{0\}^{L - u} \times \{1\}^{u}\times \zo^{R}$, and the set of constraints $A_u = \{P(x_u, x_v): v \in R\}$. We show that this set of constraints and assignments satisfies the requirements of \cref{def:NRDforVCSPs}:
    \begin{enumerate}
        \item Consider any assignment $x \in \Psi_{u'}$ for $u \neq u$. Then, for any constraint in $A_u$ of the form $P(x_u, x_v)$, we see
        \[
        P(x_u, x_v) = P(0, x_v) = 0,
        \]
        as $x_u = 0$ in any assignment from $\Psi_{u'}$. Thus, this collection of assignments satisfies the block-diagonal condition.
        \item Note that no set $A_u$ is of size $1$, and so the singleton condition is trivially satisfied.
        \item Otherwise, consider any vector of values $\varphi \in \{1, 2\}^{A_u}$. We claim that for each such vector $v$, there exists an assignment $x \in \Psi_u$ such that for every constraint $(u,v) \in A_u$,
        \[
        \varphi_{(u, v)} = P(x_u, x_v).
        \]
        For this, we simply look at the assignment $x$ such that
        \[
        x_i = \begin{cases}
            0 \text{ if $i \in L - u$} \\
            1 \text{ if $i = u$} \\
            \varphi_{(u, i)} - 1 \text{ if $i \in R$}
        \end{cases}.
        \]
        Indeed, for this $x$, one can then verify that $P(x_u, x_v)$ for $v \in R$ satisfies \[
        P(x_u, x_v) = P(1, \varphi_{(u, v)} - 1) = \begin{cases}
            2 \text{ if } \varphi_{(u, v)} = 2 \\
            1 \text{ if } \varphi_{(u, v)} = 1 \\
        \end{cases}.
        \]
        Thus, this set of constraints and assignments satisfies the complete block condition.
    \end{enumerate}

To conclude then, the above implies that $\mathrm{DNRD}(P, n) \geq \sum_{u \in L}|A_u| = \frac{n}{2} \cdot \frac{n}{2} = \Omega(n^2)$.
\end{proof}

We then obtain the following as an immediate corollary from \cref{cor:VCSPcharacterization}:

\begin{corollary}\label{cor:Plowerbound}
    There exists $\eps = \Omega(1)$ such that $\mathrm{SPR}(P, n, \eps) = \Omega(n^2)$.
\end{corollary}

\begin{remark}
    Note that if one replaces the predicate $P$ with its ``booleanization'' $P': \zo^r \rightarrow \zo$ such that 
    \[
P'(00) = 0, \quad P'(01) = 0, \quad P'(10) = 1, \quad P'(11) = 1,
\]
then $\mathrm{SPR}(P', n, \eps) = \widetilde{O}( n / \eps^2)$ (see, for instance, \cite{khanna2024code}). The sparsification lower bound of \cref{cor:Plowerbound} is \emph{inherent} to the valued setting.
\end{remark}

We view it as a very interesting open question to build a more complete characterization of this ``discrete'' non-redundancy in valued CSPs, akin to the work of \cite{khanna2024code, khanna2025efficient, brakensiek2025redundancy, brakensiek2025richness, brakensiek2026classification}.

\section{Applications to Stronger Graph Sparsification}

As an additional application of our framework, we study basic variants of graph sparsification. 

\subsection{Background}
To do so, we first formally introduce the notion of cut sparsification, which dates back to the work of \cite{BK96, talagrand1990embedding}:

\begin{definition}[Cut Sparsification]
    Given a graph $G = (V, E)$ and an assignment $x \in \R^V$, we define the \emph{cut-energy} of $G$ on assignment $x$ to be exactly
    \[
    \sum_{e  = (u,v) \in E} |x_u - x_v|.
    \]
    For a parameter $\eps > 0$, a $(1 \pm \eps)$ \emph{cut-sparsifier} of $G$ is a re-weighted subset of edges $E' \subseteq E$ along with weights $w: E' \rightarrow \R^{\geq 0}$ such that for every $x \in \R^V$,
    \[
    \left ( \sum_{e  = (u,v) \in E'} w(e) \cdot |x_u - x_v|\right ) \in (1 \pm \eps) \cdot \left ( \sum_{e  = (u,v) \in E} |x_u - x_v|\right ).
    \]
\end{definition}

The celebrated works of \cite{BK96, talagrand1990embedding} show that such cut sparsification \emph{is} indeed achievable while retaining only $\widetilde{O} \left ( n / \eps^2\right )$ many edges from the graph $G$. 

In a separate direction, several works have also studied \emph{spectral sparsification} of graphs (see, for instance, \cite{spielman2008graph, BSS09, ST11}):

\begin{definition}[Spectral Sparsification]
    Given a graph $G = (V, E)$ and an assignment $x \in \R^V$, we define the \emph{spectral-energy} of $G$ on assignment $x$ to be exactly
    \[
    \sum_{e  = (u,v) \in E} |x_u - x_v|^2.
    \]
    For a parameter $\eps > 0$, a $(1 \pm \eps)$ \emph{spectral sparsifier} of $G$ is a re-weighted subset of edges $E' \subseteq E$ along with weights $w: E' \rightarrow \R^{\geq 0}$ such that for every $x \in \R^V$,
    \[
    \left ( \sum_{e  = (u,v) \in E'} w(e) \cdot |x_u - x_v|^2\right ) \in (1 \pm \eps) \cdot \left ( \sum_{e  = (u,v) \in E} |x_u - x_v|^2\right ).
    \]
\end{definition}

Again in this regime, \cite{spielman2008graph, BSS09, ST11} have been able to show the existence of sparsifiers which retain only $O(n / \eps^2)$ many edges. We now introduce a \emph{new} definition, called $p$-spectral sparsification:

\begin{definition}[$p$-Spectral Sparsification]
    Given a graph $G = (V, E)$ a parameter $p > 0$, and an assignment $x \in \R^V$, we define the $p$-\emph{spectral-energy} of $G$ on assignment $x$ to be exactly
    \[
    \sum_{e  = (u,v) \in E} |x_u - x_v|^p.
    \]
    For a parameter $\eps > 0$, a $(1 \pm \eps)$ $p$-\emph{spectral sparsifier} of $G$ is a re-weighted subset of edges $E' \subseteq E$ along with weights $w: E' \rightarrow \R^{\geq 0}$ such that for every $x \in \R^V$,
    \[
    \left ( \sum_{e  = (u,v) \in E'} w(e) \cdot |x_u - x_v|^p\right ) \in (1 \pm \eps) \cdot \left ( \sum_{e  = (u,v) \in E} |x_u - x_v|^p\right ).
    \]
\end{definition}

Note that this notion of energy has implicitly been studied before: indeed, the work of \cite{JLLS23} showed that graphs $G = (V, E)$ admit $(1 \pm \eps)$ $p$-spectral sparsifiers with $\widetilde{O} \left (\frac{\max(n, n^{p/2})} {\eps^2} \right )$ many retained edges. However, for $p > 2$, this result loses polynomial factors in $n$ in the sparsifier size (and for $p \geq 4$, it is completely trivial). 

In this section, we prove the following strengthening of their result:

\begin{theorem}\label{thm:pSpectralSparsifiers}
    Let $G = (V, E)$ be a graph on $n$ vertices, let $\eps > 0$, and let $p \geq 1$. Then, $G$ admits a $(1 \pm \eps)$ $p$-spectral sparsifier of size $O\left ( n \cdot \mathrm{poly}(\log^p(n), \eps^{-p})\right )$.
\end{theorem}

This theorem shows that, among others, when $p$ is constant (i.e.,  $p = O(1)$), there exists a $p$-spectral sparsifier which retains only $\widetilde{O}\left (n \cdot \mathrm{poly}(\eps^{-1}) \right )$ many edges.

\subsection{Notation and Connection to $\CVNRD$}

To prove \cref{thm:pSpectralSparsifiers}, we first require the notion of the ``$p$-spectral code'' of a graph:

\begin{definition}
    Let $G = (V, E)$ be a graph (potentially a multi-graph) on $n$ vertices and $m$ edges. We let $C_{G, p} \subseteq \R_{\geq0}^m$ denote the $p$-spectral code of $G$ which contains one vector $v^{(x)}$ for each assignment $x \in \R^V$, where for an edge $e = (u,v) \in E$, the corresponding entry in the vector $v^{(x)}$ is $(v^{(x)})_e = |x_u - x_v|^p$. 
\end{definition}

\begin{remark}\label{rmk:graphCodeEquivalent}
    Observe that, for a graph $G$, building a $(1 \pm \eps)$ code sparsifier of $C_{G, p}$ is in fact \emph{equivalent} to building a $(1 \pm \eps)$ $p$-spectral sparsifier of the graph $G$. 
\end{remark}

Inspired by this connection to $C_{G, p}$, we now define the $\CVNRD$ of a graph $G$:

\begin{definition}\label{def:pCVNRDofGraph}
    For a graph $G = (V, E)$, and parameters $\eps > 0, \chi > 0, \rho > 0, p \geq 1$, the continuous-valued non-redundancy of $G$'s $p$-spectral energy (denoted by $\CVNRD_p(G, \eps, \chi, \rho)$) is defined as the largest integer $\ell$ such that there exist disjoint sets of edges $F_1, \dots F_k \subseteq E$, along with sets of potential assignments $\Psi_1, \dots \Psi_k \in \R^V$ such that:
    \begin{enumerate}
        \item $\sum_{i = 1}^k |F_i| = \ell$. 
        \item For $i \in [k]$, if $|F_i| = 1$, then let $b_2^{(i)} = |x_u - x_v|^p >0$ for the unique assignment $x \in \Psi_i$ and edge $e = (u,v) \in F_i$.
        \item For each $F_i: i \in [k]$, there exist values $b^{(i)}_1, b^{(i)}_2$, with $b^{(i)}_1 < b^{(i)}_2( 1- \eps)$, such that for any set of edges $B \subseteq F_i$, there is an assignment $x \in \Psi_i$ such that for $e = (u,v) \in B$, $|x_u - x_v|^p \in b^{(i)}_2 \cdot [1, 1 + \chi \cdot \eps ]$, and for $e = (u,v) \in F_i - B$, $|x_u - x_v|^p \in b^{(i)}_1 \cdot [1, 1 + \chi \cdot \eps ]$.
    \item For each $F_i: i \in [k]$, and for every $x \in \Psi_i$, $\sum_{j \in [k]: j \neq i} \sum_{e = (u,v) \in F_j} |x_u - x_v|^p \leq \rho \cdot b^{(i)}_2 \cdot |F_i|$.
    \item For each $F_i: i \in [k]$, and for every $x \in \Psi_i$, $\max_{j \in [k]: j \neq i} \max_{e = (u,v) \in F_j} |x_u - x_v|^p \leq \rho \cdot b^{(i)}_2$.
    \end{enumerate}
\end{definition}

We have the following proposition:

\begin{proposition}\label{prop:translateGraphCodes}
     For a graph $G = (V, E)$, and parameters $\eps > 0, \chi > 0, \rho > 0, p \geq 1$, 
     \[
     \CVNRD(C_{G, p}, \eps, \chi, \rho) = \CVNRD_p(G, \eps, \chi, \rho).
     \]
\end{proposition}

\begin{proof}
    Any witness of $\CVNRD(C_{G, p}, \eps, \chi, \rho)$ can be directly translated into a witness of $\CVNRD_p(G, \eps, \chi, \rho)$ and vice versa.
\end{proof}

\subsection{Bounding the $\CVNRD$ of Graphs}

We will show the following theorem:

\begin{theorem}\label{thm:boundCVNRDGraph}
    Let $G = (V, E)$ be a graph (potentially a multi-graph) on $n$ vertices and $m$ edges. Let $p \geq 1, \eps > 0$. Then, \[
    \CVNRD_p \left (G, \eps, \frac{1}{100 \log^2(m)}, \frac{\eps^p}{(4p \log(n))^p \cdot2} \right ) = O( n).\]
\end{theorem}
To prove this, we begin with the following claim which shows that in any $\CVNRD$ witness of a graph $G$ (with the above choices of $\rho, \chi$), there can be \emph{no short cycles}:

\begin{claim}\label{clm:NoShortCycles}
    Let $p, G, \eps$ be given and let $F_1, \dots F_k, \Psi_1, \dots \Psi_k$ be defined as in \cref{def:pCVNRDofGraph}, i.e., forming a valid witness of $\CVNRD_p \left (G, \eps, \frac{1}{100 \log^2(m)}, \rho \right )$. Then, for any $i \in [k]$, $F_i$ does not contain any cycles of length $\leq \log^2(m)$.
\end{claim}

\begin{proof}
    Indeed, let us fix $F_i$, and let us suppose for the sake of contradiction that there is a cycle in $F_i$ of length $\leq \log^2(m)$. Let us denote this cycle by $C$, and let us refer to the path taken by the cycle as $v_0 \rightarrow v_1 \rightarrow \dots \rightarrow v_{\ell} = v_0$. There are two cases for us to consider:
    \begin{enumerate}
        \item Suppose that $|C| = \ell$ is odd. Then, we know there is an assignment $x \in \Psi_i$ such that every edge $(v_j,v_{j+1}) \in C$ receives energy $|x_{v_j} - x_{v_{j+1}}|^p \in b^{(i)}_2 [1, 1 + \eps / 100 \log^2(m)]$. In particular, this means that $|x_{v_j} - x_{v_{j+1}}| \in (b_2^{(i)})^{1/p} \cdot [1, (1 + \eps / 100 \log^2(m))^{1/p}]$. Thus, we can write $x_{v_j} - x_{v_{j+1}} = \sigma_{v_j, v_{j+1}} \cdot \left ( b_2^{(i)} \right )^{1/p} + \delta_{v_j, v_{j+1}}$, where $\sigma_{v_j, v_{j+1}} \in \pm 1$, and 
        \[
        |\delta_{v_j, v_{j+1}}| \leq \left ( b_2^{(i)} \right )^{1/p} \cdot \left [  \left ( 1 + \frac{\eps}{100 \log^2(m)} \right ) ^{1/p} - 1\right ] \leq \left ( b_2^{(i)} \right )^{1/p} \cdot \frac{\eps}{100 p \log^2(m)}.
        \]

        We refer to each one of these terms $x_{v_j} - x_{v_{j+1}}$ as a \emph{potential difference}. In particular, we can see that 
        \begin{align}\label{eq:sumDifferencesCycle1}
        \sum_{j = 0}^{\ell-1} [x_{v_j} - x_{v_{j+1}}] = 0,
        \end{align}
        as we are summing over edges that form a cycle (and so each vertex appears once in the positive direction, and once in the negative direction). At the same time, we have that 
        \[
        \sum_{j = 0}^{\ell-1} \left [x_{v_j} - x_{v_{j+1}} \right ] = \sum_{j = 0}^{\ell-1} \left [ \sigma_{v_j, v_{j+1}} \cdot \left ( b_2^{(i)} \right )^{1/p} + \delta_{v_j, v_{j+1}} \right ].
        \]
        Then, we can see that 
        \[
        \left | \sum_{j = 0}^{\ell-1} \left [ \sigma_{v_j, v_{j+1}} \cdot \left ( b_2^{(i)} \right )^{1/p} + \delta_{v_j, v_{j+1}} \right ] \right | 
        \]
        \[
        \geq \left | \sum_{j = 0}^{\ell-1}  \sigma_{v_j, v_{j+1}}  \left ( b_2^{(i)} \right )^{1/p}  \right | - \left | \sum_{j = 0}^{\ell-1} \delta_{v_j, v_{j+1}}\right |
        \]
        \[
        \geq \left ( b_2^{(i)} \right )^{1/p} - \ell \cdot \left ( b_2^{(i)} \right )^{1/p} \cdot \frac{\eps}{100 p \log^2(m)} 
        \]
        \[
        \left ( b_2^{(i)} \right )^{1/p} -  \left ( b_2^{(i)} \right )^{1/p} \cdot \frac{\eps}{100 p} > 0.
        \]
        Note that in the penultimate equality, we have used the fact that $\ell$ is odd, and that each $\sigma_{v_j, v_{j+1}} \in \pm 1$ to conclude that $\left | \sum_{j = 0}^{\ell-1}  \sigma_{v_j, v_{j+1}}  \left ( b_2^{(i)} \right )^{1/p}  \right | \geq \left ( b_2^{(i)} \right )^{1/p}$. Finally, the conclusion from the above is that $\sum_{j = 0}^{\ell-1} \left [x_{v_j} - x_{v_{j+1}} \right ] > 0$,
        which is a contradiction with \cref{eq:sumDifferencesCycle1}. So, there can be no cycles of odd length.
        \item Otherwise, let us suppose that $|C| = \ell$ is even. Then, we know there is an assignment $x \in \Psi_i$ such that every edge $(v_j,v_{j+1}) \in C$ for $j \in \{0, \dots \ell -2\}$ receives energy \[
        |x_{v_j} - x_{v_{j+1}}|^p \in b^{(i)}_2 [1, 1 + \eps / 100 \log^2(m)],\]
        while $(v_{\ell-1}, v_{\ell})$ receives energy \[
        |x_{v_{\ell-1}} - x_{v_{\ell}}|^p\in b^{(i)}_1 \cdot [1, 1 + \eps / 100 \log^2(m)].
        \]
        In particular, this means that for $j \in \{0, \dots \ell -2\}$, \[
        |x_{v_j} - x_{v_{j+1}}| \in (b_2^{(i)})^{1/p} \cdot [1, (1 + \eps / 100 \log^2(m))^{1/p}],
        \]
         while \[
         |x_{v_{\ell-1}} - x_{v_{\ell}}| \in (b_1^{(i)})^{1/p} \cdot [1, (1 + \eps / 100 \log^2(m))^{1/p}].
         \]
        Thus, for $j \in \{0, \dots \ell -2\}$ we can write $x_{v_j} - x_{v_{j+1}} = \sigma_{v_j, v_{j+1}} \cdot \left ( b_2^{(i)} \right )^{1/p} + \delta_{v_j, v_{j+1}}$, where $\sigma_{v_j, v_{j+1}} \in \pm 1$, and $
        |\delta_{v_j, v_{j+1}}| \leq \left ( b_2^{(i)} \right )^{1/p} \cdot \frac{\eps}{100 p \log^2(m)}.$ Similarly, $x_{v_{\ell-1}} - x_{v_{\ell}} = \sigma_{v_{\ell-1}, v_{\ell}} \cdot \left ( b_1^{(i)} \right )^{1/p} + \delta_{v_{\ell-1}, v_{\ell}}$, with $\sigma_{v_{\ell-1}, v_{\ell}} \in \{\pm 1\}$, and $
        |\delta_{v_{\ell-1}, v_{\ell}}| \leq \left ( b_1^{(i)} \right )^{1/p} \cdot \frac{\eps}{100 p \log^2(m)}.$
Now, as in \cref{eq:sumDifferencesCycle1}, we know \[
        \sum_{j = 1}^{\ell } [x_{v_j} - x_{v_{j+1}}] = 0.
        \]
        At the same time, 
        \[
        \sum_{j = 1}^{\ell } \left [ x_{v_j} - x_{v_{j+1}} \right ] = \sigma_{v_j, v_{j+1}} \cdot \left ( b_2^{(i)} \right )^{1/p} + \delta_{v_j, v_{j+1}} + \sum_{j = 1}^{\ell-1 } \left [ \sigma_{v_j, v_{j+1}} \cdot \left ( b_2^{(i)} \right )^{1/p} + \delta_{v_j, v_{j+1}} \right ].
        \]
        Because $\ell$ is even, $\ell-1$ is odd. So, by the above, we know that 
        \[
        \left | \sum_{j = 1}^{\ell -1} [x_{v_j} - x_{v_{j+1}}] \right | = \left | \sum_{j = 1}^{\ell-1 } \left [ \sigma_{v_j, v_{j+1}} \cdot \left ( b_2^{(i)} \right )^{1/p} + \delta_{v_j, v_{j+1}} \right ] \right | \geq  \left ( b_2^{(i)} \right )^{1/p} - \sum_{j = 1}^{\ell-1} \left | \delta_{v_j, v_{j+1}} \right | 
        \]
        \[
        \geq \left ( b_2^{(i)} \right )^{1/p} - \frac{\left ( b_2^{(i)} \right )^{1/p} \cdot \ell \cdot \eps}{100p \log^2(m)} > \left ( b_2^{(i)} \right )^{1/p} (1 - \eps / 100p).
        \]
        But, then we can see that 
        \[
        \left | \sum_{j = 1}^{\ell } [x_{v_j} - x_{v_{j+1}} ]\right | \geq \left | \sum_{j = 1}^{\ell -1} [x_{v_j} - x_{v_{j+1}}] \right | - \left | x_{v_{\ell-1}} - x_{v_{\ell}} \right |
        \]
        \[
        \geq \left ( b_2^{(i)} \right )^{1/p} (1 - \eps / 100p) - \left (b_1^{(i)} \right)^{1/p} \cdot (1 + \eps / 100 \log^2(m))^{1/p} 
        \]
        \begin{align}\label{eq:finalEnergyDiff}
        \geq \left ( b_2^{(i)} \right )^{1/p} (1 - \eps / 100p) - \left (b_1^{(i)} \right )^{1/p} \cdot \left (1 + \frac{\eps}{100 p \log^2(m)} \right ).
        \end{align}
        Finally, recall that $b_1^{(i)} < b_2^{(i)} \cdot (1 - \eps)$ by \cref{def:pCVNRDofGraph}. Thus, 
        \[
         \left ( b_1^{(i)} \right )^{1/p} < \left ( b_2^{(i)} \right )^{1/p} \cdot (1 - \eps)^{1/p} < \left ( b_2^{(i)} \right )^{1/p} \cdot \left (1 - \frac{\eps}{2p} \right ).
        \]
        Plugging this in to \cref{eq:finalEnergyDiff}, we see that 
        \[
        \left ( b_2^{(i)} \right )^{1/p} (1 - \eps / 100p) - \left (b_1^{(i)} \right )^{1/p} \cdot \left (1 + \frac{\eps}{100 p \log^2(m)} \right )
        \]
        \[
        \geq \left ( b_2^{(i)} \right )^{1/p} \left [ \left (1 - \frac{\eps}{100p} \right ) - \left (1 + \frac{\eps}{100 p \log^2(m)} \right ) \cdot \left ( 1 - \frac{\eps}{2p}\right ) \right ] \geq \left ( b_2^{(i)} \right )^{1/p} \cdot \left ( \frac{\eps}{4p}\right ) > 0.
        \]
        Thus, we see that $\left | \sum_{j = 1}^{\ell } [x_{v_j} - x_{v_{j+1}} ]\right | > 0$.
    \end{enumerate}

    Thus, in both cases we reach a contradiction, and so $F_i$ must not contain any cycles of length $\leq \log^2(m)$.
\end{proof}

We are now ready to prove \cref{thm:boundCVNRDGraph}:

\begin{proof}[Proof of \cref{thm:boundCVNRDGraph}]
    Consider any witness to the $\CVNRD$, denoted by $F_1, \dots F_k$, with assignments $\Psi_1, \dots \Psi_k$, and assume that $|F_1 \cup \dots \cup F_k| \geq 100n$. Then, there must exist some cycle $C \subseteq F_1 \cup \dots \cup F_k$ such that $|C| \leq \log(n)$. Importantly, by \cref{clm:NoShortCycles}, this cycle \emph{must} use edges from multiple different blocks. Let $F_i$ denote one of those blocks which has edges in $C$. We let $\ell_i = |C \cap F_i|$ denote the number of edges in $C$ which are from $F_i$, and we let $\ell = |C| \leq \log(n)$ denote the size of the cycle. We let $P = v_0 \rightarrow v_1 \rightarrow \dots \rightarrow v_{\ell}$ denote the sequence of edges traversed in the cycle, and, as in \cref{clm:NoShortCycles}, we observe that, for any assignment $x \in \R^V$ of vertex potentials, it must be that $\sum_{(v_j, v_{j+1}) \in P} (x_{v_j} - x_{v_{j+1}}) = 0$. Now, for each cycle $C$, we have two cases:
    \begin{enumerate}
        \item If $\ell_i$ is odd, then we consider an assignment $x \in \Psi_i$ such that, for every $e = (v_j,v_{j+1}) \in C \cap F_i$, $|x_{v_j} - x_{v_{j+1}}|^p \in b^{(i)}_2 [1, 1 + \eps / 100 \log^2(m)]$. Thus, for these $j$ (as in the proof of \cref{clm:NoShortCycles}) we can write $x_{v_j} - x_{v_{j+1}} = \sigma_{v_j, v_{j+1}} \cdot \left ( b_2^{(i)} \right )^{1/p} + \delta_{v_j, v_{j+1}}$, where $\sigma_{v_j, v_{j+1}} \in \pm 1$, and $|\delta_{v_j, v_{j+1}}| \leq \left ( b_2^{(i)} \right )^{1/p} \cdot \frac{\eps}{100 p \log^2(m)}.$

        Because $\ell$ is odd, it must be that 
        \[
        \left | \sum_{(v_j, v_{j+1}) \in P \cap F_i} \left [ \sigma_{v_j, v_{j+1}} \cdot \left ( b_2^{(i)} \right )^{1/p} + \delta_{v_j, v_{j+1}}\right ] \right | \geq  \left ( b_2^{(i)} \right )^{1/p} - \sum_{(v_j, v_{j+1}) \in P \cap F_i} \left | \delta_{v_j, v_{j+1}} \right | 
        \]
        \[
        \geq \left ( b_2^{(i)} \right )^{1/p} - \frac{\left ( b_2^{(i)} \right )^{1/p} \cdot \ell \cdot \eps}{100 \log^2(m)} \geq \frac{\left ( b_2^{(i)} \right )^{1/p}}{2}.
        \]

        Because the sum of potential differences along the cycle is $0$, we then also know that 
        \[
        \left | \sum_{(v_j, v_{j+1}) \in P \cap \overline{F_i}} (x_{v_j} - x_{v_{j+1}}) \right | \geq \frac{\left ( b_2^{(i)} \right )^{1/p}}{2}.
        \]
        Using the fact that the cycle is of length at most $\log(n)$, there must then be some edge $e' = (v_j, v_{j+1}) \in C \setminus F_i$ such that
        \begin{align}\label{eq:energyfromnonFi}
       \left | x_{v_j} - x_{v_{j+1}}\right |^p \geq \cdot \left ( \frac{\left ( b_2^{(i)} \right )^{1/p}}{2 \log(n)}\right )^p \geq \frac{ b_2^{(i)}}{2^p \log^{p}(n)}.
        \end{align}
        \item If $\ell_i$ is even, we consider an assignment $x \in \Psi_i$ such that for all but one of the edges $e = (v_j, v_{j+1}) \in P \cap F_i$, $|x_{v_j} - x_{v_{j+1}}|^p \in b^{(i)}_2\cdot[1, 1 + \eps / 100 \log^2(m)]$, and for the final edge $e^*$ in $P \cap F_i$, $|x_{v_j} - x_{v_{j+1}}|^p \in b^{(i)}_1\cdot[1, 1 + \eps / 100 \log^2(m)]$. 
        
        Importantly, as in the proof of \cref{clm:NoShortCycles}, we see that for all but the final edge in $P \cap F_i$ we can write $x_{v_j} - x_{v_{j+1}} = \sigma_{v_j, v_{j+1}} \cdot \left ( b_2^{(i)} \right )^{1/p} + \delta_{v_j, v_{j+1}}$, where $\sigma_{v_j, v_{j+1}} \in \pm 1$, and $
        |\delta_{v_j, v_{j+1}}| \leq \left ( b_2^{(i)} \right )^{1/p} \cdot \frac{\eps}{100 p \log^2(m)}.$ For the final edge $e^* = (v_j, v_{j+1}) \in P \cap F_i$, we can write $x_{v_{j}} - x_{v_{j+1}} = \sigma_{v_{j}, v_{j+1}} \cdot \left ( b_1^{(i)} \right )^{1/p} + \delta_{v_{j}, v_{j+1}}$, with $\sigma_{v_{j}, v_{j+1}} \in \{\pm 1\}$, and $
        |\delta_{v_{j}, v_{j+1}}| \leq \left ( b_1^{(i)} \right )^{1/p} \cdot \frac{\eps}{100 p \log^2(m)}.$

        Thus, letting $j^*$ denote the index of this final edge $e^*$ in $P \cap F_i$, we have 
        \[
        \sum_{(v_j, v_{j+1}) \in P \cap F_i} \left [ x_{v_j} - x_{v_{j+1}} \right ] 
        \]
        \[
        = \sigma_{v_{j^*}, v_{j^*+1}} \cdot \left ( b_1^{(i)} \right )^{1/p} + \delta_{v_{j^*}, v_{j^*+1}} + \sum_{(v_j, v_{j+1}) \in P \cap F_i \setminus \{v_{j^*}, v_{j^* + 1}\}} \left [ \sigma_{v_j, v_{j+1}} \cdot \left ( b_2^{(i)} \right )^{1/p} + \delta_{v_j, v_{j+1}} \right ].
        \]
        Because $\ell_i$ is even, $\ell_i-1$ is odd. So, as in the proof of \cref{clm:NoShortCycles}, we see that 
\[
        \left | \sum_{(v_j, v_{j+1}) \in P \cap F_i \setminus \{v_{j^*}, v_{j^* + 1}\}} [x_{v_j} - x_{v_{j+1}}] \right | 
        \]
        \[
        = \left | \sum_{(v_j, v_{j+1}) \in P \cap F_i \setminus \{v_{j^*}, v_{j^* + 1}\}} \left [ \sigma_{v_j, v_{j+1}} \cdot \left ( b_2^{(i)} \right )^{1/p} + \delta_{v_j, v_{j+1}} \right ] \right | 
        \]
        \[
        \geq  \left ( b_2^{(i)} \right )^{1/p} - \sum_{(v_j, v_{j+1}) \in P \cap F_i \setminus \{v_{j^*}, v_{j^* + 1}\}} \left | \delta_{v_j, v_{j+1}} \right |
        \]
        \[
        \geq \left ( b_2^{(i)} \right )^{1/p} - \frac{\left ( b_2^{(i)} \right )^{1/p} \cdot \ell \cdot \eps}{100p \log^2(m)} > \left ( b_2^{(i)} \right )^{1/p} (1 - \eps / 100p).
        \]
        But, then we can see that 
        \[
        \left | \sum_{(v_j, v_{j+1}) \in P \cap F_i} (x_{v_j} - x_{v_{j+1}}) \right | \geq \left ( b_2^{(i)} \right )^{1/p} (1 - \eps / 100p) - \left ( b_1^{(i)} \right )^{1/p} - \delta_{v_{j^*}, v_{j^*+1}}
        \]
        \[
        \geq\left ( b_2^{(i)} \right )^{1/p} (1 - \eps / 100p)  - \left ( b_1^{(i)} \right )^{1/p} \cdot \left (1 + \frac{\eps}{100 p \log^2(m)} \right ).
        \]
        Again using the relation that $b^{(i)}_1 < b_2^{(i)} ( 1- \eps)$ from the definition of \cref{def:pCVNRDofGraph}, we then see that
        \[
         \left | \sum_{(v_j, v_{j+1}) \in P \cap F_i} (x_{v_j} - x_{v_{j+1}}) \right | 
         \]
         \[
         \geq\left ( b_2^{(i)} \right )^{1/p} \left [ \left (1 - \frac{\eps}{100p} \right ) - \left (1 + \frac{\eps}{100 p \log^2(m)} \right ) \cdot \left ( 1 - \frac{\eps}{2p}\right ) \right ] 
         \]
         \[
         \geq \left ( b_2^{(i)} \right )^{1/p} \cdot \left ( \frac{\eps}{4p}\right )
        \]
        Because the sum of potential differences along the cycle is $0$, we then also know that 
        \[
        \left | \sum_{(v_j, v_{j+1}) \in P \cap \overline{F_i}} (x_{v_j} - x_{v_{j+1}}) \right | \geq \left ( b_2^{(i)} \right )^{1/p} \cdot \left ( \frac{\eps}{4p}\right ).
        \]
        Using the fact that the cycle is of length at most $\log(n)$, there must then be some edge $(v_j, v_{j+1}) \in C \setminus F_i$ such that 
        \begin{align}\label{eq:energyFromNonFi2}
        |x_{v_j} - x_{v_{j+1}}|^p \geq \left ( \frac{\left ( b_2^{(i)} \right )^{1/p} \cdot \left ( \frac{\eps}{4p}\right )}{\log(n)} \right )^p \geq b_2^{(i)} \cdot \frac{\eps^p}{(4p \log(n))^p}.
        \end{align}
    \end{enumerate}

    With the above established, the key observation is the following: if we consider the assignment $x \in \Psi_i$ which agrees with the above cases (i.e., if $\ell_i$ is odd, it gives energy $\approx b_2^{(i)}$ to every edge in the $C \cap F_i$, and if $\ell_i$ is even it gives energy $\approx b_2^{(i)}$ to all but one edge in $C \cap F_i$).
    
    For this assignment $x$, we then see that
    \[
    \max_{(u,v) \in C \cap \overline{F_i}} \left | x_u - x_v \right |^p \geq b_2^{(i)} \cdot \frac{\eps^p}{(4p \log(n))^p},
    \]
    where we have used \cref{eq:energyfromnonFi} and \cref{eq:energyFromNonFi2}.

    But, this means that for this assignment $x$, we have violated the definition of $\CVNRD$: indeed, for this $x$, \[
    \max_{j \in [k]: j \neq i} \max_{e = (u,v) \in F_j} |x_u - x_v|^p \geq b_2^{(i)} \cdot \frac{\eps^p}{(4p \log(n))^p}.
    \] 
    However, \cref{def:pCVNRDofGraph} with our choice of $\rho = \frac{\eps^p}{(4p \log(n))^p \cdot 2}$ would require that 
\[
    \max_{j \in [k]: j \neq i} \max_{e = (u,v) \in F_j} |x_u - x_v|^p \leq b_2^{(i)} \cdot \frac{\eps^p}{(4p \log(n))^p \cdot 2},
    \] 
    which is violated. 

    Thus, \[
    \CVNRD_p \left (G, \eps, \frac{1}{100 \log^2(m)}, \frac{\eps^p}{(4p \log(n))^p \cdot 2} \right ) = O( n).\]
\end{proof}

Finally, we can prove \cref{thm:pSpectralSparsifiers}:

\begin{proof}[Proof of \cref{thm:pSpectralSparsifiers}.]
    By \cref{prop:translateGraphCodes}, we know that for $\eps > 0, \chi > 0, \rho > 0, p \geq 1$, 
     \[
     \CVNRD(C_{G, p}, \eps, \chi, \rho) \leq \CVNRD_p(G, \eps, \chi, \rho).
     \]
     We also know by \cref{rmk:graphCodeEquivalent} that it suffices to build a $(1 \pm \eps)$-sparsifier of $C_{G, p}$ in order to build a $(1 \pm \eps)$ $p$-spectral sparsifier of $G$. 

     Next, we invoke \cref{thm:CVNRDSparsifierParam} which states that 
    \[
    \mathrm{SPR}(C_{G, p}, \eps) = O\left ( \frac{\CVNRD(C_{G,p}, \eps', \chi, \rho)}{\chi^6 \rho}\cdot \mathrm{poly}(\log(m), \eps^{-1})\right ),
    \]
    where $\eps ' = \frac{\eps^4}{10^{13} \log^8(m)}$.

    Now, by \cref{prop:translateGraphCodes} and \cref{thm:boundCVNRDGraph},  we know that \[
    \CVNRD \left (C_{G,p}, \eps', \frac{1}{100 \log^2(m)}, \frac{(\eps')^p}{(4p \log(n))^p \cdot 2} \right ) = O( n).
    \]
    Thus, 
    \[
    \mathrm{SPR}(C_{G, p}, \eps) = O\left ( \frac{n}{\left ( \frac{1}{100 \log^2(m)} \right )^6 \left ( \frac{(\eps')^p}{(4p \log(n))^p \cdot 3 \log(n)}\right )}\cdot \mathrm{poly}(\log(n), \eps^{-1})\right ) 
    \]
    \[
    = O\left ( n \cdot \mathrm{poly}(\log^p(n), \eps^{-p})\right ).
    \]
    This concludes the claim. 
\end{proof}

\section{Applications to Stronger Hypergraph Sparsification}

Just as our framework can be used to study variants of graph sparsification, the same approach works for studying \emph{hypergraph} sparsification. As in the previous section, we begin with some background about hypergraph sparsification.

\subsection{Background}
To do so, we first formally introduce the notion of hypergraph cut sparsification, which dates back to the work of \cite{KK15, CKN20}:

\begin{definition}[Cut Sparsification]
    Given a hypergraph $H = (V, E)$ and an assignment $x \in \R^V$, we define the \emph{cut-energy} of $G$ on assignment $x$ to be exactly
    \[
    \sum_{e  \in E} \max_{u,v \in e}|x_u - x_v|.
    \]
    For a parameter $\eps > 0$, a $(1 \pm \eps)$ \emph{cut-sparsifier} of $H$ is a re-weighted subset of hyperedges $E' \subseteq E$ along with weights $w: E' \rightarrow \R^{\geq 0}$ such that for every $x \in \R^V$,
    \[
    \left ( \sum_{e \in E'} w(e) \cdot \max_{u,v \in e}|x_u - x_v|\right ) \in (1 \pm \eps) \cdot \left ( \sum_{e  = (u,v) \in E} \max_{u,v \in e}|x_u - x_v|\right ).
    \]
\end{definition}

The work of \cite{CKN20} shows that such cut sparsification \emph{is} indeed achievable while retaining only $\widetilde{O} \left ( n / \eps^2\right )$ many hyperedges from the hypergraph $H$. 

In a separate direction, several works have also studied \emph{spectral sparsification} of hypergraphs, initiated by the work of Soma and Yoshida \cite{SY19}:

\begin{definition}[Spectral Sparsification]
    Given a hypergraph $H = (V, E)$ and an assignment $x \in \R^V$, we define the \emph{spectral-energy} of $H$ on assignment $x$ to be exactly
    \[
    \sum_{e\in E} \max_{u,v \in e}|x_u - x_v|^2.
    \]
    For a parameter $\eps > 0$, a $(1 \pm \eps)$ \emph{spectral sparsifier} of $H$ is a re-weighted subset of hyperedges $E' \subseteq E$ along with weights $w: E' \rightarrow \R^{\geq 0}$ such that for every $x \in \R^V$,
    \[
    \left ( \sum_{e  \in E'} w(e) \cdot \max_{u,v \in e}|x_u - x_v|^2\right ) \in (1 \pm \eps) \cdot \left ( \sum_{e  = (u,v) \in E} \max_{u,v \in e}|x_u - x_v|^2\right ).
    \]
\end{definition}

Again in this regime, \cite{JLS22, Lee23} have been able to show the existence of sparsifiers which retain only $\widetilde{O}(n / \eps^2)$ many hyperedges. We now introduce a \emph{new} definition, called $p$-spectral sparsification:

\begin{definition}[$p$-Spectral Sparsification]
    Given a hypergraph $H = (V, E)$ a parameter $p > 0$, and an assignment $x \in \R^V$, we define the $p$-\emph{spectral-energy} of $H$ on assignment $x$ to be exactly
    \[
    \sum_{e \in E} \max_{u,v \in e}|x_u - x_v|^p.
    \]
    For a parameter $\eps > 0$, a $(1 \pm \eps)$ $p$-\emph{spectral sparsifier} of $H$ is a re-weighted subset of hyperedges $E' \subseteq E$ along with weights $w: E' \rightarrow \R^{\geq 0}$ such that for every $x \in \R^V$,
    \[
    \left ( \sum_{e\in E'} w(e) \cdot \max_{u,v \in e}|x_u - x_v|^p\right ) \in (1 \pm \eps) \cdot \left ( \sum_{e  \in E} \max_{u,v \in e}|x_u - x_v|^p\right ).
    \]
\end{definition}

Note that this notion of energy has implicitly been studied before: the work of \cite{JLLS23} showed that hypergraphs $H = (V, E)$ admit $(1 \pm \eps)$ $p$-spectral sparsifiers with $\widetilde{O} \left (\frac{\max(n, n^{p/2})} {\eps^2} \right )$ many retained edges. However, for $p > 2$, this result loses polynomial factors in $n$ in the sparsifier size. 

In this section, we prove the following strengthening of their result:

\begin{theorem}\label{thm:pSpectralSparsifiersHypergraph}
    Let $H = (V, E)$ be a hypergraph on $n$ vertices and $m$ hyperedges, let $\eps > 0$, and let $p \geq 1$. Then, $G$ admits a $(1 \pm \eps)$ $p$-spectral sparsifier of size $O\left ( n \cdot \mathrm{poly}(\log^p(m), \eps^{-p})\right )$.
\end{theorem}

This theorem shows that, among others, when $p$ is constant (i.e.,  $p = O(1)$), there exists a $p$-spectral sparsifier which retains only $\widetilde{O}\left (n \cdot \mathrm{poly}(\eps^{-1}) \right )$ many hyperedges.

\subsection{Notation and Connection to $\CVNRD$}

To prove \cref{thm:pSpectralSparsifiersHypergraph}, we first require the notion of the ``$p$-spectral code'' of a graph:

\begin{definition}
    Let $H = (V, E)$ be a hypergraph on $n$ vertices and $m$ hyperedges. We let $C_{H, p} \subseteq \R_{\geq0}^m$ denote the $p$-spectral code of $H$ which contains one vector $v^{(x)}$ for each assignment $x \in \R^V$, where for a hyperedge $e \in E$, the corresponding entry in the vector $v^{(x)}$ is $(v^{(x)})_e = \max_{u,v \in e}|x_u - x_v|^p$. 
\end{definition}

\begin{remark}\label{rmk:hypergraphCodeEquivalent}
    Observe that, for a hypergraph $H$, building a $(1 \pm \eps)$ code sparsifier of $C_{H, p}$ is in fact \emph{equivalent} to building a $(1 \pm \eps)$ $p$-spectral sparsifier of the hypergraph $H$. 
\end{remark}

Inspired by this connection to $C_{H, p}$, we now define the $\CVNRD$ of a hypergraph $H$:

\begin{definition}\label{def:pCVNRDofHypergraph}
    For a hypergraph $H = (V, E)$, and parameters $\eps > 0, \chi > 0, \rho > 0, p \geq 1$, the continuous-valued non-redundancy of $H$'s $p$-spectral energy (denoted by $\CVNRD_p(H, \eps, \chi, \rho)$) is defined as the largest integer $\ell$ such that there exist disjoint sets of edges $F_1, \dots F_k \subseteq E$, along with sets of potential assignments $\Psi_1, \dots \Psi_k \in \R^V$ such that:
    \begin{enumerate}
        \item $\sum_{i = 1}^k |F_i| = \ell$. 
        \item For $i \in [k]$, if $|F_i| = 1$, then let $b_2^{(i)} = \max_{u,v \in e}|x_u - x_v|^p >0$ for the unique assignment $x \in \Psi_i$, and hyperedge $e \in F_i$.
        \item For each $F_i: i \in [k]$, there exist values $b^{(i)}_1, b^{(i)}_2$, with $b^{(i)}_1 < b^{(i)}_2( 1- \eps)$, such that for any set of edges $B \subseteq F_i$, there is an assignment $x \in \Psi_i$ such that for $e \in B$, $\max_{u, v \in e}|x_u - x_v|^p \in b^{(i)}_2 \cdot [1, 1 + \chi \cdot \eps ]$, and for $e = (u,v) \in F_i - B$, $\max_{u, v \in e}|x_u - x_v|^p \in b^{(i)}_1 \cdot [1, 1 + \chi \cdot \eps ]$.
    \item For each $F_i: i \in [k]$, and for every $x \in \Psi_i$, $\sum_{j \in [k]: j \neq i} \sum_{e \in F_j} \max_{u, v \in e}|x_u - x_v|^p \leq \rho \cdot b^{(i)}_2 \cdot |F_i|$.
    \item For each $F_i: i \in [k]$, and for every $x \in \Psi_i$, $\max_{j \in [k]: j \neq i} \max_{e \in F_j} \max_{u, v \in e}|x_u - x_v|^p \leq \rho \cdot b^{(i)}_2$.
    \end{enumerate}
\end{definition}

We have the following proposition:

\begin{proposition}\label{prop:translateHypergraphCodes}
     For a hypergraph $H = (V, E)$, and parameters $\eps > 0, \chi > 0, \rho > 0, p \geq 1$, 
     \[
     \CVNRD(C_{H, p}, \eps, \chi, \rho) = \CVNRD_p(H, \eps, \chi, \rho).
     \]
\end{proposition}

\begin{proof}
    Any witness of $\CVNRD(C_{H, p}, \eps, \chi, \rho)$ can be directly translated into a witness of $\CVNRD_p(H, \eps, \chi, \rho)$ and vice versa.
\end{proof}

\subsection{A Connection to Graph-Theoretic Sauer-Shelah}

In order to streamline the proof of \cref{thm:pSpectralSparsifiersHypergraph}, we use a result of \cite{cesa1998graph}. As in their work, we let $N$ denote a finite set of symbols, and we let $G = (N, R)$ be an undirected graph over $N$. The $m$th power of $G$, denoted $G^m$, has vertex set $N^m$ and edge set $R^m$, where an edge $\{(u_1, \dots u_m),  (v_1, \dots v_m)\} \in R^m$ if there exists an index $i \in [m]$ such that $\{u_i, v_i \} \in R$.

For a subset $A = \{i_1, \dots i_{\ell}\} \subseteq [m]$, and a set $F \subseteq N^m$, we say that the projection of $F$ onto $A$ is 
\[
F|_A = \{(v_{i_1}, \dots v_{i_{\ell}}): (v_1, \dots v_m) \in F \}.
\]
A set $Q \subseteq N^m$ is a cube in $G^m$ if $Q = \{u_1, v_1\} \times \{u_2, v_2\} \times \dots \times \{u_m, v_m\}$, where each $\{u_i, v_i\} \in R$. We say that $A, Q$ is a $d$-dimensional projected cube of a set $F \subseteq V^m$ if $A \subseteq [m]$, $|A| = d$, and $Q \subseteq F|_A$ is a $d$-dimensional cube in $G^d$.

Lastly, \cite{cesa1998graph} defines $h(G, m, d)$ to be the \emph{smallest} non-negative integer $h$ such that, every clique $F$ in $G^m$ with $|F| > h$ contains a $(d+1)$-cube. The main theorem of \cite{cesa1998graph} is the following:

\begin{theorem}\label{thm:graphTheoreticSS}
    For any undirected graph $G = (N, R)$ and any $m > d \geq 0$, 
    \[
    h(G, m, d) < 2 \cdot (2m|R|)^{\lceil \log \sum_{i = 0}^d \binom{m}{i} |R|^i\rceil }.
    \]
\end{theorem}

In particular, one consequence of the above theorem is the following:

\begin{corollary}\label{cor:find0ashatter}
    Let $C \subseteq \{0, 1, 2, \dots N-1\}^m$ be such that, for every subset $B \subseteq [m]$, there is a codeword $c \in C$ such that, for $j \in B$, $c_j = 0$, and for $j \in [m]-B$, $c_j \neq 0$.

    Then, there exists a subset $A \subseteq [m]$, $|A| \geq \frac{m-1}{\log(2mN)^2}$ along with a set of symbols $a_i: i \in A$ such that, for every set $B' \subseteq A$, there is a codeword $c \in C$ such that for $j \in B'$, $c_j = 0$, and for $j \in A - B'$, $c_j = a_j$.
\end{corollary}

\begin{proof}
    We define the graph $G = (N, R)$ with the understanding that $N$ is in correspondence with the symbols $\{0, 1, 2, \dots N-1\}$. The edges $R$ are then exactly $\{0, i\}: i \in \{1, 2, \dots N-1\}$.

    Next, we observe that $C$ is a clique in $G^m$: indeed, consider any two vectors $c, c' \in C$, where $c$ witnesses set $B$ and $c'$ witnesses set $B'$. Because $B \neq B'$ we let $j \in B \Delta B'$. For this index $j$, we then see that exactly one of $c_j, c'_j$ is $0$, and the other takes value $\neq 0$. Thus, $(c_j, c_{j'}) \in R$, and so $c, c'$ have an edge in $G^m$.

    Thus, we have a clique of size $2^m$ in $G^m$. By \cref{thm:graphTheoreticSS}, for any $d$ such that 
    \[
    2^m \geq 2 \cdot (2m|R|)^{\lceil \log \sum_{i = 0}^d \binom{m}{i} |R|^i\rceil },
    \]
    we then know that there must be a cube of dimension $d$ inside $C$. In particular, using $R = N$, we can then choose $d$ such that 
    \[
    \left \lceil \log \sum_{i = 0}^d \binom{m}{i} N^i \right \rceil \cdot \log(2m N) \leq m-1,
    \]
    equivalently 
    \[
    \left \lceil \log \sum_{i = 0}^d \binom{m}{i} N^i \right \rceil \leq \frac{m-1}{\log(2mN)}.
    \]
    It suffices then to choose $d$ such that 
    \[
    \log \left ( (2mN)^d \right ) \leq \frac{m-1}{\log(2mN)},
    \]
    and so we can choose $d = \frac{m-1}{\log(2mN)^2}$.

    For this choice of $d$, by \cref{thm:graphTheoreticSS}, we know there exists a subset $A \subseteq [m]$ along with $\{(u_i, v_i)\}: i \in A$, $(u_i, v_i) \in R$ such that 
    \[
    \prod_{i\in A} \{u_i, v_i\} \subseteq C|_A.
    \]

    Because the only edges in $R$ are of the form $(0, a)$ for $a \in \{1, 2, \dots N-1\}$, we then see that this guarantees a subset $A \subseteq [m]$, $|A| \geq \frac{m-1}{\log(2mN)^2}$ along with a set of symbols $a_i: i \in A$ such that, for every set $B' \subseteq A$, there is a codeword $c \in C$ such that for $j \in B'$, $c_j = 0$, and for $j \in A - B'$, $c_j = a_j$.
\end{proof}

Importantly, this now gives us the following ``graph simulation'' statement about hypergraphs:

\begin{claim}\label{clm:graphSimulate}
    Let $H = (V, E)$ be a hypergraph, and let $E = F_1 \cup \dots \cup F_k$ be a decomposition of the hyperedges such that, for parameters $\eps, \chi, \rho, \ell$, there exists a set of assignments $\Psi_i: i \in [k]$ such that:
    \begin{enumerate}
    \item $\sum_{i = 1}^k |F_i| = \ell$. 
    \item For $i \in [k]$, if $|F_i| = 1$, then let $b_2^{(i)} = \max_{u,v \in e}|x_u - x_v|^p >0$ for the unique assignment $x \in \Psi_i$, and hyperedge $e \in F_i$.
    \item For each $F_i: i \in [k]$, there exist values $b^{(i)}_1, b^{(i)}_2$, with $b^{(i)}_1 < b^{(i)}_2( 1- \eps)$, such that for any set of edges $B \subseteq F_i$, there is an assignment $x \in \Psi_i$ such that for $e \in B$, $\max_{u, v \in e}|x_u - x_v|^p \in b^{(i)}_2 \cdot [1, 1 + \chi \cdot \eps ]$, and for $e = (u,v) \in F_i - B$, $\max_{u, v \in e}|x_u - x_v|^p \in b^{(i)}_1 \cdot [1, 1 + \chi \cdot \eps ]$.
    \item For each $F_i: i \in [k]$, and for every $x \in \Psi_i$, $\sum_{j \in [k]: j \neq i} \sum_{e \in F_j} \max_{u, v \in e}|x_u - x_v|^p \leq \rho \cdot b^{(i)}_2 \cdot |F_i|$.
    \item For each $F_i: i \in [k]$, and for every $x \in \Psi_i$, $\max_{j \in [k]: j \neq i} \max_{e \in F_j} \max_{u, v \in e}|x_u - x_v|^p \leq \rho \cdot b^{(i)}_2$.
    \end{enumerate}
    Then, there exists a graph $G_H$ (potentially with parallel edges) on the same vertex set, with edges $E'_1, \dots E'_k$ and assignments $\Psi'_i \subseteq \Psi_i: i \in [k]$ with each $|E'_i| = \Omega\left (\frac{|F_i|}{\log^2(mn)} \right )$ such that:
    \begin{enumerate}
    \item For $i \in [k]$, if $|E'_i| = 1$, then let $b_2^{(i)} = |x_u - x_v|^p >0$ for the unique assignment $x \in \Psi'_i$, and edge $e \in E'_i$.
    \item For each $E'_i: i \in [k]$, there exist values $b^{(i)}_1, b^{(i)}_2$, with $b^{(i)}_1 < b^{(i)}_2( 1- \eps)$, such that for any set of edges $B \subseteq E'_i$, there is an assignment $x \in \Psi'_i$ such that for $e \in B$, $|x_u - x_v|^p \in b^{(i)}_2 \cdot [1, 1 + \chi \cdot \eps ]$, and for $e = (u,v) \in E'_i - B$, $|x_u - x_v|^p \leq b^{(i)}_1 \cdot (1 + \chi \cdot \eps )$.
    \item For each $E'_i: i \in [k]$, and for every $x \in \Psi'_i$, $\max_{j \in [k]: j \neq i} \max_{e \in E'_j} |x_u - x_v|^p \leq \rho \cdot b^{(i)}_2$.
    \end{enumerate}
\end{claim}

\begin{proof}
    Let us consider a single block $F_i$ in the hypergraph $H$. Now, we consider the set of assignments $\Psi_i$ and the corresponding set of vectors in $C_{H, p} \subseteq \R_{\geq0}^m$ that are generated by $\Psi_i$. In particular, in the coordinates corresponding to $F_i$, we know that there exist values $b^{(i)}_1, b^{(i)}_2$, with $b^{(i)}_1 < b^{(i)}_2( 1- \eps)$, such that for any set of edges $B \subseteq F_i$, there is an assignment $x \in \Psi_i$ such that for $e \in B$, $\max_{u, v \in e}|x_u - x_v|^p \in b^{(i)}_2 \cdot [1, 1 + \chi \cdot \eps ]$, and for $e = (u,v) \in F_i - B$, $\max_{u, v \in e}|x_u - x_v|^p \in b^{(i)}_1 \cdot [1, 1 + \chi \cdot \eps ]$. For now, we restrict our attention to $\{c^{(x)}|_{F_i}: c \in \Psi_i \}$.

    We now do a transformation on these vectors. To properly define this transformation, we observe that whenever $\max_{u, v \in e}|x_u - x_v|^p \in b^{(i)}_2 \cdot [1, 1 + \chi \cdot \eps ]$, it \emph{must} be the case that there is some pair of vertices $u,v \in e$ such that $|x_u - x_v|^p \in b^{(i)}_2 \cdot [1, 1 + \chi \cdot \eps ]$. For a hyperedge $e$, and an assignment $x$ such that $\max_{u, v \in e}|x_u - x_v|^p \in b^{(i)}_2 \cdot [1, 1 + \chi \cdot \eps ]$, we call such a pair $(u,v)$ the \emph{witnessing pair} (if there are multiple witnessing pairs, we let the witnessing pair be the one which is lexicographically first). 
    Now, for any such vector $c^{(x)}|_{F_i} \in \{c^{(x)}|_{F_i}: c \in \Psi_i \}$, we construct $\hat{c}^{(x)}|_{F_i}$ such that:
    \[
    \hat{c}^{(x)}|_{F_i} =  \begin{cases}
        0 \text{ if }\max_{u, v \in e}|x_u - x_v|^p \in b^{(i)}_1 \cdot [1, 1 + \chi \cdot \eps ] \\
        (u,v) \text{ if }\max_{u, v \in e}|x_u - x_v|^p \in b^{(i)}_2 \cdot [1, 1 + \chi \cdot \eps ] , \text{ and }(u,v) \text{ is a witnessing pair}.
    \end{cases}
    \]

    Because there are $\leq \binom{n}{2}$ many choices of witnessing pairs, we can WLOG view this as a code $W : = \{c^{(x)}|_{F_i}: c \in \Psi_i \} \subseteq \{0, 1, 2, \dots \binom{n}{2} \}^{F_i}$. Now, we can observe that $W$ exactly satisfies the conditions of \cref{cor:find0ashatter} with $N = \binom{n}{2}+1$, and $m = |F_i|$. In particular, we can then find a set $F'_i \subseteq F_i$ along with a set of witnessing pairs $(u_j, v_j): j \in F'_i$ such that, for $B' \subseteq F'_i$, there is a codeword $c \in W$ such that for $j \in B'$, $c_j = 0$, and for $j \in F'_i - B'$, $c_j = (u_j, v_j)$, with $|F'_i| \geq \frac{|F'_i|}{\log(2|F_i| n^2)^2} = \Omega \left ( \frac{|F_i|}{\log(|F_i|n)^2}\right )$, when $|F_i|\geq 2$.

    Equivalently, if we undo our $\hat{\cdot}$ mapping, this same set $F'_i$ corresponds with a set of hyperedges along with a set of witnessing pairs $(u_j, v_j): j \in F'_i$ such that, for $B' \subseteq F'_i$, there is an assignment $x \in \Psi_i$ such that for $e \in B'$, $\max_{u, v \in e}|x_u - x_v|^p \in b^{(i)}_1 \cdot [1, 1 + \chi \cdot \eps ]$, and for $j \in F'_i - B'$, $|x_{u_j} - x_{v_j}|^p \in b^{(i)}_2 \cdot [1, 1 + \chi \cdot \eps ]$. Note that this condition that $\max_{u, v \in e}|x_u - x_v|^p \in b^{(i)}_1 \cdot [1, 1 + \chi \cdot \eps ]$ implies that \begin{align}\label{eq:maxImpliesAllEdges}
    |x_{u_j} - x_{v_j}|^p \leq \max_{u, v \in e}|x_u - x_v|^p \leq b^{(i)}_1 \cdot( 1 + \chi \cdot \eps )
    \end{align}

    To conclude the theorem, we let $E'_1, \dots E'_k$ denote the sets of witnessing pairs for $e \in F'_i$. Note that if $|F_i| = 1$, then we simply define its witnessing pair to be the witnessing pair for the single assignment $x \in \Psi_i$. Under this definition, we see that the conditions of the claim are satisfied:
    \begin{enumerate}
        \item When $|F_i| = 1$, $|E'| = 1$, and when $|F_i| \geq 2$, $|E'_i| = \Omega \left (  \frac{|F_i|}{\log(|F_i| n)^2}\right)$.
        \item When $|E'_i| = 1$, we simply select a single assignment from $\Psi_i$ such that $|x_u - x_v|^p \geq b_2^{(i)}$ and define this to be $\Psi'_i$.
        \item When $|E'_i| \geq 2$, for $B' \subseteq E'_i$, there exists $x \in \Psi_i$ such that for $e \in B'$, $|x_u - x_v|^p \leq b^{(i)}_1 \cdot (1 + \chi \cdot \eps)$ (see \cref{eq:maxImpliesAllEdges}), and for $j \in F'_i - B'$, $|x_{u_j} - x_{v_j}|^p \in b^{(i)}_2 \cdot [1, 1 + \chi \cdot \eps ]$, by the definition of being a witnessing pair.
        \item The final condition of our claim holds because of \cref{eq:maxImpliesAllEdges}: 
        \[
        \max_{j \in [k]: j \neq i} \max_{e \in E'_j} |x_u - x_v|^p \leq \max_{j \in [k]: j \neq i} \max_{e \in F_j} \max_{u, v \in e} |x_u - x_v|^p\rho \cdot b^{(i)}_2,
        \]
        where we are using that each witnessing pair in $E'_j$ is contained in some hyperedge in $F_j$.
    \end{enumerate}

\end{proof}

\subsection{Concluding the Theorem}

With the preceding section, we are now ready to conclude \cref{thm:pSpectralSparsifiersHypergraph}.

\begin{proof}[Proof of \cref{thm:pSpectralSparsifiersHypergraph}.]
Let $H = (V, E)$ be a hypergraph, and let $F_1 \cup \dots F_k$ be a maximum size witness of $\CVNRD$ with parameters $\eps, \chi, \rho$. In particular, such a witness satisfies:

\begin{enumerate}
    \item $\sum_{i = 1}^k |F_i| = \ell$. 
    \item For $i \in [k]$, if $|F_i| = 1$, then let $b_2^{(i)} = \max_{u,v \in e}|x_u - x_v|^p >0$ for the unique assignment $x \in \Psi_i$, and hyperedge $e \in F_i$.
    \item For each $F_i: i \in [k]$, there exist values $b^{(i)}_1, b^{(i)}_2$, with $b^{(i)}_1 < b^{(i)}_2( 1- \eps)$, such that for any set of edges $B \subseteq F_i$, there is an assignment $x \in \Psi_i$ such that for $e \in B$, $\max_{u, v \in e}|x_u - x_v|^p \in b^{(i)}_2 \cdot [1, 1 + \chi \cdot \eps ]$, and for $e = (u,v) \in F_i - B$, $\max_{u, v \in e}|x_u - x_v|^p \in b^{(i)}_1 \cdot [1, 1 + \chi \cdot \eps ]$.
    \item For each $F_i: i \in [k]$, and for every $x \in \Psi_i$, $\sum_{j \in [k]: j \neq i} \sum_{e \in F_j} \max_{u, v \in e}|x_u - x_v|^p \leq \rho \cdot b^{(i)}_2 \cdot |F_i|$.
    \item For each $F_i: i \in [k]$, and for every $x \in \Psi_i$, $\max_{j \in [k]: j \neq i} \max_{e \in F_j} \max_{u, v \in e}|x_u - x_v|^p \leq \rho \cdot b^{(i)}_2$.
    \end{enumerate}

Thus, we can invoke \cref{clm:graphSimulate} which shows that there exists a graph $G_H$ on the same vertex set, with edges $E'_1, \dots E'_k$ and assignments $\Psi'_i \subseteq \Psi_i: i \in [k]$ with each $|E'_i| = \Omega\left (\frac{|F_i|}{\log^2(mn)} \right )$ such that:
    \begin{enumerate}
    \item For $i \in [k]$, if $|E'_i| = 1$, then let $b_2^{(i)} = |x_u - x_v|^p >0$ for the unique assignment $x \in \Psi'_i$, and edge $e \in E'_i$.
    \item For each $E'_i: i \in [k]$, there exist values $b^{(i)}_1, b^{(i)}_2$, with $b^{(i)}_1 < b^{(i)}_2( 1- \eps)$, such that for any set of edges $B \subseteq E'_i$, there is an assignment $x \in \Psi'_i$ such that for $e \in B$, $|x_u - x_v|^p \in b^{(i)}_2 \cdot [1, 1 + \chi \cdot \eps ]$, and for $e = (u,v) \in E'_i - B$, $|x_u - x_v|^p \leq b^{(i)}_1 \cdot (1 + \chi \cdot \eps )$.
    \item For each $E'_i: i \in [k]$, and for every $x \in \Psi'_i$, $\max_{j \in [k]: j \neq i} \max_{e \in E'_j} |x_u - x_v|^p \leq \rho \cdot b^{(i)}_2$.
    \end{enumerate}

In particular, choosing $\rho = \frac{\eps^p}{(4p \log(n))^p \cdot2}$ and $\chi = \frac{1}{100 \log^2(m)}$, we can now exactly replicate the proof of \cref{thm:boundCVNRDGraph}.\footnote{Note that the conditions we have on the graph are now weaker; notably the lower energy edges only satisfy $|x_u - x_v|^p \leq b^{(i)}_1 \cdot (1 + \chi \cdot \eps )$ as opposed to $|x_u - x_v|^p \in b^{(i)}_1 \cdot (1, 1 + \chi \cdot \eps )$, and we have no bound on the sum of off-diagonal energies. Neither of these stronger conditions are used in the proof of \cref{thm:boundCVNRDGraph}.} This implies that it must be the case that 
\[
\sum_{i \in [k]} |E'_i| =O(n),
\]
so equivalently, 
\[
\sum_{i \in [k]} |F_i| =O(n \log^2(mn)).
\]
Thus, \begin{align}\label{eq:boundCVNRDHypergraph}
    \CVNRD_p \left (H, \eps, \frac{1}{100 \log^2(m)}, \frac{\eps^p}{(4p \log(n))^p \cdot 2} \right ) = O( n \log^2(mn)).
    \end{align}

    By \cref{prop:translateHypergraphCodes}, we know that for $\eps > 0, \chi > 0, \rho > 0, p \geq 1$, 
     \[
     \CVNRD(C_{H, p}, \eps, \chi, \rho) \leq \CVNRD_p(H, \eps, \chi, \rho).
     \]
     We also know by \cref{rmk:hypergraphCodeEquivalent} that it suffices to build a $(1 \pm \eps)$-sparsifier of $C_{H, p}$ in order to build a $(1 \pm \eps)$ $p$-spectral sparsifier of $G$. 

     Next, we invoke \cref{thm:CVNRDSparsifierParam} which states that 
    \[
    \mathrm{SPR}(C_{H, p}, \eps) = O\left ( \frac{\CVNRD(C_{H,p}, \eps', \chi, \rho)}{\chi^6 \rho}\cdot \mathrm{poly}(\log(m), \eps^{-1})\right ),
    \]
    where $\eps ' = \frac{\eps^4}{10^{13} \log^8(m)}$.

    Now, by \cref{prop:translateHypergraphCodes} and \cref{eq:boundCVNRDHypergraph},  we know that \[
    \CVNRD \left (C_{H,p}, \eps', \frac{1}{100 \log^2(m)}, \frac{(\eps')^p}{(4p \log(n))^p \cdot 2} \right ) = O( n \log^2(mn)).
    \]
    Thus, 
    \[
    \mathrm{SPR}(C_{H, p}, \eps) = O\left ( \frac{n}{\left ( \frac{1}{100 \log^2(m)} \right )^6 \left ( \frac{(\eps')^p}{(4p \log(m))^p \cdot 2}\right )}\cdot \mathrm{poly}(\log(m), \eps^{-1})\right ) 
    \]
    \[
    = O\left ( n \cdot \mathrm{poly}(\log^p(m), \eps^{-p})\right ).
    \]
    This concludes the theorem. 
\end{proof}

\section{Applications to Bounded Submodular Sparsification}\label{sec:submodular}

In this section, we prove \cref{thm:intro-submodular}. Recall that in this setting, we are given a universe of $n$ elements, and $m$ submodular functions $f_1, \dots f_m: 2^{[n]}\rightarrow \{0, 1, 2, \dots k \}$, for some $k = O(1)$. Our goal is to compute a $(1 \pm \eps)$ sparsifier of this collection of submodular functions. We show: 
\begin{theorem}
    Let $k = O(1)$ be an integer, and let $f_1, \dots f_m: 2^{[n]} \rightarrow \{0, 1, 2, 3, \dots k\}$ be submodular functions. Then, there exists $T \subseteq [m]$ along with weights $w: T \rightarrow \R_{\geq 0}$ such that for every $S \subseteq [n]$,
\[
\sum_{i \in T} w(i) \cdot f_i(S) \in (1 \pm \eps) \cdot \sum_{i \in [m]} f_i(S),
\]
and $|T| = \widetilde{O}(n^2 / \eps^4)$.
\end{theorem}

To start, observe that given this collection of submodular functions, we can define an \emph{equivalent code} $C \subseteq \{0,1, 2 \dots k\}^m$ such that $C$ contains a codeword for each $S \subseteq [n]$, where 
\[
c^{(S)}_{i} = f_i(S).
\]
Under this mapping, any $(1 \pm \eps)$ sparsifier of $C$ is inherently a $(1 \pm \eps)$ sparsifier of $f_1,\dots f_m$ (and vice versa). By \cref{thm:DiscreteDomainRVNRD}, we know that the sparsifiability of this code $C$ is bounded by $\widetilde{O}(\BADNRD(C) / \eps^4)$, where $\BADNRD(C)$ is defined as the maximum integer $\ell$ such that there exists $A_1, \dots A_p \subseteq [m]$ and $C_1, \dots C_p \subseteq C$, with $\sum_{i = 1}^p |A_i| = \ell$, and:
    \begin{enumerate}
    \item For every $i\in[p]$, if $|A_i| = 1$, then $(C_i)|_{A_i} \in \{1, 2, \dots k\}$.
        \item For every $i \in [p]$, if $|A_i| \geq 2$, then $(C_i)|_{A_i} = \{b_0, b_1\}^{|A_i|}$ for some $\{b_0, b_1\} \subseteq \{0,1, \dots k\}$.
        \item For every $j \neq i$, $(C_i)|_{A_j} = 0^{A_j}$.
    \end{enumerate}

Thus, the remainder of this section focuses on bounded the maximum size ``$\BADNRD$ witness'' that can be constructed within any collection of submodular functions. Formally, we show:

\begin{theorem}\label{thm:boundBADNRDSubmod}
Let
\[
f_1,\dots,f_m:2^{[n]}\to \{0, 1, 2, \dots k\}
\]
be nonnegative submodular functions which constitute a $\BADNRD$ witness. I.e., that $[m]$ admits a partition into $A_1, \dots A_p$, along with collections of sets $\mathcal{S}_1, \dots \mathcal{S}_p \subseteq 2^{[n]}$ such that:
\begin{enumerate}
    \item For every $i\in[p]$, if $|A_i| = 1$, then for $j \in A_i$, and for every $S \in \calS_i$, $f_j(S) \in \{1, 2, \dots k\}$.
        \item For every $i \in [p]$, if $|A_i| \geq 2$, then there exists $\{a_i, b_i\} \subseteq \{0,1, \dots k\}$ such that for any $v \in \{a_i, b_i\}^{A_i}$, there exists $S \in \calS_i$ such that for every $j \in A_i$, $f_j(S) = v_j$.
        \item For every $i' \neq i$, for every $S \in \calS_i$ and every $j \in A_{i'}$, $f_j(S) = 0$.  
\end{enumerate}
Then,
\[
m\le n^2+n.
\]
\end{theorem}

\begin{remark}
    As an immediate consequence, \cref{thm:boundBADNRDSubmod} implies that $\BADNRD(C) \leq n^2 + n$ for any code $C$ which is generated as the evaluations of submodular functions $f_1, \dots f_m$.
\end{remark}

\subsection{Notation}

We identify points of $\{0,1\}^n$ with subsets of $[n]$. Recall that a set function
\[
f:2^{[n]}\to \mathbb R
\]
is submodular if, for every $S,T\subseteq[n]$,
\[
f(S)+f(T)\ge f(S\cap T)+f(S\cup T).
\]

\subsection{Zero Sets of Nonnegative Submodular Functions}
To start, we recap basic properties of the behvaior of $0$ sets of submodular functions. This will be of particular importance to us because the definition of $\BADNRD$ enforces that submodular functions always evaluate to $0$ on the ``off-diagonal'' terms of the witness. 

A family $\calL\subseteq 2^{[n]}$ is called a sublattice if it is closed under union and intersection.

\begin{lemma}
\label{lem:zero-sublattice}
Let
\[
f:2^{[n]}\to \mathbb R_{\ge 0}
\]
be submodular. Then the zero set
\[
Z(f):=\{S\subseteq[n]:f(S)=0\}
\]
is a sublattice of $2^{[n]}$.
\end{lemma}

\begin{proof}
Let $S,T\in Z(f)$. Then
\[
f(S)=f(T)=0.
\]
By submodularity,
\[
0=f(S)+f(T)\ge f(S\cap T)+f(S\cup T).
\]
Both terms on the right are nonnegative, so both must be zero. Hence
\[
S\cap T\in Z(f)
\qquad\text{and}\qquad
S\cup T\in Z(f).
\]
Thus $Z(f)$ is closed under intersections and unions.
\end{proof}

We will repeatedly use the following immediate consequence.

\begin{lemma}
\label{lem:zero-operation}
Let
\[
f:2^{[n]}\to \mathbb R_{\ge 0}
\]
be submodular, and suppose $f(A)=0$. Then for every $S\subseteq[n]$,
\[
f(S\cap A)\le f(S)
\qquad\text{and}\qquad
f(S\cup A)\le f(S).
\]
\end{lemma}

\begin{proof}
By submodularity,
\[
f(S)+f(A)\ge f(S\cap A)+f(S\cup A).
\]
Since $f(A)=0$, this gives
\[
f(S)\ge f(S\cap A)+f(S\cup A).
\]
Both terms on the right are nonnegative, so each is at most $f(S)$.
\end{proof}

\subsection{Canonical Data of a Sublattice}

Now, in order to prove our bound on the size of any $\BADNRD$ witness, we need a proxy to quantify the ``degrees of freedom'' of submodular functions. We quantify this with the following definitions. 
Let $\calL\subseteq 2^{[n]}$ be a nonempty sublattice.
Define its \emph{common core} by
\[
C_{\calL}:=\bigcap_{A\in\calL} A.
\]
Thus $C_{\calL}$ is the set of elements that appear in every member of $\calL$.
Define its \emph{support} by
\[
U_{\calL}:=\bigcup_{A\in\calL} A.
\]
Thus $U_{\calL}$ is the set of elements that appear in at least one member of $\calL$.
Finally, for each $u\in[n]$, define
\[
M_{\calL}(u)
:=
\bigcap_{\substack{A\in\calL\\ u\in A}} A.
\]
If no member of $\calL$ contains $u$, then we use the convention
\[
M_{\calL}(u)=[n].
\]
When $u\in U_{\calL}$, the set $M_{\calL}(u)$ belongs to $\calL$ and is the smallest member of $\calL$ containing $u$.

We now define a finite list of possible defects of a set $S$ relative to $\calL$; these will essentially be our ``degress of freedom'' going forward. 

Let
\[
\Gamma_n
:=
\{\alpha_v:v\in[n]\}
\cup
\{\beta_u:u\in[n]\}
\cup
\{\gamma_{u,v}:u,v\in[n],\ u\neq v\}.
\]
Thus
\[
|\Gamma_n|=n+n+n(n-1)=n^2+n.
\]

For $S\subseteq[n]$, define
\[
\mathrm{Def}_{\calL}(S)\subseteq \Gamma_n
\]
as follows.

First, $
\alpha_v\in \mathrm{Def}_{\calL}(S)$
if and only if $
v\in C_{\calL}\setminus S.$
This records that $S$ is missing an element that every member of $\calL$ contains.

Second, $
\beta_u\in \mathrm{Def}_{\calL}(S)$
if and only if $
u\in S\setminus U_{\calL}.$
This records that $S$ contains an element that no member of $\calL$ contains.

Third, $
\gamma_{u,v}\in \mathrm{Def}_{\calL}(S)$
if and only if \[
u\in S,
\qquad
v\in M_{\calL}(u)\setminus S,
\qquad
u\neq v. \]
This records that $S$ contains $u$ but fails to contain some element that is always present in the smallest member of $\calL$ containing $u$. The next lemma shows that these defects exactly describe membership in $\calL$.

\begin{lemma}
\label{lem:defect-zero}
Let $\calL\subseteq 2^{[n]}$ be a nonempty sublattice. For $S\subseteq[n]$,
\[
S\in\calL
\]
if and only if
\[
\mathrm{Def}_{\calL}(S)=\emptyset.
\]
\end{lemma}

\begin{proof}
First suppose $S\in\calL$. Since $C_{\calL}$ is the intersection of all members of $\calL$, we have $
C_{\calL}\subseteq S.$
Thus no defect of type $\alpha_v$ occurs. Likewise, because $U_{\calL}$ is the union of all members of $\calL$, we have $
S\subseteq U_{\calL}.$
Thus no defect of type $\beta_u$ occurs.

Finally, if $u\in S$, then $S$ is one of the members of $\calL$ containing $u$. Therefore $
M_{\calL}(u)\subseteq S.$
Thus no defect of type $\gamma_{u,v}$ occurs. Hence $
\Def_{\calL}(S)=\emptyset.$

Conversely, suppose $
\Def_{\calL}(S)=\emptyset.$
Then $
C_{\calL}\subseteq S$ 
and $
S\subseteq U_{\calL}.$ If $S=\emptyset$, then $C_{\calL}=\emptyset$. Since $C_{\calL}$ is an intersection of members of $\calL$ and $\calL$ is intersection-closed, we have
\[
\emptyset=C_{\calL}\in\calL.
\]

Now suppose $S\neq\emptyset$. Since $S\subseteq U_{\calL}$, for every $u\in S$ there is at least one member of $\calL$ containing $u$, so
\[
M_{\calL}(u)\in\calL.
\]
Since no defect of type $\gamma_{u,v}$ occurs, we have
\[
M_{\calL}(u)\subseteq S
\qquad
\text{for every }u\in S.
\]
Define
\[
W:=\bigcup_{u\in S} M_{\calL}(u).
\]
Because $\calL$ is union-closed, $W\in\calL$. Each $M_{\calL}(u)$ is contained in $S$, so $
W\subseteq S.$
On the other hand, $u\in M_{\calL}(u)$ for every $u\in S$, so $
S\subseteq W.$
Therefore $W=S$, and hence $S\in\calL$.
\end{proof}

\subsection{A Domination Lemma}

The next lemma shows that these defects exactly capture the ability to translate a set $S$ into a set $T$ purely via intersections and unions with the sub-lattice $\calL$:

\begin{lemma}
\label{lem:representation}
Let $\calL\subseteq 2^{[n]}$ be a nonempty sublattice, and let $S,T\subseteq[n]$. Suppose
\[
\Def_{\calL}(T)\subseteq \Def_{\calL}(S).
\]
Then there exist sets
\[
A\in \calL\cup\{\emptyset\}
\qquad\text{and}\qquad
B\in \calL\cup\{[n]\}
\]
such that
\[
T=A\cup (S\cap B).
\]
Moreover, if $A=\emptyset$, then the operation of union with $A$ is unnecessary; if $B=[n]$, then the operation of intersection with $B$ is unnecessary.
\end{lemma}

\begin{proof}
Let
\[
X:=T\setminus S
\qquad\text{and}\qquad
Y:=T\cap S.
\]

We first build the part of $T$ outside $S$.

Suppose $u\in X$. Then $u\in T$ and $u\notin S$. If $u\notin U_{\calL}$, then $
\beta_u\in\Def_{\calL}(T). $
But since $u\notin S$, we have $
\beta_u\notin\Def_{\calL}(S) $
contradicting $
\Def_{\calL}(T)\subseteq \Def_{\calL}(S).$
Therefore every $u\in X$ lies in $U_{\calL}$, and so $M_{\calL}(u)\in\calL$.

We claim that
\[
M_{\calL}(u)\subseteq T
\qquad
\text{for every }u\in X.
\]
Indeed, if some $v\in M_{\calL}(u)$ were not in $T$, then $u\in T$ and $v\notin T$, so $
\gamma_{u,v}\in \Def_{\calL}(T).$
Here $v\neq u$ because $u\in T$. But $u\notin S$, so
$
\gamma_{u,v}\notin \Def_{\calL}(S),
$
again a contradiction.

If $X\neq\emptyset$, define $
A:=\bigcup_{u\in X} M_{\calL}(u).$
Then $A\in\calL$, because $\calL$ is union-closed. Also $
X\subseteq A\subseteq T.$
If $X=\emptyset$, set $
A:=\emptyset.$

We now build the part of $T$ inside $S$.
First suppose that $Y$ contains some element $u\notin U_{\calL}$. Since $u\notin U_{\calL}$, no member of $\calL$ contains $u$, and therefore $
M_{\calL}(u)=[n].$
We claim that in this case $
S\subseteq T.$
Indeed, if there existed $v\in S\setminus T$, then $v\in M_{\calL}(u)\setminus T$, and since $u\in T$, we would have $
\gamma_{u,v}\in \Def_{\calL}(T).$
But $u\in S$ and $v\in S$, so $
\gamma_{u,v}\notin \Def_{\calL}(S),$ which is a contradiction. Hence $S\subseteq T$.

In this case set $
B:=[n].$
Then $
S\cap B=S=T\cap S=Y.$ Now suppose that every element of $Y$ lies in $U_{\calL}$. If $Y\neq\emptyset$, define $
B:=\bigcup_{u\in Y} M_{\calL}(u).$
Then $B\in\calL$. Clearly $
Y\subseteq S\cap B.$
We claim that the reverse containment also holds. Let $
v\in S\cap B.$
Then for some $u\in Y$, $
v\in M_{\calL}(u).$
If $v\notin T$, then $u\in T$ and $v\notin T$, so
$
\gamma_{u,v}\in \Def_{\calL}(T).$
But $u\in S$ and $v\in S$, so $\gamma_{u,v}\notin \Def_{\calL}(S)$, which is a contradiction. Therefore $v\in T$, and hence $
S\cap B\subseteq T\cap S=Y.$
Thus $ S\cap B=Y$.

Finally, suppose $Y=\emptyset$. Define $B:=C_{\calL}.$
Then $B\in\calL$. We claim that $
S\cap B=\emptyset.$
If some $v\in S\cap C_{\calL}$ existed, then $v\notin T$ because $Y=T\cap S=\emptyset$. Hence $\alpha_v\in \Def_{\calL}(T).$
But $v\in S$, so $
\alpha_v\notin \Def_{\calL}(S)$, which is a
contradiction. Therefore $
S\cap B=\emptyset=Y.$

In all cases, we have constructed $A$ and $B$ such that
\[
X\subseteq A\subseteq T
\]
and
\[
S\cap B=Y.
\]
Therefore
\[
A\cup(S\cap B)
=
A\cup Y.
\]
Since $X\subseteq A$ and $A\subseteq T=X\cup Y$, this union is exactly $T$. Hence $
T=A\cup(S\cap B).$
\end{proof}

We now combine this representation with submodularity.

\begin{lemma}[Submodular domination from defect containment]
\label{lem:domination}
Let
\[
f:2^{[n]}\to \mathbb R_{\ge 0}
\]
be submodular. Let $\calL\subseteq 2^{[n]}$ be a nonempty sublattice such that $f(A) = 0$ for every $A \in \calL$.
Then, for any $S, T \subseteq [n]$ such that 
\[
\Def_{\calL}(T)\subseteq \Def_{\calL}(S),
\]
it must be that $f(T)\le f(S).$
\end{lemma}

\begin{proof}
By Lemma~\ref{lem:representation}, there are
\[
A\in\calL\cup\{\emptyset\}
\qquad\text{and}\qquad
B\in\calL\cup\{[n]\}
\]
such that $T=A\cup(S\cap B).$

If $B\in\calL$, then $f(B)=0$, and Lemma~\ref{lem:zero-operation} gives $
f(S\cap B)\le f(S).$
If $B=[n]$, then $
S\cap B=S,$
so again $
f(S\cap B)\le f(S).$

Now consider $A$. If $A\in\calL$, then $f(A)=0$, and Lemma~\ref{lem:zero-operation} gives $
f(A\cup(S\cap B))\le f(S\cap B).$
If $A=\emptyset$, then no union operation is needed, and $
A\cup(S\cap B)=S\cap B.$
In either case, $
f(T)\le f(S).$
\end{proof}

The above immediately implies that if two sets share the same defects, then they must in fact be equal under a function $f$.

\begin{corollary}[Constancy on equal defect sets]
\label{cor:constancy}
Let $
f:2^{[n]}\to \mathbb R_{\ge 0} $
be submodular. Let $\calL\subseteq 2^{[n]}$ be a nonempty sublattice such that $f(A) = 0$ for every $A \in \calL$.
If $
\Def_{\calL}(S)=\Def_{\calL}(T),$
then $
f(S)=f(T).$
\end{corollary}

\subsection{Concluding the Upper Bound}

Finally, we now prove Theorem~\ref{thm:boundBADNRDSubmod}.

\begin{proof}[Proof of Theorem~\ref{thm:boundBADNRDSubmod}]
Discard empty blocks. If only one block remains, say $A_1=[m]$, then the shattering condition requires $2^m$ distinct value patterns to be realized by points of $2^{[n]}$. Since the domain has size $2^n$, we get
\[
2^m\le 2^n,
\]
and hence
\[
m\le n\le n^2+n.
\]

Thus, we  assume there are at least two nonempty blocks. For every block $A_i$ and every subset $B\subseteq A_i$, choose one witnessing set $
S_{i,B}\subseteq[n]$
such that
\[
f_j(S_{i,B})=a_i \quad (j\in B),
\]
\[
f_j(S_{i,B})=b_i \quad (j\in A_i\setminus B),
\]
and
\[
f_j(S_{i,B})=0 \quad (j\notin A_i).
\]
Let
\[
\calW_i:=\{S_{i,B}:B\subseteq A_i\}.
\]

For each block $A_i$, define
\[
\calL_{-i}
:=
\left\langle
\bigcup_{h\neq i}\calW_h
\right\rangle,
\]
the smallest sublattice containing all witnesses from blocks other than $i$. Note that since there are at least two nonempty blocks, each $\calL_{-i}$ is nonempty. Now, fix $i$ and $j\in A_i$. By the shattering condition,
\[
f_j(S)=0
\qquad
\text{for every }S\in\calW_h,\ h\neq i.
\]
By Lemma~\ref{lem:zero-sublattice}, the zero set of $f_j$ is a sublattice. Therefore it contains the sublattice generated by all these outside witnesses. Hence
\[
f_j(S)=0
\qquad
\text{for every }S\in\calL_{-i}.
\]

For each block $i$, define its set of block defects by
\[
D_i
:=
\bigcup_{B\subseteq A_i}
\Def_{\calL_{-i}}(S_{i,B})
\subseteq \Gamma_n.
\]

We first show that $|A_i| \leq |D_i|$ for every $i$. 
Indeed, as $S$ ranges over $\calW_i$, every defect set $
\Def_{\calL_{-i}}(S)$
is a subset of $D_i$. Therefore there are at most $
2^{|D_i|} $
possible defect sets among the witnesses in $\calW_i$.

For every $j\in A_i$, the function $f_j$ vanishes on $\calL_{-i}$. Thus Corollary~\ref{cor:constancy} says that if two witnesses $
S,S'\in\calW_i$
have the same defect set relative to $\calL_{-i}$, then $
f_j(S)=f_j(S')$
for every $j\in A_i$. Hence the full value vector
\[
\bigl(f_j(S)\bigr)_{j\in A_i}
\]
is determined by the defect set
\[
\Def_{\calL_{-i}}(S).
\]

On the other hand, the block-shattering condition says that all vectors in $
\{a_i,b_i\}^{A_i}$
occur. Since $a_i\neq b_i$, there are exactly $
2^{|A_i|}$
such vectors. Therefore
\[
2^{|A_i|}\le 2^{|D_i|}.
\]
Thus
\[
|A_i|\le |D_i|.
\]

Finally, it remains to prove that the sets $D_i$ are pairwise disjoint. So, suppose $i\neq h$. First consider a label of the form $\alpha_v$. If $
\alpha_v\in D_i,$
then there is some witness $S\in\calW_i$ such that $
v\in C_{\calL_{-i}}$ and $v\notin S$.
The condition $v\in C_{\calL_{-i}}$ means that every set in $\calL_{-i}$ contains $v$. In particular, every witness in $\calW_h$ contains $v$, because $\calW_h\subseteq\calL_{-i}$. Therefore no witness from $\calW_h$ can certify $\alpha_v\in D_h$. Hence
\[
\alpha_v\notin D_h.
\]

Next consider a label of the form $\beta_u$. If $
\beta_u\in D_i,$
then some witness $S\in\calW_i$ contains $u$, while $
u\notin U_{\calL_{-i}}.$
The latter condition means that no set in $\calL_{-i}$ contains $u$. Since $\calW_h\subseteq\calL_{-i}$, no witness in $\calW_h$ contains $u$. Therefore $
\beta_u\notin D_h.$

Finally consider a label of the form $\gamma_{u,v}$. If $
\gamma_{u,v}\in D_i,$
then there is some witness $S\in\calW_i$ such that
\[
u\in S,
\qquad
v\notin S,
\qquad
v\in M_{\calL_{-i}}(u).
\]
The condition $v\in M_{\calL_{-i}}(u)$ means that every member of $\calL_{-i}$ containing $u$ also contains $v$. Since $\calW_h\subseteq\calL_{-i}$, every witness in $\calW_h$ that contains $u$ must also contain $v$. Therefore no witness in $\calW_h$ can certify $\gamma_{u,v}\in D_h$. Hence
\[
\gamma_{u,v}\notin D_h.
\]

Thus
\[
D_i\cap D_h=\emptyset
\qquad
\text{for all }i\neq h.
\]

Since the $D_i$ are pairwise disjoint subsets of $\Gamma_n$, we have
\[
\sum_i |D_i|\le |\Gamma_n|=n^2+n.
\]
Combining this with $|A_i|\le |D_i|$ for every $i$, we obtain
\[
m
=
\sum_i |A_i|
\le
\sum_i |D_i|
\le
n^2+n.
\]
This proves the theorem.
\end{proof}

\section*{Acknowledgments}

ChatGPT was used in generating figures used in the technical overview and preparing the text in \cref{sec:submodular}. 

The authors would like to thank Arpon Basu, Joshua Brakensiek, Venkatesan Guruswami, and Pravesh Kothari for useful and inspiring discussions.

\bibliographystyle{alpha}
\bibliography{ref}

\appendix

\section{Sparsifying Weighted Continuous Codes With Bounded Aspect Ratio}\label{sec:chainLength}

In this section, we show how to sparsify \emph{weighted} continuous codes \emph{with bounded aspect ratio} to their \emph{real-valued chain length}. As before, we will focus on codes $C \subseteq \left ( \{0\} \cup [1,k] \right )^m$. Recall that in this situation, the real-valued chain length is defined as:

\begin{definition}\label{def:continuousRVCL}
    Let $C \subseteq \left ( \{0\} \cup [1,k] \right )^m$. The real-valued chain-length of $C$ with parameter $\eps$, denoted by $\mathrm{BACCL}(C, \eps)$, is the maximum integer $\ell$ such that there exists $A_1, \dots A_p \subseteq [m]$ and $C_1, \dots C_p \subseteq C$, with $\sum_{i = 1}^p |A_i| = \ell$, and:
    \begin{enumerate}
        \item For every $i \in [p]$, if $|A_i| = 1$, then $(C_i)|_{A_i} \neq \{ 0\}$.
        \item  For every $i \in [p]$ if $|A_i| \geq 2$, then there is a choice of vector $\gamma \in [1,k]^{A_i}$, such that for every $B \subseteq A_i$, there is a codeword $c \in C_i$ such that for $b \in B$, $v_b \geq \gamma_b + \eps$, and for $b \in A_i - B$, $v_b \leq \gamma_b - \eps$.
        \item For every $j > i$, $(C_i)|_{A_j} = 0^{A_j}$.
    \end{enumerate}
\end{definition}

For this definition, we have the following characterization:

\begin{theorem}\label{thm:RVCLcontinuous}
    Let $C \subseteq \left ( \{0\} \cup [1,k] \right )^m$ with $k = O(1)$ and let $\eps > 0$. Then,
    \begin{enumerate}
        \item There exists $\eps'$ such that $\eps' = \Omega(\eps)$ and $\WS(C, \eps') = \Omega(\eps \cdot \mathrm{BACCL}(C, \eps))$. \label{item:continuousRVCLLB}
        \item For $\eps' = \eps / 256\log(m)$, $\WS(C, \eps) = \widetilde{O}(\mathrm{BACCL}(C, \eps') / \eps^2)$.\label{item:continuousRVCLUB}
    \end{enumerate}
\end{theorem}

\subsection{Preliminaries}

$\mathrm{BACCL}$ satisfies many convenient properties, similar to $\CL$. 

Importantly, we require the notion of contraction:

\begin{definition}
    For a code $C \subseteq \R^m$ and a set $S \subseteq [m]$, 
    \[
    \mathrm{Contract}(C, S) = \{ c \in C: c|_S = 0 \}.
    \]
\end{definition}

To start, we will take advantage of the following claim which governs how $\mathrm{BACCL}$ compounds in contracted instances(note that here we use $\mathrm{BACCL}(C)$ to mean $\mathrm{BACCL}(C, \eps)$, where here $\eps$ can be arbitrarily specified):

\begin{claim}\label{clm:contractRVCL}
    Let $C \subseteq \R^m$, and let $S \subseteq [m]$. Then,
    \[
    \mathrm{BACCL}(C) \geq \mathrm{BACCL}(\Contract(C, S)) + \mathrm{BACCL}(C|_S).
    \]
\end{claim}

\begin{proof}
Let $D = \Contract(C, S)$, let $A_1, \dots A_{\ell}$, $(C|_S)_1, \dots (C|_S)_{\ell}$ denote the witnesses of the $\mathrm{BACCL}$ of $C|_S$, and let $A'_1, \dots A'_{\ell'}$, $D_1, \dots D_{\ell'}$ denote the witnesses of the $\mathrm{BACCL}$ of $D$.

Now, importantly, we have the property that for any codeword $c \in D = \Contract(C, S)$, it must be the case that $c |_S = 0$, as if $c |_S \neq 0$, then $c$ would have been removed during the contraction operation on $S$. Thus, we can create a new witness for the chain length of $C$ by considering the sequence of sets $\{A'_1, \dots A'_{\ell'},  A_1, \dots A_{\ell}\}$, and $\{D_1, \dots D_{\ell'}, (C|_S)_1, \dots (C|_S)_{\ell} \}$, which we re-write as $\{B_1, \dots B_{\ell + \ell'} \}$, $\{C_1, \dots C_{\ell'}\}$. To see why this satisfies the definition of $\mathrm{BACCL}$, observe that the complete diagonal requirement is satisfied trivially, because these sets were witnesses to $\mathrm{BACCL}$ in $\Contract(C, S)$ and $C|_S$ respectively. Thus, all that remains is to show that they satisfy the upper-triangularity condition. 

The upper-triangularity condition trivially holds for pairs $(i,j) \in \binom{[\ell + \ell']}{2}$ where $i, j \leq \ell'$, and pairs where $i, j \geq \ell'+1$ because these sets witnessed $\mathrm{BACCL}$ in $\Contract(C, S)$ and $C|_S$ respectively. Because of our observation in the preceding paragraph regarding the structure of contracted codes, we then also get that the upper-triangularity condition holds when $i \leq \ell', j > \ell'$. 

In total then, we see that \[
\mathrm{BACCL}(C) \geq \sum_{i \in [\ell]} |A_i| + \sum_{i \in [\ell']} |A'_{\ell}| \geq \mathrm{BACCL}(C|_S) + \mathrm{BACCL}(\Contract(C, S)),\] as we desire.
\end{proof}

We can also see the following relationship between $\mathrm{BACCL}$ and $\CL$ (on appropriately mapped instances, see \cref{def:hatCodes}). 

\begin{remark}\label{rmk:RVCLvsCL}
    Immediately, we can observe that $\mathrm{BACCL}(C) \geq \CL(\hat{C})$, as any witness to a $\CL$ lower bound in $\hat{C}$ is also a valid witness to $\mathrm{BACCL}$. Likewise, $\mathrm{BACNRD}(C) \leq \mathrm{BACCL}(C)$, as any witness to $\mathrm{BACNRD}$ is automatically a witness to $\mathrm{BACCL}$.
\end{remark}

We now proceed to prove the aforementioned lower bound. 

\subsection{Lower Bound}

As in the proof of the lower bound in the non-redundancy case case, the key turns out to be the ability to argue that each complete sub-code needs most of its support retained.

\begin{claim}\label{clm:diffWeightGeneral}
    Let $C_i \subseteq (\{0 \} \cup [1,k])^{A_i}$ be an unweighted code such that there is a choice of vector $\gamma \in [1,k]^{A_i}$, such that for every $B \subseteq A_i$, there is a codeword $c \in C_i$ such that for $b \in B$, $v_b \geq \gamma_b + \eps$, and for $b \in A_i - B$, $v_b \leq \gamma_b - \eps$. Let $\hat{S} \subseteq A_i$, $|\hat{S}| \leq \eps |A_i| / 100 k$ be an arbitrary subset of the rows. Then, for any assignment of weights $\widetilde{w}$ to $\hat{S}$, there exist codewords $v, v' \in C_i$ such that:
    \begin{enumerate}
        \item $\widetilde{\wt}(v|_{\hat{S}}) \geq (1 + \eps /k) \cdot \widetilde{\wt}(v'|_{\hat{S}})$
        \item $\wt(v') \geq \wt(v)$.
    \end{enumerate}
\end{claim}

\begin{proof}
    Let $\hat{S}$ be given. We consider the codeword $v$ defined with the set $B = \hat{S}$, and consider the codeword $v'$ defined with $B = A_i - \hat{S}$. Immediately, we can see that for every $b \in \hat{S}$, it is the case that 
    \[
    \wt(v_b) \geq \gamma_b + \eps,
    \]
    while \[
    \wt(v'_b) \leq \gamma_b - \eps,
    \]
    and thus it must be the case that $\widetilde{\wt}(v|_{\hat{S}}) \geq (1 + \eps /k) \cdot \widetilde{\wt}(v'|_{\hat{S}})$.

    At the same time, we know that for $b \in A_i - \hat{S}$, $v_b \leq \gamma_b - \eps$ and $v'_b \geq \gamma_b + \eps$.

    Thus, we claim that it is also the case that the original codewords $v, v'$ satisfy $\wt(v) \leq \wt(v')$. Indeed, to see why, we have that 
    \[
    \wt(v) \leq |\hat{S}\cap A_i| \cdot k + \sum_{i \in A_i - \hat{S}} \gamma_i - \eps
    \]
    and 
    \[
     \wt(v') \geq \sum_{i \in A_i - \hat{S}} \gamma_i + \eps,
    \]
where here we have used that the entire weight of $v, v'$ is contributed by coordinates in $A_i$ (thus we do not sum over coordinates from other parts). 
Thus, 
    \[
    \wt(v') - \wt(v) \geq \sum_{i \in A_i - \hat{S}} (\gamma_i + \eps) - |\hat{S}\cap A_i| \cdot k - \sum_{i \in A_i - \hat{S}} (\gamma_i - \eps) \geq 2 \eps \cdot |A_i - \hat{S}| - |\hat{S}\cap A_i| \cdot k
    \]
    \[
    \geq 2\eps (1 - \eps/100k) \cdot |A_i| - k \cdot \eps |A_i|/100k \geq \eps |A_i| - \eps |A_i|/100 \geq 0.
    \]
\end{proof}

\begin{proof}[Proof of \cref{thm:RVCLcontinuous}, \cref{item:continuousRVCLLB}.]
    To start, we let $\eps' = \eps / 8k$, and let $\gamma = \left ( 1000 m \cdot k / \eps\right )^4$. Observe that $\eps' = \Omega(\eps)$.

    Now, let the code $C \subseteq (\{0\} \cup [1,k])^m$ be given, and let the sets $A_1, \dots A_p \subseteq [m]$ and $C_1, \dots C_p \subseteq C$ be the subsets of coordinates and codewords which witness the real-valued chain length (with parameter $\eps$).

    Now, we create the following weighting scheme: for $j \in [p]$, and for $i \in A_j$, we assign $w_i = \gamma^{j}$. Under this weighting scheme, observe that for any $c \in C_j$, $\wt(c) \leq m \cdot k \cdot \gamma^j$, as the contribution to its weight is dominated by the (at most $m$) coordinates of weight $\gamma^j$.

    Importantly, this also means that in any sparsifier $\hat{S} \subseteq [m]$, along with weights $\hat{w}_i: i \in \hat{S}$, if $i \in A_j$, then $\hat{w}_i \leq 4m \cdot k \cdot \gamma^j$. To see why, if we suppose for the sake of contradiction that $\hat{w}_i > 4m \cdot k \cdot \gamma^j$, then this means there is some codeword $c$ in $C_j$ such that $c_i \neq 0$, and therefore the sparsifier reports its weight as
    \[
    \widehat{\wt}((c)|_{\hat{S}}) \geq 4m \cdot k \cdot \gamma^j \geq 2 \cdot \wt(c).
    \]
    Thus, the resulting set of coordinates is not a $(1 \pm \eps)$ sparsifier of the original code (for $\eps < 1$). 

    So, we know that in the resulting sparsifier, if $i \in A_j$, then $\hat{w}_i \leq 4m \cdot k \cdot \gamma^j$. Next, let us focus on a single set of rows $A_j$ and the corresponding set of codewords $C_j$. We let the sparsifier retain rows $\hat{S}$, with weights $\hat{w}_i: i \in \hat{S}$, and we let $\hat{S_j} = \hat{S} \cap A_j$ denote the subset of coordinates in $A_j$ that are retained. We suppose for the sake of contradiction that $|\hat{S}_j| \leq \eps |A_j|/100 k$. There are two cases:
    \begin{enumerate}
        \item If $|A_j| = 1$, then this means $\hat{S}_j = \emptyset$. Then, let $c$ be any codeword in $C_j$. In the original code $C$, we know that the weight of $c$ is \[
        \wt(c) \geq \gamma^j.
        \]
        However, because no coordinates from $A_j$ were retained, only coordinates from $A_{< j}$ can contribute to the sparsifier's weight for $c$. Here, we know that 
        \[
        \widehat{\wt}(c|_{\hat{S}}) \leq \sum_{j' < j} \sum_{i \in A_{j'}} c_i \cdot 4m \cdot k \cdot \gamma^{j'} \leq 4m ^2 k^2 \cdot \gamma^{j-1} \leq \frac{4}{1000^4}\cdot \gamma^j,
        \]
        and is thus not a $(1\pm \eps)$ sparsifier, as the weight of $c$ is not being preserved. Thus, $A_j$ must be kept in the sparsifier. 
        \item If $|A_j| \geq 2$, and $|\hat{S}_j| \leq \eps |A_j / 100k$ then we can now invoke \cref{clm:diffWeightGeneral} to conclude that there exist codewords $(c_1)|_{A_j}, (c_2)|_{A_j}$ such that $\wt((c_2)|_{A_j}) \geq \wt((c_1)|_{A_j}) \geq |A_j| \cdot \gamma^{j}$, but at the same time $\widehat{\wt}((c_1)|_{\hat{S}_j}) \geq (1 + \eps/k) \cdot \widehat{\wt}((c_2)|_{\hat{S}_j})$.
        
        Importantly, because these codewords are in $C_j$, we also know that $(c_1)_{A_{>j}} = 0$, and $(c_2)_{A_{>j}} = 0$, and thus the contribution to the weight of $c_1$ and $c_2$ is coming purely from $A_{\leq j}$. So, we can see that 
        \[
        \wt \left ( c_1 \right ) = \wt \left ( (c_1)|_{A_j} \right ) + \wt \left ( (c_1)|_{<A_j} \right ) \leq \wt \left ( (c_1)|_{A_j} \right ) + m k \cdot \gamma^{j-1}
        \]
        \[
        \leq \wt \left ( (c_1)|_{A_j} \right ) \cdot (1 +  \frac{mk}{\gamma}).
        \]
        At the same time, we have that 
        \[
        \wt(c_2) \geq \wt \left ( (c_2)|_{A_j} \right ).
        \]

        Now, we can also bound the weights that are reported by the sparsifier: To start, we trivially have that 
        \[
        \widehat{\wt}((c_1)|_{\hat{S}}) \geq \widehat{\wt}((c_1)|_{\hat{S}_j}).
        \]
        It remains only to upper bound $\widehat{\wt}((c_2)$. For this, recall that as stated above, we know that for $i \in A_j$, $\hat{w}_i \leq 4m \cdot k \cdot \gamma^j$. Thus, we can see that
        \[
        \widehat{\wt}(c_2|_{\hat{S}}) = \sum_{i \in \hat{S}_j} \hat{w}_i \cdot (c_2)_i + \sum_{i \in \hat{S}_{<j}} c_i \cdot \hat{w}_i \leq \widehat{\wt}((c_2)|_{\hat{S}_j}) + m k \cdot 4m \cdot k \cdot \gamma^{j-1}.
        \]
        Importantly, because $\hat{S}, \hat{w}$ is a $(1 \pm \eps)$ sparsifier of $C, w$, we know that $\hat{w}(c_2|_{\hat{S}}) \geq \frac{\wt(c_2)}{2} \geq \gamma^j/2$, and so \[
        \widehat{\wt}((c_2)|_{\hat{S}_j}) + m k \cdot 4m \cdot k \cdot \gamma^{j-1} \leq \widehat{\wt}((c_2)|_{\hat{S}_j}) \cdot (1 + \frac{8m^2k^2}{\gamma}).
        \]

        All together then, we see that 
        \begin{align}
        \frac{\wt(c_1)}{\wt(c_2)} \leq \frac{\wt(c_1|_{A_j}) \cdot (1 + mk / \gamma)}{\wt(c_2|_{A_j})} \leq (1 + mk / \gamma) \leq (1 + \eps / 100k),\label{eq:boundOrigWeightGeneral}
        \end{align}
        but that 
        \begin{align}
        \frac{\widehat{\wt}((c_1)|_{\hat{S}})}{\widehat{\wt}((c_2)|_{\hat{S}})} \geq \frac{\widehat{\wt}((c_1)|_{\hat{S_j}})}{\widehat{\wt}((c_2)|_{\hat{S_j}}) \cdot (1 + 8m^2 k^2 / \gamma)} \geq \frac{1 + \eps/k}{1 + 8m^2 k^2 / \gamma} \geq 1 + 99\eps / 100k.
        \label{eq:boundHatRatioGeneral}
        \end{align}

        Now, observe that if $\hat{S}, \hat{w}$ were a $(1 \pm \eps') = (1 \pm \eps/8k)$-sparsifier of $C, w$, then $\frac{\wt(c_2)}{\widehat{\wt}(c_2|_{\hat{S}})} \leq \frac{1}{1-\eps/8k}$ and $\frac{\widehat{\wt}(c_1|_{\hat{S}})}{\wt(c_1)} \leq \frac{1+\eps/8k}{1}$. Together, this means that 
        \[
        \frac{\wt(c_2) \cdot \widehat{\wt}(c_1|_{\hat{S}})}{\widehat{\wt}(c_2|_{\hat{S}}) \cdot \wt(c_1)} \leq \frac{1+\eps/8k}{1-\eps/8k} \leq 1 + \eps/3k.
        \]
        However, by dividing \cref{eq:boundHatRatioGeneral} by \cref{eq:boundOrigWeightGeneral}, we can also see that 
        \[
        \frac{\wt(c_2) \cdot \widehat{\wt}(c_1|_{\hat{S}})}{\widehat{\wt}(c_2|_{\hat{S}}) \cdot \wt(c_1)} \geq \frac{1 + 99 \eps/100k}{1 + \eps/100k} > 1 + \eps/2k > 1 + \eps/3k,
        \]
        hence yielding a contradiction. Thus, it must be the case that $|\hat{S}_j| \geq \eps |A_j| / 100k$.
    \end{enumerate}

    We thus see that, for every $j \in [k]$, $|\hat{S}_j| \geq \eps |A_j| / 100k$. Thus, any $(1 \pm \eps')$ sparsifier must preserve $\sum_{j \in [k]} |\hat{S}_j| \geq \sum_{j \in [k]} \eps |A_j| / 100k = \Omega(\eps \cdot \mathrm{BACCL}(C, \eps))$ rows, thereby yielding the lower bound. 
\end{proof}

\subsection{Upper Bound}

Now, we proceed to a proof of \cref{thm:RVCLcontinuous}, \cref{item:continuousRVCLUB}. To start, we first show that for \emph{unweighted} codes, the real-valued chain length serves as an upper bound on the sparsifiability. 

\subsubsection{Sparsifying Unweighted Codes to $\mathrm{BACCL}$}

Here, we prove the following lemma:

\begin{lemma}\label{lem:sparsifyRVCLUnweighted}
    Let $C \subseteq (\{0\} \cup [1,k])^m$ be an unweighted code. Then, for any $\eps > 0$, $C$ admits a $(1 \pm \eps)$ sparsifier of size $\widetilde{O}(\mathrm{BACCL}(C, \eps/256\log(m) / \eps^2)$.
\end{lemma}

\begin{proof}
    We simply use \cref{alg:sparsifyRVNRDContinuous}. From before, we know that this algorithm returns a $(1 \pm \eps)$ sparsifier with high probability. It remains only to bound the size $|S|$ of the resulting sparsifier. For this, By \cref{clm:RVNRDsparsifierSizeBasicContinuous}, we already know that 
    \[
    |S| \leq \sum_{j = 1}^{\ell} \sum_{d \in \{1, 2, 4, \dots, m\}} |S^{(j)}_d| + |I^{(j)}_d| \leq 4 \log^2(m) \cdot \max_{j \in [\log_{3/2}(m)]} \max_{d \in \{1, 2, 4, \dots m\}} \max  \left ( |S^{(j)}_d|, |I^{(j)}_d| \right ),
    \]
    where $S^{(j)}_d, I^{(j)}_d$ are the sets constructed by \cref{alg:sparsifyRVNRDContinuous}. Now, we bound the maximum size of any such set $S^{(j)}_d, I^{(j)}_d$. Immediately, when $d \leq 10000k^2 \log^3(m) / \eps'^2$ (recall that here, $\eps' = \eps/16\log(m)$), we know that we can remove \emph{all} codewords of weight $[d/2, d]$ removing only $\widetilde{O}(\NRD(\hat{C}) // \eps^2)$ many rows (by \cref{clm:boundSupportSizeSmalld}). Then, using the relationship between $\NRD$ and $\CL$ (\cref{rmk:RVCLvsCL}), we can also see that $|I^{(j)}_d| = \widetilde{O}(\mathrm{BACNRD}(\hat{C}, \eps'/16) / \eps^2) = \widetilde{O}(\mathrm{BACCL}(\hat{C}, \eps'/16) / \eps^2)$. Note that for $d \leq 10000k^2 \log^3(m) / \eps'^2$, it is the case that $S^{(j)}_d = \emptyset$, and so these sets are trivially bounded in size. 

    All that remains then is to bound the size of $S^{(j)}_d, I^{(j)}_d$ when $d > 10000k^2 \log^3(m) / \eps'^2$. First, recall that $I^{(j)}_d$ is defined using \cref{thm:NRDdecomposition} with parameter $\lambda = 10000k^2 \log^4(m) / \eps'^2$. Thus, we know that 
    \[
    |I^{(j)}_d| \leq 2 \cdot (10000k^2 \log^4(m) / \eps'^2) \cdot \NRD(\hat{C}) \cdot \log(4m) = \widetilde{O}(\mathrm{BACCL}(C, \eps'/16) / \eps'^2),
    \]
    where we have again used the relationship between $\NRD$ and $\CL$ (\cref{rmk:RVCLvsCL}). Finally, it remains only to bound the size of $S^{(j)}_d$. Importantly, recall that Let $S_d^{(j)}$ is the smallest set of coordinates such that
$\left | \mathrm{Cover} \left ((C_{[d/2, d]})|_{\overline{S^{(j)}_d} \cap \overline{I^{(j)}_d}}, \eps'/4 \right )\right | \leq 2^{\eps'^2 d / 10000 k^2\log(m)}$. Thus, if we let $L_d = |S^{(j)}_d|$, we can invoke \cref{lem:constructUpperTriangular} using the code $(C_{[d/2, d]})|_{\overline{I^{(j)}_d}}$ (along with the same choices of $d, \eps'$, and $k$), using the key fact that for any set of size smaller than $L_d$, the size of any cover must be sufficiently large. In particular, this guarantees that one can find sets of coordinates $A_1, \dots A_p$, and sets of codewords $C_1, \dots C_p$ such that:
	\begin{enumerate}
		\item For every $i \in [p]$, there exists a $\gamma \in [1,k]^{A_i}$ such that, for every $B \subseteq A_i$ there exists a codeword $c \in C_i$ such that, for $b \in B$, $c_b \geq \gamma_b + \eps'/16$ and for $b \in A_i - B$, $c_b \leq \gamma_b - \eps'/16$.
		\item For $i < j \in [p]$, $C_i |_{A_j} = 0$.
		\item Each $|A_i| = \Omega(\eps'^2 d)$ and $\sum_{i = 1}^p |A_i| = \widetilde{\Omega}(\eps'^2 \cdot L_d)$.
		\item For $i \in [p]$, there exists $\hat{c} \in \hat{C}$ such that $C_i \subseteq  \Class(\hat{c})$.
	\end{enumerate}
The above coordinates and subcodes exactly satisfy the conditions of \cref{def:continuousRVCL}, and therefore show that $\mathrm{BACCL}((C_{[d/2, d]})|_{\overline{I^{(j)}_d}}, \eps'/16) = \widetilde{\Omega}(\eps'^2 L_d) = \widetilde{\Omega}(\eps'^2 |S_d^{(j)}|)$. Now, because $(C_{[d/2, d]})|_{\overline{I^{(j)}_d}}$ is a restriction of the rows and codewords of $C$, it must also therefore be the case (by \cref{rmk:inheritance}) that 
\[
\mathrm{BACCL}(C, \eps'/16) \geq \mathrm{BACCL}((C_{[d/2, d]})|_{\overline{I^{(j)}_d}}, \eps'/16), 
\]
and thus we also have 
\[
\mathrm{BACCL}(C, \eps'/16) = \widetilde{\Omega}(\eps'^2 |S_d^{(j)}|),
\]
which implies that 
\[
|S_d^{(j)}| = \widetilde{O}(\mathrm{BACCL}(C, \eps'/16) / \eps'^2).
\]

So, we have shown that all sets $I_d^{(j)}, S^{(j)}_d$ are bounded in size by $\widetilde{O}(\mathrm{BACCL}(C, \eps'/16) / \eps'^2)$. Thus, using our starting inequality, we know that our sparsifier size is bounded by 
\[
|S| \leq 4 \log^2(m) \cdot \widetilde{O}(\mathrm{BACCL}(C, \eps'/16) / \eps'^2) = \widetilde{O}(\mathrm{BACCL}(C, \eps/256\log(m)) / \eps^2),
\]
as we desire.
\end{proof}

Now, we show how to use this result as a subroutine to construct sparsifiers for arbitrary \emph{weighted} codes. 

\subsubsection{Sparsifying Codes With Polynomially-Bounded Weights}\label{sec:polynomialWeightReduction}

In the preceding section, we showed how to construct a sparsifier on \emph{unweighted} instances of size $\widetilde{O}(\mathrm{BACCL}(C, \eps / 256 \log(m)) / \eps^2)$. However, ultimately, our goal was to provide a characterization of \emph{weighted} sparsifiability. As a building block in this direction, in this section, we focus on sparsifying codes whose weights are polynomially bounded. We present the main lemma below:

\begin{lemma}\label{lem:polyWeightedSparsification}
Let $C \subseteq (\{0\} \cup [1,k])^m$ be a code such that every coordinate has weight $\in [W_{\min}, W_{\max}]$. Then, for any $\eps > 0$, $C$ admits a $(1 \pm \eps)$ sparsifier of size $\widetilde{O}(\mathrm{BACCL}(C, \eps/256\log(m)) \cdot \mathrm{polylog}(m \cdot W_{\max} / (W_{\min}\eps)) / \eps^2)$.
\end{lemma}

\begin{proof}
    Our strategy is simple: we will just duplicate coordinates a number of times proportional to their weight, and then sparsify this now unweighted code. WLOG, we assume that $W_{\min} = 1$, as otherwise, we can simply divide all weights by $W_{\min}$, and then at the end multiply all weights by this factor of $W_{\min}$.

    Our key claim is the following:
    \begin{claim}\label{clm:duplicateCoordsApprox}
        Let $C \subseteq (\{0\} \cup [1,k])^m$ be a weighted code with all weights $\geq 1$, let $\eps \in (0,1)$, and let $C_{\mathrm{unweighted}}$ be the result of duplicating each coordinate $i$ $\lfloor10  w_i /  \eps \rfloor$ times. Then, for any $(1 \pm \eps/10)$ sparsifier of $C_{\mathrm{unweighted}}$ (denoted $\widetilde{C}_{\mathrm{unweighted}}$), $\frac{\eps}{10} \cdot \widetilde{C}_{\mathrm{unweighted}}$ is a $(1 \pm \eps/10)^2$ sparsifier of $C$.
    \end{claim}

    \begin{proof}
        Consider any codeword $c \in C$. We use $\wt(c)$ to denote the weight of this codeword in $C$, and we use $\wt_{\mathrm{unweighted}}(c)$ to denote the corresponding weight of $c$ in $C_{\mathrm{unweighted}}$. Now, let us fix a single coordinate $i \in [m]$. We will compare the contribution of the weight from $c_i$ in $C$ to the contribution of the weight from $c_i$ in $C_{\mathrm{unweighted}}$. 

        Indeed, in $C$, the contribution from coordinate $i$ is exactly $c_i \cdot w_i$. In $C_{\mathrm{unweighted}}$, the contribution from coordinate $i$ becomes $c_i \cdot \lfloor10  w_i / \eps \rfloor$. In particular, we can see that 
        \[
        (\eps/10) \cdot c_i \cdot \lfloor10  w_i / \eps \rfloor \geq (\eps/10) \cdot c_i \cdot (10  w_i / \eps -1) \geq w_i c_i - \eps c_i / 10.
        \]
        At the same time,
        \[
        (\eps/10) \cdot c_i \cdot \lfloor10  w_i / \eps \rfloor \leq w_i \cdot c_i,
        \]
        and so all together, we have that 
        \[
        (\eps/10) \cdot c_i \cdot \lfloor10  w_i / \eps \rfloor \in (1 \pm \eps/10) \cdot w_i c_i,
        \]
        where we are using that $w_i \geq 1$.

        In particular, this means that for every coordinate, the contribution from the copies of the coordinate in $\frac{\eps}{10} \cdot \widetilde{C}_{\mathrm{unweighted}}$ is within a $(1 \pm \eps/10)$ factor of the contribution in $C$, so summing across all coordinates means that $\frac{\eps}{10} \cdot \widetilde{C}_{\mathrm{unweighted}}$ is a $(1 \pm \eps/10)$ sparsifier of $C$. The stated claim above then follows by composition of the sparsification bound. 
    \end{proof}

    To conclude the lemma, we simply invoke \cref{clm:duplicateCoordsApprox}. The resulting unweighted code has at most $m \cdot W_{\max} / (W_{\min} \eps)$ many coordinates, and the $\mathrm{BACCL}$ is unchanged, as we have simply duplicated rows. Thus, we can construct a $(1 \pm \eps/10)$ sparsifier of $C_{\mathrm{unweighted}}$ which preserves only $O(\mathrm{BACCL}(C) \cdot \polylog(m \cdot W_{\max} / W_{\min}\eps)/\eps^2)$ many coordinates, and is in turn a $(1 \pm \eps/10)^2 \in (1 \pm \eps)$ sparsifier of $C$ by using \cref{lem:sparsifyRVCLUnweighted}.
\end{proof}

\subsubsection{Sparsifying Codes With Unbounded Weights}

Now, we perform a weight class decomposition trick to break the code down into groups, where in each group, the ratios of coordinate weights are bounded. We then invoke \cref{lem:polyWeightedSparsification} to sparsify each one of these groups. We introduce this procedure formally below:

\begin{definition}\label{def:weightDecomposition}
    Let $C \subseteq (\{0\} \cup [1,k])^m$, let $w_i\geq 1: i \in [m]$ be a set of weights for the coordinates of $C$, and let $\eps$ be the desired accuracy for the resulting sparsifier. We let $\gamma = (1000 m k / \eps)^3$, and then define $C^{[\gamma^{a}, \gamma^{a+1}]}$ to be the code $C|_{T^{[\gamma^{a}, \gamma^{a+1}]}}$, where $T^{[\gamma^{a}, \gamma^{a+1}]} = \{ i \in [m]: w_i \in [\gamma^{a}, \gamma^{a+1}]$.

    We let the code $\Ceven = \bigcup_{a \in \Z^{\geq0}} C|_{T^{[\gamma^{2a}, \gamma^{2a+1}]}}$, and we let $\Codd = \bigcup_{a \in \Z^{\geq0}} C|_{T^{[\gamma^{2a+1}, \gamma^{2a+2}]}}$.
\end{definition}

Going forward, we focus without loss of generality on creating a sparsifier of $\Ceven$. A similar argument will allow us to create a sparsifier of $\Codd$, from which we can then take the union of the two sparsifiers. 

Now, we require the notion of the \emph{largest non-empty weight class}:

\begin{definition}\label{def:contractedWeightClasses}
    Given the code $\Ceven$ as defined above, we let $C^{\sup, 1}$ denote the \emph{largest} non-empty weight class in $\Ceven$. I.e., $C^{\sup, 1} = C|_{T^{[\gamma^{2a, 2a+1}]}} $ for the maximum $a \in \Z^{\geq 0}$ such that $C|_{T^{[\gamma^{2a, 2a+1}]}}$ is non-empty. We let $T^{(1)} = \Supp(C^{\sup, 1})$, and we let $C^{(\mathrm{even, 2)}} = \Contract(\Ceven, T^{(1)})$. We then define $C^{\sup, 2}$ to be the largest non-empty weight class in $C^{(\mathrm{even, 2)}}$, $T^{(2)} = \Supp(C^{\sup, 2})$, $C^{(\mathrm{even, 3)}} = \Contract(\Ceven, T^{(1)} \cup T^{(2)})$, and so on. 
\end{definition}

Note that the above procedure will terminate, as in each round, we contract on some non-trivial subset of the coordinates. We let $t$ denote the number of iterations in the above procedure.

The key structural lemma we will use is the following:

\begin{lemma}\label{lem:unionSparsifiersCorrect}
    Let $C^{(\mathrm{even}, 1)}, \dots C^{(\mathrm{even}, t)}$ be defined as in \cref{def:contractedWeightClasses}, and let $\widetilde{C^{(\mathrm{even}, 1)}}, \widetilde{C^{(\mathrm{even}, t)}}$ be $(1 \pm \eps/10)$ sparsifiers of $C^{(\mathrm{even}, 1)}, \dots C^{(\mathrm{even}, t)}$ respectively. Then, $\bigcup_{j = 1}^t \widetilde{C^{(\mathrm{even}, j)}}$ is a $(1 \pm \eps)$ sparsifier of $\Ceven$. 
\end{lemma}

\begin{proof}
    Consider any codeword $c \in \Ceven$, and let $j^* \in [t]$ denote the first round where $c|_{T(j^*)} \neq 0$. Importantly, because $j^*$ is the first round where this support is non-zero, it must be the case that $c|_{T^{(< j^*)}} = 0$. Additionally, because in our contraction procedure, we only delete codewords which are \emph{non-zero}, this means that $c$ has not been removed from the code by round $j^*$, i.e., $c \in C^{(\sup, j^*)}$. Now, by the lemma statement, we know that $\widetilde{C^{(\mathrm{even}, j^*)}}$ is a $(1 \pm \eps/10)$ sparsifier of $C^{(\sup, j^*)}|_{T^{(j^*)}}$, which in particular means that $\wt(c|_{T^{(j^*)}})$ is preserved to a $(1 \pm \eps/10)$ factor. 

    The remaining coordinates in the support of $c$ must be in \emph{lower level} weight classes. The intuition is that these weight classes can have only a negligible impact on the weight assigned to $c$ by the sparsifier, and indeed this will be the case.

    First, we can observe that 
    \begin{align}
    \frac{\wt(c|_{T^{(j^*)}})}{\wt(c)} = \frac{\wt(c|_{T^{(j^*)}})}{\wt(c|_{T^{(j^*)}}) + \wt(c|_{T^{(>j^*)}})} \geq \frac{(1000mk/\eps)^{3\ell^*}}{(1000mk/\eps)^{3\ell^*} + km \cdot (100m/\eps)^{3\ell^*-3}} \label{eq:boundOnOrig}
\end{align}
    \[
    \geq \frac{1}{1 + \frac{km}{(1000km/\eps)^3}} \geq 1 - \eps/100.
    \]
    Here, we are using $\ell^*$ to denote the minimum weight class of $T^{(j^*)}$, i.e., these coordinates have weights in the range $[\gamma^{\ell^*, \ell^*+1}]$. Importantly, for all coordinates in $T^{(>j^*)}$, all coordinates have weight $\leq \gamma^{\ell^*-1}$.
    
    Now, because $\widetilde{C^{(\mathrm{even}, j^*)}}$ is a $(1 \pm \eps/10)$ sparsifier of $C^{(\sup, j^*)}|_{T^{(j^*)}}$, this means that
    \[
    \widetilde{\wt}(c) \geq \widetilde{\wt}(c|_{T^{(j^*)}})) \geq (1 - \eps/10) \cdot \wt(c|_{T^{(j^*)}}) \geq (1 - \eps/10) \cdot (1 - \eps/100) \wt(c) \geq (1 - \eps) \wt(c).
    \]

    So, it remains only to show the upper bound on the sparsifier weight. Here, we start by remarking, for every $j \in [t]$, because $\widetilde{C^{(\mathrm{even}, j)}}$ is a $(1 \pm \eps/10)$ sparsifier of $C^{(\sup, j)}|_{T^{(j)}}$, if we let $\ell$ denote the weight class of $C^{(\mathrm{even}, j)}$ (i.e., the weights are of size $\gamma^{\ell, \ell+1}$), then it must be the case that the weights assigned by $\widetilde{C^{(\mathrm{even}, j)}}$ are $\leq 8 mk^2 \cdot \gamma^{\ell+1}$. This is because if we assume for the sake of contradiction that some $w_i: i \in T^{(j)}$ is given such a large weight, then there is some codeword $y \in C^{(\mathrm{even}, j)}$, for which $y_i \geq 1$, whose weight is now reported to be $\widetilde{\wt}(y_i) \geq 8 mk^2 \cdot \gamma^{\ell+1}$. But, because $y \in C^{(\mathrm{even}, j)}$, $\Supp(y)$ is only defined on coordinates of weight $\leq \gamma^{\ell+1}$, and thus $\wt(y) \leq mk \cdot \gamma^{\ell+1}$. Thus, any sparsifier assigning a weight $w_i \geq 8 mk^2 \cdot \gamma^{\ell+1}$ is not a $(1 \pm \eps/10)$ sparsifier for the codeword $y$.

    With this in hand, we can now see that 
    \begin{align}
    \frac{\widetilde{\wt}(c)}{\widetilde{\wt}(c|_{T^{(j^*)}})} = \frac{\widetilde{\wt}(c|_{T^{(j^*)}}) + \widetilde{\wt}(c|_{T^{(>j^*)}})}{\widetilde{\wt}(c|_{T^{(j^*)}})} \leq 1 + \frac{m \cdot 8m k^2\cdot \gamma^{\ell^*-1}}{\wt(c) /2} \label{eq:boundOnTilde}
    \end{align}
    \[
    \leq 1 + \frac{16m^2 \cdot k^2 \gamma^{\ell^*-1}}{\gamma^{\ell^*}} \leq 1 + \frac{16m^2k^2}{\gamma} \leq 1 + 2\eps/100.
    \]

    So, using \cref{eq:boundOnTilde}, the fact that $\widetilde{C^{(\mathrm{even}, j)}}$ is a $(1 \pm \eps/10)$ sparsifier of $C^{(\sup, j)}|_{T^{(j)}}$, and that $\wt(c|_{T^{(j^*)}})  \leq \wt(c)$, we see
    \[
    \widetilde{\wt}(c) \leq (1 + 2\eps/100) \cdot \widetilde{\wt}(c|_{T^{(j^*)}}) \leq (1 + \eps/100) \cdot (1 + \eps/10) \wt(c|_{T^{(j^*)}})  \leq (1 + \eps) \wt(c).
    \]
    This then yields the claim, as we have shown that $\widetilde{\wt}(c) \in (1 \pm \eps) \wt(c)$.
\end{proof}

Finally, we are then ready to prove our main theorem:

\begin{proof}[Proof of \cref{thm:RVCLcontinuous}, \cref{item:continuousRVCLUB}]
    We construct $\Ceven$ and $\Codd$, and perform the same sparsification procedure on both. We focus on $\Ceven$: we create $C^{(\mathrm{even}, 1)}, \dots C^{(\mathrm{even}, t)}$ as in \cref{def:contractedWeightClasses}. We then construct $(1 \pm \eps/10)$ sparsifiers $\widetilde{C^{(\mathrm{even}, 1)}}, \widetilde{C^{(\mathrm{even}, t)}}$ of $C^{(\mathrm{even}, 1)}, \dots C^{(\mathrm{even}, t)}$ respectively, using \cref{lem:polyWeightedSparsification}. Immediately then, we can invoke \cref{lem:unionSparsifiersCorrect} to see that $\bigcup_{j = 1}^t \widetilde{C^{(\mathrm{even}, j)}}$ is a $(1 \pm \eps)$ sparsifier of $\Ceven$. 

    All that remains is to show that the resulting sparsifier is not too large. For this, observe that each $C^{(\mathrm{even}, j)}$ has weights in the range $[\gamma^{\ell}, \gamma^{\ell+1}]$ for some parameter $\ell$, and thus the value $\frac{W_{\max}}{W_{\min}}  = \gamma = (1000mk / \eps)^3$. Thus, for each $C^{(\mathrm{even}, j)}$, we can create a $(1 \pm \eps/10)$ sparsifier of size $\leq O(\mathrm{BACCL}(C^{(\mathrm{even}, j)}, \eps/2560\log(m)) \cdot \polylog(m / \eps) / \eps^2)$, and so the resulting sparsifier size is 
    \[
    \sum_{j = 1}^t O(\mathrm{BACCL}(C^{(\mathrm{even}, j)}, \eps/2560\log(m)) \cdot \polylog(m / \eps) / \eps^2) \]
    \[
    = O \left ( \frac{\sum_{j = 1}^t \mathrm{BACCL}(C^{(\mathrm{even}, j)}, \eps/2560\log(m)) \cdot \polylog(m / \eps)}{\eps^2}\right ).
    \]
    To conclude, we use \cref{clm:contractRVCL} to see that $\mathrm{BACCL}(C, \eps/2560\log(m)) \geq \sum_{j = 1}^t \mathrm{BACCL}(C^{(\mathrm{even}, j), },\eps/2560\log(m))$, as these are exactly sub-matrices that are obtained via contraction of coordinates. Thus, our resulting $(1 \pm \eps)$ sparsifier preserves $\widetilde{O}(\mathrm{BACCL}(C, \eps/2560\log(m)) / \eps^2)$ many coordinates, as we desire.
\end{proof}

\end{document}